\documentclass[10pt,a4paper,reqno]{amsart}
\usepackage{amsthm,amsfonts,amssymb,amsmath,euscript,comment, enumerate,slashed,bm}
\usepackage{graphicx}
\usepackage{color}
\usepackage{bbm}
\usepackage{geometry}
\usepackage{amsfonts}
\usepackage{xcolor}
\usepackage{seqsplit}

\newtheorem{thm}{Theorem}[section]

\newtheorem{cor}[thm]{Corollary}

\newtheorem{df}[thm]{Definition}
\newtheorem{rk}[thm]{Remark}
\newtheorem{proposition}[thm]{Proposition}
\newtheorem{lemma}[thm]{Lemma}

\newtheorem{conjecture}[thm]{Conjecture}

\newcommand{\m}[1]{\mathbb{#1}}

\newcommand{\q}[1]{\mathcal{#1}}

\newcommand{\wht}[1]{\widetilde{#1}}

\newcommand{\ep}{\varepsilon}

\newcommand{\f}{\frac}
\newcommand{\rd}{\partial}
\newcommand{\nab}{\nabla}
\newcommand{\alp}{\alpha}
\newcommand{\bt}{\beta}
\newcommand{\bA}{{\bf A}}

\newcommand{\gi}{(g^{-1})}
\newcommand{\mfg}{\mathfrak g}

\newcommand{\ls}{\lesssim}
\newcommand{\de}{\delta}
\newcommand{\om}{\omega}
\def\i {\infty}
\newcommand{\ud}{\mathrm{d}}

\newcommand{\pfstep}[1]{\vspace{.5em} {\it \noindent #1.}}

\numberwithin{equation}{section}

\addtocontents{toc}{\setcounter{tocdepth}{1}}

\begin{document}

\title[Trilinear compensated compactness and Burnett's conjecture]{Trilinear compensated compactness and\\ Burnett's conjecture in general relativity}

\begin{abstract}
Consider a sequence of $C^4$ Lorentzian metrics $\{h_n\}_{n=1}^{+\infty}$ on a manifold $\mathcal M$ satisfying the Einstein vacuum equation $\mathrm{Ric}(h_n)=0$. Suppose there exists a smooth Lorentzian metric $h_0$ on $\mathcal M$ such that $h_n\to h_0$ uniformly on compact sets. Assume also that on any compact set $K\subset \mathcal M$, there is a decreasing sequence of positive numbers $\lambda_n \to 0$ such that
$$\|\rd^\alp (h_n - h_0)\|_{L^\i(K)} \ls \lambda_n^{1-|\alp|},\quad |\alp|\geq 4.$$
It is well-known that $h_0$, which represents a ``high-frequency limit'', is not necessarily a solution to the Einstein vacuum equation. Nevertheless, Burnett conjectured that $h_0$ must be isometric to a solution to the Einstein--massless Vlasov system. 

In this paper, we prove Burnett's conjecture assuming that $\{h_n\}_{n=1}^{+\infty}$ and $h_0$ in addition admit a $\mathbb U(1)$ symmetry and obey an elliptic gauge condition. The proof uses microlocal defect measures --- we identify an appropriately defined microlocal defect measure to be the Vlasov measure of the limit spacetime. In order to show that this measure indeed obeys the Vlasov equation, we need some special cancellations which rely on the precise structure of the Einstein equations. These cancellations are related to a new ``trilinear compensated compactness'' phenomenon for solutions to (semilinear) elliptic and (quasilinear) hyperbolic equations.
\end{abstract}

\author{C\'ecile Huneau}
\address{CMLS, Ecole Polytechnique, 91120 Palaiseau, France}
\email{cecile.huneau@polytechnique.edu}
\author{Jonathan Luk}
\address{Department of Mathematics, Stanford University, CA 94304, USA}
\email{jluk@stanford.edu}
	
	\maketitle

\tableofcontents

\section{Introduction}

It has been known in the context of classical general relativity that ``backreaction of high frequency gravitational waves mimics effective matter fields'' (see for instance \cite{Burnett, CBHF, GW1, GW2, Isaacson1, Isaacson2}). One way to describe this phenomenon mathematically (due to Burnett \cite{Burnett}) is to consider a sequence of (sufficiently regular) Lorentzian metrics $\{h_n\}_{n=1}^{+\infty}$ on a smooth manifold $\mathcal M$ satisfying the Einstein vacuum equations
\begin{equation}\label{eq:Vac}
\mathrm{Ric}(h_n) = 0
\end{equation}
such that (in some coordinate system) the metric components admit some limit $h_0$ where $h_n\to h_0$ uniformly on compact sets and $\rd h_n \to \rd h_0$ weakly. Assume moreover that for any compact set $K$, there is some sequence of positive numbers $\lambda_n\to 0$ such that the following holds on $K$:
\begin{equation}\label{eq:bounds.caricature}
|h_n-h_0|\ls \lambda_n,\quad |\rd h_n|\ls 1,\quad |\rd^k h_n|\ls \lambda_n^{-k+1} \mbox{ for $k=2,3,4$}.
\end{equation}
Due to the nonlinearity of the Einstein equations, the limit $h_0$ does not necessarily satisfy \eqref{eq:Vac}. Instead, in general it is possible for $h_0$ to satisfy
$$\mathrm{Ric}(h_0) - \f12 h_0 R(h_0) = T$$
(where $R$ is the scalar curvature) for some non-trivial stress-energy-momentum tensor $T$. This tensor $T$ that arise in the limit can be interpreted as an effective matter field.

A question arises as to what type of effective matter field can arise in such a limiting process. In this direction, Burnett made the following conjecture\footnote{{We remark that in the original \cite{Burnett}, \eqref{eq:bounds.caricature} is only required to hold up to $k=2$. We impose the slightly stronger assumption that \eqref{eq:bounds.caricature} holds up to $k=4$ in view of the result that we prove in this paper.}}:
\begin{conjecture}[Burnett \cite{Burnett}]\label{conj:Burnett}
Given $(\mathcal M,h_n)$ and $(\mathcal M, h_0)$ above, the limit $h_0$ is isometric to a solution to the Einstein--massless Vlasov system, i.e.~the effective stress-energy-momentum tensor corresponds to that of massless Vlasov matter.
\end{conjecture}

We refer the reader to Sections~\ref{sec:EMV}--\ref{sec:null.dust} below for definitions concerning the Einstein--massless Vlasov system. We remark that in Conjecture~\ref{conj:Burnett}, ``Einstein--massless Vlasov system'' has to be appropriately formulated to include \emph{measure-valued} Vlasov fields since there are known examples for which the limits are isometric to solutions to the Einstein--null dust system. For further background on the Einstein--Vlasov system, see for instance \cite{Andreasson,Ringstrom}.

Conjecture~\ref{conj:Burnett} can be interpreted as stating that the effective matter field must be propagating with the speed of light and that the matter propagating in different directions do not directly interact, but only interact through their effect on the geometry; see \cite{Burnett}.

Our main result is a proof of Conjecture~\ref{conj:Burnett} under two additional assumptions:
\begin{enumerate}
\item ($\mathbb U(1)$ symmetry.) The sequence $\{h_n\}_{n=1}^{+\infty}$ and the limit $h_0$ all admit a $\mathbb U(1)$ symmetry (without necessarily obeying a polarization condition).
\item (Elliptic gauge.) All the metrics can be put in an elliptic gauge and the bounds \eqref{eq:bounds.caricature} hold in this gauge.
\end{enumerate}

The following is our main theorem; see Theorem~\ref{thm:main} for a precise statement.
\begin{thm}\label{thm:intro}
Conjecture~\ref{conj:Burnett} is true under the above two additional assumptions.
\end{thm}

Theorem~\ref{thm:intro} implies a fortiori that the effective stress-energy-momentum tensor is traceless, obeys the \underline{dominant} energy condition (i.e.~for every future-directed causal vector $X$, the vector $-T^{\mu}{}_{\nu}X^\nu$ is also future-directed and causal), and is non-negative in the sense that $T(X,X)\geq 0$ pointwise for every vector field $X$ (not necessarily causal). In fact, we show that these statements continue to hold even if we relax the convergence assumption to be significantly weaker than \eqref{eq:bounds.caricature}. We give an informal statement here but refer the reader to Theorem~\ref{thm:prelim} for a precise statement.
\begin{thm}\label{thm:intro.prelim}
{Assume that $h_n$, $h_0$ all admit a $\mathbb U(1)$ symmetry and are put in an elliptic gauge.} Suppose \eqref{eq:bounds.caricature} is replaced by the conditions that $h_n\to h_0$ uniformly on compact sets and $\rd h_n \rightharpoonup \rd h_0$ weakly in $L^{p_0}_{\mathrm{loc}}$ for some $p_0 > \f 83$. 

Then the effective stress-energy-momentum tensor is traceless, obeys the dominant energy condition, and is non-negative.
\end{thm}

Theorem~\ref{thm:intro.prelim} can be compared with the following theorem of Green--Wald \cite{GW1}, which to our knowledge is so far the best result towards Conjecture~\ref{conj:Burnett}:
\begin{thm}[Green--Wald \cite{GW1}]\label{thm:GW}
Assume $\{h_n\}_{n=1}^{+\infty}$ and $h_0$ are such that \eqref{eq:Vac} and \eqref{eq:bounds.caricature} hold. Then the effective stress-energy-momentum tensor is traceless and obeys the {\underline{weak}} energy condition {(i.e.~$T(X,X)\geq 0$ pointwise for {every} timelike $X$)}.
\end{thm}
Note that while its conclusion is weaker than Theorem~\ref{thm:intro.prelim}, Theorem~\ref{thm:GW} is a general result which does \underline{not} require $\mathbb U(1)$ symmetry.

While our results are gauge-dependent, it should be mentioned that a large class of non-trivial examples have been constructed under our gauge conditions. In our previous paper \cite{HL}, we have constructed sequences of solutions of Einstein vacuum equation with polarized $\m U(1)$ symmetry, which can be put in an elliptic gauge, such that \eqref{eq:bounds.caricature} are satisfied and the limit is a solution to Einstein equations coupled to $N$ null dusts. See further discussions in Section~\ref{sec:reverse}.

We now briefly discuss the proof; for more details see Section~\ref{sec:method}. Under the $\m U(1)$ symmetry assumptions, the $(3+1)$-dimensional Einstein vacuum equations reduce to the $(2+1)$-dimensional Einstein--wave map system. The rough strategy is the following:
\begin{itemize}
\item The first step of the proof is to show that only the two scalar fields corresponding to the wave map part of the system are responsible for the failure of the limit to be vacuum. This can already be viewed as a form of compensated compactness.
\item To capture and describe the defect of convergence given by the scalar fields, we rely on microlocal defect measures (introduced by Tartar \cite{Tartar} and G\'erard \cite{Gerard}). It is well-known that microlocal defect measures arising from \emph{linear} wave equations satisfy a massless Vlasov equation\footnote{In \cite{Francfort, FrancfortMurat}, a transport equation is derived only when the coefficients of the linear wave equation are time-independent. The case of a general linear wave equation in fact follows in a similar manner, except for more complicated algebraic manipulations.} \cite{Francfort, FrancfortMurat, Tartar}. 
\item We show that in our setting, despite the \emph{quasilinear} nature of the problem, the microlocal defect measure corresponding to the wave map part of the system still satisfies the massless Vlasov equation.
\end{itemize}
The most difficult part of the argument is to justify the massless Vlasov equation for the microlocal defect measure. That this holds relies on some remarkable structures and cancellations of the system, which are related to what we call a \emph{trilinear} compensated compactness phenomenon.

The remainder of the introduction will be organized as follows: In \textbf{Section~\ref{sec:method}}, we explain the ideas of the proof. In \textbf{Section~\ref{sec:discussions}}, we discuss some related problems. In \textbf{Section~\ref{sec:outline}}, we outline the remainder of the paper.

\subsection{Ideas of the proof}\label{sec:method}

\subsubsection{Microlocal defect measures}\label{sec:method.MDM}

The microlocal defect measure (see Section~\ref{sec:PSIDOs} for further details) is a measure on the cosphere bundle which identifies the ``region in phase space'' for which strong convergence fails. One important property of microlocal defect measures, especially relevant for our problem, is that microlocal defect measures arising from (approximate) solutions to hyperbolic equations themselves satisfy some transport equations. 
	
Let $u_n$ be a sequence of functions $\Omega \to \m R$, where $\Omega\subset \m R^d$ is open, which converges \emph{weakly} in $L^2(\Omega)$ to a function $u$. In general, after passing to a subsequence, $|u_n|^2 - |u|^2$ converges to a non-zero measure. The failure of the convergence $|u_n|^2 \to |u|^2$ can arise from concentrations or oscillations. The microlocal defect measure is a tool which captures both the position and the frequency of this failure of strong convergence.

For instance, if $u_n=n^\frac{d}{2}\chi(n(x-x_0))$ (with $\chi \in C^\infty_c$) so that $|u_n|^2$ concentrates to a delta measure, then the corresponding microlocal defect measure is given by $\de_{x_0} \otimes \nu$, where $\de_{x_0}$ is the spatial delta measure and $\nu$ is a uniform measure on the cotangent space. On the other hand, suppose $u_n(x)=\chi(x)\cos\left(n(x\cdot\omega)\right)$ so that $u_n$ oscillates in a particular frequency $\omega$. Then the corresponding microlocal defect measure is $|\chi|^2 \ud x \otimes \de_{[\omega]}$, where $\de_{[\omega]}$ is the delta measure concentrated at the (equivalent class of the) frequency $\omega$. See \cite{Tartar} for further discussions.

An important fact is that \textbf{microlocal defect measures arising from solutions to \underline{linear} wave equations on $(\Omega,g)$ satisfy the massless Vlasov equation on $(\Omega,g)$.} Consider the special case where $\Omega = \mathbb R^{d+1}$ and $\rd_\alp\phi_n$ a sequence of functions such that $\rd\phi_n \rightharpoonup \rd\phi_0$ weakly in $L^2$. In this case, there exists a non-negative Radon measure $\ud \nu$ on $S^*\mathbb R^{d+1}$ --- which is the microlocal defect measure --- so that 
\begin{equation}\label{eq:intro.MDM.def}
\int_{\mathbb R^{d+1}} \rd_\alp(\phi_n - \phi_0)(A\rd_\bt(\phi_n - \phi_0)) \,\ud x \to \int_{S^*\mathbb R^{d+1}} \f{a(x,\xi)\xi_\alp \xi_\bt}{|\xi|^2}\,\ud\nu,
\end{equation}
If $\phi_n$ are approximate solutions to some wave equation, then $\ud\nu$ is a (measure-valued) solution to the massless Vlasov equation \eqref{eq:massless.def.intro} and \eqref{eq:transport.def.intro}. More precisely,
\begin{enumerate}
\item Suppose
\begin{equation}\label{eq:ex.weak}
\Box_g \phi_n = f_n,\quad \|f_n \|_{L^2(\Omega)}\ls 1.
\end{equation} 
Then $\ud \nu$ is supported on the zero mass shell in the sense that for {every} $f\in C_c(\q M)$,
\begin{equation}\label{eq:massless.def.intro}
\int_{S^*\m R^{2+1}} f(x) (g^{-1})^{\alp\bt}\xi_\alp\xi_\bt \,\f{\ud \nu}{|\xi|^2} =0.
\end{equation}
\item If, instead of \eqref{eq:ex.weak}, we also have the stronger assumption
\begin{equation}\label{eq:ex.strong}
\Box_g \phi_n = f_n,\quad \|f_n - f_0\|_{L^2(\Omega)}\to 0.
\end{equation}
Then for any $C^1$ function $\widetilde{a}:T^*\mathcal M\to \mathbb R$ which is homogeneous of degree $1$ in $\xi$, 
\begin{equation}\label{eq:transport.def.intro}
\int_{S^*\m R^{2+1}} ((g^{-1})^{\alp\bt}\xi_\bt \rd_{x^{\alp}}\widetilde{a} - \f 12 (\rd_\mu g^{-1})^{\alp\bt} \xi_\alp \xi_\bt \rd_{\xi_\mu}\widetilde{a}) \,\f{\ud\nu}{|\xi|^2} = 0.
\end{equation}
\end{enumerate}

\subsubsection{Standard (bilinear) compensated compactness}

We now explain how microlocal defect measures can be applied to the Burnett conjecture. Recall that Einstein equation with $\m U(1)$ symmetry reduced to a $2+1$ dimensional system (see Section~\ref{sec:u1})
\begin{equation}\label{sysn}
\left\{
\begin{array}{l}
\Box_g \psi + \f 12 e^{-4\psi} g^{-1}(\ud \om, \ud \om) = 0,\\
\Box_g \om - 4 g^{-1} (\ud \om, \ud \psi) = 0,\\
\mathrm{Ric}_{\alp\bt}(g)= 2\partial_\alp \psi \partial_\bt \psi + \f 12 e^{-4\psi} \rd_\alp \om \rd_\bt \om.
\end{array}
\right.
\end{equation}

{Assume that} we have a sequence of solutions $\{(\psi_n, \om_n, g_n)\}_{n=1}^{+\infty}$ which satisfy \eqref{sysn}, with $g_n$ in an elliptic gauge, {which} moreover attains $C^0$-limit $(\psi_0, \om_0, g_0)$ with {the following estimates:}
\begin{equation}\label{eq:est.intro}
\|\rd^k (\psi_n - \psi_0 ,\, \om_n - \om_0,\, g_n - g_0)\|_{L^2\cap L^\i(\mathbb R^2)} \ls \lambda_n^{1-k},\quad k=0,1,\dots, 4.
\end{equation}

 The first step is to show that 
\begin{equation}\label{eq:CC.wave}
\Box_{g_n} \psi_n\rightharpoonup \Box_{g_0} \psi_0, \quad \Box_{g_n} \om_n\rightharpoonup \Box_{g_0} \om_0;
\end{equation}
\begin{equation}\label{eq:CC.null}
g_n^{-1}(\ud \om_n, \ud \om_n)\rightharpoonup g_0^{-1}(\ud \om_0, \ud \om_0),\quad g_n^{-1}(\ud \om_n, \ud \psi_n)\rightharpoonup g_0^{-1}(\ud \om_0, \ud \psi_0);
\end{equation}
\begin{equation}\label{eq:CC.Ric}
\mathrm{Ric}_{\alp\bt}(g_n) \rightharpoonup \mathrm{Ric}_{\alp\bt}(g_0).
\end{equation}
in the sense of distributions.

That \eqref{eq:CC.wave} holds is due to the divergence structure of the terms. That \eqref{eq:CC.null} is true is slightly more subtle but well-known, and is related to the standard compensated compactness: $g_0^{-1}(\ud \om_n, \ud \om_n)$ and $g_0^{-1}(\ud \om_n, \ud \psi_n)$ are \emph{null forms}, so that when \eqref{eq:est.intro} holds and that $\Box_{g_0}\om_n$ and $\Box_{g_0}\psi_n$ are bounded uniformly in $L^2\cap L^\i$, the convergence \eqref{eq:CC.null} holds.

Finally, \eqref{eq:CC.Ric} holds under our elliptic gauge condition. This is because
\begin{itemize}
\item the elliptic gauge gives strong compactness for \emph{spatial} derivatives of the metric components;
\item in this gauge the nonlinear structure is such that there are no quadratic products of time derivatives of the metric components.
\end{itemize}

Given \eqref{eq:CC.wave}--\eqref{eq:CC.Ric}, it follows that to capture how much the limit $(\psi_0,\om_0,g_0)$ deviates from solving \eqref{sysn}, we just need to understand {(for every vector field $Y \in C^\infty_c$)} the $n \to +\infty$ limit of
\begin{equation}\label{eq:intro.defect}
\int_{\mathcal M} \{ 2 (Y \psi_n) (Y \psi_n) + \f 12 e^{-4\psi_n} (Y \om_n )(Y \om_n) \} \, \mathrm{dVol}_{g_n} - \int_{\mathcal M} \{ 2 (Y \psi_0) (Y \psi_0) + \f 12 e^{-4\psi_0} (Y \om_0 )(Y \om_0) \} \, \mathrm{dVol}_{g_0}.
\end{equation}

{The deviation} of \eqref{eq:intro.defect} from $0$ is in particular captured by the microlocal defect measure. More precisely, {we define} the non-negative Radon measure $\ud \nu$ {(cf.~\eqref{eq:intro.MDM.def})} by 
\begin{equation}\label{eq:intro.defect.1}
\begin{split}
&\: \lim_{n\to +\infty} \int_{\mathcal M} \{ 2 (\rd_\alp (\psi_n-\psi_0)) (A\rd_\bt (\psi_n-\psi_0)) + \f {e^{-4\psi_0}}2  (\rd_\alp (\om_n -\om_0) )(A\rd_\bt (\om_n -\om_0)) \} \, \mathrm{dVol}_{g_0} \\
= &\: \int_{S^*\mathcal M} \f{a(x,\xi)\,\xi_\alp \xi_\bt \,\ud \nu}{|\xi|^2},
\end{split}
\end{equation}
{after passing to a subsequence (which we do not relabel). Then}
\begin{equation}\label{eq:intro.defect.2}
\lim_{n\to +\infty} \mbox{\eqref{eq:intro.defect}} = \int_{S^*\mathcal M} \langle Y,\, \xi\rangle^2 \,\f{\ud \nu}{|\xi|^2}.
\end{equation}
In particular, the limit $(\psi_0, \om_0, g_0)$ obeys the following system:
\begin{equation}\label{sysn.2}
\left\{
\begin{array}{l}
\Box_{g_0} \psi_0 + \f 12 e^{-4\psi_0} g_0^{-1}(\ud \om_0, \ud \om_0) = 0,\\
\Box_{g_0} \om_0 - 4 g_0^{-1} (\ud \om_0, \ud \psi_0) = 0,\\
\int_{\mathcal M} \mathrm{Ric}(g_0)(Y,Y)\, \mathrm{dVol}_{g_0} = \int_{\mathcal M} \{ 2(Y \psi_0)^2 + \f 12 e^{-4\psi_0} (Y \om_0)^2 \} \, \mathrm{dVol}_{g_0} + \int_{S^*\mathcal M} \langle Y,\, \xi\rangle^2 \,\f{\ud \nu}{|\xi|^2},
\end{array}
\right.
\end{equation}
where the final equation in \eqref{sysn.2} is to be understood as holding for {every} vector field $Y\in C^\infty_c$. \eqref{sysn.2} is exactly the form of the Einstein--massless Vlasov system, as long as the measure $\ud \nu$ is indeed a measure-valued (weak) solution to the massless Vlasov equation.

The main task of the paper is therefore to justify that in our \emph{quasilinear} setting, $\ud \nu$ still solves the massless Vlasov equation, i.e.~(analogs of) \eqref{eq:massless.def.intro} and \eqref{eq:transport.def.intro} still hold. Already in Section~\ref{sec:method.MDM}, we saw that \eqref{eq:massless.def.intro} only require weaker assumptions (cf.~\eqref{eq:ex.weak} and \eqref{eq:ex.strong}) and is therefore relatively straightforward. However, as we discuss below, it is much harder to obtain the transport equation \eqref{eq:transport.def.intro}.

\subsubsection{Model problem}

As we argued above, the key difficulty is to justify the transport equation for the microlocal defect measure. Observe already that in \eqref{eq:ex.strong}, one needs that $f_n \to f_0$ in the $L^2$ \underline{norm} in order to justify the transport equation. However, in our setting, we only have weak convergence so that the derivation of the transport equation must rely on some special compensation. Another issue is that the wave operator $\Box_{g_n}$ is now dependent on $n$. It is relatively straightforward to show that if $g_n$ tends to its limit $g_0$ in $C^1$, then the transport equation remains valid. However, again because of weak convergence, we need compensation in the relevant terms.

To elucidate some of the difficulties and the techniques to tackle them, consider the following simplified \emph{semilinear} model problem with $n$-dependent metrics:
\begin{equation}\label{eq:intro.model.eqn}
\begin{cases}
\Box_{g_n}\phi_n = g_n^{-1}(\ud \phi_n,\ud \phi_n),\\
g_n = - N_n^2 (dt)^2 +  (dx^1)^2 + (dx^2)^2.
\end{cases}
\end{equation}
Assume also, for simplicity in this exposition, that 
$\phi_n \to 0$ and $N_n\to 1$ pointwise with the following bounds:
\begin{equation}\label{eq:intro.model.est}
\|\rd^k\phi_n\|_{L^2\cap L^\i(\mathbb R^{2+1})} + \|\rd^k (N_n-1)\|_{L^2\cap L^\i(\mathbb R^{2+1})} \ls \lambda_n^{1-k},
\end{equation}
and that the spatial derivatives of $N_n$ (denoted by $\nab$) obey stronger estimates:
\begin{equation}\label{eq:intro.N.spatial}
\|\nab N_n\|_{L^2\cap L^\i(\mathbb R^{2+1})} \ls \lambda_n^{\f 12}.
\end{equation}
(Note that the assumptions that $\phi_n\to 0$ and $N_n \to 1$ are slight over-simplifications. On the other hand, \eqref{eq:intro.N.spatial} is a reasonable assumption in view of the elliptic gauge. See Section~\ref{sec:further.issues}.)

{Define the microlocal defect measure $\ud \nu$ according to \eqref{eq:intro.MDM.def}.} Our goal will be to show that for any $\widetilde{a}(x,\xi)$ which is homogeneous of order $+1$ in $\xi$, 
\begin{equation}\label{eq:intro.goal}
\begin{split}
0 = \int_{S^*\mathbb R^{2+1}} (- \xi_t \rd_t + \xi_i \rd_i) \widetilde{a} \,\f{\ud\nu}{|\xi|^2}.
\end{split}
\end{equation}

We derive \eqref{eq:intro.goal} using an energy identity. Let $A$ be a pseudo-differential operator with principal symbol $a = \f{\widetilde{a}}{\xi_t}$. A long (but unilluminating from the point of view of this discussion) computation yields
\begin{align}
&\: \int_{\mathbb R^{2+1}}
\{\frac{\partial_t \phi_n}{N_n} [\partial_t,A] \frac{\partial_t \phi_n}{N_n} 
	+\delta^{ij} (N_n \partial_j \phi_n) [A, \partial_i](\frac{\partial_t \phi_n}{N_n})\} \,\ud x\label{eq:intro.main.0}\\
	&\: + \int_{\mathbb R^{2+1}} \{ 
	\delta^{ij} (N_n \partial_j \phi_n) [\partial_t,A] (\frac{1}{N_n} \partial_i \phi_n)
	+ [A,\partial_j]( N_n\partial_i \phi_n)\delta^{ij} ( \frac{\partial_t \phi_n}{N_n})\}
	\ud x \label{eq:intro.main}\\
&\: + \int_{\mathbb R^{2+1}} \{ (\rd_i \phi_n) \de^{ij} (\rd_t N_n) A \f{\rd_j \phi_n}{N_n} - (\rd_i \phi_n)\de^{ij} N_n A\f{(\rd_t N_n) (\rd_j \phi_n)}{N_n^2} \} \,\ud x \label{eq:intro.SA} \\
&\: + \int_{\mathbb R^{2+1}} \{ (\rd_i \phi_n)\de^{ij} N_n A\f{(\rd_j N_n) (\rd_t \phi_n)}{N_n^2} - (\rd_t \phi_n)(\rd_j N_n)\de^{ij} A\f{(\rd_i \phi_n)}{N_n} \} \,\ud x \label{eq:intro.trivial}\\
= &\: \int_{\mathbb R^{2+1}} \{ \f{\rd_t\phi_n}{N_n} A(N_n g_n^{-1}(\ud\phi_n, \ud\phi_n)) + A(\f{\rd_t\phi_n}{N_n}) N_n g_n^{-1}(\ud\phi_n, \ud\phi_n) \} \,\ud x \label{eq:intro.wave}\\
&\:-\int_{\mathbb R^{2+1}}  \f{\rd_t \phi_n}{N_n} \de^{ij} \{(N_n^2 - 1) \rd_i A \f{\rd_j\phi_n}{N_n}-  \rd_i A ((N_n^2 - 1) \f{\rd_j \phi_n}{N_n})   \} \,\ud x.\label{eq:intro.hard}
\end{align}
{(See Section~\ref{sec:energy.id} for details of a similar computation.)}

By \eqref{eq:intro.MDM.def}, (after passing to a subsequence if necessary {and using that $\xi_t^2=\delta^{ij}\xi_i \xi_j$ on the support of $\ud\nu$}) as $n\to +\infty$, ${\eqref{eq:intro.main.0}+}\eqref{eq:intro.main} \to 2\times \mbox{RHS of \eqref{eq:intro.goal}}$. It therefore suffices to show that the other lines all tend to $0$ as $n\to +\infty$.

That \eqref{eq:intro.SA} and \eqref{eq:intro.trivial} tend to $0$ are relatively straightforward: these rely respectively on the self-adjointness (up to error) of $A$ and \eqref{eq:intro.N.spatial}. 

However, that \eqref{eq:intro.wave} and \eqref{eq:intro.hard} both tend to $0$ is more subtle. This requires trilinear compensated compactness. We now turn to that.

\subsubsection{Trilinear compensated compactness}

There are two types of trilinear compensated compactness that we use. The first kind relates to term \eqref{eq:intro.wave}. We call this \emph{trilinear compensated compactness for three waves} as it is a trilinear term in the derivatives of $\phi_n$, and the compensated compactness relies in particular on good bounds for $\Box_{g_n}\phi_n$. The second kind of trilinear compensated compactness relates to the term \eqref{eq:intro.hard}. We call this \emph{elliptic-wave trilinear compensated compactness} since it relies on both $\phi_n$ satisfying wave estimates and $N_n$ satisfying good spatial derivative estimate \eqref{eq:intro.N.spatial} (which in the actual problem is obtained via elliptic estimates for $N_n$).

\textbf{Trilinear compensated compactness for three waves.} In fact each term in \eqref{eq:intro.wave} tends to $0$. We discuss only a simpler statement, which captures already the main idea involved. We argue that for $\phi_n$ satisfying \eqref{eq:intro.model.est}, 
\begin{equation}\label{eq:phi.trilinear.model}
(\partial_t \phi_n) g_n^{-1} (\ud \phi_n,\ud \phi_n)\rightharpoonup 0
\end{equation}
in the sense of distributions.

 To this end, first observe that by \eqref{eq:intro.model.eqn} and \eqref{eq:intro.model.est}, 
\begin{equation}\label{eq:intro.model.Box.est} 
\|\Box_{g_n} \phi_n \|_{L^1\cap L^\i(\m R^{2+1})} \ls 1.
\end{equation}
Then notice that we can write
 $$g_n^{-1} (\ud \phi_n,\ud \phi_n) = \f 12 \Box_{g_n} (\phi_n^2)- \phi_n(\Box_{g_n} \phi_n).$$
 It follows that for $\chi \in C^\infty_c(\m R^{2+1})$,
 $$ \int_{\m R^{2+1}} \chi (\partial_t \phi_n) g_n^{-1} (\ud \phi_n,\ud \phi_n) \, N_n\, \ud x 
 =  \frac{1}{2}\int_{\m R^{2+1}} \chi(x)(\rd_t\phi_n)\Box_{g_n} (\phi_n^2) \, N_n\,\ud x - \int_{\m R^{2+1}} \chi(x)(\rd_t\phi_n)\phi_n\Box_{g_n} (\phi_n) \, N_n\,\ud x.$$
 The second term clearly $\to 0$ by \eqref{eq:intro.model.est} and \eqref{eq:intro.model.Box.est}. After integrating by parts, the first term can be written as a term taking the form of the second terms plus $O(\lambda_n)$ error, which then implies \eqref{eq:phi.trilinear.model}.

\textbf{Elliptic-wave trilinear compensated compactness.} We now turn to the term \eqref{eq:intro.hard}. Using the estimates in \eqref{eq:intro.model.est}, it follows that \eqref{eq:intro.hard} has the same limit as
\begin{equation}\label{eq:intro.hard.again}
\int_{\mathbb R^{2+1}} \partial_t \phi_n \de^{ij} \{ (N_n^2-1)\rd_i A \rd_j \phi_n - \rd_i A [(N_n^2 - 1) \rd_j \phi_n] \} \,\ud x.
\end{equation}

If $N_n\to 1$ in $C^1$, then \eqref{eq:intro.hard.again} can be easily handled using the Calder\'on commutator estimate (see Lemma~\ref{lem:PSIDOs}.6), which gives
$$|\eqref{eq:intro.hard.again}| \ls \|\rd_t \phi_n\|_{L^2(\mathbb R^{n+1})} \|N_n^2 -1\|_{C^1(\mathbb R^{2+1})} \|\rd_j\phi_n\|_{L^2(\mathbb R^{2+1})} \to 0.$$
The main issue is therefore that while $N_n-1$ and $\nabla (N_n-1)$ indeed converge uniformly, the term $\rd_t N_n$ only converges to $0$ \emph{weakly}. We therefore need the more precise structure in \eqref{eq:intro.hard.again} and argue in Fourier space.

To illustrate the idea, assume that $A$ is simply a Fourier multiplier, i.e.~its symbol $a(x, \xi) = m(\xi)$ is independent of $x$. This indeed captures the main difficulty. In this case, since $\phi_n$ is real-valued, we can assume also that $m$ is even.

Under these assumptions, we can rewrite \eqref{eq:intro.hard.again} up to terms tending to $0$.
\begin{equation*}
\begin{split}
&\:  \left| \eqref{eq:intro.hard.again} - \int_{\mathbb R^{2+1}} \partial_t \phi_n \de^{ij} \{ (N_n^2-1) A \rd^2_{ij} \phi_n - A [(N_n^2 - 1) \rd^2_{ij} \phi_n] \} \,\ud x \right| \\
 = & \: \left| \int_{\mathbb R^{2+1}} \partial_t \phi_n \de^{ij} A [\rd_i(N_n^2 - 1) (\rd_j \phi_n)] \,\ud x\right| \ls \|\rd_t \phi_n\|_{L^2(\mathbb R^{n+1})} \|\nab(N_n^2 -1)\|_{C^0(\mathbb R^{2+1})} \|\rd_j\phi_n\|_{L^2(\mathbb R^{2+1})} \to 0.
\end{split}
\end{equation*}
Then we compute (cf.~Proposition \ref{prop:main.N.term})
\begin{equation}\label{eq:intro.hard.yet.again}
\begin{split}
&\: \int_{\mathbb R^{2+1}} \partial_t \phi_n \de^{ij} \{ (N_n^2-1) A \rd^2_{ij} \phi_n - A [(N_n^2 - 1) \rd^2_{ij} \phi_n]\} \,\ud x \\
= &\: \frac{i}{2}\int_{\m R^{2+1}\times \m R^{2+1} }(\xi_t|\eta_i|^2+ \eta_t|\xi_i|^2)\widehat{(N_n^2-1)}(\eta-\xi)\widehat{\phi_n}(-\eta)\widehat{\phi_n}(\xi)(m(\xi)-m(\eta))\, \ud\xi \,\ud\eta,
\end{split}
\end{equation}
where we decomposed $\xi$ and $\eta$ into their time and spatial parts: $\xi=(\xi_t,\xi_i), \eta = (\eta_t,\eta_i)$.
	
	Roughly speaking $ (\xi_t|\eta_i|^2+ \eta_t|\xi_i|^2)$ corresponds to three derivatives, and hence contributes roughly to $O(\lambda_n^{-3})$ in size (see \eqref{eq:intro.model.est}). This is just enough to show that the \eqref{eq:intro.hard.yet.again} is bounded using the estimates \eqref{eq:intro.model.est}. To deduce that \eqref{eq:intro.hard.yet.again} in fact tends to $0$, observe
	\begin{itemize}
	\item our main enemy is when $N_n^2-1$ has high-frequency in $t$, i.e.~$|\eta_t -\xi_t|$ is large (since we have better estimates for spatial derivatives of $N_n$; see \eqref{eq:intro.N.spatial});
	\item we can gain with factors of $\xi_i - \eta_i$ (corresponding to spatial derivatives of $N_n^2 -1$) or $\xi_t^2 - |\xi_i|^2$ or $\eta_t^2 - |\eta_i|^2$ (corresponding to $\Box_{g_0}$ acting on $\phi_n$).
	\end{itemize}
	
Now the Fourier multiplier in \eqref{eq:intro.hard.yet.again} can be written as 	
	$$\xi_t|\eta_i|^2+ \eta_t|\xi_i|^2 = \eta_t(\xi_i+\eta_i)(\xi_i-\eta_i)+ |\eta_i|^2(\xi_t+\eta_t).$$
	The first term contains a factor of $(\xi_i-\eta_i)$ which as mentioned above corresponds to a spatial derivative of $N_n$ and behaves better.
	For the other term, we rewrite
	$$ |\eta_i|^2(\xi_t+\eta_t)=|\eta_i|^2 \frac{\xi_t^2-\eta_t^2}{\xi_t-\eta_t}=|\eta_i|^2 \frac{\xi_t^2-|\xi_i|^2}{\xi_t-\eta_t}
	+|\eta_i|^2 \frac{|\eta_i|^2-\eta_t^2}{\xi_t-\eta_t}+|\eta_i|^2 \frac{(\xi_i+\eta_i)\cdot (\xi_i - \eta_i)}{\xi_t-\eta_t}.$$
	When $\xi_t - \eta_t$ is large, we can make use of the gain in $\xi_t^2-|\xi_i|^2$, $|\eta_i|^2-\eta_t^2$ or $(\xi_i - \eta_i)$ to conclude that this term behaves better than expected.
	
\subsubsection{Further issues}\label{sec:further.issues}
	
We finally discuss a few additional issues that we encounter in the proof, but are not captured by the simplified model problem above.
\begin{enumerate}
  \item (Spacetime cutoff) Our solution is a priori only defined in a subset of $\mathbb R^{2+1}$, with estimates that hold only locally. We therefore need to introduce and control appropriate cutoff functions.
  \item (Estimates for metric components) The estimates for the metric coefficients has to be derived using the elliptic equations that they satisfy. 
	\begin{enumerate}
	\item To show that \eqref{eq:intro.N.spatial} holds for the metric components, we use the fact that the metric components satisfy (semilinear) elliptic equations due to our gauge condition.
	\item There is in fact further structure for the estimates for the metric components: while the spatial derivatives of \emph{all} metric components obey a better estimate of the form \eqref{eq:intro.N.spatial}, the $\rd_t$ derivative of the metric component of $\gamma$ (see \eqref{g.form}) also obeys a better estimate due to the gauge condition. This fact is crucially used.
	\end{enumerate}
	\item (Non-trivial limit for wave variables) In general $\phi_n$ does not tend to $0$, but instead tends to a non-trivial limit $\phi_0$ (with estimates $\|\rd^k (\phi_n - \phi_0)\|_{L^2\cap L^\i} \ls \lambda_n^{1-k}$). 
	\begin{enumerate}
	\item The non-triviality of $\phi_0$ already means that (in addition to an analogue of \eqref{eq:intro.main}--\eqref{eq:intro.hard}) we need to derive an energy identity for the limit spacetime and take difference appropriately.
	\item More seriously, we need an additional ingredient, which is \underline{not} captured by our model problem.  In general, when the limit $\phi_0$ is not identically $0$, the corresponding trilinear compensated compactness statement gives (see Proposition~\ref{prop:trilinear}),
	$$(\rd_t\phi_n) g_0^{-1}(\ud \phi_n, \ud\phi_n) - 2 (\rd_t\phi_n) g_0^{-1}(\ud \phi_n, \ud\phi_0) \rightharpoonup -(\rd_t\phi_0) g_0^{-1}(\ud \phi_0, \ud\phi_0)$$
	in the sense of distributions. In other words, in our model problem, if we assume $\phi_n\to \phi_0 \not\equiv 0$ (but still assuming $N_n \to 1$), we get 
	\begin{equation}\label{eq:intro.problem}
	\begin{split}
	&\: \int_{\mathbb R^{2+1}} \f{\rd_t\phi_n}{N_n} A(N_n g_n^{-1}(\ud\phi_n, \ud\phi_n))  \,\ud x \\
	\to &\: 2 \int_{S^*\mathbb R^{2+1}}  a(x,\xi) (g_0^{-1})^{\alp\bt} \xi_t \xi_\alp (\rd_\bt\phi_0)  \,\f{\ud \nu}{|\xi|^2} + \int_{\mathbb R^{2+1}} \rd_t\phi_0 A(g_0^{-1}(\ud\phi_0, \ud\phi_0))  \,\ud x \\
	\neq &\: \int_{\mathbb R^{2+1}} \rd_t\phi_0 A(g_0^{-1}(\ud\phi_0, \ud\phi_0))  \,\ud x,
	\end{split}
	\end{equation}
	which does \underline{not} cancel off the corresponding term in the energy identity for $\phi_0$. 
	
	The actual system, despite its complications, is in fact better in the sense that all the terms involving the microlocal defect measure as in \eqref{eq:intro.problem} \underline{cancel}! This cancellation is related to the Lagrangian structure of the wave map system.
	\end{enumerate}
	\item (Freezing coefficients) Since the equation for $\phi_n$ is quasilinear, we can not take the Fourier transform as in the model problems. To overcome this difficulty, we will introduce a partition of our domains into ball of radius $\lambda_n^{\ep_0}$ {(with well-chosen $\ep_0$)}, and show that in each of these balls the metric coefficients can be well-approximated (in terms of $\lambda_n$) by constants {so as} to carry out our argument. See Sections~\ref{sec:freeze.coeff} and \ref{sec:elliptic.wave.tri}.
\end{enumerate}

Finally, let us emphasize that {in all the above discussions} we have relied very heavily on the structure of the terms involved. Indeed it is easy to slightly modify the terms so that the argument fails.

\subsection{Discussions}\label{sec:discussions}

\subsubsection{The reverse Burnett conjecture}\label{sec:reverse}
Already in \cite{Burnett}, Burnett suggested that a reverse version of Conjecture~\ref{conj:Burnett} may also hold, in the sense that any sufficiently regular solution to the Einstein--massless Vlasov system can be approximated weakly by a sequence of high frequency vacuum spacetimes in the sense of \eqref{eq:bounds.caricature}. 

Like Conjecture~\ref{conj:Burnett}, in full generality the reverse Burnett conjecture remains open. On the other hand, some results have been achieved in the $\mathbb U(1)$-symmetric polarized case in our previous \cite{HL}. More precisely, given a generic small and regular polarized $\m U(1)$-symmetric solution to the Einstein--null dust system with a finite number of families of null dust which are angularly separated in an appropriate sense, we proved that it can arise as a weak limit\footnote{Note however that the convergence we obtained was slightly weaker than \eqref{eq:bounds.caricature}; see \cite{HL} for precise convergence rates.} of solutions to the Einstein vacuum system.

Note that the Einstein--null dust system is indeed a special case of the Einstein--massless Vlasov system, where at each spacetime point the Vlasov measure is given as a finite sum of delta measures in the cotangent space; see Section~\ref{sec:null.dust}. In fact, since finite sums of delta measures form a weak-* dense subset of finite Radon measures, one can even hope that the results in \cite{HL} can be extended to a larger class of solutions to the Einstein--massless Vlasov system.


\subsubsection{Trilinear compensated compactness} 
To the best of our knowledge, the phenomenon of trilinear compensated compactness has previously only been studied in the classical work \cite{JJR}. The work considers three sequences of functions $\{\phi_{1,i}\}_{i=1}^{+\infty}$, $\{\phi_{2,i}\}_{i=1}^{+\infty}$ and $\{\phi_{3,i}\}_{i=1}^{+\infty}$ on $\mathbb R^3$, each of which has a weak-$L^2$ limit and moreover $X_j \phi_{j,i}$ is bounded in $L^2$ uniformly in $i$ for some smooth vector fields $X_1$, $X_2$ and $X_3$. It is proven that \emph{under suitable assumptions of $X_j$}, the product $\phi_{1,i}\phi_{2,i}\phi_{3,i}$ converges in the sense of distributions to the product of the weak limits.

\subsection{Outline of the paper}\label{sec:outline}
{The remainder of the paper is structured as follows. In \textbf{Section~\ref{sec:prelim}}, we begin with an introduction to various notions important for our setup, including the symmetry and gauge conditions, and the notion of measure-valued solutions to the Einstein--massless Vlasov system. In \textbf{Section~\ref{sec.notations}}, we then introduce the notation used for the remainder of the paper. In \textbf{Section~\ref{sec:main.results}}, we give the precise statements of the main results of the paper. In \textbf{Section~\ref{sec:PSIDOs}}, we recall some standard facts about pseudo-differential operators and microlocal defect measures.

Starting in Section~\ref{sec:MDM.psi.om}, we begin with the proof of the main results. In \textbf{Section~\ref{sec:MDM.psi.om}}, we apply derive some simple facts about the microlocal defect measures in our setting. In \textbf{Section~\ref{sec:pf.thm.prelim}}, we prove our first main theorem, Theorem~\ref{thm:prelim} (cf.~Theorem~\ref{thm:intro.prelim}).

In the remaining sections, we prove our other main theorem, Theorem~\ref{thm:main} (cf.~Theorem~\ref{thm:intro}). \textbf{Section~\ref{sec:preliminaries}} gives some preliminary observations. In \textbf{Section~\ref{sec:energy.id}}, we derive the main energy identities (cf.~\eqref{eq:intro.main}--\eqref{eq:intro.hard}) which will be used to derive the transport equation of the microlocal defect measures. In \textbf{Section~\ref{sec:energy.id.n.1}}, we first handle the easier terms in deriving the transport equation. In the next two sections we handle terms for which we need \emph{trilinear compensated compactness}: terms requiring elliptic-wave compensated compactness will be treated in \textbf{Section~\ref{sec:elliptic.wave.tri}}; and terms requiring three-waves compensated compactness will be treated in \textbf{Section~\ref{sec:cc}}. The proof is finally concluded in \textbf{Section~\ref{sec:final}}.
}

\subsection*{Acknowledgements} C.~Huneau is supported by the ANR-16-CE40-0012-01. J.~Luk is supported by a Terman fellowship, a Sloan fellowship and the NSF grant DMS-1709458.

\section{Setup and preliminaries}\label{sec:prelim}

\subsection{$\mathbb U(1)$ symmetry}\label{sec:u1}

For the remainder of the paper, fix a $T>0$ and take as our ambient $(3+1)$-dimensional manifold ${}^{(4)}\mathcal M = \mathcal M\times \mathbb R$, where $\mathcal M = (0,T)\times \mathbb R^2$. Introduce coordinates $(t,x^1,x^2)$ on $\mathcal M$ and $(t,x^1,x^2,x^3)$ on ${}^{(4)}\mathcal M$ in the obvious manner.

Consider a Lorentzian metric ${}^{(4)}g$ on ${}^{(4)}\mathcal M$ with a $\mathbb U(1)$ symmetry, i.e.~${}^{(4)}g$ takes the form 
\begin{equation}\label{eq:4g}
^{(4)}g= e^{-2\psi}g+ e^{2\psi}(\ud x^3+ \mathfrak A_\alp \ud x^{\alp})^2,
\end{equation}
where $g$ is a Lorentzian metric on $\mathcal M$, $\psi$ is a real-valued function on $\mathcal M$ and $\mathfrak A_{\alp}$ is a real-valued $1$-form on $\mathcal M$.

Under these assumptions, it is well known that the Einstein vacuum equations for $({}^{(4)}\mathcal M,{}^{(4)}g)$ reduces to the following $(2+1)$-dimensional Einstein--wave map system for $(\mathcal M, g, \psi, \om)$ (see for instance \cite{livrecb}),
\begin{equation}\label{eq:U1vac}
\left\{
\begin{array}{l}
\Box_g \psi + \f 12 e^{-4\psi} g^{-1}(\ud \om, \ud \om) = 0,\\
\Box_g \om - 4 g^{-1} (\ud \om, \ud \psi) = 0,\\
\mathrm{Ric}_{\alp\bt}(g)= 2\partial_\alp \psi \partial_\bt \psi + \f 12 e^{-4\psi} \rd_\alp \om \rd_\bt \om,
\end{array}
\right.
\end{equation}
where $\om$ is a real-valued function which relates to $\mathfrak A_\alp$ via the relation
\begin{equation}\label{eq:mfkA}
(\ud \mathfrak A)_{\alp\bt}  = \rd_\alp \mathfrak A_\bt - \rd_\bt \mathfrak A_\alp = \f 12 e^{-4\psi} (g^{-1})^{\lambda\de} \in_{\alp\bt\lambda} \rd_\de \om,
\end{equation}
where $\in_{\alp\bt\lambda}$ denotes the volume form corresponding to $g$.

\subsection{Elliptic gauge}\label{sec:elliptic.gauge}

We will work in a particular elliptic gauge for the $(2+1)$-dimensional metric $g$ on $\mathcal M$ (cf.~\eqref{eq:4g}). More precisely, we will assume that $g$ takes the form
\begin{equation}\label{g.form}
g=-N^2dt^2 + e^{2\gamma}\delta_{ij}(dx^i + \beta^i dt)(dx^j + \beta^jdt).
\end{equation}
such that the following relation is satisfied
\begin{equation}\label{trace.free}
\rd_t \gamma -\bt^i \rd_i \gamma - \f 12 \rd_i \bt^i = 0,
\end{equation}
where in \eqref{g.form} and \eqref{trace.free} (and in the remainder of the paper), repeated lower case Latin indices are summed over $i,j=1,2$.

We remark that the condition \eqref{trace.free} ensures that the constant-$t$ hypersurfaces have zero mean curvature and the condition \eqref{g.form} ensures that the metric on a constant-$t$ hypersurface induced by $g$ is conformal to the flat metric.

Given the form of the metric in \eqref{g.form}, the inverse metric $g^{-1}$ takes the form
\begin{equation}\label{g.inverse}
g^{-1} = - \f 1{N^2} (\rd_t - \bt^i \rd_i) \otimes (\rd_t - \bt^j \rd_j) + e^{-2\gamma} \de^{ij} \rd_i \otimes \rd_j.
\end{equation}

Assuming that a metric $g$ on $\mathcal M$ obeys \eqref{g.form} and \eqref{trace.free}, the metric components $N$, $\gamma$ and $\bt^i$ satisfy the following elliptic equations; see \cite[Appendix~B]{HL}:
\begin{align}
&{\de^{ik}}\partial_{{k}} H_{ij}= -\frac{e^{2\gamma}}{N} \mathrm{Ric}_{0j}, \label{elliptic.1}\\
&\Delta \gamma = -\frac{e^{2\gamma}}{N^2}G_{00} - \frac{1}{2}e^{-2\gamma}|H|^2,\label{elliptic.2}\\
&\Delta N =Ne^{-2\gamma}|H|^2 -\f 12 e^{2\gamma} N R + \f{e^{2\gamma}}{N} G_{00},\label{elliptic.3}\\
& (\mathfrak L\beta)_{ij}=2Ne^{-2\gamma}H_{ij},\label{elliptic.4}
\end{align}
where $e_0=\rd_t-\bt^i\rd_i$, $\mathrm{Ric}_{\alp\bt}$ is the Ricci tensor, $R$ is the scalar curvature, $G_{\alp\bt} = \mathrm{Ric}_{\alp\bt} - \f 12 R g_{\alp\bt}$ is the Einstein tensor, and $\mathfrak L$ is the conformal Killing operator given by 
\begin{equation}\label{L.def}
(\mathfrak L\beta)_{ij}:=\delta_{j\ell}\rd_i\beta^\ell+\delta_{i\ell}\rd_j\beta^\ell-\delta_{ij}\rd_k\beta^k.
\end{equation}
Moreover, assuming \eqref{g.form} and \eqref{trace.free}, the spatial components of the Ricci tensor is given by (see \cite[Proposition~B2]{HL})
\begin{align}
\label{Rij} \mathrm{Ric}_{ij} = &\delta_{ij}\left(-\Delta \gamma-\frac{1}{2N}\Delta N\right)
-\frac{1}{N}(\partial_t-\beta^k\partial_k)H_{ij}-2e^{-2\gamma} H_i{ }^{\ell} H_{j\ell} \\
\notag &+\frac{1}{N}\left(\partial_j \beta^k H_{ki}+\partial_i \beta^k H_{kj}\right)-\frac{1}{N}\left( \partial_i \partial_j N -\frac{1}{2}\delta_{ij}\Delta N -\left(\delta_i^k\partial_j \gamma +\delta_j^k\partial_i \gamma-\delta_{ij} \de^{\ell k}\partial_{\ell} \gamma\right)\partial_k N \right).
\end{align}

In the particular case where the vacuum equations \eqref{eq:U1vac} are satisfied, \eqref{elliptic.1}--\eqref{elliptic.3} take the following form:
\begin{align}
&{\de^{ik}}\partial_{{k}} H_{ij}= -\frac{ e^{2\gamma}}{N}\left(2(e_0 \psi)( \partial_j \psi) + \f 12 e^{-4\psi} (e_0\om)(\rd_j\om)\right), \label{elliptic.wave.1}\\
&\Delta \gamma = -(|\nabla \psi|^2 +\f 14 e^{-4\psi} |\nabla \om|^2) -\frac{e^{2\gamma}}{N^2} ((e_0 \psi)^2+\f 14 e^{-4\psi} (e_0\om)^2) - \frac{1}{2}e^{-2\gamma}|H|^2,\label{elliptic.wave.2}\\
&\Delta N =Ne^{-2\gamma}|H|^2+ \frac{e^{2\gamma}}{N}(2(e_0 \psi)^2+\f 12 e^{-4\psi}(e_0\om)^2).\label{elliptic.wave.3}
\end{align}
Combining \eqref{elliptic.4} and \eqref{elliptic.wave.1}, we also obtain the following second order elliptic equation for $\bt^j$ when \eqref{eq:U1vac} are satisfied
\begin{equation}\label{elliptic.wave.4}
\Delta \beta^j={\delta^{ik}}\delta^{j\ell}\rd_k\left(\log(Ne^{-2\gamma})\right)(L\beta)_{i\ell}-4 \delta^{ij}(e_0 \phi)( \partial_i \phi) \end{equation}
Moreover, \eqref{Rij} takes the following form when \eqref{eq:U1vac} are satisfied:
\begin{equation}\label{Rij.wave}
\begin{split}
&\: 2\partial_i \psi \partial_j \psi + \f 12 e^{-4\psi} \rd_i \om \rd_j \om \\
= &\: \delta_{ij}\left(-\Delta \gamma-\frac{1}{2N}\Delta N\right)
-\frac{1}{N}(\partial_t-\beta^k\partial_k)H_{ij}-2e^{-2\gamma} H_i{ }^{\ell} H_{j\ell} \\
 &\: +\frac{1}{N}\left(\partial_j \beta^k H_{ki}+\partial_i \beta^k H_{kj}\right)-\frac{1}{N}\left( \partial_i \partial_j N -\frac{1}{2}\delta_{ij}\Delta N -\left(\delta_i^k\partial_j \gamma +\delta_j^k\partial_i \gamma-\delta_{ij} \de^{\ell k}\partial_{\ell} \gamma\right)\partial_k N \right).
\end{split}
\end{equation}

Finally, given a metric in the gauge \eqref{g.form}, the wave operator takes the following form:
\begin{equation}\label{eq:wave}
\Box_g f = \f{1}{\sqrt{-\det g}} \rd_\alp \Big(\gi^{\alp\bt} \sqrt{-\det g} \rd_\bt f \Big) = -\f{e_0^2 f}{N^2} - e^{-2\gamma} \de^{ij} \rd^2_{ij} f + \f{e_n N}{N^3} e_0 f + \f{e^{-2\gamma}}N \de^{ij} \rd_i N \rd_j f.
\end{equation}

\subsection{Measure solutions to the Einstein--massless Vlasov system}\label{sec:EMV}

\begin{df}[Measure solutions to Vlasov equation]\label{def:Vlasov}
Let $(\mathcal M,g)$ be a $C^1$ Lorentzian manifold. We say that a non-negative finite Radon measure $\ud\mu$ on $T^*\mathcal M$ with $\int_{T^*\mathcal M} |\xi|^2 \,\ud\mu<+\infty$ solves the massless Vlasov equation if the following two conditions both hold:
\begin{enumerate}
\item $\ud\mu$ is supported on the zero mass shell $\{(x,\xi)\in T^*\mathcal M : (g^{-1})^{\alp\bt} \xi_\alp \xi_\bt = 0\}$.
\item For every function $a(x,\xi)\in C^\infty_c(T^*\mathcal M\setminus \{0\})$, it holds that 
\begin{equation}\label{eq:transport.orig}
\begin{split}
\int_{T^*\mathcal M\setminus\{0\}} ((g^{-1})^{\alp\bt}\xi_\bt \rd_{x^{\alp}}a - \f 12 (\rd_\mu g^{-1})^{\alp\bt} \xi_\alp \xi_\bt \rd_{\xi_\mu}a) \,\ud\mu = 0.
\end{split}
\end{equation}
\end{enumerate}
\end{df}

Definition~\ref{def:Vlasov} is indeed a generalization of the ``usual'' Vlasov equation, where $\ud\mu$ is absolutely continuous with respect to the natural measure on the zero mass shell. More precisely, 
\begin{proposition}
Let $(x^0, x^1, \dots, x^n)$ be a system of local coordinates on $U\subset \mathcal M$. Introduce a local coordinate system $(\bar{x}^0, \bar{x}^1,\dots, \bar{x}^n, \bar{\xi}_1, \dots, \bar{\xi}_n) :=(x^0, x^1, \dots, x^n, \xi_1, \dots, \xi_n)$ on the zero mass shell restricted to $U$ (which is a $(2n+1)$-dimensional sub-manifold of the cotangent bundle). Here, and in the proof, we use the bar in $\rd_{\bar{x}^\alp}$, etc.~to indicate that the derivative is to be understood as the coordinate derivative with respect to the coordinate system on the zero mass shell. On the zero mass shell, $\xi_0$ will be understood as a function of $(x^0, x^1, \dots, x^n, \xi_1, \dots, \xi_n)$ defined implicitly by $(g^{-1})^{\alp\bt} \xi_\alp \xi_\bt = 0$. {(Note that $\xi_0$ is well-defined locally.)}

Suppose $f:\{(x,\xi)\in T^* U : (g^{-1})^{\alp\bt} \xi_\alp \xi_\bt = 0{,\, \xi \neq 0}\} \to [0,+\infty)$ is a $C^1$ function satisfying the equation
\begin{equation}\label{eq:classical.Vlasov}
(g^{-1})^{\alp\bt}\xi_\alp\rd_{\bar{x}^\bt} f - \f 12 \rd_{\bar{x}^i} (g^{-1})^{\alp\bt} \xi_\alp\xi_\bt \rd_{\bar{\xi}_i} f= 0.
\end{equation}
Then, for $\ud\mu:= f\,\f{\ud x^0\,\ud x^1\,\cdots\,\ud x^n\,\ud \xi_1\,\cdots \,\ud \xi_n}{{|(g^{-1})^{\alp0}\xi_\alp|}}$, \eqref{eq:transport.orig} holds for {every} $a\in C^\infty_c(T^*U\setminus\{0\})$.
\end{proposition}
\begin{proof}
{For the proof, we fix either $\xi_0>0$ or $\xi_0<0$. (The argument is identical in the two cases.)} 

We first compute the transformation 
\begin{equation}\label{eq:mass.shell.transform}
\rd_{\bar{x}^\alp} = \rd_{x^\alp} - \f 12 \f{\rd_{x^\alp}\gi^{\bt\nu}\xi_\bt\xi_\nu}{\gi^{0\mu} \xi_\mu}\rd_{\xi_0}, \quad \rd_{\bar{\xi}_i} = \rd_{\xi_i} - \f{\gi^{i\bt}\xi_\bt}{\gi^{0\mu} \xi_\mu}\rd_{\xi_0}, 
\end{equation}
where, here, and in the rest of the proof, Greek indices run through $0,1,\dots, n$ and Latin indices run through $1,\dots, n$.

It follows that 
$$(g^{-1})^{\alp\bt}\xi_\bt \rd_{x^{\alp}} - \f 12 (\rd_\mu g^{-1})^{\alp\bt} \xi_\alp \xi_\bt \rd_{\xi_\mu}  = (g^{-1})^{\alp\bt}\xi_\alp\rd_{\bar{x}^\bt} - \f 12 \rd_{\bar{x}^i} (g^{-1})^{\alp\bt} \xi_\alp\xi_\bt \rd_{\bar{\xi}_i}.$$
Therefore, the LHS of \eqref{eq:transport.orig} in the coordinate system we introduced reads
\begin{equation}\label{eq:need.to.check.for.classical}
\int_{\mathbb R^n} \int_{U} [(g^{-1})^{\alp\bt}\xi_\alp\rd_{\bar{x}^\bt}a - \f 12 \rd_{\bar{x}^i} (g^{-1})^{\alp\bt} \xi_\alp\xi_\bt \rd_{\bar{\xi}_i} a]\,f \,\f{\ud x^0\,\ud x^1\,\cdots\,\ud x^n\,\ud \xi_1\,\cdots \,\ud \xi_n}{(g^{-1})^{\mu 0}\xi_\mu}.
\end{equation}

Next, notice that by \eqref{eq:mass.shell.transform},
\begin{equation*}
\begin{split}
&\:  \rd_{\bar{x}^\bt} \Big[ \f{(g^{-1})^{\alp\bt}\xi_\alp}{(g^{-1})^{\mu 0}\xi_\mu} \Big]  \\
= &\: \f{ \rd_{\bar{x}^i} (g^{-1})^{\alp i}\xi_\alp}{(g^{-1})^{\mu 0}\xi_\mu} - \f 12 \f{\gi^{0i}\rd_{x^i}\gi^{\sigma\nu}\xi_\sigma\xi_\nu}{[\gi^{0\mu} \xi_\mu]^2} - \f{(g^{-1})^{\alp i}\xi_\alp \rd_{\bar{x}^i} (g^{-1})^{\nu 0}\xi_\nu}{[(g^{-1})^{\mu 0}\xi_\mu]^2}  \\
&\: + \f 12 \f{\gi^{00} (g^{-1})^{\alp i}\xi_\alp\rd_{x^i}\gi^{\sigma\nu}\xi_\sigma\xi_\nu}{[\gi^{0\mu} \xi_\mu]^2} + \f{ \rd_{x^0} (g^{-1})^{\alp 0}\xi_\alp}{(g^{-1})^{\mu 0}\xi_\mu} - \f12 \f{(g^{-1})^{0 0}\rd_{x^0}\gi^{\alp\nu}\xi_\alp\xi_\nu}{[\gi^{0\mu} \xi_\mu]^2} \\
&\:  - \f{(g^{-1})^{\alp 0}\xi_\alp \rd_{x^0} (g^{-1})^{\nu 0}\xi_\nu}{[(g^{-1})^{\mu 0}\xi_\mu]^2} + \f 12 \f{ \rd_{x^0} \gi^{\sigma\nu}\xi_\sigma\xi_\nu (g^{-1})^{\alp 0}\xi_\alp (g^{-1})^{00}}{[(g^{-1})^{\mu 0}\xi_\mu]^3} \\
=&\: \f{ \rd_{\bar{x}^i} (g^{-1})^{\alp i}\xi_\alp}{(g^{-1})^{\mu 0}\xi_\mu} - \f{(g^{-1})^{\alp i}\xi_\alp \rd_{\bar{x}^i} (g^{-1})^{\nu 0}\xi_\nu}{[(g^{-1})^{\mu 0}\xi_\mu]^2} - \f 12 \f{\gi^{0i}\rd_{x^i}\gi^{\sigma\nu}\xi_\sigma\xi_\nu}{[\gi^{0\mu} \xi_\mu]^2} + \f 12 \f{\gi^{00} (g^{-1})^{\alp i}\xi_\alp\rd_{x^i}\gi^{\sigma\nu}\xi_\sigma\xi_\nu}{[\gi^{0\mu} \xi_\mu]^2}
\end{split}
\end{equation*}
and 
\begin{equation*}
\begin{split}
 - \f 12 \rd_{\bar{\xi}_i} \Big[ \f{\rd_{\bar{x}^i} (g^{-1})^{\alp\bt} \xi_\alp\xi_\bt}{(g^{-1})^{\mu 0}\xi_\mu} \Big] 
= &\: - \f{\rd_{\bar{x}^i} (g^{-1})^{\alp i} \xi_\alp}{(g^{-1})^{\mu 0}\xi_\mu} + \f{\gi^{i\bt}\xi_\bt \rd_{\bar{x}^i} (g^{-1})^{\alp 0} \xi_\alp}{[\gi^{\mu 0} \xi_\mu]^2} \\
&\: +  \f 12 \f{\gi^{0i} \rd_{\bar{x}^i} (g^{-1})^{\alp\bt} \xi_\alp \xi_\bt}{[(g^{-1})^{\mu 0}\xi_\mu]^2} - \f12 \f{\gi^{i\nu}\xi_\nu \rd_{\bar{x}^i} (g^{-1})^{\alp\bt} \xi_\alp\xi_\bt \gi^{00} }{[(g^{-1})^{\mu 0}\xi_\mu]^3}.
\end{split}
\end{equation*}
Therefore,
\begin{equation*}
\begin{split}
&\: \rd_{\bar{x}^\bt} \Big[\f{(g^{-1})^{\alp\bt}\xi_\alp}{(g^{-1})^{\mu 0}\xi_\mu} \Big] - \f 12 \rd_{\bar{\xi}_i} \Big[\f{\rd_{\bar{x}^i} (g^{-1})^{\alp\bt} \xi_\alp\xi_\bt}{(g^{-1})^{\mu 0}\xi_\mu} \Big] = 0. 
\end{split}
\end{equation*}

Therefore, integrating by parts in \eqref{eq:need.to.check.for.classical} and using \eqref{eq:classical.Vlasov}, we obtain that 
$\eqref{eq:need.to.check.for.classical} = 0,$
as desired. \qedhere
\end{proof}

\begin{df}[Measure solutions to the Einstein--massless Vlasov system]\label{def:measEmV}
Let $(\mathcal M, g)$ be a $C^2$ Lorentzian manifold and $\ud \mu$ be a non-negative finite Radon measure on $T^*\mathcal M$. We say that $(\mathcal M, g,\ud \mu)$ is a \emph{measure solution} to the Einstein--massless Vlasov system if the following both hold:
\begin{enumerate}
\item For every smooth and compactly supported vector field $Y$,
$$\int_{\mathcal M} \mathrm{Ric}(Y,Y)\,\ud \mathrm{Vol}_g = \int_{T^*\mathcal M} \langle \xi,Y\rangle^2 \, \ud \mu.$$
\item $\ud \mu$ is a measure solution to the massless Vlasov system in the sense of Definition~\ref{def:Vlasov}.
\end{enumerate}
\end{df}

\subsection{Radially-averaged measure solutions to the Einstein--massless Vlasov system}\label{sec:radavg.EMV}

It will be convenient for us to define a notion of \emph{radially-averaged measure solution} to the Einstein--massless Vlasov system. Strictly speaking, this is related but is a distinct notion from that of a measure solution in Definition~\ref{def:measEmV}. It is however easy to see that any measure solution in the sense of Definition~\ref{def:measEmV} naturally induces a radially-averaged measure solution. Conversely, given a radially-averaged measure solution, one can construct a measure solution in the sense of Definition~\ref{def:measEmV}; see Lemma~\ref{lem:relationbetweenformulation}. One reason for introducing this notion is that this is the natural class of solutions that we construct using the microlocal defect measure.

Before we proceed to the definition of a radially-averaged measure solution, let us first define the cosphere bundle
$$S^*\mathcal M := \cup_{x\in \mathcal M} S^*_x \mathcal M := \cup_{x\in \mathcal M} \big((T^*_x \mathcal M \setminus\{0\})/{\sim}\big),$$
where we have quotiented out by the equivalence relation $\xi \sim \eta$ if $\xi = \lambda \eta$ for some $\lambda>0$. 

A continuous function on $S^*\mathcal M$ can be naturally identified with a continuous function on $T^*\mathcal M$ which is homogeneous of order $0$ in $\xi$. Therefore a Radon measure on $S^*\mathcal M$ naturally acts on continuous function on $T^*\mathcal M$ which is homogeneous of order $0$ in $\xi$.

We are now ready to define radially-averaged measure solutions to the Einstein--massless Vlasov system:
\begin{df}[Radially-averaged measure solutions to the Einstein--massless Vlasov system]\label{def:radmeasEmV}
Let $(\mathcal M, g)$ be a $C^2$ Lorentzian manifold and $\ud \nu$ be a non-negative finite Radon measure on $S^*\mathcal M$. We say that $(\mathcal M, g,\ud \nu)$ is a radially-averaged measure solution to the Einstein--massless Vlasov system if the following both hold:
\begin{enumerate}
\item For every smooth and compactly supported vector field $Y$,
$$\int_{\mathcal M} \mathrm{Ric}(Y,Y)\,\ud \mathrm{Vol}_g = \int_{S^*\mathcal M} \langle \xi,Y\rangle^2 \, \f{\ud \nu}{|\xi|^2},$$
where $|\xi|^2 = \sum_{\alp=0}^n \xi_\alp^2$.
\item $\ud \nu$ is supported on the zero mass shell in the sense that for {every} $f\in C_c(\q M)${,}
$$\int_{S^*\mathcal M} f(x) (g^{-1})^{\alp\bt}\xi_\alp\xi_\bt \,\f{\ud \nu}{|\xi|^2} =0.$$
\item For any $C^1$ function $\widetilde{a}:T^*\mathcal M\to \mathbb R$ which is homogeneous of degree $1$ in $\xi$, 
\begin{equation}\label{eq:transport.def}
\int_{{S}^*\mathcal M} ((g^{-1})^{\alp\bt}\xi_\bt \rd_{x^{\alp}}\widetilde{a} - \f 12 (\rd_\mu g^{-1})^{\alp\bt} \xi_\alp \xi_\bt \rd_{\xi_\mu}\widetilde{a}) \,\f{\ud\nu}{|\xi|^2} = 0.
\end{equation}
\end{enumerate}
\end{df}

The relation between a measure solution to the Einstein--massless Vlasov system (Definition~\ref{def:measEmV}) and a radially-averaged measure solution to the Einstein--massless Vlasov system (Definition~\ref{def:radmeasEmV}) is clarified in the following lemma:
\begin{lemma}\label{lem:relationbetweenformulation}
Given a measure solution $(\mathcal M, g, \ud\mu)$ to the Einstein--massless Vlasov system, there exists a radially-averaged measure solution $(\mathcal M, g, \ud\nu)$ to the Einstein--massless Vlasov system (with the same $(\mathcal M, g)$). This is also true conversely if $(\mathcal M, g)$ is globally hyperbolic.
\end{lemma}
\begin{proof}
\pfstep{Forward direction} This is the easier direction, and can in some sense be viewed as taking average in the radial direction in $\xi$. More precisely, given $\varphi\in C_0(S^*\mathcal M)$ (thought of as a continuous function homogeneous of order $0$ in $v$), define a map $I: C_0(S^*\mathcal M)\to \mathbb R$ by
$$I(\varphi) := \int_{T^*\mathcal M} \varphi \, |\xi|^2\ud \mu.$$
Since $\ud\mu$ is non-negative, $I$ is a non-negative map. By the Riesz--Markov representation theorem, it follows that there exists non-negative $\ud \nu$ such that 
$$I(\varphi) = \int_{S^*\mathcal M} \varphi \, \ud \nu.$$
Since $|\xi|^2\ud \mu$ is finite (by definition), it follows that $\ud \nu$ is finite. The Einstein equation also follows by definition.

\pfstep{Converse direction} This is harder and there is some choice available in the construction.

Given a globally hyperbolic $(\mathcal M, g)$, pick a Cauchy hypersurface $\Sigma_0$ and define the set 
 (with two connected components)
$$\underline{S}:=\{(x,\xi)\in T^*\mathcal M: x\in \Sigma_0,\, |\xi|^2=1,\, g^{-1}(x)(\xi,\xi) = 0 \}.$$ 

Define now the set $S$ as {the set of points in $T^*\mathcal M$ which lie in a geodesic starting from $\underline{S}$}. Note that $S$ is a co-dimensional $2$ submanifold of $T^*\mathcal M \setminus \{0\}$. Moreover, the vector field $(g^{-1})^{\alp\bt}\xi_\bt \rd_{x^{\alp}} - \f 12 (\rd_\mu g^{-1})^{\alp\bt} \xi_\alp \xi_\bt \rd_{\xi_\mu}$ is by definition tangential to $S$.

Given $\varphi \in C_0(T^*\mathcal M)$, define $\varphi^* \in C_0(S^*\mathcal M)$ as the function such that $\varphi^*\restriction_{S} = \varphi \restriction_S$ which is homogeneous of order $0$ in $\xi$. Define a map $J:C_0(T^*\mathcal M)\to \mathbb R$ by
$$J(\varphi) := \int_{S^*\mathcal M}\varphi^*\,\ud\nu.$$
This is non-negative by the non-negativity of $\ud \nu$. Hence, by the Riesz--Markov representation theorem, it follows that there exists non-negative $\ud \mu$ such that 
$$J(\varphi) = \int_{T^*\mathcal M} \varphi \, |\xi|^2 \ud \mu.$$
Note that $|\xi|^2\ud\mu$ is finite since $\ud \nu$ is finite. The Einstein equation also follows by definition.

To see that $\ud\mu$ is supported on the zero mass shell, it suffices to note that by definition, $S$, on which by definition $\ud \mu$ is supported, is a subset of the zero mass shell by construction.

Finally we show that \eqref{eq:transport.orig} holds. Take $a(x,\xi)\in C^\infty_c(T^*\mathcal M\setminus \{0\})$. Define $\widetilde{a}$ so that $\widetilde{a}\restriction_{S} = a$ but such that $\widetilde{a}$ is homogeneous of order $1$. Therefore, using \eqref{eq:transport.def}, we know that \eqref{eq:transport.orig} holds with $\widetilde{a}$ in the place of $a$. However, since $\ud \mu$ is supported on $S$ (by construction), it follows that in fact \eqref{eq:transport.orig} holds for $a$. \qedhere
\end{proof}

\begin{rk}[$\ud \nu$ can be chosen to be even] In the ``forward'' direction of the above proof, we could have instead defined
$$I(\varphi) := \int_{T^*\mathcal M} \f 12(\varphi(\xi) + \varphi(-\xi)) \, |\xi|^2\ud \mu(\xi),$$
so that $I$ is even, i.e.~$I(\varphi)= 0$ for {every} odd function $\varphi$. Consequently, $\ud\nu$ is also even. In fact, the measure solution to the Vlasov equation that we will eventually construct is even. 
\end{rk}

\subsection{Restricted Einstein--massless Vlasov system in $\mathbb U(1)$ symmetry}\label{sec:restricted}
The final notion that we introduce in this section is that of the restricted Einstein--massless Vlasov system in $\mathbb U(1)$ symmetry. By ``restricted'', we mean that we are not considering general $(3+1)$-dimensional solutions to the Einstein--massless Vlasov system for which the metric admits a $\mathbb U(1)$ symmetry, but instead we require that massless Vlasov measure to be supported in the cotangent bundle corresponding to the $(2+1)$-dimensional (instead of the $(3+1)$-dimensional) manifold.

Since we have already introduced and contrast both measure solutions and radially-averaged measure solutions for the Einstein--massless Vlasov system (cf.~Sections~\ref{sec:EMV} and \ref{sec:radavg.EMV}), we will directly define the notion of radially-averaged measure solutions for the restricted Einstein--massless Vlasov system in $\mathbb U(1)$ symmetry.

\begin{df}[Radially-averaged measure solutions for the \underline{restricted} Einstein--massless Vlasov system in \underline{$\mathbb U(1)$ symmetry}]\label{def:the.final.def}
Let $(^{(4)}\mathcal M, ^{(4)}g)$ be a $(4+1)$-dimensional $C^2$ Lorentzian manifold which is $\mathbb U(1)$ symmetric as in \eqref{eq:4g}, i.e.~the metric takes the form
$$
^{(4)}g= e^{-2\psi}g+ e^{2\psi}(\ud x^3+ \mathfrak A_\alp \ud x^{\alp})^2,
$$
for $g$, $\psi$, $\mathfrak A$ independent of $x^3$. Let $\ud\nu$ be a non-negative finite Radon measure on $S^*\mathcal M$.

We say that $(^{(4)}\mathcal M, ^{(4)}g, \ud\nu)$ is a radially-averaged measure solution for the restricted Einstein--massless Vlasov system in $\mathbb U(1)$ symmetry if 
\begin{enumerate}
\item the following equations are satisfied:
\begin{equation}\label{eq:U1vac.vlasov}
\left\{
\begin{array}{l}
\Box_g \psi + \f 12 e^{-4\psi} g^{-1}(\ud \om, \ud \om) = 0,\\
\Box_g \om - 4 g^{-1}(\ud \om , \ud \psi) = 0,\\
\int_{\mathcal M} \mathrm{Ric}(g)(Y,Y) \, \mathrm{dVol}_g=  \int_{\mathcal M} [2 (Y \psi)^2 + \f 12 e^{-4\psi} (Y\om)^2] \, \mathrm{dVol}_g + \int_{{S}^*\mathcal M} \langle \xi,Y\rangle^2 \, \f{\ud \nu}{|\xi|^2},
\end{array}
\right.
\end{equation}
for every $C^\infty_c$ vector field $Y$, where $\om$ relates to $\mathfrak A_\alp$ via \eqref{eq:mfkA};
\item (2) and (3) in Definition~\ref{def:radmeasEmV} both hold.
\end{enumerate}
\end{df}

\subsection{Null dust and massless Vlasov}\label{sec:null.dust}

In this subsection, we show that a solution to the null dust system is a measure solution to the massless Vlasov system. In particular, this shows that solutions to Einstein--null dust system considered in \cite{HLHF} can indeed be viewed within the framework of this paper.

For simplicity, let us just consider the case where $F_{\bA}$ is compactly supported.

\begin{lemma}
Let $(\mathcal M, g)$ be a $C^2$ Lorentzian manifold. Suppose for a finite set $\mathcal A$, $\{(F_{\bA}, u_{\bA})\}_{\bA\in \mathcal A}$ is a compactly supported solution to the null dust system on $(\mathcal M, g)$, i.e.~$F_{\bA}:\mathcal M\to \mathbb R$ is a compactly supported $C^1$ function and $u_{\bA}:\mathcal M\to \mathbb R$ is a $C^2$ function satisfying
\begin{enumerate}
\item $g^{-1}(\ud u_{\bA},\ud u_{\bA}) = 0, \quad \ud u_{\bA}\neq 0$ for all $\bA\in \mathcal A$,
\item $2(g^{-1})^{\alp\bt}(\rd_\bt u_{\bA}) \rd_{\alp} F_{\bA} + (\Box_g u_{\bA})F_{\bA} = 0$ for all $\bA\in \mathcal A$.
\end{enumerate}
Then the measure $\ud\mu$ on $T^*\mathcal M$ defined by
\begin{equation}\label{def:mu.dust}
\ud\mu := \sum_{\bA\in \mathcal A} F_{\bA}^2 \de_{v = du_{\bA}} \mathrm{dVol}_g
\end{equation}
is a measure solution to the massless Vlasov equation on $(\mathcal M,g)$ (cf.~Definition~\ref{def:Vlasov}).
\end{lemma}
\begin{proof}
That $\ud\mu$ is supported on the zero mass shell follows immediately from \eqref{def:mu.dust} and $g^{-1}(du_{\bA},du_{\bA}) = 0$. It remains therefore to verify the transport equation in Definition~\ref{def:Vlasov}.

For this we need a preliminary calculation. First, since $(g^{-1})^{\alp\bt}\rd_\alp u \rd_\bt u = 0$, we have
$$(\rd_\sigma (g^{-1})^{\alp\bt})\rd_\alp u \rd_\bt u + 2 (g^{-1})^{\alp\bt}(\rd_\bt u)(\rd_\alp \rd_\sigma u) = 0.$$
Therefore, given any $a\in C^\infty_c(T^*\mathcal M)$, viewing $a(x, du(x))$ as a function on $\mathcal M$ (and emphasizing this by calling the coordinates $\bar{x}$), we have
\begin{equation}\label{a.for.dust}
\begin{split}
&\: (g^{-1})^{\alp\bt}(\rd_\bt u) \rd_{\bar{x}^\alp} (a(\bar{x}, du(\bar{x}))) \\
= &\: (g^{-1})^{\alp\bt}(\rd_\bt u) \rd_{x^\alp} a + (g^{-1})^{\alp\bt}(\rd_\bt u) (\rd_{\xi_\sigma} a) (\rd_\alp \rd_\sigma u)\\
= &\: (g^{-1})^{\alp\bt}(\rd_\bt u) \rd_{x^\alp} a - \f 12 (\rd_\sigma (g^{-1})^{\alp\bt}) (\rd_\alp u)(\rd_\bt u) (\rd_{\xi_\sigma} a).
\end{split}
\end{equation}

We now check that the transport equation in Definition~\ref{def:Vlasov} using \eqref{a.for.dust} and integrating by parts:
\begin{equation*}
\begin{split}
&\: \int_{T^*\mathcal M} \Big( (g^{-1})^{\alp\bt}\xi_\bt \rd_{x^\alp} a - \f 12 (\rd_\sigma (g^{-1})^{\alp\bt}) \xi_\alp \xi_\bt (\rd_{\xi_\sigma} a) \Big)\, \ud \mu \\
=&\: \sum_{\bA\in \mathcal A} \int_{\mathcal M} \Big((g^{-1})^{\alp\bt}(\rd_\bt u_{\bA}) \rd_{x^\alp} a - \f 12 (\rd_\sigma (g^{-1})^{\alp\bt}) (\rd_\alp u_{\bA})(\rd_\bt u_{\bA}) (\rd_{\xi_\sigma} a) \Big) F^2_{\bA} \, \mathrm{dVol}_g \\
=&\: \sum_{\bA\in \mathcal A} \int_{\mathcal M} \Big((g^{-1})^{\alp\bt}(\bar{x})(\rd_\bt u_{\bA})(\bar{x}) \rd_{\bar{x}^\alp} \big(a(\bar{x}, du_{\bA}(\bar{x}))\big)\Big) F^2_{\bA}(\bar{x}) \, \mathrm{dVol}_g(\bar{x}) \\
=&\: - \sum_{\bA\in \mathcal A} \int_{\mathcal M} \Big(2(g^{-1})^{\alp\bt}(\rd_\bt u_{\bA}) \rd_{\alp} F_{\bA} + (\Box_g u_{\bA})F_{\bA} \Big)(\bar{x}) F_{\bA}(\bar{x}) a\big(\bar{x}, du_{\bA}(\bar{x})\big) \, \mathrm{dVol}_g(\bar{x}) =0,
\end{split}
\end{equation*}
where in the last line we used the equation satisfied by $F_{\bA}$. This concludes the proof. \qedhere
\end{proof}

\section{Notations and function spaces}\label{sec.notations}

{\bf Ambient space and coordinates.}
In this paper, we will be working on the ambient manifold $\q M:=(0,T)\times \mathbb R^2$ (although often we only restrict to subsets $\Omega\subset \Omega'\subset\Omega''$, cf.~Section~\ref{sec:compact.reduction}); see Section~\ref{sec:main.results}. The space will be equipped with a system of coordinates $(t,x^1,x^2)$. We often write $x = (t,x^1,x^2)$. We will use $x^i$ with the lower case Latin index $i,j=1,2$. We will also sometimes denote $x^t=t$.

Let $T^*\mathcal M$ be the cotangent bundle. The standard coordinates on $T^*\mathcal M$ will be given by $(x,\xi) = (x^t,x^1,x^2,\xi_t,\xi_1,\xi_2)$.

When there is no risk of confusion, we write $\rd_i = \rd_{x^i}$.

{\bf Indices.} We will use the following conventions:
\begin{itemize}
\item Lower case Latin indices run through the spatial indices $1,2$, while lower case Greek indices run through all $t,1,2$.
\item Repeat indices are always summed over: where lower case Latin indices sum over the spatial indices $1,2$ and lower case Greek indices sum over all indices $t,1,2$.
\end{itemize}

{\bf Metrics.}
\begin{itemize}
\item $g_n$ and $g_0$ denote the metrics introduced in Section~\ref{sec:main.results}, which both take the form \eqref{g.form} and \eqref{trace.free}.
\item We denote by $\mfg_n \in \{\log N_n, \bt^i_n, \gamma_n\}$ and $\mfg_0 \in \{\log N_0, \bt^i_0, \gamma_0\}$ the metric coefficients of $g_n$ and $g_0$ respectively.
\item $\de_{ij}$ (and $\de^{ij}$) denotes the Euclidean metric.
\end{itemize}

{\bf Norms for tensors and derivatives.} 
\begin{itemize}
\item Given a rank-$r$ covariant tensor $\eta_{\mu_1\cdots \mu_r}$, define
$$|\eta|^2:= \sum_{\mu_1,\dots, \mu_r=t,1,2} |\eta_{\mu_1\cdots \mu_r}|^2,\quad |\eta_{i_1\cdots i_r}|^2:=\sum_{j_1,\dots, j_r = 1,2} |\eta_{j_1\cdots j_r}|^2.$$
(The second definition is a slight abuse of notation, by which we mean unless otherwise stated, we will also implicitly take the sum. Similarly below.)
\item {The above} notation is in particular used for $(x,\xi)\in T^*\mathcal M$ where we denote
$$|\xi|^2:= \sum_{\mu = t,1,2} |\xi_\mu|^2,\quad |\xi_i|^2:= \sum_{j=1,2} |\xi_j|^2.$$
\item Likewise, given a scalar function $f:\mathbb R^{2+1}\to \mathbb R$, we define
$$|\rd f|^2:= |\rd_t f|^2+\sum_{i=1}^2 |\rd_{x^i} f|^2, \quad |\rd_i f|^2=\sum_{j=1}^2 |\rd_{x^j} f|^2.$$
\item A similar notation will be used for higher coordinate derivatives (even though they are not tensors), i.e.
$$|\rd^k f|^2:= \sum_{\mu_1,\dots, \mu_k=t,1,2} |\rd_{\mu_1,\dots, \mu_k}^k f|^2.$$
\end{itemize}

{\bf Constants.} Conventions for constants will be discussed in the beginning of Section~\ref{sec:preliminaries}.

{\bf Differential operators.}
\begin{itemize}
\item $\Delta$ denotes the spatial Laplacian on $\m R^2$ with respect to the \underline{spatial Euclidean metric}, i.e.
$$\Delta u = \sum_{i=1}^2 \rd_i^2 u.$$
\item $\Box_{g_0}$ and $\Box_{g_n}$ denote the Laplace--Beltrami operators with respect to $g_0$ and $g_n$ respectively.
 (see also \eqref{eq:Box.g0} and \eqref{eq:Box.gn}).
\item {$\Box_{g_0,A}$ and $\Box_{g_n,A}$ are operators to be defined respectively in \eqref{def:Box.g0} and \eqref{def:Box.gn}.}
\item $(e_0)_0$ and $(e_0)_n$ denote the vector fields $(e_0)_0=\rd_t-\beta^i_0\rd_{x^i}$ and $(e_0)_n=\rd_t-\beta_n^i\rd_{x^i}$ respectively (where {$\beta_n$ and $\bt_0$} will be introduced in \eqref{g.form}). 
\item $\mathfrak L$ denotes the Euclidean conformal Killing operator acting on vectors on $\m R^2$ to give a symmetric traceless (with respect to the Euclidean metric $\delta$) covariant $2$-tensor, i.e., 
$$(\mathfrak L\eta)_{ij}:=\delta_{j\ell}\rd_i\eta^\ell+\delta_{i\ell}\rd_j\eta^\ell-\delta_{ij}\rd_k\eta^k.$$
\end{itemize}

\textbf{Fourier transforms.} We will denote \underline{spacetime} Fourier transform by $\mbox{ }\widehat{ }\mbox{ }$ and \underline{spatial} Fourier transform by $\mathcal F_{\mathrm{spa}}$. We will take the following normalizations: 
$$\widehat{f}(\xi):= \f{1}{(2\pi)^{\f 32}}\int_{\mathbb R^{2+1}} e^{-ix^\mu \xi_\mu} f(x) \, \ud x,$$
$$\mathcal F_{\mathrm{spa}}(f)(t,\xi_k):=\f{1}{2\pi}\int_{\mathbb R^{2}} e^{-ix^k \xi_k} f(t,x^j)\,\ud x^1\,\ud x^2.$$
Fourier multipliers will be denoted as follows for $m:\mathbb R^{2+1}\to \mathbb R$:
$$(m(\f 1i \nabla) f)(x) := \f{1}{(2\pi)^3} \int_{\mathbb R^{2+1}} e^{i(x^\mu-y^\mu) \xi_\mu} m(\xi) f(y)\,\ud y.$$

{\bf Functions spaces.} Unless otherwise stated, all function spaces will be understood on $\mathbb R^{2+1}$. Define the following norms for a scalar function $f:\mathbb R^{2+1}\to \mathbb R$:
$$\|f\|_{L^p}:= (\int_{\mathbb R^{2+1}} |f|^p(x)\,\ud x)^{\f 1p},\quad p \in [1,+\infty),\qquad \|f\|_{L^\i}:= \mathrm{esssup}_{x\in \mathbb R^{2+1}} |f|(x),$$
$$\|f\|_{W^{m,p}}:= \sum_{|\alp|\leq m} \|\rd^\alp f\|_{L^p},\quad m\in \mathbb N\cup \{0\},\,p\in [1,+\infty).$$
Define also the corresponding function spaces in the obvious way. We will denote $H^m:= W^{m,2}$. Define also the norm
$$\|f\|_{H^{-1}}:= \Big( \int_{\mathbb R^{2+1}} (1+|\xi|^2)^{-1} |\widehat{f}|^2(\xi)\,\ud \xi \Big)^{\f 12}$$
and the corresponding function space. 

We will also use the above function spaces for tensors on $\m R^{2+1}$, where the norms in the case of tensors are understood componentwise (with respect to the $(t,x^1,x^2)$ coordinates).

\section{Main results}\label{sec:main.results}

Let $\mathcal M := (0,T)\times \mathbb R^2$. Suppose $\{(\psi_n,\om_n,g_n)\}_{n=1}^{+\infty}$ is a sequence such that $\psi_n$, $\om_n$ are $C^4$ real-valued function on $\mathcal M$ and $g_n$ is a $C^4$ Lorentzian metric on $\mathcal M$ satisfying the following four conditions:
\begin{enumerate}
\item (Solving the equations)
$(\psi_n,\om_n,g_n)$ satisfies \eqref{eq:U1vac} for all $n\in \mathbb N$.
\item (Gauge condition)
The metric $g_n$ is put into a form satisfying \eqref{g.form} and \eqref{trace.free} for all $n\in \mathbb N$.
\item (Local uniform convergence) There exists a limit $(\psi_0,\om_0,g_0)$ which is smooth and $g_0$ satisfies \eqref{g.form} and \eqref{trace.free}. 
Assume that the following convergences hold: 
\begin{enumerate}
\item $\psi_n\to \psi_0$, $\om_n\to \om_0$ uniformly on (spacetime) compact sets.
\item For $\mfg \in \{\log N,\, \bt^i,\, \gamma\}$ (being the metric components; cf.~\eqref{g.form}), $\mfg_n \to \mfg_0$ uniformly on compact sets.
\end{enumerate}
\item (Weak convergence of the derivatives) Let $p_0 \in (\f 83, +\infty)$. With $(\psi_0,\om_0,g_0)$ as above, the following convergences hold:
\begin{enumerate}
\item $\rd\psi_n\rightharpoonup \rd\psi_0$, $\rd\om_n\rightharpoonup \rd\om_0$ weakly in $L^{p_0}_{\mathrm{loc}}$.
\item For $\mfg \in \{\log N,\, \bt^i,\, \gamma\}$, $\rd\mfg_n \rightharpoonup \rd\mfg_0$ weakly in $L^{p_0}_{\mathrm{loc}}$.
\end{enumerate}
\end{enumerate}

\begin{thm}\label{thm:prelim}
Given $\{(\psi_n,\om_n,g_n)\}_{n=1}^{+\infty}$ and $(\psi_0,\om_0,g_0)$ such that the conditions (1)--(4) above hold. Then there exists a non-negative Radon measure $\ud \nu$ on $S^*\mathcal M$ such that $(\mathcal M,\psi_0,\om_0,g_0,\ud\nu)$ satisfy the following conditions:
\begin{enumerate}
\item $\ud \nu$ is supported on $\{(x,\xi) \in S^*\mathcal M: g^{-1}_0(\xi,\xi) = 0\}$;
\item the following system of equations hold:
\begin{equation}\label{eq:U1vac.vlasov.again}
\left\{
\begin{array}{l}
\Box_{g_0} \psi_0 + \f 12 e^{-4\psi_0} g^{-1}_0(\ud \om_0, \ud \om_0) = 0,\\
\Box_{g_0} \om_0 - 4 g_0^{-1}(\ud \om_0 , \ud \psi_0) = 0,\\
\int_{\mathcal M} \mathrm{Ric}(g_0)(Y,Y) \, \mathrm{dVol}_{g_0}=  \int_{\mathcal M} [2 (Y \psi_0)^2 + \f 12 e^{-4\psi_0} (Y\om_0)^2] \, \mathrm{dVol}_{g_0} + \int_{S^*\mathcal M} \langle \xi,Y\rangle^2 \, \f{\ud \nu}{|\xi|^2},
\end{array}
\right.
\end{equation}
for every $C^\infty_c$ vector field $Y$.
\end{enumerate}

In particular, the effective stress-energy-momentum tensor $T_{\mu\nu}$ is traceless, non-negative and obeys the dominant energy condition.
\end{thm}

The above theorem has the advantage that the assumptions are very weak. On the other hand, we also do not get the full Burnett's conjecture in that we do not show that the limit is isometric to a solution to the Einstein--massless Vlasov system. In order to obtain the stronger result, we impose the following additional assumption:

\begin{enumerate}\setcounter{enumi}{4}
\item (Estimates)
For every compact subset $K\subset \mathcal M$, there exists a sequence $\{\lambda_n\}_{n=1}^{\infty}\subset (0,1]$ (depending on $K$) with $\lim_{n\to +\infty} \lambda_n = 0$ such that for $\mathfrak g \in \{\log N,\, \bt^i, \, \gamma\}$,
\begin{equation}\label{assumption.0}
\sup_n \lambda_n^{-1} \|(\psi_n-\psi_0,\om_n-\om_0,\mathfrak{g}_n-\mathfrak{g}_0)\|_{L^\i(K)} <+\infty,
\end{equation}
\begin{equation}\label{assumption.1}
\sup_n  \|(\partial \psi_n,\rd\om_n, \rd\mathfrak{g}_n)\|_{L^\i(K)} <+\infty.
\end{equation}
\begin{equation}\label{assumption.2}
\sup_n  \lambda_n^{k-1} \|(\partial^k \psi_n,\rd^k\om_n,\rd^k \mathfrak{g}_n)\|_{L^\i(K)} <+\infty, \quad \mbox{for $k=2,3,4$}.
\end{equation}
\end{enumerate}

\begin{thm}\label{thm:main}
Given $\{(\psi_n,\om_n,g_n)\}_{n=1}^{+\infty}$ and $(\psi_0,\om_0,g_0)$ such that the conditions (1)--(5) above hold. Then there exists a non-negative Radon measure $\ud\nu$ on $S^*\mathcal M$ such that $(\mathcal M,\psi_0,\om_0,g_0,\ud\nu)$ is a radially-averaged measure solution to the restricted Einstein--massless Vlasov equations in $\mathbb U(1)$ symmetry in the sense of Definition~\ref{def:the.final.def}.
\end{thm}

\begin{rk}
Even though it is most convenient for the proof to formulate Theorem~\ref{thm:main} so that the Einstein part of the system is satisfied in the weak sense (see \eqref{eq:U1vac.vlasov}), it follows a posteriori that the Einstein part of the system is also satisfied classically for an appropriately defined stress-energy-momentum tensor. 

This can be formulated as follows. Let ${\bf\pi}:S^*\mathcal M\to \mathcal M$ be the natural projection map. It follows from Theorem~\ref{thm:main} that after defining 
$$T_{\alp\bt}(p) := \liminf_{r\to 0} \frac{3}{4\pi r^3}\int_{{\bf\pi}^{-1}(B(p,r))} \frac{\xi_\alp\xi_\bt}{|\xi|^2}\,\ud \nu$$
(where $B(p,r)$ is the coordinate ball), $T_{\alp\bt}$ is continuous and that the Einstein equation
$$\mathrm{Ric}_{\alp\bt}(g)= 2\partial_\alp \psi \partial_\bt \psi + \f 12 e^{-4\psi} \rd_\alp \om \rd_\bt \om + T_{\alp\bt}$$
holds classically.
\end{rk}

\subsection{Reduction to compact sets}\label{sec:compact.reduction}

It will be technically convenient to reduce our theorems to corresponding cut-off versions. 

Fix an open and precompact $\Omega\subset \mathcal M$. Let $\Omega'\subset \mathcal M$ be an open and precompact set containing $\overline{\Omega}${, the closure of $\Omega$}. Let $\chi$ be a non-negative function in $C^\infty_c$ such that
\begin{equation}\label{def:chi}
\chi = 1 \mbox{ on $\Omega$},\quad \mathrm{supp}(\chi)\subset \Omega'.
\end{equation}  
We will show that (cf.~Section~\ref{sec:existence} below) for every such $\Omega$, $\Omega'$ and $\chi$, there exists a non-negative Radon measure $\ud\nu$ on $S^*\mathbb R^{2+1}$ such that for some subsequence $n_k$, the following holds for any $0$-th order pseudo-differential operator $A$ with principal symbol $a$ (cf.~Section~\ref{sec:PSIDOs}):
\begin{equation}\label{eq:MDM.def.structure}
\begin{split}
&\: \lim_{k\to +\infty}2\int_{\mathbb R^{2+1}} \rd_\alp \big(\chi(\psi_{n_k}-\psi_0)\big)\Big[A\rd_\bt(\chi(\psi_{n_k}-\psi_0))\Big] \,\mathrm{dVol}_{g_{n_k}} \\
&\: \quad + \lim_{k\to +\infty} \f 12 \int_{\mathbb R^{2+1}} e^{-4\psi_0} \rd_\alp \big(\chi(\om_{n_k}-\om_0) \big)\Big[A\rd_\bt(\chi(\om_{n_k}-\om_0)) \Big] \,\mathrm{dVol}_{g_{n_k}} = \int_{S^*\mathbb R^{2+1}} a \xi_\alp \xi_\bt \f{\ud \nu}{|\xi|^2}.
\end{split}
\end{equation}

We now claim that in order to prove Theorems~\ref{thm:prelim} and \ref{thm:main}, it suffices to show that for every $\Omega$, $\Omega'$ and $\chi$ as above, $(\Omega,\psi_0,\om_0,g_0,\ud\nu\restriction_{\Omega})$ verifies the conclusion of Theorems~\ref{thm:prelim} and \ref{thm:main}. More precisely,
\begin{proposition}\label{prop:compact.reduction}
Let $(\Omega,\psi_0,\om_0,g_0,\ud\nu)$ be as defined above.
\begin{enumerate}
\item Suppose that under the assumptions of Theorem~\ref{thm:prelim}, for every $\Omega$, $\Omega'$ and $\chi$ above, $\ud\nu$ is supported in $\{(x,\xi) \in S^*\mathcal M: g^{-1}_0(\xi,\xi) = 0\}$ and \eqref{eq:U1vac.vlasov.again} holds in $\Omega$ with any $Y\in C^\infty_c(\Omega)$. Then Theorem~\ref{thm:prelim} holds.
\item Suppose that under the assumptions of Theorem~\ref{thm:main}, for every $\Omega$, $\Omega'$ and $\chi$ above, $(\Omega,\psi_0,\om_0,g_0,\ud\nu\restriction_{\Omega})$ is a radially-averaged measure solution to the restricted Einstein--massless Vlasov equations in $\mathbb U(1)$ symmetry in the sense of Definition~\ref{def:the.final.def}. Then Theorem~\ref{thm:main} holds.
\end{enumerate}
\end{proposition}
\begin{proof}
We will define a Radon measure on all of $\mathcal M$ under the assumptions.

Let $\{\Omega_i\}_{i=1}^{+\infty}$ be an exhaustion of $\mathcal M$, i.e.~$\Omega_i \subset \Omega_{i+1}$ and $\cup_{i=1}^{+\infty} \Omega_i = \mathcal M$. Define $\chi_i$ by \eqref{def:chi} with $\Omega = \Omega_i$ and $\Omega' = \Omega_{i+1}$.

We now proceed inductively in $i$. First, define $\ud \nu_1$ as a Radon measure on $\Omega_1$ as in the assumption of the proposition. Now, for ever $i\geq 1$, suppose we are given a Radon measure $\ud \nu_i$ on $\Omega_i$ so that the convergence in \eqref{eq:MDM.def.structure} holds and the properties as in the statement of the proposition are satisfied. We can then define a Radon measure $\ud \nu_{i+1}$ on $\Omega_{i+1}$ by considering a further subsequence of $(\psi_{n_k}, \om_{n_k}, g_{n_k})$ as in the assumption of the proposition. {Notice that by definition $\ud \nu_{i+1}\restriction_{\Omega_i} = \ud\nu_i$ for all $i \in \mathbb N$.} By considering a diagonal subsequence, we {therefore obtain} that there is a \emph{fixed} subsequence $n_k$ such that the following holds for every $i\in \mathbb N$:
\begin{equation*}
\begin{split}
&\: \lim_{k\to +\infty}2\int_{\mathbb R^{2+1}} \rd_\alp \big(\chi_i(\psi_{n_k}-\psi_0) \big) \big[A\rd_\bt(\chi_i(\psi_{n_k}-\psi_0)) \big] \,\mathrm{dVol}_{g_{n_k}} \\
&\: \quad + \lim_{k\to +\infty} \f 12 \int_{\mathbb R^{2+1}} e^{-4\psi_0} \rd_\alp \big(\chi_i(\om_{n_k}-\om_0) \big)\big[A\rd_\bt(\chi_i(\om_{n_k}-\om_0)) \big] \,\mathrm{dVol}_{g_{n_k}} = \int_{S^*\mathbb R^{2+1}} a \xi_\alp \xi_\bt \f{\ud \nu_i}{|\xi|^2},
\end{split}
\end{equation*}
 
Define $\ud\nu_\infty$ as follows. Let $\varphi \in C_c(S^*\mathcal M)$. Then there exists $\Omega_i$ such that $\mathrm{supp}\varphi\subset S^*\Omega_i$. Define
$$\ud\nu_\infty(\varphi):= \ud\nu_{i}(\varphi).$$
Note that this is well-defined (and independent of the particular choice of $i$).

In each of the cases (1) and (2), it is then easy to verify that $(\mathcal M, \psi_0,\om_0,g_0,\ud\nu_\infty)$ obeys the desired conclusion. \qedhere
\end{proof}

In view of the above proposition, \textbf{from now on we fix $\Omega$, $\Omega'$ and $\chi$ as above.} It will suffice to prove that the conditions in Proposition~\ref{prop:compact.reduction} hold.

It will be convenient to fix also two open, precompact sets $\Omega'' \subset \Omega'''\subset \mathcal M$ such that $\overline{\Omega'}\subset \Omega''$ and $\overline{\Omega''}\subset \Omega'''$. Define
\begin{equation}\label{def:tildechi}
\widetilde{\chi} = 1 \mbox{ on $\Omega''$},\quad \mathrm{supp}(\widetilde{\chi})\subset \Omega'''.
\end{equation}

\section{Preliminaries on pseudo-differential operators and microlocal defect measures}\label{sec:PSIDOs}

In this section, we recall useful notions on pseudo-differential operators and microlocal defect measures. Everything in this section is standard and is mainly included to fix notations.

For the remainder of this section, fix $k \in \mathbb N$ (which will be taken as $3=2+1$ in later sections). Denote by $T^*\mathbb R^k$ the cotangent bundle of $\mathbb R^k$ with coordinates $(x,\xi)\in \mathbb R^k \times \mathbb R^k$.

\begin{df}
\begin{enumerate}
\item For $m\in \mathbb Z$, define the symbol class 
$$S^m:= \{a:T^*\mathbb R^k \to \mathbb C: a\in C^{\infty},\, \forall \alp,\bt,\, \exists C_{\alp,\bt}>0,\, |\rd_x^\alp \rd_\xi^\bt a(x,\xi)| \leq C_{\alp,\bt} (1+|\xi|)^{m-|\bt|} \}.$$
\item Given a symbol $a \in S^m$, define the operator $\mathrm{Op}(a):\mathcal S(\mathbb R^k)\to \mathcal S(\mathbb R^k)$ by
$$(\mathrm{Op}(a)u)(x) := \f{1}{(2\pi)^k}\int_{\mathbb R^k} \int_{\mathbb R^k} e^{i(x-y)\cdot \xi} a(x,\xi) u(y) \,\ud y\, \ud \xi.$$
For {$A= \mathrm{Op}(a)$ as} above, we say that $A$ is a pseudo-differential operator of order $m$ with symbol $a$. If moreover {$a(x,\xi) = a_{\mathrm{prin}}(x,\xi)\chi(\xi) + a_{\mathrm{error}}$, where $a_{\mathrm{prin}}(x,\lambda\xi) = \lambda^m a(x,\xi)$ for all $\lambda>0$, $\chi(\xi)$ is a cutoff $\equiv 1$ for large $|\xi|$ and $\equiv 0$ near $0$, and $a_{\mathrm{error}} \in S^{m-1}$,} we say that ${a_{\mathrm{prin}}}$ is the principal symbol 
of $A$.
\end{enumerate}
\end{df}

We record for convenience some standard facts.
\begin{lemma}\label{lem:PSIDOs}
\begin{enumerate}
\item (\cite[Theorem~2, p.237]{Stein}) Let $a_1\in S^{m_1}$, $a_2 \in S^{m_2}$. Then $\exists c\in S^{m_1+m_2}$ such that $\mathrm{Op}(a_1) \circ\mathrm{Op}(a_2) = \mathrm{Op}(c)$, where 
$$c(x,\xi) - a_1(x,\xi) a_2(x,\xi) \in S^{m_1+m_2-1}.$$
\item (\cite[Theorem~2, p.237]{Stein}) Let $a_1\in S^{m_1}$, $a_2 \in S^{m_2}$. Then $\exists c\in S^{m_1+m_2-1}$ such that $\mathrm{Op}(a_1) \circ\mathrm{Op}(a_2) - \mathrm{Op}(a_2) \circ\mathrm{Op}(a_1) = \mathrm{Op}(c)$, where 
$$c(x,\xi) + i\{ a_1, a_2\} \in S^{m_1+m_2-2},\quad \{a_1,a_2\}:= \rd_{\xi_\mu} a_1 \rd_{x^\mu} a_2 -  \rd_{x^\mu} a_1  \rd_{\xi_\mu} a_2.$$
\item (\cite[Proposition, p.259]{Stein}) Let $a \in S^{m}$. Then $\mathrm{Op}(a)^*$ (the $L^2$-adjoint of $\mathrm{Op}(a)$) satisfies
$$\mathrm{Op}(a)^*- \mathrm{Op}(\overline{a}) \in S^{m-1}.$$
\item (Calder\'on--Zygmund theorem \cite[Proposition~5, p.251]{Stein}) A pseudo-differential operator $A$ of order $m$ extends to a bounded map $W^{n+m,p}(\mathbb R^k) \to W^{n,p}(\mathbb R^k)$, $\forall n \in \mathbb N\cup \{0\}$, $\forall m\in \mathbb Z$, $\forall p\in (1,+\infty)$.
\item (Rellich--Kondrachov theorem) A pseudo-differential operator $A$ of order $-1$ extends to a \underline{compact} map $:L^2(\mathbb R^k) \to L^2_{\mathrm{loc}}(\mathbb R^k)$.
\item (Calder\'on commutator theorem \cite[Corollary, p.309]{Stein})
Let $u:\mathbb R^n\to \mathbb R$ be a Lipschitz function for which there exists $M>0$ so that $|u(x)-u(y)|\leq M|x-y|$ for all $x,\,y\in \mathbb R^n$. Let $T$ be a pseudo-differential operator of order $1$. Then $[T, u] \in \mathcal B(L^2(\mathbb R^n), L^2(\mathbb R^n))$, i.e.~that it is a bounded linear map on $L^2(\mathbb R^n)$. In {addition}\footnote{The precise statement in \cite{Stein} only asserts that $[T, u] \in \mathcal B(L^2(\mathbb R^n), L^2(\mathbb R^n))$ (without explicitly saying that the {operator norm} is proportional to $M$). Nevertheless, \eqref{eq:quan.Calderon} follows from the closed graph theorem. Let $(\mathbf{Lip},\,\|\cdot \|_{\mathbf{Lip}})$ be the Banach space of equivalence classes of Lipschitz functions, where two functions are equivalent if they differ by a constant, and $\|u \|_{\mathbf{Lip}}:= \sup_{x\neq y} \f{|u(x)-u(y)|}{|x-y|}$. The corollary on p.309 in \cite{Stein} implies that there is a map $\Phi:\mathbf{Lip} \to \mathcal B(L^2(\mathbb R^n), L^2(\mathbb R^n))$ given by $[u]\mapsto [T,u]$. By the closed graph theorem, in order to obtain \eqref{eq:quan.Calderon}, it suffices to show that if $\lim_{j\to+\infty}\| [u_j] \|_{\mathbf{Lip}}= 0$ and $\lim_{j\to+\infty} \|[T, u_j] - S\|_{\mathcal B(L^2(\mathbb R^n), L^2(\mathbb R^n))} = 0$ for some $S \in \mathcal B(L^2(\mathbb R^n), L^2(\mathbb R^n))$, then $S= 0$. To show this, pick a representative $u_j$ such that $u_j(0) = 0$. In particular it follows that $\||x|^{-1}u_j\|_{L^\infty} \to 0$ as $j\to +\infty$. Now for any $\varphi\in L^2(\mathbb R^n)$, $T(u_j\varphi)$ and $u_j T(\varphi)$ both tend to $0$ in the sense of distribution{s} as $j\to +\infty$. Therefore, $S=0$ as required.}, for every $f\in \mathcal S(\mathbb R^n)$,
\begin{equation}\label{eq:quan.Calderon}
\|T(uf) - u(Tf)\|_{L^2(\mathbb R^n)} \ls M \|f\|_{L^2(\mathbb R^n)},
\end{equation}
where the implicit constant depends only on $T$.
\end{enumerate}
\end{lemma}

We now turn to the discussion of microlocal defect measures, following \cite{Gerard} (see also \cite{Tartar}). We first need some preliminary definitions.

Let $S^*\mathbb R^k$ be the cosphere bundle given by $S^*\mathbb R^k:= (T^*\mathbb R^k\setminus\{0\})/\sim$, where $(x,\xi)\sim (y,\eta)$ if and only if $x=y$ and $\xi = \lambda \eta$ for some $\lambda>0$. From now on, we identify a function on $S^*\mathbb R^k$ with a function on $T^*\mathbb R^k\setminus\{0\}$ which is homogeneous of order $0$ in $\xi$, i.e.~$a(x,\lambda\xi) = a(x,\xi)$, $\forall \lambda>0$.

\begin{df}\label{def:nonneg.measure}
We say that $\ud\mu$ is a non-negative $(N\times N)$-complex-matrix-valued Radon measure on $S^*\mathbb R^k$ if  $\ud\mu$ is a map $\ud \mu: C_c(S^*\mathbb R^k)\to \mathbb C^{N\times N}$ 
\begin{enumerate}
\item obeying the estimate $\|\ud \mu (\varphi) \|\leq C_K \|\varphi\|_{C(K)}$ for every compact set $K\subset S^*\mathbb R^k$ (for some $C_K>0$ depending on $K$), and
\item satisfying $\ud \mu(\varphi)$ is a positive semi-definite Hermitian matrix whenever $\varphi$ is a non-negative function.
\end{enumerate}
\end{df}

\begin{df}
Let $\ud\mu$ be a non-negative $(N\times N)$-complex-matrix-valued Radon measure on $S^*\mathbb R^k$ and $\slashed{a}:S^*\mathbb R^k\to \mathbb C^{N\times N}$ be a continuous matrix-valued function on $S^*\mathbb R^k$.

Define $\mathrm{tr}\, (\slashed{a}(x,\xi) \,\ud \mu)$ to be the (scalar-valued) Radon measure on $S^*\mathbb R^k$ given by
$$(\mathrm{tr}\, (\slashed{a}(x,\xi) \,\ud \mu))(\varphi) := \mathrm{tr}\,[\slashed{a}(x,\xi) \cdot (\ud \mu(\varphi))].$$
\end{df}

\begin{thm}[Existence of microlocal defect measures, Theorem~1\footnote{Note that this is a specialization of the original theorem of G\'erard. In the original paper, the domain is any open set in $\mathbb R^k$ and $u_n$ may take value in any separable Hilbert space, instead of $\mathbb C^N$.} in \cite{Gerard}]\label{thm:existenceMDM}
Suppose $\{u_n\}_{n=1}^{+\infty} \in L^2(\mathbb R^k;\mathbb C^N)$ be a bounded sequence such that $u_n \rightharpoonup 0$ weakly in $L^2(\mathbb R^k;\mathbb C^N)$.

Then there exists a subsequence $\{u_{n_k}\}_{k=1}^{+\infty}$ and a non-negative $(N\times N)$-complex-matrix-valued Radon measure $\ud\mu$ on $S^*\mathbb R^k$ such that for every $\mathbb C^{N\times N}$-valued pseudo-differential operator $\slashed{A}$ of order $0$ with principal symbol $\slashed{a} \in C_c(S^*\mathbb R^k; \mathbb C^{N\times N})$,
$$\lim_{k\to +\infty} \int_{\mathbb R^k} \langle \slashed{A} u_{n_k}, u_{n_k} \rangle_{\mathbb C^N} \,\ud x= \int_{S^*\mathbb R^k} \mathrm{tr}\, (\slashed{a}(x,\xi) \,\ud \mu).$$
\end{thm}

The measure $\ud \mu$ in Theorem~\ref{thm:existenceMDM} is called a \emph{microlocal defect measure}. Following \cite{Gerard}, if the conclusion of Theorem~\ref{thm:existenceMDM} holds for $\{u_n\}_{n=1}^{+\infty}$ (without passing to a subsequence), we say that $\{u_n\}_{n=1}^{+\infty}$ is a \emph{pure} sequence.

\begin{thm}[Localization of microlocal defect measures, Corollary~2.2 in \cite{Gerard}]\label{thm:localization}
Let $\{u_n\}$ be a pure sequence of $L^2(\mathbb R^k,\mathbb C^N)$, of microlocal defect measure $\ud \mu$. Let $P$ be an $m$-th order differential operator with principal symbol $p = \sum_{|\alp|= m} a_\alp (i\xi)^\alp$ for some smooth $(N\times N)$-matrices $a_\alp$. If $\{Pu_n\}_{n=1}$ is relatively compact in $H^{-m}_{\mathrm{loc}}(\mathbb R^k,\mathbb C^N)$, then 
$p\,\ud\mu = 0$.
\end{thm}

\section{Microlocal defect measures for $\psi$ and $\om$}\label{sec:MDM.psi.om}

We begin to prove Theorem~\ref{thm:prelim} by studying the properties of the microlocal defect measures. For the remainder of this section, we work under the assumption of Theorem~\ref{thm:prelim}.

\subsection{Existence of the microlocal defect measures}\label{sec:existence}

Consider now the sequence of functions $\chi(\psi_n-\psi_0)$ and $\chi(\om_n-\om_0)$ (cf.~\eqref{def:chi}). We are now in a setting to apply the existence theorem (Theorem~\ref{thm:existenceMDM}) to obtain microlocal defect measures.

\begin{proposition}[Existence of microlocal defect measures]\label{prop:existence.MDM}
There exist Radon measures $\ud \sigma^\psi_{\alp\bt}$, $\ud \sigma^\om_{\alp\bt}$ and $\ud \sigma^{\mathrm{cross}}_{\alp\bt}$ on $S^*\mathbb R^{2+1}$ such that for any zeroth order (scalar) pseudo-differential operators $\{A^{\alp\bt}\}_{\alp,\bt = t,1,2}$ with principal symbols $\{a^{\alp\bt}\}_{\alp,\bt = t,1,2}$, the following holds up to a subsequence (which we do not relabel):
$$\lim_{n\to \infty} \int_{\mathbb R^{2+1}} \rd_\alp (\chi(\psi_n-\psi_0)) A^{\alp\bt} \rd_\bt(\chi(\psi_n-\psi_0)) \, \mathrm{dVol}_{g_0} = \int_{S^*\mathbb R^{2+1}} a^{\alp\bt} \,\ud \sigma^\psi_{\alp\bt},$$
$$\lim_{n\to \infty} \int_{\mathbb R^{2+1}} \rd_\alp (\chi(\om_n-\om_0)) A^{\alp\bt} \rd_\bt(\chi(\om_n-\om_0)) \, \mathrm{dVol}_{g_0} = \int_{S^*\mathbb R^{2+1}} a^{\alp\bt} \,\ud \sigma^\om_{\alp\bt},$$
$$\lim_{n\to \infty} \int_{\mathbb R^{2+1}} \rd_\alp (\chi(\psi_n-\psi_0)) A^{\alp\bt} \rd_\bt(\chi(\om_n-\om_0)) \, \mathrm{dVol}_{g_0} = \int_{S^*\mathbb R^{2+1}} a^{\alp\bt} \,\ud \sigma^{\mathrm{cross}}_{\alp\bt},$$
$$\lim_{n\to \infty} \int_{\mathbb R^{2+1}} \rd_\alp (\chi(\om_n-\om_0)) A^{\alp\bt} \rd_\bt(\chi(\psi_n-\psi_0)) \, \mathrm{dVol}_{g_0} = \int_{S^*\mathbb R^{2+1}} a^{\alp\bt} \,(\ud \sigma^{\mathrm{cross}})^*_{\alp\bt},$$
where ${ }^*$ denotes the Hermitian conjugate.

Moreover, $\ud \sigma^\psi_{\alp\bt}$ and $\ud \sigma^\om_{\alp\bt}$ are non-negative in the sense of Definition~\ref{def:nonneg.measure}.
\end{proposition}
\begin{proof}
Applying Theorem~\ref{thm:existenceMDM} with
\[ u_n = 
\left[
\begin{array}{c}
\rd_t (\chi(\psi_n-\psi_0)) \\
\rd_{x^1} (\chi(\psi_n-\psi_0)) \\
\rd_{x^2} (\chi(\psi_n-\psi_0)) \\
\rd_t (\chi(\om_n-\om_0)) \\
\rd_{x^1} (\chi(\om_n-\om_0)) \\
\rd_{x^2} (\chi(\om_n-\om_0))
\end{array}
\right],
\]
we obtain a non-negative $(6\times 6)$-complex-matrix-valued Radon measure $\ud\mu$. Since $\ud\mu$ is Hermitian by Theorem~\ref{thm:existenceMDM}, $\ud\mu$ takes a block diagonal form as follows
\[
\ud \mu=
\left[
\begin{array}{c|c}
\ud \sigma^\psi & \ud \sigma^{\mathrm{cross}} \\
\hline
(\ud \sigma^{\mathrm{cross}})^* & \ud \sigma^\om
\end{array}
\right].
\]
It is then easy to check that the components $\ud \sigma^\psi_{\alp\bt}$, $\ud \sigma^\om_{\alp\bt}$ and $\ud \sigma^{\mathrm{cross}}_{\alp\bt}$ of the corresponding measures have the properties as claimed. (Note that we have in particular used $\overline{\rd_\alp (\chi(\psi_n-\psi_0))} = \rd_\alp (\chi(\psi_n-\psi_0))$, etc.~in the expressions in the proposition.) \qedhere
\end{proof}

Without loss of generality (by passing to a subsequence if necessary), \textbf{we will assume from now on that the sequence is pure.}

\subsection{First properties of the microlocal defect measures}

In this subsection, we prove some properties of the microlocal defect measures.
\begin{proposition}\label{prop:dsigma.is.symmetric}
$\ud \sigma^{\mathrm{cross}}_{\alp\bt}$ is symmetric, i.e.
$$\ud \sigma^{\mathrm{cross}}_{\alp\bt} = \ud \sigma^{\mathrm{cross}}_{\bt\alp}.$$
\end{proposition}
\begin{proof}
This amounts to 
\begin{equation*}
\begin{split}
&\: \lim_{n\to +\infty} \int_{\mathbb R^{2+1}} \rd_\alp (\chi(\psi_n-\psi_0)) A^{\alp\bt} \rd_\bt(\chi(\om_n-\om_0)) \, \mathrm{dVol}_{g_0} \\
= &\: \lim_{n\to +\infty} \int_{\mathbb R^{2+1}} \rd_\bt (\chi(\psi_n-\psi_0)) A^{\alp\bt} \rd_\alp(\chi(\om_n-\om_0)) \, \mathrm{dVol}_{g_0}.
\end{split}
\end{equation*}
This can be seen by noting that $[A^{\alp\bt},\rd_\alp] : {L^2} \to L^2_{\mathrm{loc}}$ is {bounded by Lemma~\ref{lem:PSIDOs}.2 (and that the corresponding contribution $\to 0$ since $\om_n \to \om_0$ in $L^2$)}, integrating by parts and using assumptions (3) and (4) of Theorem~\ref{thm:prelim}. \qedhere
\end{proof}

\begin{proposition}[Microlocal defect measures are effectively given by $\ud \nu^\psi$ and $\ud \nu^\om$]\label{prop:nu}
There exist non-negative Radon measures $\ud \nu^\psi$, $\ud \nu^\om$ on $S^*\mathbb R^{2+1}$ such that 
$$\ud \sigma^\psi_{\alp\bt} = \f{\xi_\alp\xi_\bt}{|\xi|^2}\ud \nu^\psi,\quad  \ud \sigma^\om_{\alp\bt} = \f{\xi_\alp\xi_\bt}{|\xi|^2}\ud \nu^\om,$$
where $|\xi|^2:= \sum_{\mu=0}^2 |\xi_\mu|^2$.
\end{proposition}
\begin{proof}
We will focus on $\ud\nu^\psi$ in the exposition. $\ud\nu^{\om}$ can be treated similarly.

\pfstep{Step~1: Defining an auxiliary measure $\ud \rho^\psi_\bt$} Using the identity $\rd_\alp\rd_\mu \psi_n = \rd_\mu\rd_\alp \psi_n$ and Theorem~\ref{thm:localization}, it follows that for every $\bt$, 
$$\f{\xi_\mu}{|\xi|} \ud \sigma^\psi_{\alp\bt} = \f{\xi_\alp}{|\xi|} \ud \sigma^\psi_{\mu\bt}.$$
This implies that
\begin{enumerate}
\item $\f{\ud \sigma^\psi_{\alp\bt}}{\xi_\alp}$ (to be understood \underline{without} summing repeated indices) is a well-defined Radon measure for every $\bt$. (To see this, note that at each point in $T^*\mathcal M \setminus\{0\}$, some component $\xi_\alp \neq 0$.)
\item $\f{\ud \sigma^\psi_{\alp\bt}}{\xi_\alp} = \f{\ud \sigma^\psi_{\alp'\bt}}{\xi_{\alp'}}$ for every $\alp$, $\alp'$.
\end{enumerate}
With the above observations, we can thus define the measure $\ud \rho^\psi_\bt := \f{|\xi| \ud \sigma^\psi_{\alp\bt}}{\xi_\alp}$.

\pfstep{Step~2: Defining $\ud \nu^{\psi}$} Since $\ud \sigma^\psi_{\alp\bt}$ is Hermitian (by Proposition~\ref{prop:existence.MDM}), for $\ud \rho^\psi_\bt$ defined as in Step~1,
\begin{equation}\label{the.rho.identity}
\xi_\alp \ud \rho^\psi_\bt = \xi_\bt \overline{\ud \rho^\psi_\alp}.
\end{equation}
Arguing as in Step~1 above, we know that $\f{\ud \rho^\psi_\alp}{\xi_\alp}$ is well-defined. We define 
$$\ud \nu^\psi := \f{|\xi|\ud \rho^\psi_\alp}{\xi_\alp}.$$

\pfstep{Step~3: Non-negativity of $\ud \nu^{\psi}$} Finally, using Proposition~\ref{prop:existence.MDM}, one sees that $\ud \nu^{\psi}$ is non-negative. \qedhere
\end{proof}

We record the following result, which follows from Propositions~\ref{prop:existence.MDM}, \ref{prop:nu} and simple algebraic manipulations.
\begin{cor}\label{cor:nu}
For $\ud \nu^\psi$, $\ud \nu^\om$ as in Proposition~\ref{prop:nu}, it holds that for every zeroth pseudo-differential operator $A$ with principal symbol $a$ which is {homogeneous} of order $0$, and supported in $S^*\Omega$,
\begin{equation}\label{eq:cor.nu.1}
\int_{\mathbb R^{2+1}} \rd_\alp (\chi\psi_n) (A(\rd_\bt (\chi\psi_n)))\,\mathrm{dVol}_{g_0} \to \int_{\mathbb R^{2+1}} \rd_\alp (\chi\psi_0) (A(\rd_\bt (\chi\psi_0)))\,\mathrm{dVol}_{g_0} + \int_{S^*\mathbb R^{2+1}} a \xi_\alp \xi_\bt \f{\ud \nu^\psi}{|\xi|^2},
\end{equation}
\begin{equation}\label{eq:cor.nu.2}
\int_{\mathbb R^{2+1}} \rd_\alp (\chi\om_n) (A(\rd_\bt (\chi\om_n)))\,\mathrm{dVol}_{g_0} \to \int_{\mathbb R^{2+1}} \rd_\alp (\chi\om_0) (A(\rd_\bt (\chi\om_0)))\,\mathrm{dVol}_{g_0} + \int_{S^*\mathbb R^{2+1}} a \xi_\alp \xi_\bt \f{\ud \nu^\om}{|\xi|^2}.
\end{equation}
Moreover,
\begin{equation}\label{eq:cor.nu.3}
\int_{\mathbb R^{2+1}} \rd_\alp (\chi\psi_n) (A(\rd_\bt (\chi\om_n)))\,\mathrm{dVol}_{g_0} \to \int_{\mathbb R^{2+1}} \rd_\alp (\chi\psi_0) (A(\rd_\bt (\chi\om_0)))\,\mathrm{dVol}_{g_0} + \int_{S^*\mathbb R^{2+1}} a \,\ud \sigma_{\alp\bt}^{\mathrm{cross}},
\end{equation}
and
\begin{equation}\label{eq:cor.nu.4}
\int_{\mathbb R^{2+1}} \rd_\alp (\chi\om_n) (A(\rd_\bt (\chi\psi_n)))\,\mathrm{dVol}_{g_0} \to \int_{\mathbb R^{2+1}} \rd_\alp (\chi\om_0) (A(\rd_\bt (\chi\psi_0)))\,\mathrm{dVol}_{g_0} + \int_{S^*\mathbb R^{2+1}} a \, (\ud \sigma^{\mathrm{cross}})^*_{\alp\bt}.
\end{equation}
\end{cor}

\subsection{Microlocal defect measures are supported on the light cones}

Our goal in this subsection is to use Theorem~\ref{thm:localization} to show that the microlocal defect measures are supported on the light cones.

\begin{lemma}\label{lem:Box.psi.decomp}
$\Box_{g_0} (\chi (\psi_n - \psi_0))$ and $\Box_{g_0} (\chi (\om_n - \om_0))$ admit the following decomposition:
$$\Box_{g_0} (\chi (\psi_n - \psi_0)) = \rd_\alp (\xi^{(\psi)}_n)^\alp + \eta_n^{(\psi)},\quad \Box_{g_0} (\chi (\om_n - \om_0)) = \rd_\alp (\xi^{(\om)}_n)^\alp + \eta_n^{(\om)},$$
where $\xi^{(\psi)}_n$, $\xi^{(\om)}_n$ are vector fields compactly supported in $\Omega'$ {(recall the definition of $\Omega'$ in Section~\ref{sec:compact.reduction})} which converges to $0$ in the $L^2$ norm; and $\eta^{(\psi)}_n$, $\eta^{(\om)}_n$ are functions compactly supported in $\Omega'$ which are uniformly bounded in $L^{\f{p_0}{2}}$ (for $p_0$ as in assumption (4) of Theorem~\ref{thm:prelim}).
\end{lemma}
\begin{proof}
We will prove the decomposition for $\Box_{g_0} (\chi (\psi_n - \psi_0))$; $\Box_{g_0} (\chi (\om_n - \om_0))$ can be treated similarly.

First we write
\begin{equation}\label{Box.psi.n.bounded}
\Box_{g_0} (\chi (\psi_n - \psi_0)) = \underbrace{(\Box_{g_0} - \Box_{g_n})(\chi\psi_n)}_{=:\mathrm{I}_n} + \underbrace{\Box_{g_n} (\chi\psi_n) }_{=:\mathrm{II}_n} \underbrace{- \Box_{g_0} (\chi\psi_0)}_{=:\mathrm{III}_n}.
\end{equation}

Clearly each term is compactly supported in $\Omega'$.

Term $\mathrm{I}_n$ can be computed further as follows:
\begin{equation*}
\begin{split}
\mathrm{I}_n = &\: \underbrace{\rd_\alp(((g_0^{-1})^{\alp\bt} - (g_n^{-1})^{\alp\bt}) \rd_{\bt} (\chi\psi_n))}_{=:\mathrm{I}_{a,n}} \underbrace{- (\rd_\alp ((g_0^{-1})^{\alp\bt} - (g_n^{-1})^{\alp\bt})) \rd_\bt (\chi\psi_n)}_{=:\mathrm{I}_{b,n}} \\
&\: + \underbrace{(\f{1}{\sqrt{-\det g_0}} \rd_\alp((g_0^{-1})^{\alp\bt} \sqrt{-\det g_0}) -\f{1}{\sqrt{-\det g_n}} \rd_\alp((g_n^{-1})^{\alp\bt} \sqrt{-\det g_n})  ) \rd_\bt (\chi\psi_n)}_{=:\mathrm{I}_{c,n}}.
\end{split}
\end{equation*}

Under the assumptions of Theorem~\ref{thm:prelim}, $\mathrm{I}_{b,n}$ and $\mathrm{I}_{c,n}$ are both uniformly bounded in $L^{\f{p_0}{2}}$.

For term $\mathrm{I}_{a,n}$, note that by assumptions (3) and (4) of Theorem~\ref{thm:prelim} (and H\"older's inequality), $((g_0^{-1})^{\alp\bt} - (g_n^{-1})^{\alp\bt}) \rd_{\bt} (\chi\psi_n) \to 0$ in the $L^2$ norm. 

For the term $\mathrm{II}_n$ in \eqref{Box.psi.n.bounded}, we note that by \eqref{eq:U1vac}, assumptions (3), (4) of Theorem~\ref{thm:prelim} and H\"older's inequality, it follows that $\mathrm{II}_n$ is uniformly bounded in $L^{\f{p_0}{2}}$.

Finally, the term $\mathrm{III}_n$ in \eqref{Box.psi.n.bounded} is smooth and \emph{independent of $n$}. It is therefore uniformly bounded in $L^{\f{p_0}{2}}$.

Combining the above results and letting
$$(\xi_n^{(\psi)})^{\alp} := ((g_0^{-1})^{\alp\bt} - (g_n^{-1})^{\alp\bt}) \rd_{\bt} (\chi\psi_n), \quad \eta_n^{(\psi)}:= \Box_{g_0} (\chi (\psi_n - \psi_0))-\rd_\alp (\xi_n^{(\psi)})^{\alp},$$
we obtain the desired result. \qedhere
\end{proof}

\begin{proposition}[Support of microlocal defect measures]\label{prop:psiom.localized}
Let $\ud \nu^\psi$, $\ud \nu^\om$ be as in Proposition~\ref{prop:nu}. Then
$$\f{(g_0^{-1})^{\alp\bt}\xi_\alp \xi_\bt}{|\xi|^2} \, \ud \nu^\psi = 0 = \f{(g_0^{-1})^{\alp\bt}\xi_\alp \xi_\bt}{|\xi|^2} \, \ud \nu^\om.$$
\end{proposition}
\begin{proof}
We will only prove the equality for $\ud \nu^\psi$. The equality for $\ud \nu^\om$ can be treated in the same manner.

\pfstep{Step~1: Compactness of $\Box_{g_0} (\chi (\psi_n - \psi_0))$ in $H^{-1}_{\mathrm{loc}}$} We use the decomposition $\Box_{g_0} (\chi (\psi_n - \psi_0)) = \rd_\alp (\xi^{(\psi)}_n)^\alp + \eta_n^{(\psi)}$ given by Lemma~\ref{lem:Box.psi.decomp}.

Since $(\xi^{(\psi)}_n)^\alp \to 0$ in the $L^2$ norm, $\rd_\alp (\xi^{(\psi)}_n)^\alp$ converges to $0$ in $H^{-1}_{\mathrm{loc}}$ (and hence is compact).

On the other hand, we know that $\eta_n^{(\psi)}$ is uniformly bounded in $L^{\f{p_0}{2}}$, where $p_0 \in (\f 83, +\infty)$ (cf.~assumption (4)). In $(2+1)$ dimensions, since $\f{p_0}{2} > \f 43$, $L^{\f{p_0}{2}}$ embeds compactly into $H^{-1}_{\mathrm{loc}}$. (This can be proven by a duality argument after recalling that $H^{-1}$ is the dual of $H^1$.) It follows that $\{\eta_n^{(\psi)}\}_{n=1}^{+\infty}$ is compact in $H^{-1}_{\mathrm{loc}}$.

Putting all the above considerations together, it follows that $\Box_{g_0} (\chi (\psi_n - \psi_0))$ is compact in $H^{-1}_{\mathrm{loc}}$.

\pfstep{Step~2: Application of Theorem~\ref{thm:localization}} By Theorem~\ref{thm:localization} and the compactness obtained in Step~1, we obtain that, for any index $\bt$,
$$\f{(g_0^{-1})^{\alp\sigma}\xi_\sigma}{|\xi|} \, \ud \sigma^\psi_{\alp\bt} = 0.$$
This implies, via Proposition~\ref{prop:nu}, that for any index $\bt$,
$$\f{(g_0^{-1})^{\alp\sigma}\xi_\sigma \xi_\alp\xi_\bt}{|\xi|^3} \, \ud \nu^\psi = 0.$$
For every $(x,\xi)\in S^*\mathbb R^{2+1}$, $\xi_\bt\neq 0$ for some $\bt$. Hence, we obtain the desired conclusion. \qedhere
\end{proof}

\section{The proof of Theorem~\ref{thm:prelim}}\label{sec:pf.thm.prelim}

In this section, we prove Theorem~\ref{thm:prelim}. We continue to work under the assumptions of Theorem~\ref{thm:prelim}. As discussed in Section~\ref{sec:compact.reduction}, with $\ud \nu = 2\ud\nu^\psi + \f 12 e^{-4\psi_0} \ud\nu^\om$, it suffices to show that $(\Omega,g_0,\psi_0,\om_0,\ud\nu)$ obeys the conclusion of Theorem~\ref{thm:prelim}.

We have already proven that $\ud\nu$ is supported on the null cones by Proposition~\ref{prop:psiom.localized}. We therefore only need to prove \eqref{eq:U1vac.vlasov.again}. The two wave equations will be proven in \textbf{Section~\ref{sec:wave.equation}}; the equation for the geometry will be proven in \textbf{Section~\ref{sec:limiting.T}}. These results can be viewed as consequences of (bilinear) compensated compactness. We then put all these together in \textbf{Section~\ref{sec:conclusion.prelim}}.

\subsection{Wave equations for the limits $\psi_0$ and $\om_0$}\label{sec:wave.equation}

We begin with a simple (bilinear) compensated compactness type result related to the null forms.
\begin{lemma}\label{lem:limitwave}
Let $\{\phi^{(1)}_n\}_{n=1}^{+\infty}$ and $\{\phi^{(2)}_n\}_{n=1}^{+\infty}$ be two sequences of real-valued smooth functions on $\mathcal M = (0,T)\times \mathbb R^3$. Assume that there exist smooth functions $\phi^{(1)}_0$ and $\phi^{(2)}_0$ on $\mathcal M$ such that the following hold for some $p_0\in (\f 83, +\infty)$:
\begin{enumerate}
\item For any (spacetime) compact subset $K$ of $\mathcal M$,
$$\|\phi^{(i)}_n - \phi^{(i)}_0\|_{L^{ \max\{2,\f{p_0}{p_0-2}\} }(K)}\to 0.$$
\item For any (spacetime) compact subset $K$ of $\mathcal M$,
$$\max_i \sup_n \|\rd \phi^{(i)}_n\|_{L^2(K)}  <+\infty.$$
\item $\Box_{g_0}\phi^{(i)}_n$ admits a decomposition $\Box_{g_0}\phi^{(i)}_n = \rd_\alp (\xi^{(i)}_n)^\alp + \eta_n^{(i)}$ for some vector field $(\xi^{(i)}_n)^\alp$ and some function $\eta_n^{(i)}$ such that for any (spacetime) compact subset $K$ of $\mathcal M$, $(\xi^{(i)}_n)^\alp \to 0$ in the $L^2(K)$ norm and $\eta_n^{(i)}$ is uniformly bounded in the $L^{\f{p_0}{2}}(K)$ norm.
\end{enumerate}
Then as $n\to +\infty$,
$$g_0^{-1}(d\phi^{(1)}_n, d\phi^{(2)}_n) \rightharpoonup g_0^{-1}(d\phi^{(1)}_0, d\phi^{(2)}_0)\quad\mbox{in the sense of distributions}.$$
\end{lemma}
\begin{proof}
Let $\vartheta \in C^\infty_c(\mathcal M)$. We want to show that
\begin{equation}\label{Murat.type.goal}
\int_{\mathbb R^{2+1}} \vartheta g_0^{-1}(d\phi^{(1)}_n, d\phi^{(2)}_n)\, \ud x \to \int_{\mathbb R^{2+1}} \vartheta g_0^{-1}(d\phi^{(1)}_0, d\phi^{(2)}_0)\, \ud x.
\end{equation}

We write $$g_0^{-1}(d\phi^{(1)}_n, d\phi^{(2)}_n) = \underbrace{\f 12\Box_{g_0}(\phi^{(1)}_n \phi^{(2)}_n)}_{=:\mathrm{I}} \underbrace{-\f 12(\Box_{g_0}\phi^{(1)}_n) \phi^{(2)}_n}_{=:\mathrm{II}} \underbrace{-\f 12 \phi^{(1)}_n (\Box_{g_0}\phi^{(2)}_n) }_{=:\mathrm{III}}$$ (and similarly for $g_0^{-1}(d\phi^{(1)}_0, d\phi^{(2)}_0)$).
We handle each of these terms below.

\pfstep{Step~1: Term $\mathrm{I}$} To handle the term $\mathrm{I}$, simply note that the assumptions and H\"older's inequality implies that $\phi^{(1)}_n\phi^{(2)}_n\to \phi^{(1)}_0\phi^{(2)}_0$ strongly in $L^1$ (on any compact set). Since $\Box_{g_0}$ is a smooth differential operator, it follows that $\f 12\Box_{g_0}(\phi^{(1)}_n \phi^{(2)}_n)$ converges to $\f 12\Box_{g_0}(\phi^{(1)}_0 \phi^{(2)}_0)$ as distributions.

\pfstep{Step~2: Terms $\mathrm{II}$ and $\mathrm{III}$} We then consider the term $\mathrm{II}$; the term $\mathrm{III}$ is clearly similar.

\pfstep{Step~2(a): Contribution from $\f 12\rd_\alp (\xi^{(1)}_n)^\alp \phi^{(2)}_n$} Using the $L^2$ norm convergence of $(\xi^{(1)}_n)^\alp$ and the $L^2$ norm boundedness of $\rd\phi^{(2)}_n$, a simple integration by parts and H\"older's inequality imply that $\f 12\rd_\alp (\xi^{(1)}_n)^\alp \phi^{(2)}_n \rightharpoonup 0$ in the sense of distributions. 

\pfstep{Step~2(b): Contribution from $\f 12\eta_n^{(1)} \phi^{(2)}_n$} Since $g_0$ is smooth, $\Box_{g_0}\phi^{(1)}_n \rightharpoonup \Box_{g_0}\phi^{(1)}_0$ in the sense of distributions. The assumptions then imply that $\eta^{(1)}_n\rightharpoonup \Box_{g_0}\phi^{(1)}_0$ in the sense of distributions. Using now the $L^{\f{p_0}{2}}$ boundedness of $\eta^{(1)}_n$ and the norm convergence of $\phi^{(i)}_n - \phi^{(i)}_0$ in $L^{ \f{p_0}{p_0-2} }$, it follows that for any $\vartheta \in C^\infty_c(\mathcal M)$,
$$\int_{\mathbb R^{2+1}} \vartheta \eta_n^{(1)} \phi^{(2)}_n \,\ud x\to \int_{\mathbb R^{2+1}} \vartheta (\Box_{g_0}\phi^{(1)}_0) \phi^{(2)}_0 \,\ud x.$$

Combining Steps~1 and 2, we have proven \eqref{Murat.type.goal}. \qedhere

\end{proof}

Using Lemma~\ref{lem:limitwave}, we obtain the following equation for $\chi\psi_0$ and $\chi\om_0$.
\begin{proposition}\label{prop:limitwave}
$\chi\psi_0$ obeys (classically) the wave equation
\begin{equation}\label{eq:psi0}
\Box_{g_0} (\chi\psi_0) - 2 g_0^{-1}( \ud\chi,\ud \psi_0) - \psi_0 \Box_{g_0} \chi + \f 12 \chi e^{-4\psi_0} g_0^{-1}( \ud\om_0, \ud \om_0) = 0.
\end{equation}
$\chi\om_0$ obeys (classically) the wave equation
\begin{equation}\label{eq:om0}
\Box_{g_0} (\chi\om_0) - 2g_0^{-1}( \ud\chi,\ud \om_0) - \om_0 \Box_{g_0} \chi - 4 \chi g_0^{-1}( \ud\om_0, \ud \psi_0) = 0.
\end{equation}
\end{proposition}
\begin{proof}
We will focus the exposition on \eqref{eq:psi0}. \eqref{eq:om0} can be treated similarly.

Since $\chi\psi_0$ is smooth, it suffices to show that \eqref{eq:psi0} holds in the sense of distributions, i.e.~we want to show that for any $\eta \in C^\infty_c(\mathbb R^{2+1})$,
\begin{equation}\label{eq:psi0.weak}
\begin{split}
\underbrace{\int_{\mathbb R^{2+1}} (\Box_{g_0}\eta) \chi \psi_0  \sqrt{-\det g_0} \,\ud x}_{=:\mathrm{I}} + \underbrace{\f 12 \int_{\mathbb R^{2+1}} \eta {\chi} e^{-4\psi_0} g_0^{-1}(\ud \om_0, \ud \om_0)  \sqrt{-\det g_0} \,\ud x}_{=:\mathrm{II}} &\\
+\underbrace{\int_{\mathbb R^{2+1}} \eta\left( - 2 g_0^{-1}(\ud \chi, \ud \psi_0) - \psi_0 \Box_{g_0} \chi \right) \sqrt{-\det g_0} \,\ud x}_{=:\mathrm{III}} &= 0.
\end{split}
\end{equation}
We note that by assumption (4) of Theorem~\ref{thm:prelim}, $\rd\psi_n$ and $\rd \mfg_n$ converge respectively to $\rd\psi_0$ and $\rd\mfg_0$ weakly in $L^{p_0}_{\mathrm{loc}}$. Therefore, using also the locally uniform convergence of $\psi_n$ and $\mfg_n$ (in assumption (3) of Theorem~\ref{thm:prelim}), we obtain
\begin{equation}\label{limitwave.1}
\mathrm{I}+\mathrm{III}= \lim_{n\to +\infty} \int_{\mathbb R^{2+1}} \left((\Box_{g_n}\eta) \chi \psi_n - 2\eta g_n^{-1}(\ud \chi, \ud \psi_n) - \eta \psi_n \Box_{g_n} \chi \right) \sqrt{-\det g_n} \,dx.
\end{equation}
For the term $\mathrm{II}$ in \eqref{eq:psi0.weak}, we compute using the uniform convergence of $\psi_n$ and $g_n$ (on compact sets) and Lemma~\ref{lem:limitwave}.
Note that Lemma~\ref{lem:limitwave} indeed applied to $g_0^{-1}(\ud \om_n, \ud \om_n)$ since by assumptions (3), (4) of Theorem~\ref{thm:prelim} and Lemma~\ref{lem:Box.psi.decomp}, $\om_n$ obeys the assumptions of Lemma~\ref{lem:limitwave}. Hence, we obtain
\begin{equation}\label{limitwave.2}
\begin{split}
\mathrm{II} = &\: \f 12 \lim_{n\to+\infty} \int_{\mathbb R^{2+1}} \eta \chi e^{-4\psi_0} g_0^{-1}(\ud \om_n, \ud \om_n)  \sqrt{-\det g_0} \,\ud x \\
= &\: \f 12 \lim_{n\to+\infty} \int_{\mathbb R^{2+1}} \eta {\chi} e^{-4\psi_n} g_n^{-1}(\ud \om_n, \ud \om_n)  \sqrt{-\det g_n} \,\ud x.
\end{split}
\end{equation}

Combining \eqref{limitwave.1} and \eqref{limitwave.2}, we obtain
\begin{equation}\label{eq:psi0.weak.final}
\begin{split}
 \mathrm{I} + \mathrm{II} + \mathrm{III}
= &\: \lim_{n\to +\infty} \int_{\mathbb R^{2+1}} \left((\Box_{g_n}\eta) \chi \psi_n - 2\eta g_n^{-1}(\ud \chi, \ud \psi_n) - \eta \psi_n \Box_{g_n} \chi \right) \sqrt{-\det g_n} \,dx \\
&\: + \f 12 \lim_{n\to+\infty} \int_{\mathbb R^{2+1}} \eta {\chi} e^{-4\psi_n} g_n^{-1}(\ud \om_n, \ud \om_n)  \sqrt{-\det g_n} \,\ud x = 0,
\end{split}
\end{equation}
where in the last line we have used the fact that for every $n \in \mathbb N$, the wave equation
$$\Box_{g_n} (\chi\psi_n) - 2 g_n^{-1}(\ud \chi, \ud \psi_n) - \psi_n \Box_{g_n} \chi + \f 12 {\chi} e^{-4\psi_n} g_n^{-1}(\ud \om_n, \ud \om_n) = 0$$
holds. We have thus proven \eqref{eq:psi0.weak}.  \qedhere
\end{proof}

\subsection{The limiting stress-energy-momentum tensor}\label{sec:limiting.T}

\begin{proposition}\label{prop:T.compute.1}
There is a subsequence $n_k$ such that for every vector field $Y\in C^\infty_c(\Omega)$,
\begin{equation*}
\begin{split}
\int_{\mathbb R^{2+1}} \mathrm{Ric}(g_0) (Y,Y) \,\mathrm{dVol}_{g_0} = \lim_{k\to +\infty} \int_{\mathbb R^{2+1}} [2 (Y \psi_{n_k})^2 + \f 12 e^{-4\psi_{n_k}} (Y\om_{n_k})^2] \, \mathrm{dVol}_{g_{n_k}}.
\end{split}
\end{equation*}
\end{proposition}
\begin{proof}

\pfstep{Step~1: $\rd_i \gamma_n$ and $(H_n)_{ij}$ have strong subsequential $L^2_{\mathrm{loc}}$ limits} In this step, we show that on any fixed compact set, after choosing a subsequence $n_k$, $\rd_i \gamma_{n_k}$ and $(H_{n_k})_{ij}$ have strong $L^2$ limits. Since $p_0 \in (\f 83, +\infty)$, $W^{1,\f{p_0}{2}}_{\mathrm{loc}}$ embeds \underline{compactly} into $L^{2}_{\mathrm{loc}}$ (in $(2+1)$ dimensions). Therefore, it suffices to show that for any fixed compact set, $\rd_i \gamma_n$ and $(H_n)_{ij}$ are uniformly bounded in $W^{1, \f{p_0}{2}}$. By assumptions, we already know that $\rd_i \gamma_n$ and $(H_n)_{ij}$ are uniformly bounded in $L^{\f{p_0}{2}}$ (in fact also $L^{p_0}$) on any compact set; we therefore need to show that the same holds true for all first derivatives of $\rd_i \gamma_n$ and $(H_n)_{ij}$.

By \eqref{elliptic.wave.1}, \eqref{elliptic.wave.2} and \eqref{elliptic.wave.3}, $\Delta \gamma_n$, $\Delta N$ and $\de^{ik}\rd_k (H_n)_{ij}$ are all uniformly bounded in $L^{\f{p_0}{2}}$ in any fixed compact set. Standard $L^p$ elliptic theory\footnote{Note that $H_n$ is traceless. In two (spatial) dimensions, this implies that a bound on the divergence of $H_n$ also gives a bound on the curl of $H_n$. Hence, we indeed have an elliptic estimate of the type
$$\sum_{i,j,k} \|\rd_k (H_n)_{ij}\|_{L^{\f{p_0}{2}}(U_1)} \ls \sum_j \|\de^{ik}\rd_k (H_n)_{ij}\|_{L^{\f{p_0}{2}}(U_2)} + \sum_{i,j}\|(H_n)_{ij}\|_{L^{\f{p_0}{2}}(U_2)}$$
for $U_1\subset U_2\subset \mathbb R^2$, each set being an open and precompact subset of the next set.
} (applied for each fixed $t$) implies that 
$$\rd^2_{ij} \gamma_n,\,\rd^2_{ik} N,\, \rd_k (H_n)_{ij}$$
are all uniformly bounded in $L^{\f{p_0}{2}}$ on any fixed compact set.

Using the above, and also \eqref{Rij.wave}, the assumptions of Theorem~\ref{thm:prelim} and H\"older's inequality, we also obtain that
$$\rd_t (H_n)_{ij}$$
is uniformly bounded in $L^{\f{p_0}{2}}$ on any fixed compact set.

It remains to bound $\rd_t \rd_i \gamma_n$. For this, first note that by \eqref{elliptic.wave.4}, the assumptions of Theorem~\ref{thm:prelim} and the above bounds, we see that $\Delta\bt^i_n$ is uniformly bounded in $L^{\f{p_0}{2}}$ on any fixed compact set. Elliptic theory then implies that 
$$\rd^2_{jk}\bt^i_n$$ 
is uniformly bounded in $L^{\f{p_0}{2}}$ on any fixed compact set.

Now we use \eqref{trace.free}, take a spatial derivative, and apply the above estimates. We see that
$$\rd_t \rd_i \gamma_n$$
is uniformly bounded in $L^{\f{p_0}{2}}$ on any fixed compact set.

The above discussions imply that indeed $\rd_i \gamma_n$ and $(H_n)_{ij}$ have strong subsequential $L^2_{\mathrm{loc}}$ limits.

\pfstep{Step~2: Weak convergence of the Ricci tensor} We now turn to the expressions for the Ricci tensor as given in \eqref{elliptic.1}, \eqref{elliptic.2}, \eqref{elliptic.3} and \eqref{Rij}. Notice that in each of the terms which is quadratic in the first derivative of metric, there is at least one factor of $\rd_i \gamma_{n_k}$ or $(H_{n_k})_{ij}$. By Step~1 and the Cauchy--Schwarz inequality, it follows that $\mathrm{Ric}(g_{n_k})$ converges to $\mathrm{Ric}(g_0)$ in the sense of distributions (where $n_k$ is the subsequence as in Step~1).

\pfstep{Step~3: Putting everything together} By Step~2 and assumption (3) of Theorem~\ref{thm:prelim}, it follows that for any smooth vector field $Y$ supported in $\Omega$, as $k\to +\infty$,
$$\int_{\mathbb R^{2+1}} \mathrm{Ric}(g_{n_k}) (Y,Y) \,\mathrm{dVol}_{g_{n_k}} \to \int_{\mathbb R^{2+1}} \mathrm{Ric}(g_0) (Y,Y) \,\mathrm{dVol}_{g_0}.$$
On the other hand, since $(\psi_n,\om_n,g_n)$ satisfies \eqref{eq:U1vac} for all $n\in \mathbb N$, we know that for every $n_k \in \mathbb N$,
$$\int_{\mathbb R^{2+1}} \mathrm{Ric}(g_{n_k}) (Y,Y) \,\mathrm{dVol}_{g_{n_k}} = \int_{\mathbb R^{2+1}} [2 (Y \psi_{n_k})^2 + \f 12 e^{-4\psi_{n_k}} (Y\om_{n_k})^2] \, \mathrm{dVol}_{g_{n_k}}.$$
The conclusion follows. \qedhere
\end{proof}

\begin{proposition}\label{prop:T.compute.2}
Let 
\begin{equation}\label{def:dnu}
\ud\nu:= 2\ud \nu^\psi + \f 12 e^{-4\psi_0} \, \ud \nu^\om.
\end{equation}
Then the limiting metric $g_0$ satisfies
\begin{equation*}
\begin{split}
\int_{\mathbb R^{2+1}} \mathrm{Ric}(g_0) (Y,Y) \,\mathrm{dVol}_{g_0}= \int_{\mathbb R^{2+1}} \left( 2(Y\psi_0)^2 + \f12 e^{-4\psi_0} (Y\om_0)^2 \right)\, \mathrm{dVol}_{g_0} + \int_{S^*\mathbb R^{2+1}} (Y^\alp \xi_\alp)^2 \, \ud\nu
\end{split}
\end{equation*}
for every vector field $Y\in C^\infty_c(\Omega)$.
\end{proposition}
\begin{proof}
Since $\psi_k$ and $g_k$ converge uniformly on compact sets, they in particular converge uniformly on $\Omega$. Therefore, taking $n_k$ as the subsequence in Proposition~\ref{prop:T.compute.1}, 
\begin{equation*}
\begin{split}
\lim_{k\to +\infty} \int_{\mathbb R^{2+1}} [2 (Y \psi_{n_k})^2 + \f 12 e^{-4\psi_{n_k}} (Y\om_{n_k})^2] \, \mathrm{dVol}_{g_{n_k}} 
= &\: \lim_{k\to +\infty} \int_{\mathbb R^{2+1}} [2 (Y \psi_{n_k})^2 + \f 12 e^{-4\psi_0} (Y\om_{n_k})^2] \, \mathrm{dVol}_{g_0}.
\end{split}
\end{equation*}
Now using the fact that $\chi \equiv 1$ on the support of $Y$ and Corollary~\ref{cor:nu}, we obtain
\begin{equation*}
\begin{split}
&\: \lim_{k\to +\infty} \int_{\mathbb R^{2+1}} [2 (Y \psi_{n_k})^2 + \f 12 e^{-4\psi_0} (Y\om_{n_k})^2] \, \mathrm{dVol}_{g_0} \\
= &\: \int_{\mathbb R^{2+1}} [2 (Y \psi_0)^2 + \f 12 e^{-4\psi_0} (Y\om_0)^2] \, \mathrm{dVol}_{g_0} + \int_{S^*\mathbb R^{2+1}} (Y^\alp \xi_\alp)^2 \, \ud\nu.
\end{split}
\end{equation*}
The desired conclusion therefore follows from Proposition~\ref{prop:T.compute.1}. \qedhere
\end{proof}

\subsection{Conclusion of the proof of Theorem~\ref{thm:prelim}}\label{sec:conclusion.prelim}

We now conclude the proof of Theorem~\ref{thm:prelim}:
\begin{proof}[Proof of Theorem~\ref{thm:prelim}]
First, $\ud \nu$ is supported on $\{(x,\xi) \in S^*\mathcal M: g^{-1}_0(\xi,\xi) = 0\}$ in view of Proposition~\ref{prop:psiom.localized}.

To check that the three equations in \eqref{eq:U1vac.vlasov.again} are verified, note that the first two equations are verified due to Proposition~\ref{prop:limitwave} (and the fact that $\chi = 1$ on $\Omega$), while the last equation is verified thanks to Proposition~\ref{prop:T.compute.2}. 

Finally, using Proposition~\ref{prop:compact.reduction}, we have completed the proof of Theorem~\ref{thm:prelim}. \qedhere
\end{proof}

\section{Beginning of the proof of Theorem~\ref{thm:main}}\label{sec:preliminaries}

\textbf{From now on and for the remainder of the paper, we prove Theorem~\ref{thm:main}.} We will therefore work under the assumptions of Theorem~\ref{thm:main}. The main goal from now on will be to show that with the additional assumption (5) of Theorem~\ref{thm:main}, we can show moreover that the measure $\ud\nu$ satisfies a transport equation on $\Omega$ (where we have used the reduction in Proposition~\ref{prop:compact.reduction}).

\textbf{From now on, unless otherwise stated, let $A$ be a zeroth order pseudo-differential operator with real symbol $a(x,\xi)$. Assume moreover that $a(x,\xi)$ is supported in $S^*\Omega$.}

We introduce now conventions that we will use for the remainder of the paper. We use the convention that $\lambda_n$ refers to the sequence of constants in assumption (5) of Theorem~\ref{thm:main} \textbf{with $K = \overline{\Omega'''}$} (cf.~Section~\ref{sec:compact.reduction}).

From now on, we use the convention that for two non-negative quantities $B_1$ and \textbf{$B_2$, $B_1\ls B_2$ means there exists $C>0$ depending potentially on $T$, $\psi_0$, $\om_0$, $g_0$, $\Omega$, $\Omega'$, $\Omega''$, $\Omega'''$ and $A$, but \underline{independent of $n$}, such that 
$$B_1\leq C B_2.$$}

\textbf{We will also use the big-O and small-o conventions,} i.e.~for a non-negative quantity $B$ (depending on $n$) and a positive function $f(n)$ of $n$, $B=O(f(n))$ means $B\ls C \cdot f(n)$, while $B=o(f(n))$ means $\f{B}{f(n)}\to 0$ and $n\to +\infty$.

In this section, we carry out various preliminary steps. In \textbf{Section~\ref{sec:dmetric}}, we begin with some convergence estimates for the derivatives of the metric which follow from the elliptic equations (and are stronger than \eqref{assumption.1}). In \textbf{Section~\ref{sec:freeze.coeff}}, we discuss the freezing of coefficients, which will be used in various places later. Finally, in \textbf{Section~\ref{sec:reduction}}, we discuss a reduction allowing us to consider only a subclass of pseudo-differential operators $A$ later.

\subsection{Convergence of the derivatives of the metric components}\label{sec:dmetric}

\begin{proposition}\label{prop:spatial.imp}
Let $\widetilde{\chi}$ be as in \eqref{def:tildechi}. 
$$\| \rd_i(\widetilde{\chi}(\gamma_n-\gamma_0))\|_{L^\i} + \| \rd_i(\widetilde{\chi}(\bt^j_n-\bt^j_0))\|_{L^\i} + \| \rd_i(\widetilde{\chi}(N_n-N_0))\|_{L^\i} \ls \lambda_n^{\f 12}.$$
\end{proposition}
\begin{proof}

In view of the elliptic equations \eqref{elliptic.wave.2}, \eqref{elliptic.wave.3} and \eqref{elliptic.wave.4} satisfied by $\gamma$, $N$ and $\bt^j$, it suffices to show that for smooth $u_n,\,u_0:\mathbb R^{3+1} \to \mathbb R$ ($n\in \mathbb N$) such that 
\begin{equation}\label{spatial.imp.1}
\|\widetilde{\chi}(u_n-u_0)\|_{L^\i}\ls \lambda_n
\end{equation} and 
\begin{equation}\label{spatial.imp.2}
\|\Delta(\widetilde{\chi}(u_n-u_0))\|_{L^\i}\ls 1,
\end{equation}
we have $\|\rd_i(\widetilde{\chi}(u_n-u_0))\|_{L^\i} \ls \lambda_n^{\f 12}$. {This is a standard interpolation estimates; we include below a proof for completeness.}

Let $\Theta:[0,+\infty)\to \mathbb R$ be a non-negative smooth cutoff function such that 
$$\Theta \geq 0,\quad \Theta(x) = 1 \mbox{ for $x\in [0,1]$},\quad \Theta(x) = 0 \mbox{ for $x\geq 2$}.$$
For every fixed\footnote{Note that in fact for $t\notin [0,T]$, the term \eqref{spatial.imp.decomp} vanishes.} $t\in \mathbb R$, we take the spatial Fourier transform $\mathcal F_{\mathrm{spa}}$ and then decompose into a low-spatial-frequency part and a high-spatial-frequency part as follows:
\begin{equation}\label{spatial.imp.decomp}
\widetilde{\chi} (u_n - u_0)(t,\xi_i) = \underbrace{\mathcal F_{\mathrm{spa}}^{-1}\Theta(\lambda_n^{\f 12}|\xi_i|)\mathcal F_{\mathrm{spa}}(\widetilde{\chi} (u_n - u_0))(t,\xi_i)}_{=:\mathrm{I}} + \underbrace{\mathcal F_{\mathrm{spa}}^{-1}(1-\Theta(\lambda_n^{\f 12}|\xi_i|))\mathcal F_{\mathrm{spa}}(\widetilde{\chi} (u_n - u_0))(t,\xi_i)}_{=:\mathrm{II}}.
\end{equation}

For the term $\mathrm{I}$, we apply Bernstein's inequality and \eqref{spatial.imp.1} to obtain 
$$\left\|\rd_j\left(\mathcal F_{\mathrm{spa}}^{-1}\left(\Theta(\lambda_n^{\f 12}|\xi_i|)\mathcal F_{\mathrm{spa}}(\widetilde{\chi} (u_n - u_0))(t,\xi_i)\right)\right)\right\|_{L^\i_x}(t) \ls \lambda_n^{-\f 12} \|\widetilde{\chi} (u_n - u_0)\|_{L^\i_x}(t) \ls \lambda_n^{-\f 12} \lambda_n = \lambda_n^{\f 12}.$$
Taking supremum over $t$ implies the desired estimate for this term.

For the term $\mathrm{II}$, define first $P_{\mathrm{spa},k}$ the \underline{spatial} standard Littlewood--Paley projection to spatial frequency $|\xi_i|\sim 2^k$. {Denote the corresponding Fourier multiplier by $m_{LP}(2^{-k}\xi_i)$ where $m_{LP}$} is a radial smooth {spatial} function supported in an annulus. 

Now note that for each fixed $t\in \mathbb R$ and for each Littlewood--Paley piece, 
\begin{equation*}
\begin{split}
&\: \left\|\rd_j\left(\mathcal F_{\mathrm{spa}}^{-1}\left((1-\Theta(\lambda_n^{\f 12}|\xi_i|))\mathcal F_{\mathrm{spa}} P_{\mathrm{spa},k}(\widetilde{\chi} (u_n - u_0))(t,\xi_i)\right)\right)\right\|_{L^\i_x}(t)\\
\ls & \:\left\|\mathcal F_{\mathrm{spa}}^{-1} \left( \frac{i\xi_{{j}}}{|\xi|^2}(1-\Theta(\lambda_n^{\f 12}|\xi_i|))
{m_{LP}(2^{-k}\xi_i)}\mathcal F_{\mathrm{spa}} \left(\Delta(\widetilde{\chi} (u_n - u_0)) \right)(t,\xi_i)\right)
\right\|_{L^\i_x}(t)\\
\ls & \:\left\|\mathcal F_{\mathrm{spa}}^{-1} \left( \frac{i\xi_{{j}}}{|\xi|^2}
{m_{LP}(2^{-k}\xi_i)}\right)\right\|_{L^1_x}\left\|\Delta(\widetilde{\chi} (u_n - u_0))
\right\|_{L^\i_x}(t)\\
&\:+\left\|\mathcal F_{\mathrm{spa}}^{-1} \left( \frac{i\xi_{{j}}}{|\xi|^2}
{m_{LP}(2^{-k}\xi_i)}\right)\right\|_{L^1_x}\left\|\mathcal F_{\mathrm{spa}}^{-1} \left( \Theta(\lambda_n^{\f 12}|\xi_i|)\mathcal F_{\mathrm{spa}} \left(\Delta(\widetilde{\chi} (u_n - u_0)) \right)(t,\xi_i)\right)
\right\|_{L^\i_x}(t)\\
\ls & \:2^{-k}\left(1+\|\mathcal F_{\mathrm{spa}}^{-1} \Theta(\lambda_n^{\f 12}|\xi_i|)\|_{L^1_x}\right)\left\|\Delta(\widetilde{\chi} (u_n - u_0))
\right\|_{L^\i_x}(t)\\\\
\ls &\: 2^{-k}\left\|\Delta(\widetilde{\chi} (u_n - u_0)) \right\|_{L^\i_x}(t) {\ls 2^{-k}},
\end{split}
\end{equation*}
where in the last estimate we used \eqref{spatial.imp.2}.

Now, summing up all the Littlewood--Paley pieces with $2^k\gtrsim \lambda_n^{-\f 12}$, we obtain that for every fixed $t\in \mathbb R$,
$$\left\|\rd_j\left(\mathcal F_{\mathrm{spa}}^{-1}\left((1-\Theta(\lambda_n^{\f 12}|\xi_i|))\mathcal F_{\mathrm{spa}}(\widetilde{\chi} (u_n - u_0))(t,\xi_i)\right)\right)\right\|_{L^\i_x}(t) \ls \sum_{k: 2^k \gtrsim \lambda_n^{-\f 12}} 2^{-k} \ls \lambda_n^{\f 12}.$$
Taking supremum over $t$ then implies the desired estimate.
\qedhere
\end{proof}

\begin{proposition}\label{prop:dtgamma.imp}
Let $\widetilde{\chi}$ be as in \eqref{def:tildechi}. 
$$\| \rd_t(\widetilde{\chi}(\gamma_n-\gamma_0))\|_{L^\i} \ls \lambda_n^{\f 12}.$$
\end{proposition}
\begin{proof}
This is an immediate consequence of \eqref{trace.free} and the estimates in \eqref{assumption.0} and Proposition~\ref{prop:spatial.imp}. \qedhere
\end{proof}

\subsection{Freezing coefficients}\label{sec:freeze.coeff}

One trick that we will repeatedly use is to to freeze coefficients. This will be used in Section~\ref{sec:reduction}, but will again be useful when we capture some trilinear cancellations; see already Section~\ref{sec:elliptic.wave.tri}. In this subsection, we will introduce some relevant notations and prove some basic estimates.

\textbf{Fix some $\ep_0\in (\f 16, \f 12)$} (for the remainder of the paper). For each $n\in \mathbb N$, choose finitely many (spacetime) balls of radius $\lambda_n^{\ep_0}$ (with respect to the $(t,x^1,x^2)$ coordinates), labeled by $\{B_\alp\}_{\alp}$ so that $\Omega'\subset \cup_{\alp} B_\alp \subset \overline{\cup_{\alp} B_\alp} \subset \Omega''$ (cf.~Section~\ref{sec:compact.reduction}). Note that this gives $O(\lambda_n^{-3\ep_0})$ balls, each with volume $O(\lambda_n^{3\ep_0})$. Introduce a partition of unity $\{\zeta_\alp^3\}_\alp$ adapted to these balls so that $\mathrm{supp}(\zeta_\alp) \subset B_\alp$ and
$$\sum_\alp \zeta_\alp^3 = 1\quad \mbox{on $\Omega'$}.$$
Due to the choice of $B_\alp$, $\zeta_\alp$ can be chosen so that for every $r\in [1,2]$,
\begin{equation}\label{zeta.prop}
\|\rd^k \zeta_\alp^r\|_{L^\infty} \ls \lambda_n^{-k\ep_0},\quad k=0,1,2,3.
\end{equation}

The following is an immediate consequence of mean value theorem:
\begin{proposition}\label{prop:constants}
Let $b:\Omega'' \to \mathbb R$ be a $C^1$ function. Then for every fixed $n\in \mathbb N$ and fixed $\alp$ as above, there exist \underline{constants} $b_{c,\alp}$ (depending on $\alp$) such that (with implicit constants depending on the $C^1$ norm of $b$ but independent of $n$ or $\alp$)
$$\|b-b_{c,\alp}\|_{L^\infty(B_\alp)} \ls \lambda_n^{\ep_0}.$$
Moreover, the constants satisfy
$$\sup_\alp |b_{c,\alp}|\ls 1.$$

In particular, for every $\alp$, there exist uniformly bounded \underline{constants} $N_{c,\alp}$, $\bt_{c,\alp}^i$ and $\gamma_{c,\alp}$ (depending on $\alp$) such that
$$\|\log N_0-\log N_{c,\alp}\|_{L^\infty(B_\alp)} + \|\bt_0^i-\bt^i_{c,\alp}\|_{L^\infty(B_\alp)} + \|\gamma_0 - \gamma_{c,\alp}\|_{L^\infty(B_\alp)} \ls \lambda_n^{\ep_0}.$$
\end{proposition}

\begin{proposition}\label{prop:psi.localized}
For every $\alp$, every $r\in [1,2]$ and every $p\in [1,+\infty]$,
\begin{align}
\|\rd^k (\zeta_\alp^r\chi(\psi_n - \psi_0))\|_{L^p}\ls &\: \lambda_n^{1-k+\f{3\ep_0}{p}},\quad k=0,1,2,3,\label{eq:psi.localized.1}\\
\|\rd_\mu (\zeta_\alp^r\chi(\psi_n - \psi_0)) - \zeta_\alp^r \rd_\mu(\chi(\psi_n - \psi_0))\|_{L^p} \ls &\: \lambda_n^{1+\ep_0(-1+\f 3p)},\label{eq:psi.localized.2}\\
\|\rd^2_{\mu\nu} (\zeta_\alp^r\chi(\psi_n - \psi_0)) - \zeta_\alp^r \rd^2_{\mu\nu}(\chi(\psi_n - \psi_0))\|_{L^p} \ls &\: \lambda_n^{\ep_0(-1+\f 3p)}\label{eq:psi.localized.3}.
\end{align}
Similarly, for every $\alp$, every $r\in [1,2]$ and every $p\in [1,+\infty]$,
\begin{align*}
\|\rd^k (\zeta_\alp^r\chi(\om_n - \om_0))\|_{L^p}\ls &\: \lambda_n^{1-k+\f{3\ep_0}{p}},\quad k=0,1,2,3,\\
\|\rd_\mu (\zeta_\alp^r\chi(\om_n - \om_0)) - \zeta_\alp^r \rd_\mu(\chi(\om_n - \om_0))\|_{L^p} \ls &\: \lambda_n^{1+\ep_0(-1+\f 3p)},\\
\|\rd^2_{\mu\nu} (\zeta_\alp^r\chi(\om_n - \om_0)) - \zeta_\alp^r \rd^2_{\mu\nu}(\chi(\om_n - \om_0))\|_{L^p} \ls &\: \lambda_n^{\ep_0(-1+\f 3p)}.
\end{align*}
\end{proposition}
\begin{proof}
We will only discuss \eqref{eq:psi.localized.1}--\eqref{eq:psi.localized.3}; the remaining bounds can be derived in an identical manner.

\pfstep{Step~1: Proof of \eqref{eq:psi.localized.1}} By \eqref{assumption.0}, \eqref{assumption.1}, \eqref{assumption.2} and \eqref{zeta.prop}, and using $\ep_0<\f 12$, we have, for $k=0,1,2,3$,
$$\|\rd^k (\zeta_\alp^r \chi(\om_n - \om_0))\|_{L^\infty}\ls \lambda_n^{1-k}.$$
Now since $\mathrm{supp}(\rd^k (\zeta_\alp^r \chi(\psi_n - \psi_0)))\subset B_\alp$, and $B_\alp$ has volume $O(\lambda_n^{3\ep_0})$, we obtain the desired conclusion for all $p\in [1,+\infty]$.

\pfstep{Step~2: Proof of \eqref{eq:psi.localized.2} and \eqref{eq:psi.localized.3}} The proof of \eqref{eq:psi.localized.2} and \eqref{eq:psi.localized.3} is similar to that of \eqref{eq:psi.localized.1} except in this case we are computing the commutator so at least one derivative hits on $\zeta_\alp^r$. This results in the better bounds. \qedhere
\end{proof}

\begin{proposition}\label{prop:g.localized}
For every $\alp$, for $p\in [1,+\infty]$ and for $\mfg \in \{\log N,\, \bt^j, \, \gamma\}$,
$$\|\zeta_\alp(\mfg_n - \mfg_0)\|_{L^p}\ls \lambda_n^{1+\f{3\ep_0}{p}},\quad \|\rd (\zeta_\alp (\mfg_n - \mfg_0))\|_{L^p}\ls \lambda_n^{\f{3\ep_0}{p}},\quad \|\rd^2 (\zeta_\alp (\mfg_n-\mfg_0)\|_{L^p} \ls \lambda_n^{-1+\f{3\ep_0}{p}},$$
$$\|\rd_i (\zeta_\alp (\mfg_n - \mfg_0))\|_{L^p}\ls \lambda_n^{\f 12+\f{3\ep_0}{p}},\quad \|\Delta (\zeta_\alp (\mfg_n - \mfg_0))\|_{L^p} \ls \lambda_n^{\f{3\ep_0}{p}}.$$
\end{proposition}
\begin{proof}
The estimates for $\rd^k(\zeta_\alp (\mfg_n-\mfg_0))$ ($k=0,1,2$) are similar to Proposition~\ref{prop:psi.localized}; we omit the details. 

The last two estimates assert that there is an improvement associated to spatial derivative. First, using Proposition~\ref{prop:spatial.imp} {(which is applicable since $\mathrm{supp}(\zeta_\alp)\subset \Omega''$)}, \eqref{zeta.prop}, \eqref{assumption.0} and $\ep_0<\f 12$, we obtain 
$$\|\rd_i (\zeta_\alp (\mfg_n - \mfg_0))\|_{L^\infty}\ls \max\{\lambda_n^{1-\ep_0},\lambda_n^{\f 12}\}\ls \lambda_n^{\f 12}.$$
(Here, the $\lambda_n^{1-\ep_0}$ error arises when the derivative acts on $\zeta_\alp$, while the $\lambda_n^{\f 12}$ error arises when the derivatives acts on $\mfg_n - \mfg_0$.) Taking the $L^p$ norm over $B_\alp$ yields the desired claim for all $p\in [1,+\infty]$.

Finally, using the equations \eqref{elliptic.wave.2}--\eqref{elliptic.wave.4} (which together with \eqref{assumption.0} and \eqref{assumption.1} give an $L^\i$ bound for $\Delta \gamma_n$, $\Delta \log N_n$ and $\Delta \bt^j_n$), the estimates in \eqref{assumption.0}, \eqref{assumption.1}, Proposition~\ref{prop:spatial.imp}, \eqref{zeta.prop} and $\ep_0 <\f 12$, we obtain
$$\|\Delta (\zeta_\alp (\mfg_n - \mfg_0))\|_{L^\infty}\ls \max\{\lambda_n^{1-2\ep_0},\lambda_n^{\f 12-\ep_0}, 1\}\ls 1.$$
As before, taking the $L^p$ norm over $B_\alp$ yields the desired claim for all $p\in [1,+\infty]$. \qedhere
\end{proof}

One important consequence of freezing the coefficients is that in every $B_\alp$, the $\Box_{g_0}$ operator is comparable to a \emph{constant coefficient} operator:
\begin{proposition}\label{prop:approximate.wave}
For every $\alp$, let $\widetilde{\Box}_{c,\alp}$ be the constant coefficient second order differential operator defined by
$$\widetilde{\Box}_{c,\alp}:= -\f 1{N_{c,\alp}^2}(\rd_t - \bt_{c,\alp}^i \rd_i) (\rd_t - \bt_{c,\alp}^j \rd_j) + e^{-2\gamma_{c,\alp}} \de^{k\ell} \rd^2_{k\ell},$$
where the constants $N_{c,\alp}$, $\bt_{c,\alp}^i$ and $\gamma_{c,\alp}$ are defined in Proposition~\ref{prop:constants} {(cf.~\eqref{eq:wave})}.

Then, for every $\alp$, for every $n\in \mathbb N$ and for every $p\in [1,+\infty]$,
$$\|\rd^k \widetilde{\Box}_{c,\alp} (\zeta_\alp\chi(\psi_n-\psi_0))\|_{L^p}\ls \lambda_n^{-1-k+\ep_0} \cdot \lambda_n^{\f{3\ep_0}{p}},\quad k=0,1,2.$$
Similarly, for every $\alp$, for every $n\in \mathbb N$ and for every $p\in [1,+\infty]$,
$$\|\rd^k\widetilde{\Box}_{c,\alp} (\zeta_\alp\chi(\om_n-\om_0))\|_{L^p}\ls \lambda_n^{-1-k+\ep_0} \cdot \lambda_n^{\f{3\ep_0}{p}},\quad k=0,1,2.$$
\end{proposition}
\begin{proof}
We will only prove the estimates for $\psi_n-\psi_0$; the estimates for $\om_n-\om_0$ are similar and will be omitted.

{First, notice that after using the first equation in \eqref{eq:U1vac} and the bounds \eqref{assumption.0}--\eqref{assumption.2} in a similar (but easier) argument as in the proof of Lemma~\ref{lem:Box.psi.decomp}, we obtain}
%
\begin{equation}\label{box.of.the.difference}
\|\rd^k \Box_{g_0} (\chi(\psi_n - \psi_0))\|_{L^\i} \ls \lambda_n^{-k},\quad k=0,1,2.
\end{equation}

Next, we compute
\begin{equation}\label{eq:diff.with.constants}
\begin{split}
&\:   \Box_{g_0} - \widetilde{\Box}_{c,\alp} \\
= &\:  - \Big(\f{1}{N_0^2} - \f{1}{N_{c,\alp}^2} \Big) \rd^2_t + 2 \Big(\f{\bt_0^i}{N_0^2} - \f{\bt_{c,\alp}^i}{N_{c,\alp}^2} \Big)\rd^2_{ti} + \Big[ (e^{-2\gamma} - e^{-2\gamma_{c,\alp}})\de^{ij}-(\f{\bt_0^i\bt_0^j}{N_0^2} - \f{\bt_{c,\alp}^i\bt_{c,\alp}^j}{N_{c,\alp}^2}) \Big]\rd^2_{ij} \\
&\:  + \f 1{\sqrt{-\det g_0}}\Big[\rd_\mu ((g_0^{-1})^{\mu\nu} \sqrt{-\det g_0} ) \Big] \rd_\nu
\end{split}
\end{equation}
Therefore, it follows that by \eqref{box.of.the.difference}, Propositions~\ref{prop:constants} and \ref{prop:psi.localized} and \eqref{assumption.0}--\eqref{assumption.2}, we have
\begin{equation}\label{eta.of.const.box}
\begin{split}
&\:\|\zeta_\alp \rd^k \widetilde{\Box}_{c,\alp} (\chi(\psi_n-\psi_0))\|_{L^p} \\
\ls &\: \|\zeta_\alp \rd^k \Box_{g_0} (\chi(\psi_n - \psi_0))\|_{L^p} + \|\zeta_\alp \rd^k [(\Box_{g_0} - \widetilde{\Box}_{c,\alp}) (\chi(\psi_n-\psi_0))]\|_{L^p}  \\
\ls &\:  \lambda_n^{-k} \times\lambda_n^{\f{3\ep_0}{p}} + \lambda_n^{-k-1+\ep_0} \times\lambda_n^{\f{3\ep_0}{p}} \ls \lambda_n^{-k-1+\ep_0}\lambda_n^{\f{3\ep_0}{p}},\quad k=0,1,2.
\end{split}
\end{equation}
Here, $\lambda_n^{\f{3\ep_0}{p}}$ comes from the volume of the support of $\zeta_\alp$. For the $ \rd^k \Box_{g_0} (\chi(\psi_n - \psi_0))$ term, we applied \eqref{box.of.the.difference}. For the main $\rd^k [(\Box_{g_0} - \widetilde{\Box}_{c,\alp}) (\chi(\psi_n-\psi_0))]$ term, when there are $k+2$ derivatives (which is the maximum possible) hitting on $\psi_n - \psi_0$, we use the expression \eqref{eq:diff.with.constants} together with the bounds in Propositions~\ref{prop:constants} and \ref{prop:psi.localized}; when fewer derivatives hit on $\psi_n-\psi_0$, this is slightly easier.

Finally, since $\ep_0<1$, when $k=0,1,2$, the following commutator can be estimated above by
\begin{equation}\label{commutator.of.two.derivatives.with.eta}
\begin{split}
&\: \| \rd^{k+2}_{\mu_1\cdots \mu_{k+2}} (\zeta_\alp \chi(\psi_n-\psi_0)) - \zeta_{\alp} \rd^{k+2}_{\mu_1\cdots \mu_{k+2}} (\chi(\psi_n-\psi_0))\|_{L^p} \\
\ls &\: \sum_{\ell = 0}^{k+1} \| (\rd^{k+2-\ell} \zeta_\alp) (\rd^{\ell} (\chi(\psi_n-\psi_0)))\|_{L^p} \ls \sum_{\ell=0}^{k+1} \lambda_n^{-(k+2-\ell)\ep_0}\lambda_n^{1-\ell}\lambda_n^{\f{3\ep_0}{p}} \ls \lambda_n^{-k-\ep_0}\lambda_n^{\f{3\ep_0}{p}}.
\end{split}
\end{equation}
Therefore, by \eqref{eta.of.const.box} and \eqref{commutator.of.two.derivatives.with.eta}, we obtain that for $k=0,1,2$,
$$\| \rd^k \widetilde{\Box}_{c,\alp} (\zeta_\alp\chi (\psi_n-\psi_0))\|_{L^p} \ls \lambda_n^{-k-1+\ep_0}\lambda_n^{\f{3\ep_0}{p}} + \lambda^{-k-\ep_0}\lambda_n^{\f{3\ep_0}{p}} \ls \lambda_n^{-k-1+\ep_0}\lambda_n^{\f{3\ep_0}{p}},$$
where in the last estimate we used $\ep_0<\f 12$. \qedhere
\end{proof}

\subsection{Main preliminary reduction}\label{sec:reduction}

\begin{proposition}\label{prop:only.even}
Let $a(x,\xi) = b(x) m(\xi)$. Suppose $m(\xi)$ is homogeneous of order 0 and is \underline{odd}, i.e.~$m(\xi) = -m(-\xi)$ for all $\xi \in \mathbb R^3$. Then 
$$\int_{S^*\mathbb R^{2+1}} \Big( (g_0^{-1})^{\alp\bt}\xi_\bt \rd_{x^{\alp}}(\f{(\xi_t - \bt_0^i \xi_i) a}{N_0}) - \f 12 (\rd_\mu g_0^{-1})^{\alp\bt} \xi_\alp \xi_\bt \rd_{\xi_\mu}(\f{(\xi_t - \bt_0^i \xi_i) a}{N_0}) \Big)  \, \f{\ud \nu^\psi}{|\xi|^2} = 0$$
and 
$$\int_{S^*\mathbb R^{2+1}} e^{-4\psi_0} \Big( (g_0^{-1})^{\alp\bt}\xi_\bt \rd_{x^{\alp}}(\f{(\xi_t - \bt_0^i \xi_i) a}{N_0}) - \f 12 (\rd_\mu g_0^{-1})^{\alp\bt} \xi_\alp \xi_\bt \rd_{\xi_\mu}(\f{(\xi_t - \bt_0^i \xi_i) a}{N_0}) \Big) \, \f{\ud \nu^\om}{|\xi|^2} = 0.$$
\end{proposition}
\begin{proof}
We will only prove the first equality as the second one can be achieved in an identical manner.

It is easy to check that $\f{1}{|\xi|^2}(g_0^{-1})^{\alp\bt}\xi_\bt \rd_{x^{\alp}}(\f{(\xi_t - \bt_0^i \xi_i) a}{N_0})$ and $\f{1}{|\xi|^2}(\rd_\mu g_0^{-1})^{\alp\bt} \xi_\alp \xi_\bt \rd_{\xi_\mu}(\f{(\xi_t - \bt_0^i \xi_i) a}{N_0})$ are both odd in $\xi$. In particular, these terms can be written as a finite sum of terms of form $f(x,\xi) = \underline{b}(x) \underline{m}(\xi)$, where $\underline{m}(\xi)$ is homogeneous of degree $0$ and \textbf{odd}.

It therefore suffices to show that for every {such} $f(x,\xi) = \underline{b}(x) \underline{m}(\xi)$, 
$$\int_{S^*\mathbb R^{2+1}} f(x,\xi) \, \ud \nu^{\psi} = 0.$$
Equivalently, since $\xi_t - \bt_0^i \xi_i \neq 0$ on the support of $\ud \nu^\psi$ (by Proposition~\ref{prop:psiom.localized} and the form of the metric), it suffices to show that for $f(x,\xi)= \underline{b}(x) \underline{m}(\xi)$ as above,
\begin{equation}\label{goal.for.f}
\int_{S^*\mathbb R^{2+1}} f(x,\xi) \f{(\xi_t - \bt_0^i \xi_i)^2}{|\xi|^2}\, \ud \nu^{\psi} = 0.
\end{equation}

To proceed, given $\underline{b}$ as above, we freeze the coefficients as in Proposition~\ref{prop:constants} and find constants $\{\underline{b}_{c,\alp}\}_\alp$ adapted to the partition of unity introduced in Section~\ref{sec:freeze.coeff} so that the conclusion of Proposition~\ref{prop:constants} holds. 

Then, using Corollary \ref{cor:nu}, $\sum_\alp \zeta^3_\alp = 1$, Propositions~\ref{prop:constants}, \ref{prop:psi.localized}, \eqref{assumption.0} and \eqref{assumption.1}, the LHS of \eqref{goal.for.f} can be expressed as follows:
\begin{equation}\label{goal.for.f.2}
\begin{split}
\mbox{LHS of \eqref{goal.for.f}}= &\: \lim_{n\to +\infty} \int_{\mathbb R^{2+1}} (\rd_t - \bt^i_0\rd_i)(\chi(\psi_n-\psi_0)) \underline{b}\, \underline{m}(\f 1i\nabla) (\rd_t - \bt^i_0\rd_i)(\chi(\psi_n-\psi_0)) \, \ud x \\
= &\: \sum_{\alp} \lim_{n\to +\infty} \int_{\mathbb R^{2+1}} \zeta_\alp^3 (\rd_t - \bt^i_0\rd_i)(\chi(\psi_n-\psi_0)) \underline{b}\, \underline{m}(\f 1i\nabla) (\rd_t - \bt^i_0\rd_i)(\chi(\psi_n-\psi_0)) \, \ud x \\ 
= &\: \sum_{\alp} \underline{b}_{c,\alp} \lim_{n\to +\infty} \int_{\mathbb R^{2+1}} (\rd_t - \bt^i_0\rd_i)(\zeta_\alp^{\f 32} \chi(\psi_n-\psi_0)) \underline{m}(\f 1i\nabla) (\rd_t - \bt^i_0\rd_i)(\zeta_\alp^{\f 32} \chi(\psi_n-\psi_0)) \, \ud x. 
\end{split}
\end{equation}
Taking Fourier transform and using that $m$, $\zeta_\alp$, $\chi(\psi_n-\psi)$ are real, we obtain
\begin{equation}\label{goal.for.f.3}
\begin{split}
\mbox{RHS of \eqref{goal.for.f.2}} = &\: \sum_{\alp} \underline{b}_{c,\alp} \lim_{n\to +\infty} \int_{\mathbb R^{2+1}} (\xi_t - \bt^i_0\xi_i)^2 \overline{\widehat{(\zeta_\alp^{\f 32} \chi(\psi_n-\psi_0))}(\xi)} \underline{m}(\xi) \widehat{(\zeta_\alp^{\f 32} \chi(\psi_n-\psi_0))}(\xi) \, \ud \xi \\
= &\: \sum_{\alp} \underline{b}_{c,\alp} \lim_{n\to +\infty} \int_{\mathbb R^{2+1}} (\xi_t - \bt^i_0\xi_i)^2 \widehat{(\zeta_\alp^{\f 32} \chi(\psi_n-\psi_0))}(-\xi) \underline{m}(\xi) \widehat{(\zeta_\alp^{\f 32} \chi(\psi_n-\psi_0))}(\xi) \, \ud \xi.
\end{split}
\end{equation}
Since $\underline{m}$ is odd, a simple change of variable $\xi \mapsto -\xi$ shows that the last line in \eqref{goal.for.f.3} also equals
$$-\sum_{\alp} \underline{b}_{c,\alp} \lim_{n\to +\infty} \int_{\mathbb R^{2+1}} (\xi_t - \bt^i_0\xi_i)^2 \widehat{(\zeta_\alp^{\f 32} \chi(\psi_n-\psi_0))}(-\xi) \underline{m}(\xi) \widehat{(\zeta_\alp^{\f 32} \chi(\psi_n-\psi_0))}(\xi) \, \ud \xi.$$
This then implies that the term is identically zero as desired. \qedhere
\end{proof}

\begin{proposition}\label{prop:main.reduction}
Let $\ud \nu = 2\ud \nu^\psi + \f 12 e^{-4\psi_0} \, \ud \nu^\om$. Suppose the following holds for all $a(x,\xi) = b(x) m(\xi)$ with $b$, $m$ smooth, real, $m$ homogeneous of order $0$ and $m$ \underline{even}:
\begin{equation}\label{main.transport.1}
\int_{S^*\mathbb R^{2+1}} ((g_0^{-1})^{\alp\bt}\xi_\bt \rd_{x^{\alp}}(\f{(\xi_t - \bt_0^i \xi_i) a}{N_0}) - \f 12 (\rd_\mu g_0^{-1})^{\alp\bt} \xi_\alp \xi_\bt \rd_{\xi_\mu}(\f{(\xi_t - \bt_0^i \xi_i) a}{N_0})) \, \f{\ud \nu}{|\xi|^2} = 0.
\end{equation}

Then in fact \eqref{main.transport.1} holds for \underline{all} smooth real $a(x,\xi)$ which are homogeneous of order $0$.
\end{proposition}
\begin{proof}
By a standard density argument using the Stone--Weierstrass theorem we can reduce to the case where $a(x,\xi)$ takes the form $a(x,\xi) = \sum_{\mathrm{finite}} b_k(x) m_k(\xi)$. It therefore suffices to consider $a(x,\xi) = b(x) m(\xi)$. Decompose $m$ into its odd and even parts. Proposition~\ref{prop:only.even} shows that the odd part must give a zero contribution to \eqref{main.transport.1}. The conclusion follows. \qedhere
\end{proof}

\textbf{From now on we assume that $a(x,\xi) = b(x) m(\xi)$ and that $b(x)$ is real and $m(\xi)$ is real and even.} {Moreover, we will take $A$ to be a $0$-th order pseudo-differential operator $A = b(x)\widetilde{m}(\f 1i \nab)$, where $\widetilde{m}(\xi)$ is a smooth real-valued even function such that $\widetilde{m}(\xi) = m(\xi)$ for $|\xi|\geq 1$.}

One consequence of the evenness assumption is the following.
\begin{proposition}\label{prop:real.output}
Let {$A = b(x) \widetilde{m}(\f 1i\nabla)$}, where $b(x)$ is real{,} $\widetilde{m}(\xi)$ is {real and even and agrees with a real, even, homogeneous of order $0$ $m(\xi)$ for $|\xi|\geq 1$}. Then for any real function $\phi\in L^2$, we have $A\phi \in L^2$ and $A^*\phi \in L^2$ are both real (where $A^*=\widetilde{m}(\f 1i\nabla) b(x)$ denotes the $L^2$-adjoint of $A$).
\end{proposition}
\begin{proof}
It suffices to show that $\widetilde{m}(\f 1i\nabla) \phi$ is real. First, since $\phi$ is real, we have $\overline{\widehat{\phi}(\xi)} = \widehat{\phi}(-\xi)$.
Hence, 
$$\overline{\widehat{\Big[\widetilde{m}(\f 1i\nabla) \phi\Big]}(\xi)} =\overline{ \widetilde{m}(\xi)\widehat{\phi}(\xi)} = \widetilde{m}(\xi)\widehat{\phi}(-\xi) = \widetilde{m}(-\xi)\widehat{\phi}(-\xi) = \widehat{\Big[\widetilde{m}(\f 1i\nabla) \phi\Big]}(-\xi).$$
This implies that $\widetilde{m}(\f 1i\nabla) \phi$ is real. \qedhere
\end{proof}

\section{Energy identities}\label{sec:energy.id}

We continue to work under the assumptions of Theorem~\ref{thm:main} and the reductions in Sections~\ref{sec:compact.reduction} and \ref{sec:reduction}. Let $A$ be a $0$-th order pseudo-differential operator given by $A = b(x) \widetilde{m}(\f 1i\nabla)$, where the principal symbol {$a(x,\xi) = b(x) m(\xi)$ (with $m(\xi) = \widetilde{m}(\xi)$ for $|\xi|\geq 1$)} is real and supported in $T^*\Omega${,} $m(\xi)$ is homogeneous of order $0${, and $m$ and $\widetilde{m}$ are both even}.

In this section, we derive the main energy identities that will be used to prove the transport equation for the microlocal defect measure. We first introduce some notations in \textbf{Section~\ref{sec:def.boxes}}. In \textbf{Section~\ref{sec:id.0}} and \textbf{Section~\ref{sec:id.n}}, we will then derive respectively energy identities using the equations satisfied by $(\psi_0,\om_0)$ and $(\psi_n,\om_n)$.

\subsection{Definitions of $\Box_{g_0,A}$ and $\Box_{g_n,A}$}\label{sec:def.boxes}

{Using \eqref{eq:wave} for the metric $g_0$, we obtain}
\begin{equation}\label{eq:Box.g0}
\begin{split}
\Box_{g_0}\phi = -\f{e^{-2\gamma_0}}{N_0} \rd_t\big(\f{e^{2\gamma_0}}{N_0} (e_0)_0\phi \big) + \f{e^{-2\gamma_0}}{N_0} \rd_i \big(\f{\bt_0^i e^{2\gamma_0}}{N_0} (e_0)_0 \phi \big) + \f{e^{-2\gamma_0}}{N_0} \de^{ij} \rd_i \big(N_0 \rd_j\phi\big).
\end{split}
\end{equation}
Similarly,
\begin{equation}\label{eq:Box.gn}
\begin{split}
\Box_{g_n}\phi = -\f{e^{-2\gamma_n}}{N_n} \rd_t(\f{e^{2\gamma_n}}{N_n} \big(e_0)_n\phi \big) + \f{e^{-2\gamma_n}}{N_n} \rd_i \big(\f{\bt_n^i e^{2\gamma_n}}{N_n} (e_0)_n \phi \big) + \f{e^{-2\gamma_n}}{N_n} \de^{ij} \rd_i \big(N_n \rd_j\phi\big).
\end{split}
\end{equation}

Define the operator $\Box_{g_0,A}$ by
\begin{equation}\label{def:Box.g0}
\begin{split}
\Box_{g_0,A}\phi = -\f{e^{-2\gamma_0}}{N_0} \rd_t [e^{2\gamma_0} A(\f{(e_0)_0\phi}{N_0})] + \f{e^{-2\gamma_0}}{N_0} \rd_i \big[\bt_0^i e^{2\gamma_0} A(\f{(e_0)_0 \phi}{N_0}) \big] + \f{e^{-2\gamma_0}}{N_0} \de^{ij} \rd_i \big[N_0^2 A(\f{\rd_j\phi}{N_0}) \big].
\end{split}
\end{equation}
Similarly, for every $n\in \mathbb N$, define the operator $\Box_{g_n,A}$ by
\begin{equation}\label{def:Box.gn}
\begin{split}
\Box_{g_n,A}\phi = -\f{e^{-2\gamma_n}}{N_n} \rd_t \big[e^{2\gamma_n} A(\f{(e_0)_n\phi}{N_n}) \big] + \f{e^{-2\gamma_n}}{N_n} \rd_i \big[\bt_n^i e^{2\gamma_n} A(\f{(e_0)_n \phi}{N_n}) \big] + \f{e^{-2\gamma_n}}{N_n} \de^{ij} \rd_i \big[N_n^2 A(\f{\rd_j\phi}{N_n})\big].
\end{split}
\end{equation}

\subsection{Energy identities for $(\psi_0,\om_0)$}\label{sec:id.0}

We first derive some energy identities by directly integrating by parts. {One can view these as analogues of the the standard energy identities for the wave map system for $(\psi, \om)$ with multiplier $\f{e_0}{N_0}$, though now we also need to take into account the contribution of the pseudo-differential operator $A$.}

\begin{proposition}\label{prop:energy.id.0.1}
The following identities hold:
\begin{equation}\label{eq:energy.id.0.1}
\begin{split}
&\: \int_{\mathbb R^{2+1}} \f{(e_0)_0 (\chi\psi_0)}{N_0} \Box_{g_0,A}(\chi\psi_0)  \,\mathrm{dVol}_{g_0} + \int_{\mathbb R^{2+1}} A \Big(\f{(e_0)_0 (\chi\psi_0)}{N_0}\Big)  \Box_{g_0}(\chi\psi_0)  \,\mathrm{dVol}_{g_0}\\
=&\: - \int_{\mathbb R^{2+1}} \f{(e_0)_0 (\chi\psi_0)}{N_0} (\rd_t e^{2\gamma_0}) A \Big(\f{(e_0)_0 (\chi\psi_0)}{N_0}\Big) \,\ud x + \int_{\mathbb R^{2+1}} \f{(e_0)_0 (\chi\psi_0)}{N_0} (\rd_i (\bt_0^i e^{2\gamma_0})) A\Big(\f{(e_0)_0 (\chi\psi_0)}{N_0}\Big) \,\ud x\\
&\: + \int_{\mathbb R^{2+1}} N_0\rd_i (\chi\psi_0) \de^{ij} \Big\{ [A, \rd_j](\f{(e_0)_0(\chi\psi_0)}{N_0})\Big\} \,\ud x + \int_{\mathbb R^{2+1}} [\rd_i (\chi\psi_0)] \de^{ij} N_0 \Big\{[(e_0)_0,A](\f{\rd_j(\chi\psi_0)}{N_0}) \Big\} \,\ud x\\
&\: - \int_{\mathbb R^{2+1}} [\rd_i (\chi\psi_0)] \de^{ij} (\rd_k \bt^k_0) N_0 \Big[A(\f{\rd_j(\chi\psi_0)}{N_0}) \Big] \,\ud x  + \int_{\mathbb R^{2+1}} [\rd_i (\chi\psi_0)] \de^{ij} [(e_0)_0 N_0] \Big[A(\f{\rd_j(\chi\psi_0)}{N_0})\Big] \,\ud x  \\
&\: + \int_{\mathbb R^{2+1}} \Big[\rd_i (\chi\psi_0)] \de^{ij} N_0 [A(\f{((e_0)_0(\chi\psi_0))(\rd_j N_0)}{N_0^2} +\f{(\rd_j\bt_0^k)\rd_k(\chi\psi_0)}{N_0} - \f{((e_0)_0 N_0)\rd_j(\chi\psi_0)}{N_0^2}) \Big] \,\ud x \\
&\: + \int_{\mathbb R^{2+1}} (\rd_i \bt_0^k) [\rd_k (\chi\psi_0)] \de^{ij}  N_0 \Big[A(\f{\rd_j(\chi\psi_0)}{N_0}) \Big] \,\ud x + \int_{\mathbb R^{2+1}} [(e_0)_0 (\chi\psi_0)] (\rd_i N_0) \de^{ij}  \Big[A(\f{\rd_j(\chi\psi_0)}{N_0}) \Big]\,\ud x,
\end{split}
\end{equation}
and
\begin{equation}\label{eq:energy.id.0.3}
\begin{split}
&\:  \f 14 \int_{\mathbb R^{2+1}} e^{-4\psi_0} \f{(e_0)_0 (\chi\om_0)}{N_0} \Box_{g_0,A}(\chi\om_0)  \,\mathrm{dVol}_{g_0} + \f 14 \int_{\mathbb R^{2+1}} e^{-4\psi_0}  A(\f{(e_0)_0 (\chi\om_0)}{N_0})  \Box_{g_0}(\chi\om_0)  \,\mathrm{dVol}_{g_0}\\
=&\: \f 14 \int_{\mathbb R^{2+1}} e^{-4\psi_0} \Big\{ - \f{(e_0)_0 (\chi\om_0)}{N_0} (\rd_t e^{2\gamma_0}) A(\f{(e_0)_0 (\chi\om_0)}{N_0}) +  \f{(e_0)_0 (\chi\om_0)}{N_0} (\rd_i (\bt_0^i e^{2\gamma_0})) A(\f{(e_0)_0 (\chi\om_0)}{N_0}) \Big\} \,\ud x\\
&\: + \f 14\int_{\mathbb R^{2+1}} e^{-4\psi_0} \Big\{ N_0\rd_i (\chi\om_0) \de^{ij} \big\{ [A, \rd_j](\f{(e_0)_0(\chi\om_0)}{N_0}) \big\} + [\rd_i (\chi\om_0)] \de^{ij} N_0 \big\{[(e_0)_0,A](\f{\rd_j(\chi\om_0)}{N_0}) \big\} \Big\}\,\ud x\\
&\: - \f 14\int_{\mathbb R^{2+1}} e^{-4\psi_0} \Big\{ [\rd_i (\chi\om_0)] \de^{ij} (\rd_k \bt^k_0) N_0 [A(\f{\rd_j(\chi\om_0)}{N_0})] + [\rd_i (\chi\om_0)] \de^{ij} [(e_0)_0 N_0] [A(\f{\rd_j(\chi\om_0)}{N_0})] \Big\} \,\ud x  \\
&\: + \f 14\int_{\mathbb R^{2+1}} e^{-4\psi_0} [\rd_i (\chi\om_0)] \de^{ij} N_0 \Big[A(\f{((e_0)_0(\chi\om_0))(\rd_j N_0)}{N_0^2} +\f{(\rd_j\bt_0^k)\rd_k(\chi\om_0)}{N_0} - \f{((e_0)_0 N_0)\rd_j(\chi\om_0)}{N_0^2}) \Big] \,\ud x \\
&\: + \f 14\int_{\mathbb R^{2+1}} e^{-4\psi_0} \Big\{ (\rd_i \bt_0^k) [\rd_k (\chi\om_0)] \de^{ij}  N_0 [A(\f{\rd_j(\chi\om_0)}{N_0})] + [(e_0)_0 (\chi\om_0)] (\rd_i N_0) \de^{ij}  [A(\f{\rd_j(\chi\om_0)}{N_0})] \Big\}\,\ud x \\
&\: - \int_{\mathbb R^{2+1}} e^{-4\psi_0} e^{2\gamma_0} ((e_0)_0\psi_0) \f{(e_0)_0 (\chi\om_0)}{N_0}  A\Big(\f{(e_0)_0 (\chi\om_0)}{N_0}\Big) \,\ud x \\
&\: + \int_{\mathbb R^{2+1}} e^{-4\psi_0} \de^{ij} N_0 (\rd_i\psi_0) [ (e_0)_0 (\chi\om_0)]  \Big[A(\f{\rd_j(\chi\om_0)}{N_0})\Big] \, \ud x \\
&\: - \int_{\mathbb R^{2+1}} e^{-4\psi_0} ((e_0)_0 \psi_0)[ \rd_i (\chi\om_0)] \de^{ij} N_0 \Big[A(\f{\rd_j(\chi\om_0)}{N_0})\Big] \, \ud x \\
&\: + \int_{\mathbb R^{2+1}} e^{-4\psi_0} (\rd_j\psi_0)[\rd_i (\chi\om_0)] \de^{ij} N_0 A\Big( \f{(e_0)_0(\chi\om_0)}{N_0}\Big) \, \ud x.
\end{split}
\end{equation}
Here, we recall the definition of $\Box_{g_0,A}$ in \eqref{def:Box.g0}.
\end{proposition}
\begin{proof}
We first prove \eqref{eq:energy.id.0.1}. Consider each term in \eqref{def:Box.g0} and integrate by parts to obtain the following three identities.

\begin{equation}\label{energy.ibp.1}
\begin{split}
&\: - \int_{\mathbb R^{2+1}} \f{(e_0)_0 (\chi\psi_0)}{N_0} \rd_t \Big(e^{2\gamma_0} A(\f{(e_0)_0 (\chi\psi_0)}{N_0})\Big) \,\ud x \\
= &\: - \int_{\mathbb R^{2+1}} \f{(e_0)_0 (\chi\psi_0)}{N_0}  e^{2\gamma_0} \Big[\rd_t (A(\f{(e_0)_0 (\chi\psi_0)}{N_0}))\Big] \,\ud x - \int_{\mathbb R^{2+1}} \f{(e_0)_0 (\chi\psi_0)}{N_0} (\rd_t e^{2\gamma_0}) A\Big(\f{(e_0)_0 (\chi\psi_0)}{N_0}\Big) \,\ud x  \\
= &\: \int_{\mathbb R^{2+1}} \rd_t \Big(e^{2\gamma_0} \f{(e_0)_0 (\chi\psi_0)}{N_0}\Big)    \Big[A(\f{(e_0)_0 (\chi\psi_0)}{N_0})\Big] \,\ud x - \int_{\mathbb R^{2+1}} \f{(e_0)_0 (\chi\psi_0)}{N_0} (\rd_t e^{2\gamma_0}) A\Big(\f{(e_0)_0 (\chi\psi_0)}{N_0} \Big) \,\ud x .
\end{split}
\end{equation}

\begin{equation}\label{energy.ibp.2}
\begin{split}
&\: \int_{\mathbb R^{2+1}} \f{(e_0)_0 (\chi\psi_0)}{N_0} \rd_i \Big(\bt_0^i e^{2\gamma_0} A(\f{(e_0)_0 (\chi\psi_0)}{N_0})\Big) \,\ud x \\
= &\: \int_{\mathbb R^{2+1}} \f{(e_0)_0 (\chi\psi_0)}{N_0}  \bt_0^i e^{2\gamma_0} \Big[\rd_i (A(\f{(e_0)_0 (\chi\psi_0)}{N_0}))\Big] \,\ud x + \int_{\mathbb R^{2+1}} \f{(e_0)_0 (\chi\psi_0)}{N_0} (\rd_i (\bt_0^i e^{2\gamma_0})) A\Big(\f{(e_0)_0 (\chi\psi_0)}{N_0}\Big) \,\ud x\\
= &\: - \int_{\mathbb R^{2+1}} \Big[\rd_i (\bt_0^i e^{2\gamma_0}\f{(e_0)_0 (\chi\psi_0)}{N_0})\Big]   \Big[A(\f{(e_0)_0 (\chi\psi_0)}{N_0})\Big] \,\ud x + \int_{\mathbb R^{2+1}} \f{(e_0)_0 (\chi\psi_0)}{N_0} (\rd_i (\bt_0^i e^{2\gamma_0})) A\Big(\f{(e_0)_0 (\chi\psi_0)}{N_0}\Big) \,\ud x.
\end{split}
\end{equation}

\begin{equation}\label{energy.ibp.3}
\begin{split}
&\: \int_{\mathbb R^{2+1}} \f{(e_0)_0 (\chi\psi_0)}{N_0} \de^{ij} \rd_i \Big[N_0^2 A(\f{\rd_j(\chi\psi_0)}{N_0})\Big] \,\ud x \\
= &\: - \int_{\mathbb R^{2+1}} [ (e_0)_0 \rd_i (\chi\psi_0)] \de^{ij} N_0 \Big[A(\f{\rd_j(\chi\psi_0)}{N_0})\Big] \, \ud x + \int_{\mathbb R^{2+1}} (\rd_i \bt_0^k) [\rd_k (\chi\psi_0)] \de^{ij}  N_0 \Big[A(\f{\rd_j(\chi\psi_0)}{N_0})\Big] \,\ud x \\
&\: + \int_{\mathbb R^{2+1}} [(e_0)_0 (\chi\psi_0)] (\rd_i N_0) \de^{ij}  \Big[A(\f{\rd_j(\chi\psi_0)}{N_0})\Big]\,\ud x\\
= &\: \int_{\mathbb R^{2+1}} [\rd_i (\chi\psi_0)] \de^{ij} [(e_0)_0 N_0] \Big[A(\f{\rd_j(\chi\psi_0)}{N_0})\Big] \,\ud x + \int_{\mathbb R^{2+1}} [\rd_i (\chi\psi_0)] \de^{ij} N_0 \Big\{[e_0,A](\f{\rd_j(\chi\psi_0)}{N_0})\Big\} \,\ud x \\
&\: - \int_{\mathbb R^{2+1}} [\rd_i (\chi\psi_0)] \de^{ij} (\rd_k \bt^k_0) N_0 \Big[A(\f{\rd_j(\chi\psi_0)}{N_0})\Big] \,\ud x \\
&\: + \int_{\mathbb R^{2+1}} [\rd_i (\chi\psi_0)] \de^{ij} N_0 \Big[A(\rd_j \f{(e_0)_0(\chi\psi_0)}{N_0} + \f{((e_0)_0(\chi\psi_0))(\rd_j N_0)}{N_0^2} +\f{(\rd_j\bt_0^k)\rd_k(\chi\psi_0)}{N_0} - \f{((e_0)_0 N_0)\rd_j(\chi\psi_0)}{N_0^2})\Big] \,\ud x \\
&\: + \int_{\mathbb R^{2+1}} (\rd_i \bt_0^k) [\rd_k (\chi\psi_0)] \de^{ij}  N_0 \Big[A(\f{\rd_j(\chi\psi_0)}{N_0})\Big] \,\ud x + \int_{\mathbb R^{2+1}} [(e_0)_0 (\chi\psi_0)] (\rd_i N_0) \de^{ij}  \Big[A(\f{\rd_j(\chi\psi_0)}{N_0})\Big]\,\ud x\\
= &\: - \int_{\mathbb R^{2+1}} \rd_j [N_0\rd_i (\chi\psi_0)] \de^{ij} \Big[A(\f{(e_0)_0(\chi\psi_0)}{N_0})\Big] \,\ud x \\
&\: + \int_{\mathbb R^{2+1}} N_0\rd_i (\chi\psi_0) \de^{ij} \Big\{ [A, \rd_j](\f{(e_0)_0(\chi\psi_0)}{N_0}) \Big\} \,\ud x + \int_{\mathbb R^{2+1}} [\rd_i (\chi\psi_0)] \de^{ij} N_0 \Big\{[(e_0)_0,A](\f{\rd_j(\chi\psi_0)}{N_0})\Big\} \,\ud x\\
&\: - \int_{\mathbb R^{2+1}} [\rd_i (\chi\psi_0)] \de^{ij} (\rd_k \bt^k_0) N_0 \Big[A(\f{\rd_j(\chi\psi_0)}{N_0})\Big] \,\ud x  + \int_{\mathbb R^{2+1}} [\rd_i (\chi\psi_0)] \de^{ij} [(e_0)_0 N_0] \Big[A(\f{\rd_j(\chi\psi_0)}{N_0})\Big] \,\ud x  \\
&\: + \int_{\mathbb R^{2+1}} [\rd_i (\chi\psi_0)] \de^{ij} N_0 \Big[A(\f{((e_0)_0(\chi\psi_0))(\rd_j N_0)}{N_0^2} +\f{(\rd_j\bt_0^k)\rd_k(\chi\psi_0)}{N_0} - \f{((e_0)_0 N_0)\rd_j(\chi\psi_0)}{N_0^2})\Big] \,\ud x \\
&\: + \int_{\mathbb R^{2+1}} (\rd_i \bt_0^k) [\rd_k (\chi\psi_0)] \de^{ij}  N_0 \Big[A(\f{\rd_j(\chi\psi_0)}{N_0})\Big] \,\ud x + \int_{\mathbb R^{2+1}} [(e_0)_0 (\chi\psi_0)] (\rd_i N_0) \de^{ij}  \Big[A(\f{\rd_j(\chi\psi_0)}{N_0})\Big]\,\ud x.
\end{split}
\end{equation}

Combining \eqref{energy.ibp.1}--\eqref{energy.ibp.3}, and recalling \eqref{eq:Box.g0} and \eqref{def:Box.g0}, we obtain \eqref{eq:energy.id.0.1}.

Now the proof of \eqref{eq:energy.id.0.3} is similar, except that since there is an $e^{-4\psi_0}$ weight, we need to handle the extra (four) terms arising from differentiating  $e^{-4\psi_0}$. We omit the details. \qedhere

\end{proof}

Using the equations derived in Proposition~\ref{prop:limitwave}, we obtain the following energy identities, which give different ways of expressing \eqref{eq:energy.id.0.1} and \eqref{eq:energy.id.0.3}.

\begin{proposition}\label{prop:energy.id.0.2}
Let $$F^\psi_0:=2 g_0^{-1} (\ud \chi,\ud \psi_0) + \psi_0 \Box_{g_0} \chi -\f 12 \chi e^{-4\psi_0} g^{-1}_0 (\ud \om_0, \ud \om_0).$$ Then
\begin{equation}\label{eq:energy.id.0.2}
\begin{split}
&\: \int_{\mathbb R^{2+1}} \f{(e_0)_0 (\chi\psi_0)}{N_0} \Box_{g_0,A}(\chi\psi_0)  \,\mathrm{dVol}_{g_0} + \int_{\mathbb R^{2+1}} A(\f{(e_0)_0 (\chi\psi_0)}{N_0})  \Box_{g_0}(\chi\psi_0)  \,\mathrm{dVol}_{g_0}\\
=&\: \int_{\mathbb R^{2+1}} \f{(e_0)_0 (\chi\psi_0)}{N_0} \Big(\Box_{g_0,A} (\chi\psi_0) - \f{1}{\sqrt{-\det g_0}} A (\sqrt{-\det g_0}\Box_{g_0} (\chi\psi_0))\Big) \,\mathrm{dVol}_{g_0} \\
&\: + \int_{\mathbb R^{2+1}} \f{(e_0)_0 (\chi\psi_0)}{N_0} A (\sqrt{-\det g_0} F_0^\psi) \, \ud x + \int_{\mathbb R^{2+1}} A\Big(\f{(e_0)_0 (\chi\psi_0)}{N_0}\Big)  F_0^\psi  \,\sqrt{-\det g_0} \,\ud x.
\end{split}
\end{equation}
Similarly, let $$F^\om_0 = 2 g_0^{-1} (\ud \chi, \ud \om_0) + \om_0 \Box_{g_0} \chi + 4 \chi g^{-1}_0 (\ud \om_0, \ud \psi_0).$$ Then
\begin{equation}\label{eq:energy.id.0.4}
\begin{split}
&\: \f 14\int_{\mathbb R^{2+1}} e^{-4\psi_0}\f{(e_0)_0 (\chi\om_0)}{N_0} \Box_{g_0,A}(\chi\om_0)  \,\mathrm{dVol}_{g_0} + \f 14 \int_{\mathbb R^{2+1}} e^{-4\psi_0} A\Big(\f{(e_0)_0 (\chi\om_0)}{N_0}\Big)  \Box_{g_0}(\chi\om_0)  \,\mathrm{dVol}_{g_0}\\
=&\: \f 14\int_{\mathbb R^{2+1}} e^{-4\psi_0} \f{(e_0)_0 (\chi\om_0)}{N_0} \Big(\Box_{g_0,A} (\chi\om_0) - \f{1}{\sqrt{-\det g_0}} A (\sqrt{-\det g_0}\Box_{g_0} (\chi\om_0))\Big) \,\mathrm{dVol}_{g_0} \\
&\: + \f 14 \int_{\mathbb R^{2+1}} e^{-4\psi_0}\f{(e_0)_0 (\chi\om_0)}{N_0} A (\sqrt{-\det g_0} F_0^\om) \, \ud x + \f 14\int_{\mathbb R^{2+1}} {e^{-4\psi_0}} A\Big(\f{(e_0)_0 (\chi\om_0)}{N_0}\Big)  F_0^\om  \,\sqrt{-\det g_0} \,\ud x.
\end{split}
\end{equation}

\end{proposition}
\begin{proof}
This is an obvious consequence of 
$\Box_{g_0} (\chi\psi_0) = F_0^\psi$ and $\Box_{g_0} (\chi\om_0) = F_0^\om$
(which holds by Proposition~\ref{prop:limitwave}). \qedhere
\end{proof}

\subsection{Energy identities of $(\psi_n,\om_n)$}\label{sec:id.n}

We now derive analogues of Propositions~\ref{prop:energy.id.0.1} and \ref{prop:energy.id.0.2} with $(\psi_0,\om_0)$ replaced by $(\psi_n,\om_n)$. The results are given in Proposition~\ref{prop:energy.id.n.1} and \ref{prop:energy.id.n.2} below. Since the proofs are essentially the same as those for Propositions~\ref{prop:energy.id.0.1} and \ref{prop:energy.id.0.2}, they are omitted.

\begin{proposition}\label{prop:energy.id.n.1}
\begin{equation}\label{eq:energy.id.n.1}
\begin{split}
&\: \int_{\mathbb R^{2+1}} \f{(e_0)_n (\chi\psi_n)}{N_n} \Box_{g_n,A}(\chi\psi_n)  \,\mathrm{dVol}_{g_n} + \int_{\mathbb R^{2+1}} A(\f{(e_0)_n (\chi\psi_n)}{N_n})  \Box_{g_0}(\chi\psi_n)  \,\mathrm{dVol}_{g_n}\\
=&\: \underbrace{- \int_{\mathbb R^{2+1}} \f{(e_0)_n (\chi\psi_n)}{N_n} (\rd_t e^{2\gamma_n}) A(\f{(e_0)_n (\chi\psi_n)}{N_n}) \,\ud x}_{=:\mathrm{easy}_1} + \underbrace{\int_{\mathbb R^{2+1}} \f{(e_0)_n(\chi\psi_n)}{N_n} (\rd_i (\bt_n^i e^{2\gamma_n})) A(\f{(e_0)_n (\chi\psi_n)}{N_n}) \,\ud x}_{=:\mathrm{easy}_2}\\
&\: + \underbrace{\int_{\mathbb R^{2+1}} N_n\rd_i (\chi\psi_n) \de^{ij} \{ [A, \rd_j](\f{(e_0)_n(\chi\psi_n)}{N_n})\} \,\ud x}_{=:\mathrm{easy}_3} + \underbrace{\int_{\mathbb R^{2+1}} [\rd_i (\chi\psi_n)] \de^{ij} N_n \{[(e_0)_n,A](\f{\rd_j(\chi\psi_n)}{N_n})\} \,\ud x}_{=:\mathrm{hard}}\\
&\: \underbrace{- \int_{\mathbb R^{2+1}} [\rd_i (\chi\psi_n)] \de^{ij} (\rd_k \bt^k_n) N_n [A(\f{\rd_j(\chi\psi_n)}{N_n})] \,\ud x}_{=:\mathrm{easy}_4}  + \underbrace{\int_{\mathbb R^{2+1}} [\rd_i (\chi\psi_n)] \de^{ij} [(e_0)_n N_n] [A(\f{\rd_j(\chi\psi_n)}{N_n})] \,\ud x}_{=:\mathrm{medium}_1}  \\
&\: + \underbrace{\int_{\mathbb R^{2+1}} [\rd_i (\chi\psi_n)] \de^{ij} N_n A(\f{((e_0)_n(\chi\psi_n))(\rd_j N_n)}{N_n^2}) \,\ud x}_{=:\mathrm{easy}_5} + \underbrace{\int_{\mathbb R^{2+1}} [\rd_i (\chi\psi_n)] \de^{ij} N_n A (\f{(\rd_j\bt_n^k)\rd_k(\chi\psi_n)}{N_n}) \,\ud x}_{=:\mathrm{easy}_6}\\
&\: \underbrace{- \int_{\mathbb R^{2+1}} [\rd_i (\chi\psi_n)] \de^{ij} N_n A(\f{((e_0)_n N_n)\rd_j(\chi\psi_n)}{N_n^2}) \,\ud x}_{=:\mathrm{medium}_2}  + \underbrace{\int_{\mathbb R^{2+1}} (\rd_i \bt_n^k) [\rd_k (\chi\psi_n)] \de^{ij}  N_n [A(\f{\rd_j(\chi\psi_n)}{N_n})] \,\ud x}_{=:\mathrm{easy}_7} \\
&\: + \underbrace{\int_{\mathbb R^{2+1}} [(e_0)_n (\chi\psi_n)] (\rd_i N_n) \de^{ij}  [A(\f{\rd_j(\chi\psi_n)}{N_n})]\,\ud x}_{=:\mathrm{easy}_8},
\end{split}
\end{equation}
\begin{equation}\label{eq:energy.id.n.3}
\begin{split}
&\:  \f 14 \int_{\mathbb R^{2+1}} e^{-4\psi_0} \f{(e_0)_n (\chi\om_n)}{N_n} \Box_{g_0,A}(\chi\om_n)  \,\mathrm{dVol}_{g_n} + \f 14 \int_{\mathbb R^{2+1}} e^{-4\psi_0}  A(\f{(e_0)_n (\chi\om_n)}{N_n})  \Box_{g_n}(\chi\om_n)  \,\mathrm{dVol}_{g_0}\\
=&\: \f 14 \int_{\mathbb R^{2+1}} e^{-4\psi_0} \{- \f{(e_0)_n (\chi\om_n)}{N_n} (\rd_t e^{2\gamma_n}) A(\f{(e_0)_n (\chi\om_n)}{N_n}) +  \f{(e_0)_n (\chi\om_n)}{N_n} (\rd_i (\bt_n^i e^{2\gamma_n})) A(\f{(e_0)_n (\chi\om_n)}{N_n}) \} \,\ud x\\
&\: + \f 14\int_{\mathbb R^{2+1}} e^{-4\psi_0} \{ N_n\rd_i (\chi\om_n) \de^{ij} \{ [A, \rd_j](\f{(e_0)_n(\chi\om_n)}{N_n})\} + [\rd_i (\chi\om_n)] \de^{ij} N_n \{[(e_0)_n,A](\f{\rd_j(\chi\om_n)}{N_n})\} \}\,\ud x\\
&\: - \f 14\int_{\mathbb R^{2+1}} e^{-4\psi_0} \{ [\rd_i (\chi\om_n)] \de^{ij} (\rd_k \bt^k_n) N_n [A(\f{\rd_j(\chi\om_n)}{N_n})] + [\rd_i (\chi\om_n)] \de^{ij} [(e_0)_n N_n] [A(\f{\rd_j(\chi\om_n)}{N_n})] \} \,\ud x  \\
&\: + \f 14\int_{\mathbb R^{2+1}} e^{-4\psi_0} [\rd_i (\chi\om_n)] \de^{ij} N_n [A(\f{((e_0)_n(\chi\om_n))(\rd_j N_n)}{N_n^2} +\f{(\rd_j\bt_n^k)\rd_k(\chi\om_n)}{N_n} - \f{((e_0)_n N_n)\rd_j(\chi\om_n)}{N_n^2})] \,\ud x \\
&\: + \f 14\int_{\mathbb R^{2+1}} e^{-4\psi_0} \{ (\rd_i \bt_n^k) [\rd_k (\chi\om_n)] \de^{ij}  N_n [A(\f{\rd_j(\chi\om_n)}{N_n})] + [(e_0)_n (\chi\om_n)] (\rd_i N_n) \de^{ij}  [A(\f{\rd_j(\chi\om_n)}{N_n})] \}\,\ud x \\
&\: \underbrace{- \int_{\mathbb R^{2+1}} e^{-4\psi_0} e^{2\gamma_n} ((e_0)_n\psi_0) \f{(e_0)_n (\chi\om_n)}{N_n}  A(\f{(e_0)_n (\chi\om_n)}{N_n}) \,\ud x}_{=:\mathrm{extra}_1} \\
&\: + \underbrace{\int_{\mathbb R^{2+1}} e^{-4\psi_0} \de^{ij} N_n (\rd_i\psi_0) [ (e_0)_0 (\chi\om_n)]  [A(\f{\rd_j(\chi\om_n)}{N_n})] \, \ud x }_{=:\mathrm{extra}_2}\\
&\: \underbrace{- \int_{\mathbb R^{2+1}} e^{-4\psi_0} ((e_0)_n \psi_0)[ \rd_i (\chi\om_n)] \de^{ij} N_n [A(\f{\rd_j(\chi\om_n)}{N_n})] \, \ud x }_{=:\mathrm{extra}_3}\\
&\: + \underbrace{\int_{\mathbb R^{2+1}} e^{-4\psi_0} (\rd_j\psi_0)[\rd_i (\chi\om_n)] \de^{ij} N_n A( \f{(e_0)_n(\chi\om_n)}{N_n}) \, \ud x}_{=:\mathrm{extra}_4}.
\end{split}
\end{equation}
\end{proposition}

\begin{proposition}\label{prop:energy.id.n.2}
Let $$F^\psi_n:=2 g_n^{-1} (\ud \chi,\ud \psi_n) + \psi_n \Box_{g_n} \chi -\f 12 \chi e^{-4\psi_n} g^{-1}_n (\ud \om_n, \ud \om_n).$$ Then
\begin{equation}\label{eq:energy.id.n.2}
\begin{split}
&\: \int_{\mathbb R^{2+1}} \f{(e_0)_n (\chi\psi_n)}{N_n} \Box_{g_n,A}(\chi\psi_n)  \,\mathrm{dVol}_{g_n} + \int_{\mathbb R^{2+1}} A\Big(\f{(e_0)_n (\chi\psi_n)}{N_n}\Big)  \Box_{g_0}(\chi\psi_n)  \,\mathrm{dVol}_{g_n}\\
=&\: \underbrace{\int_{\mathbb R^{2+1}} \f{(e_0)_n (\chi\psi_n)}{N_n} \Big(\Box_{g_n,A} (\chi\psi_n) - \f{1}{\sqrt{-\det g_n}} A (\sqrt{-\det g_n}\Box_{g_n} (\chi\psi_n))\Big) \,\mathrm{dVol}_{g_n}}_{=:\mathrm{main commutator}} \\
&\: + \underbrace{\int_{\mathbb R^{2+1}} \f{(e_0)_n (\chi\psi_n)}{N_n} A (\sqrt{-\det g_n} F_n^\psi) \, \ud x}_{=:\mathrm{trilinear}_1} + \underbrace{\int_{\mathbb R^{2+1}} A\Big(\f{(e_0)_n (\chi\psi_n)}{N_n}\Big)  F_n^\psi  \,\sqrt{-\det g_n} \,\ud x}_{=:\mathrm{trilinear}_2}.
\end{split}
\end{equation}
Similarly, let $$F^\om_n = 2 g_n^{-1} (\ud \chi, \ud \om_n) + \om_n \Box_{g_n} \chi + 4 \chi g^{-1}_n (\ud \om_n, \ud \psi_n).$$ Then
\begin{equation}\label{eq:energy.id.n.4}
\begin{split}
&\: \f 14\int_{\mathbb R^{2+1}} e^{-4\psi_0}\f{(e_0)_n (\chi\om_n)}{N_n} \Box_{g_n,A}(\chi\om_n)  \,\mathrm{dVol}_{g_n} + \f 14 \int_{\mathbb R^{2+1}} e^{-4\psi_0} A\Big(\f{(e_0)_n (\chi\om_n)}{N_n}\Big)  \Box_{g_n}(\chi\om_n)  \,\mathrm{dVol}_{g_n}\\
=&\: \f 14\int_{\mathbb R^{2+1}} e^{-4\psi_0} \f{(e_0)_n (\chi\om_n)}{N_n} \Big(\Box_{g_n,A} (\chi\om_n) - \f{1}{\sqrt{-\det g_n}} A (\sqrt{-\det g_n}\Box_{g_n} (\chi\om_n))\Big) \,\mathrm{dVol}_{g_n} \\
&\: + \underbrace{\f 14 \int_{\mathbb R^{2+1}} e^{-4\psi_0}\f{(e_0)_n (\chi\om_n)}{N_n} A (\sqrt{-\det g_n} F_n^\om) \, \ud x}_{=:\mathrm{trilinear}_3} + \underbrace{\f 14\int_{\mathbb R^{2+1}} {e^{-4\psi_0}} A\Big(\f{(e_0)_n (\chi\om_n)}{N_n}\Big)  F_n^\om  \,\sqrt{-\det g_n} \,\ud x}_{=:\mathrm{trilinear}_4}.
\end{split}
\end{equation}

\end{proposition}

\textbf{Our goal now is to compute the limit of the RHS of \eqref{eq:energy.id.n.1}, \eqref{eq:energy.id.n.3}, \eqref{eq:energy.id.n.2} and \eqref{eq:energy.id.n.4}} as $n\to +\infty$ (allowing possibly passing to a subsequence). We then compare the resulting expression with the RHS of \eqref{eq:energy.id.0.1}, \eqref{eq:energy.id.0.3}, \eqref{eq:energy.id.0.2} and \eqref{eq:energy.id.0.4} to derive an equation for $\ud \nu$. This task will be the goal of {Sections~\ref{sec:energy.id.n.1}}--\ref{sec:final} below.

\section{Terms in Proposition~\ref{prop:energy.id.n.1}}\label{sec:energy.id.n.1}

We continue to work under the assumptions of Theorem~\ref{thm:main} and the reductions in Sections~\ref{sec:compact.reduction} and \ref{sec:reduction}. As above, {let $A$ be a $0$-th order pseudo-differential operator given by $A = b(x) \widetilde{m}(\f 1i\nabla)$, where the principal symbol $a(x,\xi) = b(x) m(\xi)$ (with $m(\xi) = \widetilde{m}(\xi)$ for $|\xi|\geq 1$) is real and supported in $T^*\Omega${,} $m(\xi)$ is homogeneous of order $0$, and $m$ and $\widetilde{m}$ are both even.}

Our goal in this section is to compute the limit (as $n\to +\infty$) of the terms on the RHSs of \eqref{eq:energy.id.n.1} and \eqref{eq:energy.id.n.3} in Proposition~\ref{prop:energy.id.n.1}. We will focus mainly on \eqref{eq:energy.id.n.1}. The terms in \eqref{eq:energy.id.n.3} can be treated mostly in a similar manner; we will explain the additional details in Proposition~\ref{prop:RHS.3}.

The terms on the RHS of \eqref{eq:energy.id.n.1} labeled as ``$\mathrm{easy}$'' will be treated in \textbf{Section~\ref{sec:easy}}. The terms on RHS of \eqref{eq:energy.id.n.1} labeled as ``$\mathrm{medium}$'' will then be treated in \textbf{Section~\ref{sec:not.so.easy}}. Note that the ``$\mathrm{hard}$'' terms will not be dealt with but need to be combined with other terms later. We then conclude the section in \textbf{Section~\ref{sec:metric.conclude}}.

\subsection{The easier terms}\label{sec:easy}

\begin{proposition}
As $n\to +\infty$, for $\mathrm{easy}_i$ being the terms in \eqref{eq:energy.id.n.1},
\begin{equation*}
\begin{split}
\sum_{i=1}^8 \mathrm{easy}_i \to &\: \mbox{corresponding terms on the RHS of \eqref{eq:energy.id.0.1}} \\
&\: -2\int_{S^*\mathbb R^{2+1}} \Big((g^{-1}_0)^{\alp\bt} (\rd_\bt X^\gamma) \xi_\alp \xi_\gamma -\f 12 X^\mu \rd_\mu (g^{-1}_0)^{\alp\gamma}\xi_\alp \xi_\gamma\Big) a \,\f{d\nu^\psi}{|\xi|^2} \\
&\: + \int_{S^*\mathbb R^{2+1}} \Big[-\de^{ij} \xi_i (\xi_t-\bt^k_0\xi_k) \rd_{x^j} a \Big]\, \f{e^{-2\gamma_0}}{N_0}\f{d\nu^\psi}{|\xi|^2},
\end{split}
\end{equation*}
where $X = \f 1{N_0} (\rd_t -\bt^i_0 \rd_i)$.
\end{proposition}
\begin{proof}
\pfstep{Step~1: Taking limits for the metric quantities} In all the ``$\mathrm{easy}_i$'' terms for $i\neq 3$, note that we have the appearance of the metric components $\gamma_n$, $\log N_n$, $\bt^j_n$ and the following derivatives $\rd_i\gamma_n$, $\rd_i \log N_n$, $\rd_i \bt^j_n$ and $\rd_t\gamma_n$. In other words, there are \underline{no} appearances of $\rd_t \bt_n^j$ and $\rd_t \log N_n$. 

Therefore, by the estimates in \eqref{assumption.0} and \eqref{assumption.1} and the convergence statements for $\rd_i\gamma_n$, $\rd_i \log N_n$, $\rd_i \bt^j_n$ and $\rd_t\gamma_n$ in Propositions~\ref{prop:spatial.imp} and \ref{prop:dtgamma.imp}, all the $\mathrm{easy}_i$ terms have the same limit (as $n\to +\infty$) if we replace all the $(\gamma_n, \log N_n, \bt^j_n)$ by $(\gamma_0, \log N_0, \bt^j_0)$. {(Notice that we can apply Propositions~\ref{prop:spatial.imp} and \ref{prop:dtgamma.imp} here disregarding the factors $\widetilde{\chi}$ because $\widetilde{\chi} \equiv 1$ on $\mathrm{supp}(\chi)$.)} 
For instance, 
$$- \int_{\mathbb R^{2+1}} \f{(e_0)_n (\chi\psi_n)}{N_n} (\rd_t e^{2\gamma_n}) A(\f{(e_0)_n (\chi\psi_n)}{N_n}) \,\ud x + \int_{\mathbb R^{2+1}} \f{(e_0)_0 (\chi\psi_n)}{N_0} (\rd_t e^{2\gamma_0}) A(\f{(e_0)_0 (\chi\psi_n)}{N_0}) \,\ud x \to 0.$$
Similarly for other ``$\mathrm{easy}_i$'' terms with $i\neq 3$.

The $i=3$ term is also similar. We only need to note additionally by Lemmas~\ref{lem:PSIDOs}.2 and \ref{lem:PSIDOs}.4 $[A,\rd_j]$ is a bounded $L^2\to L^2$ operator \emph{independent of $n$}. Hence,
$$\int_{\mathbb R^{2+1}} N_n\rd_i (\chi\psi_n) \de^{ij} \{ [A, \rd_j](\f{(e_0)_n(\chi\psi_n)}{N_n})\} \,\ud x - \int_{\mathbb R^{2+1}} N_0\rd_i (\chi\psi_n) \de^{ij} \{ [A, \rd_j](\f{(e_0)_0(\chi\psi_n)}{N_0})\} \,\ud x \to 0.$$

\pfstep{Step~2: Using the microlocal defect measures} After the reduction in Step~1, we now use Corollary~\ref{cor:nu} to take the $n\to +\infty$ limits. We treat the $i\neq 3$ (Step~2(a)) and $i=3$ cases (Step~2(b)) separately.

\pfstep{Step~2(a): All terms except for $\mathrm{easy}_3$} Consider now the sum $\displaystyle\sum_{\substack{1\leq i \leq 8 \\ i\neq 3}} \mathrm{easy}_i$. Using Step~1, Corollary~\ref{cor:nu}, and recalling that $\ud x = \f{1}{\sqrt{-\det g_0}}\, \mbox{dVol}_{g_0} = \f{e^{-2\gamma_0}}{N_0}\,\mbox{dVol}_{g_0}$, we see that $\displaystyle\sum_{\substack{1\leq i \leq 8 \\ i\neq 3}} \mathrm{easy}_i$ converges to the corresponding terms on the RHS of \eqref{eq:energy.id.0.1} \emph{plus} the following:
\begin{equation}\label{most.easy}
\begin{split}
&\:\int_{S^*\mathbb R^{2+1}}  (\xi_t - \bt_0^k\xi_k)^2 \f{e^{-2\gamma_0}[-(\rd_t e^{2\gamma_0}) + \rd_i (\bt_0^i e^{2\gamma_0})]}{N_0^3}  a \,\f{d\nu^\psi}{|\xi|^2}  - \int_{S^*\mathbb R^{2+1}} \de^{ij} \xi_i\xi_j \f{e^{-2\gamma_0}(\rd_k \bt^k_0)}{N_0} a \,\f{d\nu^\psi}{|\xi|^2} \\
&\: + \int_{S^*\mathbb R^{2+1}} \de^{ij} \xi_i (\xi_t-\bt^k_0\xi_k)(\f{e^{-2\gamma_0}(\rd_j N_0)}{N_0^2}) a \,\f{d\nu^\psi}{|\xi|^2} + \int_{S^*\mathbb R^{2+1}} \de^{ij} \xi_i \xi_k \f{e^{-2\gamma_0}(\rd_j\bt_n^k)}{N_0} a \,\f{d\nu^\psi}{|\xi|^2} \\
&\: + \int_{S^*\mathbb R^{2+1}}  \de^{ij}  \xi_j \xi_k \f{e^{-2\gamma_0}(\rd_i \bt_0^k)}{N_0} a \,\f{d\nu^\psi}{|\xi|^2} + \int_{S^*\mathbb R^{2+1}}  \de^{ij}\xi_j (\xi_t-\bt^k_0\xi_k) \f{e^{-2\gamma_0}\rd_i N_0}{N_0^2} a \,\f{d\nu^\psi}{|\xi|^2}.
\end{split}
\end{equation}
We now use the fact that $\f 1{N_0^2} (\xi_t -\bt^k_0 \xi_k)^2 = e^{-2\gamma_0} \de^{ij} \xi_i \xi_j$ on the support of $\ud \nu^\psi$ (by Proposition~\ref{prop:psiom.localized}) to derive
\begin{equation}\label{most.easy.2}
\begin{split}
\mbox{\eqref{most.easy}} =&\: - \int_{S^*\mathbb R^{2+1}}  (\xi_t - \bt_0^k\xi_k)^2 \f{e^{-2\gamma_0}(\rd_t - \bt^i_0\rd_i) e^{2\gamma_0} }{N_0^3}  a \,\f{d\nu^\psi}{|\xi|^2} + 2\int_{S^*\mathbb R^{2+1}} \de^{ij} \xi_i (\xi_t-\bt^k_0\xi_k)\f{e^{-2\gamma_0}(\rd_j N_0)}{N_0^2} a \,\f{d\nu^\psi}{|\xi|^2} \\
&\: + 2\int_{S^*\mathbb R^{2+1}} \de^{ij} \xi_i \xi_k \f{e^{-2\gamma_0}(\rd_j\bt_0^k)}{N_0} a \,\f{d\nu^\psi}{|\xi|^2}.
\end{split}
\end{equation}

For $X = \f 1{N_0} (\rd_t - \bt_0^i \rd_i)$, let us also compute {(using \eqref{g.inverse})} that 
\begin{equation}\label{deformation.1}
\begin{split}
 (g^{-1}_0)^{\alp\bt} \rd_\bt X^\gamma \xi_\alp \xi_\gamma  
= &\: -\f 1 {N_0^2} (\rd_t - \bt^j_0 \rd_j)(\f 1{N_0}) \xi_t (\xi_t - \bt^k_0 \xi_k) + \f 1 {N_0^2} (\rd_t -\bt^j_0\rd_j) (\f {\bt_0^i}{N_0}) \xi_i (\xi_t - \bt^k_0 \xi_k) \\
&\: + e^{-2\gamma_0} \de^{ij} \rd_i (\f 1{N_0}) \xi_j \xi_t - e^{-2\gamma_0} \de^{ij}  \rd_i (\f {\bt_0^k}{N_0}) \xi_j \xi_k \\
=&\: + \f 1{N_0^4} ((\rd_t - \bt^j_0 \rd_j)N_0) (\xi_t - \bt^k_0 \xi_k)^2 + \f 1{N_0^3} ((\rd_t -\bt^j_0\rd_j) \bt_0^i)\xi_i (\xi_t - \bt^k_0 \xi_k) \\
&\: - \f 1{N_0^2} e^{-2\gamma_0} \de^{ij} (\rd_i N_0) \xi_j (\xi_t - \bt^k_0 \xi_k) - \f 1{N_0} e^{-2\gamma_0} \de^{ij}  (\rd_i \bt^k_0)\xi_j \xi_k,
\end{split}
\end{equation}
and
\begin{equation}\label{deformation.2}
\begin{split}
&\: \f 12 X^\mu \rd_\mu (g^{-1}_0)^{\alp\gamma}\xi_\alp \xi_\gamma \\
 = &\: - \f 12 \f{1}{N_0} ((\rd_t - \bt^k_0\rd_k)\f 1{N_0^2}) \xi_t^2 + \f{1}{N_0} ((\rd_t - \bt^k_0\rd_k)\f {\bt^i_0}{N_0^2}) \xi_t \xi_i \\
&\: + \f 12 \f{1}{N_0} ((\rd_t - \bt^k_0\rd_k)e^{-2\gamma_0}) \de^{ij} \xi_i\xi_j - \f 12 \f{1}{N_0} ((\rd_t - \bt^k_0\rd_k)(\f{\bt^i_0\bt^j_0}{N_0^2})) \xi_i\xi_j \\
= & \f 1{N_0^4} ((\rd_t - \bt^k_0\rd_k)N_0) (\xi_t-\bt^i_0 \xi_i)^2 + \f{1}{N_0^3} ((\rd_t - \bt^k_0\rd_k)\bt^i_0) (\xi_t -\bt^j_0 \xi_j) \xi_i + \f 12 \f{1}{N_0} ((\rd_t - \bt^k_0\rd_k)e^{-2\gamma_0}) \de^{ij} \xi_i\xi_j \\
= & \f 1{N_0^4} ((\rd_t - \bt^k_0\rd_k)N_0) (\xi_t-\bt^i_0 \xi_i)^2 + \f{1}{N_0^3} ((\rd_t - \bt^k_0\rd_k)\bt^i_0) (\xi_t -\bt^j_0 \xi_j) \xi_i - \f 12 \f{e^{-4\gamma_0}}{N_0} ((\rd_t - \bt^k_0\rd_k)e^{2\gamma_0}) \de^{ij} \xi_i\xi_j .
\end{split}
\end{equation}
Subtracting \eqref{deformation.2} from \eqref{deformation.1}, it follows that 
\begin{equation*}
\begin{split}
&\: (g^{-1}_0)^{\alp\bt} (\rd_\bt X^\gamma) \xi_\alp \xi_\gamma -\f 12 X^\mu \rd_\mu (g^{-1}_0)^{\alp\gamma}\xi_\alp \xi_\gamma \\
= &\: - \f 1{N_0^2} e^{-2\gamma_0} \de^{ij} (\rd_i N_0) \xi_j (\xi_t - \bt^k_0 \xi_k) - \f 1{N_0} e^{-2\gamma_0} \de^{ij}  (\rd_i \bt^k_0)\xi_j \xi_k + \f 12 \f{e^{-4\gamma_0}}{N_0} ((\rd_t - \bt^k_0\rd_k)e^{2\gamma_0}) \de^{ij} \xi_i\xi_j.
\end{split}
\end{equation*}

By inspection, we have proven that
\begin{equation*}
\begin{split}
\mbox{\eqref{most.easy.2}} = -2\int_{S^*\mathbb R^{2+1}} \Big((g^{-1}_0)^{\alp\bt} (\rd_\bt X^\gamma) \xi_\alp \xi_\gamma -\f 12 X^\mu \rd_\mu (g^{-1}_0)^{\alp\gamma}\xi_\alp \xi_\gamma\Big) a \,\f{d\nu^\psi}{|\xi|^2}.
\end{split}
\end{equation*}

\pfstep{Step~2(b): The term $\mathrm{easy}_3$}
By Lemma~\ref{lem:PSIDOs}.2, $[A,\rd_j]$ is a $0$-th order pseudo-differential symbol with principal symbol
$$-i\{a, i \xi_j\} = - \rd_{x^j} a.$$
Therefore, using Corollary~\ref{cor:nu},
\begin{equation*}
\begin{split}
&\: \int_{\mathbb R^{2+1}} N_n\rd_i (\chi\psi_n) \de^{ij} \Big\{ [A, \rd_j](\f{(e_0)_n(\chi\psi_n)}{N_n})\Big\} \,\ud x \\
\to  &\: \int_{\mathbb R^{2+1}} N_0\rd_i (\chi\psi_0) \de^{ij} \Big\{ [A, \rd_j](\f{(e_0)_0(\chi\psi_0)}{N_0})\Big\} \,\ud x + \int_{S^*\mathbb R^{2+1}} \Big[-\de^{ij} \xi_i (\xi_t-\bt^k_0\xi_k) \rd_{x^j} a \Big]\, \f{e^{-2\gamma_0}}{N_0}\f{d\nu^\psi}{|\xi|^2}.
\end{split}
\end{equation*}
Together with Step~1, this gives the desired limit. \qedhere

\end{proof}

\subsection{The not-so-easy terms}\label{sec:not.so.easy}

\begin{proposition}\label{prop:main.medium}
The following holds after passing to a subsequence (which we do not relabel):
\begin{equation}\label{eq:main.medium}
\begin{split}
&\: \int_{\mathbb R^{2+1}} [\rd_i (\chi\psi_n)] \de^{ij} [(e_0)_n N_n] [A(\f{\rd_j(\chi\psi_n)}{N_n})] \,\ud x - \int_{\mathbb R^{2+1}} [\rd_i (\chi\psi_n)] \de^{ij} N_0 A\Big(\f{((e_0)_n N_n)\rd_j(\chi\psi_n)}{N_n^2}\Big) \,\ud x \\
&\: - \int_{\mathbb R^{2+1}} [\rd_i (\chi\psi_n)] \de^{ij} [(e_0)_0 N_0] [A(\f{\rd_j(\chi\psi_n)}{N_0})] \,\ud x + \int_{\mathbb R^{2+1}} [\rd_i (\chi\psi_n)] \de^{ij} N_0 A\Big(\f{((e_0)_0 N_0)\rd_j(\chi\psi_n)}{N_0^2}\Big) \,\ud x \to 0.
\end{split}
\end{equation}
\end{proposition}
\begin{proof}
Using \eqref{assumption.0} and \eqref{assumption.1}, it is easy to see that the two first two terms in \eqref{eq:main.medium} have the same limit as
$$ 
\int_{\mathbb R^{2+1}} [\rd_i (\chi\psi_n)] \de^{ij} [(e_0)_0 N_n] \Big[A(\f{\rd_j(\chi\psi_n)}{N_0})\Big] \,\ud x - \int_{\mathbb R^{2+1}} [\rd_i (\chi\psi_n)] \de^{ij} N_0 A\Big(\f{((e_0)_0 N_n)\rd_j(\chi\psi_n)}{N_0^2}\Big) \,\ud x.
$$
It therefore suffices to show that
\begin{equation}\label{eq:medium.goal}
\begin{split}
&\: \int_{\mathbb R^{2+1}} [\rd_i (\chi\psi_n)] \de^{ij} [(e_0)_0 (N_n-N_0)] \Big[A(\f{\rd_j(\chi\psi_n)}{N_0})\Big] \,\ud x \\
&\: - \int_{\mathbb R^{2+1}} [\rd_i (\chi\psi_n)] \de^{ij} N_0 A\Big(\f{((e_0)_0 (N_n-N_0))\rd_j(\chi\psi_n)}{N_0^2}\Big) \,\ud x \to 0.
\end{split}
\end{equation}

To prove \eqref{eq:medium.goal}, we need to rely further on the structure of the terms. We begin with the following algebraic manipulation.
\begin{equation*}
\begin{split}
\mbox{LHS of \eqref{eq:medium.goal}}= &\: \underbrace{\int_{\mathbb R^{2+1}} [\rd_i (\chi\psi_n)] \de^{ij} [(e_0)_0 (N_n-N_0)] \Big[A(\f{\rd_j(\chi(\psi_n-\psi_0))}{N_0})\Big] \,\ud x}_{=:\mathrm{I}} \\
&\: \underbrace{- \int_{\mathbb R^{2+1}} [\rd_i (\chi(\psi_n-\psi_0))] \de^{ij} N_0 A\Big(\f{((e_0)_0 (N_n-N_0))\rd_j(\chi\psi_n)}{N_0^2}\Big) \,\ud x}_{=:\mathrm{II}} \\
&\: +\underbrace{\int_{\mathbb R^{2+1}} [\rd_i (\chi(\psi_n-\psi_0))] \de^{ij} [(e_0)_0 (N_n-N_0)] \Big[A(\f{\rd_j(\chi\psi_0)}{N_0})\Big] \,\ud x}_{=:\mathrm{III}} \\
&\: \underbrace{- \int_{\mathbb R^{2+1}} [\rd_i (\chi\psi_0)] \de^{ij} N_0 A\Big(\f{((e_0)_0 (N_n-N_0))\rd_j(\chi(\psi_n-\psi_0))}{N_0^2}\Big) \,\ud x}_{=:\mathrm{IV}} \\
&\: +\underbrace{\int_{\mathbb R^{2+1}} [\rd_i (\chi\psi_0)] \de^{ij} [(e_0)_0 (N_n-N_0)] \Big[A(\f{\rd_j(\chi\psi_0)}{N_0})\Big] \,\ud x}_{=:\mathrm{V}} \\
&\: \underbrace{- \int_{\mathbb R^{2+1}} [\rd_i (\chi\psi_0)] \de^{ij} N_0 A\Big(\f{((e_0)_0 (N_n-N_0))\rd_j(\chi\psi_0)}{N_0^2}\Big) \,\ud x}_{=:\mathrm{VI}}.
\end{split}
\end{equation*}
We first consider $\mathrm{I} + \mathrm{II}$. Note that by Proposition~\ref{prop:real.output}, $A^*([\rd_i (\chi(\psi_n-\psi_0))] N_0)$ is real. Therefore,
\begin{equation*}
\begin{split}
&\: \mathrm{I} + \mathrm{II} \\
= &\: \int_{\mathbb R^{2+1}} [\rd_i (\chi\psi_n)] \de^{ij} [(e_0)_0 (N_n-N_0)] \Big[ A(\f{\rd_j(\chi(\psi_n-\psi_0))}{N_0}) \Big] \,\ud x\\
&\: - \int_{\mathbb R^{2+1}} A^*([\rd_i (\chi(\psi_n-\psi_0))] N_0)\de^{ij} \Big(\f{((e_0)_0 (N_n-N_0))\rd_j(\chi\psi_n)}{N_0^2} \Big) \,\ud x \\
=&\: - \int_{\mathbb R^{2+1}} \Big\{(A^*-A)([\rd_i (\chi(\psi_n-\psi_0))] N_0) + [A, N_0^2](\f{\rd_i(\chi(\psi_n-\psi_0))}{N_0}) \Big\}\de^{ij} \Big(\f{((e_0)_0 (N_n-N_0))\rd_j(\chi\psi_n)}{N_0^2}\Big) \,\ud x.
\end{split}
\end{equation*}
Now both $A^*-A$ and $[A, N_0^2]$ are pseudo-differential operator of orders $-1$ by Lemmas~\ref{lem:PSIDOs}.2 and \ref{lem:PSIDOs}.3 (and the fact that $a$ is real). Lemma~\ref{lem:PSIDOs}.5 then implies that after passing to a subsequence, both $(A^*-A)([\rd_i (\chi(\psi_n-\psi_0))] N_0)$ and $[A, N_0^2](\f{\rd_i(\chi(\psi_n-\psi_0))}{N_0})$ converge strongly in the $L^2$ norm. {Since the strong limit must coincide with the weak limit, the $L^2$ limit is in fact $=0$.} The Cauchy--Schwartz inequality then implies that (up to passing to a subsequence) $\mathrm{I}+\mathrm{II}\to 0$.

For the terms $\mathrm{III}$ and $\mathrm{IV}$, we show that they separately tend to $0$. To show each of these convergences, it suffices to show that $((e_0)_0 (N_n-N_0))\rd_j(\chi(\psi_n-\psi_0))$ converges to $0$ weakly in $L^2$, i.e.~the weak limit of the product coincide with the product of the weak limits. This can be viewed as a compensated compactness result: the key is that even though $(e_0)_0 (N_n-N_0)$ does not have a strong limit, we can integrate by parts to take advantage of the fact that $\rd_i (N_n-N_0)$ converges locally uniformly to $0$. More precisely, take $\vartheta \in C^\infty_c(\mathbb R^{2+1})$ (which we can do by a density argument). We then compute
\begin{equation*}
\begin{split}
&\: \int_{\mathbb R^{2+1}} \vartheta ((e_0)_0 (N_n-N_0))\rd_j(\chi(\psi_n-\psi_0)) \, \ud x \\
= &\: - \int_{\mathbb R^{2+1}} \vartheta (N_n-N_0){(e_0)_0}\rd_j(\chi(\psi_n-\psi_0)) \,\ud x + \int_{\mathbb R^{2+1}} [-{(e_0)_0}\vartheta  + \vartheta (\rd_i \bt_0^i)](N_n-N_0)\rd_j(\chi(\psi_n-\psi_0))\, \ud x \\
= &\: \int_{\mathbb R^{2+1}} [(\rd_j\vartheta) (N_n-N_0)+ \vartheta \rd_j (N_n-N_0)]{(e_0)_0}(\chi(\psi_n-\psi_0)) \,\ud x - \int_{\mathbb R^{2+1}} \vartheta (N_n-N_0) (\rd_j \bt_0^i) \rd_i (\chi(\psi_n-\psi_0))\,\ud x\\
&\: + \int_{\mathbb R^{2+1}} [-{(e_0)_0}\vartheta  + \vartheta (\rd_i \bt_0^i)](N_n-N_0)\rd_j(\chi(\psi_n-\psi_0))\, \ud x. 
\end{split}
\end{equation*}
By virtue of \eqref{assumption.0}, \eqref{assumption.1} and Proposition~\ref{prop:spatial.imp}, this $\to 0$.

Finally, $\mathrm{V}$ and $\mathrm{VI}$ both $\to 0$ by virtue of the fact that $((e_0)_0 (N_n-N_0))$ converges weakly in $L^2$ to $0$. We thus conclude the proof of \eqref{eq:medium.goal}. \qedhere
\end{proof}

\begin{proposition}
\begin{equation*}
\begin{split}
&\: \int_{\mathbb R^{2+1}} [\rd_i (\chi\psi_n)] \de^{ij} [(e_0)_n N_n] \Big[A(\f{\rd_j(\chi\psi_n)}{N_n})\Big] \,\ud x - \int_{\mathbb R^{2+1}} [\rd_i (\chi\psi_n)] \de^{ij} N_0 A\Big(\f{((e_0)_n N_n)\rd_j(\chi\psi_n)}{N_n^2}\Big) \,\ud x \\
\to &\: \int_{\mathbb R^{2+1}} [\rd_i (\chi\psi_0)] \de^{ij} [(e_0)_0 N_0] \Big[A(\f{\rd_j(\chi\psi_0)}{N_0})\Big] \,\ud x - \int_{\mathbb R^{2+1}} [\rd_i (\chi\psi_0)] \de^{ij} N_0 A\Big(\f{((e_0)_0 N_0)\rd_j(\chi\psi_0)}{N_0^2}\Big) \,\ud x.
\end{split}
\end{equation*}
\end{proposition}

\begin{proof}
By Corollary~\ref{cor:nu}, 
\begin{equation*}
\begin{split}
&\: - \int_{\mathbb R^{2+1}} [\rd_i (\chi\psi_n)] \de^{ij} [(e_0)_0 N_0] \Big[A(\f{\rd_j(\chi\psi_n)}{N_0})\Big] \,\ud x + \int_{\mathbb R^{2+1}} [\rd_i (\chi\psi_n)] \de^{ij} N_0 A\Big(\f{((e_0)_0 N_0)\rd_j(\chi\psi_n)}{N_0^2}\Big) \,\ud x \\
\to &\: - \int_{S^*\mathbb R^{2+1}} \f{e^{-2\gamma_0} ((e_0)_0 N_0)}{N_0^2} a \de^{ij} \xi_i\xi_j \,\f{\ud \nu^\psi}{|\xi|^2} + \int_{S^*\mathbb R^{2+1}} \f{e^{-2\gamma_0} ((e_0)_0 N_0)}{N_0^2} a \de^{ij} \xi_i\xi_j \,\f{\ud \nu^\psi}{|\xi|^2} = 0.
\end{split}
\end{equation*}

The result therefore follows from Proposition~\ref{prop:main.medium}. \qedhere
\end{proof}

\subsection{Putting everything together}\label{sec:metric.conclude}

We summarize what we have obtained in this section.

\begin{proposition}\label{prop:RHS.1}
Suppose $A=b(x)\widetilde{m}(\f 1i\nabla)$, where the principal symbol is real and supported in $T^*\Omega$, and $m(\xi)$ is homogeneous of order $0$ and is even. After passing to a subsequence (which we do not relabel),
\begin{equation*}
\begin{split}
&\: \mbox{RHS of \eqref{eq:energy.id.n.1}} - \int_{\mathbb R^{2+1}} [\rd_i (\chi\psi_n)] \de^{ij} N_n \Big\{ [(e_0)_n,A](\f{\rd_j(\chi\psi_n)}{N_n}) \Big\} \,\ud x \\
\to &\: \mbox{RHS of \eqref{eq:energy.id.0.1}} - \int_{\mathbb R^{2+1}} [\rd_i (\chi\psi_0)] \de^{ij} N_0 \Big\{[(e_0)_0,A](\f{\rd_j(\chi\psi_0)}{N_0}) \Big\} \,\ud x \\
&\: -2\int_{S^*\mathbb R^{2+1}} \Big((g^{-1}_0)^{\alp\bt} (\rd_\bt X^\gamma) \xi_\alp \xi_\gamma -\f 12 X^\mu \rd_\mu (g^{-1}_0)^{\alp\gamma}\xi_\alp \xi_\gamma \Big) a \,\f{d\nu^\psi}{|\xi|^2} \\
&\: + \int_{S^*\mathbb R^{2+1}} \Big[-\de^{ij} \xi_i (\xi_t-\bt^k_0\xi_k) \rd_{x^j} a \Big]\, \f{e^{-2\gamma_0}}{N_0}\f{d\nu^\psi}{|\xi|^2},
\end{split}
\end{equation*}
where $X = \f 1{N_0} (\rd_t -\bt^i_0 \rd_i)$.
\end{proposition}

We have a similar result regarding the limit of the  RHS of \eqref{eq:energy.id.n.3}.
\begin{proposition}\label{prop:RHS.3}
Suppose $A=b(x)m(\f 1i\nabla)$, where the principal symbol is real and supported in $T^*\Omega$, and $m(\xi)$ is homogeneous of order $0$ and is even. After passing to a subsequence (which we do not relabel),
\begin{equation*}
\begin{split}
&\: \mbox{RHS of \eqref{eq:energy.id.n.3}} - \f 14\int_{\mathbb R^{2+1}} e^{-4\psi_0}[\rd_i (\chi\om_n)] \de^{ij} N_n \Big\{ [(e_0)_n,A](\f{\rd_j(\chi\om_n)}{N_n}) \Big\} \,\ud x \\
\to &\: \mbox{RHS of \eqref{eq:energy.id.0.3}} - \f 14 \int_{\mathbb R^{2+1}} e^{-4\psi_0} [\rd_i (\chi\om_0)] \de^{ij} N_0 \Big\{[(e_0)_0,A](\f{\rd_j(\chi\om_0)}{N_0}) \Big\} \,\ud x \\
&\: -\f 12\int_{S^*\mathbb R^{2+1}} e^{-4\psi_0} \Big((g^{-1}_0)^{\alp\bt} (\rd_\bt X^\gamma) \xi_\alp \xi_\gamma -\f 12 X^\mu \rd_\mu (g^{-1}_0)^{\alp\gamma}\xi_\alp \xi_\gamma \Big) a \,\f{d\nu^\om}{|\xi|^2} \\
&\: + \f 14\int_{S^*\mathbb R^{2+1}} e^{-4\psi_0} \Big[-\de^{ij} \xi_i (\xi_t-\bt^k_0\xi_k) \rd_{x^j} a \Big]\, \f{e^{-2\gamma_0}}{N_0}\f{d\nu^\om}{|\xi|^2}\\
&\: + 2\int_{S^*\mathbb R^{2+1}} \f{e^{-4\psi_0}}{N_0} (g_0^{-1})^{\alp\bt} (\rd_\alp \psi_0) \xi_\bt (\xi_t-\bt^k_0 \xi_k) a\, \f{\ud \nu^\om}{|\xi|^2},
\end{split}
\end{equation*}
where $X = \f 1{N_0} (\rd_t -\bt^i_0 \rd_i)$.
\end{proposition}
\begin{proof}
Except for the terms labeled ``$\mathrm{extra}_1$''--``$\mathrm{extra}_4$'', all the other terms in \eqref{eq:energy.id.n.3} have their obvious analogues in \eqref{eq:energy.id.n.1}. We thus only focus on the terms ``$\mathrm{extra}_1$''--``$\mathrm{extra}_4$''.

Using \eqref{assumption.0} and Corollary~\ref{cor:nu}, it immediately follows that
\begin{equation}\label{eq:extra.terms.1}
\begin{split}
\sum_{i=1}^4 \mathrm{extra}_i \to &\:  \mbox{corresponding terms on RHS of \eqref{eq:energy.id.0.3}} \\
&\: - \int_{S^*\mathbb R^{2+1}} e^{-4\psi_0} \Big[\f{((e_0)_0\psi_0)}{N_0^3} (\xi_t- \bt^k_0\xi_k)^2 + \f{e^{-2\gamma_0}((e_0)_0\psi_0)}{N_0} \de^{ij}\xi_i\xi_j \Big] a\, \f{\ud \nu^\om}{|\xi|^2} \\
&\: + 2\int_{S^*\mathbb R^{2+1}} e^{-4\psi_0} \Big[\f{e^{-2\gamma_0}(\rd_i\psi_0)}{N_0} \de^{ij}(\xi_t- \bt^k_0\xi_k) \xi_j \Big] a\, \f{\ud \nu^\om}{|\xi|^2}.
\end{split}
\end{equation}

Note that by Proposition~\ref{prop:psiom.localized}, on the support of $\ud \nu^\om$, 
\begin{equation}\label{eq:transport.second.term.null}
\f 1{N_0^2} (\xi_t -\bt^k_0 \xi_k)^2 = e^{-2\gamma_0} \de^{ij} \xi_i \xi_j.
\end{equation}
Hence, a direct computation shows that \underline{on the support of $\ud \nu^\om$}, 
\begin{equation}\label{eq:extra.terms.2}
\begin{split}
&\: \f 2{N_0} (g_0^{-1})^{\alp\bt} (\rd_\alp \psi_0) \xi_\bt (\xi_t-\bt^k_0 \xi_k) \\
=&\: -\f 2{N_0^3} ((e_0)_0\psi_0) (\xi_t-\bt^k_0 \xi_k)^2 + \f {2 e^{-2\gamma_0}}{N_0} \de^{ij} (\rd_i\psi_0) (\xi_t-\bt^k_0 \xi_k) \xi_j \\
=&\: -\f 1{N_0^3} ((e_0)_0\psi_0) (\xi_t-\bt^k_0 \xi_k)^2 - \f{e^{-2\gamma_0}}{N_0}((e_0)_0\psi_0) \de^{ij} \xi_i\xi_j + \f {2 e^{-2\gamma_0}}{N_0} \de^{ij} (\rd_i\psi_0) (\xi_t-\bt^k_0 \xi_k) \xi_j.
\end{split}
\end{equation}

Therefore, using the computations leading to Proposition~\ref{prop:RHS.1} and also \eqref{eq:extra.terms.1} and \eqref{eq:extra.terms.2}, we obtain the desired conclusion. \qedhere
\end{proof}

Let us again emphasize that we have not handled the terms $\int_{\mathbb R^{2+1}} [\rd_i (\chi\psi_n)] \de^{ij} N_n \{[(e_0)_n,A](\f{\rd_j(\chi\psi_n)}{N_n})\} \,\ud x$ and $\f 14\int_{\mathbb R^{2+1}} e^{-4\psi_0}[\rd_i (\chi\om_n)] \de^{ij} N_n \{[(e_0)_n,A](\f{\rd_j(\chi\om_n)}{N_n})\} \,\ud x$. They are considerably more difficult: not only do we need to use a version of trilinear compensated compactness, but we will also need to combine this with appropriate terms on the RHS of \eqref{eq:energy.id.n.2} and \eqref{eq:energy.id.n.4} to obtain extra cancellations.

\section{The main commutator terms in Proposition~\ref{prop:energy.id.n.2} and the elliptic-wave trilinear compensated compactness}\label{sec:elliptic.wave.tri}

We continue to work under the assumptions of Theorem~\ref{thm:main} and the reductions in Sections~\ref{sec:compact.reduction} and \ref{sec:reduction}. As above, {let $A$ be a $0$-th order pseudo-differential operator given by $A = b(x) \widetilde{m}(\f 1i\nabla)$, where the principal symbol $a(x,\xi) = b(x) m(\xi)$ (with $m(\xi) = \widetilde{m}(\xi)$ for $|\xi|\geq 1$) is real and supported in $T^*\Omega$, $m(\xi)$ is homogeneous of order $0$, and $m$ and $\widetilde{m}$ are both even.}

Our goal in this section is to compute the limit of the term labeled ``$\mathrm{maincommutator}$'' in \eqref{eq:energy.id.n.2} (and the corresponding term in \eqref{eq:energy.id.n.4}). To handle this term, we will in particular need various forms of trilinear compensated compactness for special combinations of functions satisfying nonlinear elliptic and wave equations. 

To proceed, let us compute using \eqref{eq:Box.gn} and \eqref{def:Box.gn} that
\begin{align}
&\:\sqrt{-\det g_n}\left(\Box_{g_n,A} (\chi\psi_n) - \f{1}{\sqrt{-\det g_n}} A (\sqrt{-\det g_n}\Box_{g_n} (\chi\psi_n))\right) \\
= &\: -\rd_t \Big[e^{2\gamma_n} A(\f{(e_0)_n(\chi\psi_n)}{N_n})\Big] + A \rd_t \Big[e^{2\gamma_n} (\f{(e_0)_n(\chi\psi_n)}{N_n})\Big] \label{eq:main.term.commutator.gamma} \\
&\: + \de^{ij} \rd_i \Big[N_n^2 A(\f{\rd_j(\chi\psi_n)}{N_n})\Big]- {\de^{ij}} A \rd_i \Big[N_n \rd_j(\chi\psi_n) \Big] \label{eq:main.term.commutator.N}\\
&\: + \rd_i \Big[\bt_n^i e^{2\gamma_n} A(\f{(e_0)_n (\chi\psi_n)}{N_n}) \Big] - A \rd_i \Big[e^{2\gamma_n} \bt_n^i (\f{(e_0)_n(\chi\psi_n)}{N_n})\Big]. \label{eq:main.term.commutator.beta}
\end{align}

We will consider the contribution to the ``$\mathrm{maincommutator}$'' term from \eqref{eq:main.term.commutator.gamma}, \eqref{eq:main.term.commutator.N} and \eqref{eq:main.term.commutator.beta} in \textbf{Section~\ref{sec:commutator.gamma}}, \textbf{Section~\ref{sec:commutator.N}} and \textbf{Section~\ref{sec:commutator.beta}} respectively. We then put together the computations and obtain our conclusion in \textbf{Section~\ref{sec:commutator.terms.mdm}} and \textbf{Section~\ref{sec:commutator.terms.final}}.

\subsection{The term \eqref{eq:main.term.commutator.gamma}}\label{sec:commutator.gamma}

\begin{proposition}\label{prop:gamma.commute.final}
\begin{equation*}
\begin{split}
&\: \int_{\mathbb R^{2+1}} \f{(e_0)_n(\chi\psi_n)}{N_n} \Big\{ \rd_t [e^{2\gamma_n} A(\f{(e_0)_n(\chi\psi_n)}{N_n})] - A \rd_t [e^{2\gamma_n} (\f{(e_0)_n(\chi\psi_n)}{N_n})] \Big\}\,\ud x \\
&\: - \int_{\mathbb R^{2+1}} \f{(e_0)_0(\chi\psi_n)}{N_0} \Big\{ \rd_t [e^{2\gamma_0} A(\f{(e_0)_0(\chi\psi_n)}{N_0})] - A \rd_t [e^{2\gamma_0} (\f{(e_0)_0(\chi\psi_n)}{N_0})] \Big\}\,\ud x \to 0.
\end{split}
\end{equation*}
A similar statement holds after replacing $\psi_n\rightsquigarrow \om_n$, $\psi_0\rightsquigarrow \om_0$ and $\ud x\rightsquigarrow \f 14 e^{-4\psi_0} \ud x$.
\end{proposition}
\begin{proof}
We first note that
\begin{align}
&\: \f{(e_0)_n(\chi\psi_n)}{N_n} \Big\{ \rd_t [e^{2\gamma_n} A(\f{(e_0)_n(\chi\psi_n)}{N_n})] - A \rd_t [e^{2\gamma_n} (\f{(e_0)_n(\chi\psi_n)}{N_n})] \Big\} \nonumber \\
&\: -  \f{(e_0)_0(\chi\psi_n)}{N_0} \Big\{ \rd_t [e^{2\gamma_0} A(\f{(e_0)_0(\chi\psi_n)}{N_0})] - A \rd_t [e^{2\gamma_0} (\f{(e_0)_0(\chi\psi_n)}{N_0})] \Big\} \nonumber  \\
= &\: (\f{(e_0)_n(\chi\psi_n)}{N_n}  - \f{(e_0)_0(\chi\psi_n)}{N_0}) \Big\{ \rd_t [e^{2\gamma_n} A(\f{(e_0)_n(\chi\psi_n)}{N_n})] - A \rd_t [e^{2\gamma_n} (\f{(e_0)_n(\chi\psi_n)}{N_n})] \Big\}  \label{eq:gamma.commute.easy.1}\\
&\: + \f{(e_0)_0(\chi\psi_n)}{N_0} \Big\{ \rd_t [e^{2\gamma_n} A(\f{(e_0)_n(\chi\psi_n)}{N_n})] - A \rd_t [e^{2\gamma_n} (\f{(e_0)_n(\chi\psi_n)}{N_n})] \Big\} \label{eq:gamma.commute.easy.2}\\
&\: - \f{(e_0)_0(\chi\psi_n)}{N_0} \Big\{ \rd_t [e^{2\gamma_0} A(\f{(e_0)_0(\chi\psi_n)}{N_0})] - A \rd_t [e^{2\gamma_0} (\f{(e_0)_0(\chi\psi_n)}{N_0})] \Big\}.\label{eq:gamma.commute.easy.3}
\end{align}

\pfstep{Step~1: Estimating \eqref{eq:gamma.commute.easy.1}} We bound \eqref{eq:gamma.commute.easy.1} in $L^1$. First, it is easy to check {using} \eqref{assumption.0} and H\"older's inequality that
\begin{equation*}
\begin{split}
&\:\| \f{(e_0)_n(\chi\psi_n)}{N_n}  - \f{(e_0)_0(\chi\psi_n)}{N_0} \|_{L^2}\\
\ls &\: \|\bt_n^i - \bt_0^i\|_{L^\infty} \|\f{\rd_i(\chi\psi_n)}{N_n}\|_{L^2} +\|(e_0)_0(\chi\psi_n)\|_{L^2}\|\f 1{N_n} -\f 1{N_0}\|_{L^\infty} \ls \lambda_n.
\end{split}
\end{equation*}
On the other hand,
\begin{equation*}
\begin{split}
&\: \rd_t [e^{2\gamma_n} A(\f{(e_0)_n(\chi\psi_n)}{N_n})] - A \rd_t [e^{2\gamma_n} (\f{(e_0)_n(\chi\psi_n)}{N_n})] \} \\
= &\: \underbrace{e^{2\gamma_n} (\rd_t b) \widetilde{m}(\f 1i \nabla)(\f{(e_0)_n(\chi\psi_n)}{N_n})}_{=:\mathrm{I}_1} + \underbrace{b e^{2\gamma_n} \rd_t \widetilde{m}(\f 1i\nabla)(\f{(e_0)_n(\chi\psi_n)}{N_n}) - b \widetilde{m}(\f 1i\nabla) \rd_t [e^{2\gamma_n} (\f{(e_0)_n(\chi\psi_n)}{N_n})]}_{=:\mathrm{I}_2} \\
&\: + \underbrace{ (\rd_t e^{2\gamma_n}) [A(\f{(e_0)_n(\chi\psi_n)}{N_n})]}_{=:\mathrm{I}_3},
\end{split}
\end{equation*}
where we have used that $\widetilde{m}(\frac{1}{i}\nabla)$ commute{s} with $\partial_t$ and $\partial_i$. Each of $\mathrm{I}_1$, $\mathrm{I}_2$ and $\mathrm{I}_3$ can easily be seen to be bounded in $L^2$ uniformly in $n$. For $\mathrm{I}_1$, this simply follows from the assumptions \eqref{assumption.0} and \eqref{assumption.1} and the fact that $\widetilde{m}(\f 1i\nabla)$ is bounded on $L^2$. For $\mathrm{I}_2$, this is a consequence of the Calder\'on commutator theorem (Lemma~\ref{lem:PSIDOs}.6) and \eqref{assumption.0} and \eqref{assumption.1}. Finally, for $\mathrm{I}_3$, this is an immediate consequence of \eqref{assumption.0} and \eqref{assumption.1}.

Therefore, by the Cauchy--Schwarz inequality, 
\begin{equation}\label{gamma.commute.easy.final.1}
\|\mbox{\eqref{eq:gamma.commute.easy.1}}\|_{L^1} \ls \lambda_n \to 0.
\end{equation}

\pfstep{Step~2: Estimating \eqref{eq:gamma.commute.easy.2} and \eqref{eq:gamma.commute.easy.3}} The term $(\mbox{\eqref{eq:gamma.commute.easy.2}} + \mbox{\eqref{eq:gamma.commute.easy.3}})$ is more subtle. First,
\begin{align}
&\: \rd_t [e^{2\gamma_n} A(\f{(e_0)_n(\chi\psi_n)}{N_n})] - A \rd_t [e^{2\gamma_n} (\f{(e_0)_n(\chi\psi_n)}{N_n})]  -\{ \rd_t [e^{2\gamma_0} A(\f{(e_0)_0(\chi\psi_n)}{N_0})] - A \rd_t [e^{2\gamma_0} (\f{(e_0)_0(\chi\psi_n)}{N_0})]\} \nonumber \\
= &\: b  (e^{2\gamma_n}-e^{2\gamma_0}) [\rd_t \widetilde{m}(\f 1 i\nabla)(\f{(e_0)_n(\chi\psi_n)}{N_n})] - b \rd_t \widetilde{m}(\f 1 i\nabla) [(e^{2\gamma_n}-e^{2\gamma_0}) (\f{(e_0)_n(\chi\psi_n)}{N_n})]) \label{gamma.commute.easy.1}\\
&\: + b  e^{2\gamma_0} [\rd_t \widetilde{m}(\f 1i \nabla)(\f{(e_0)_n(\chi\psi_n)}{N_n} - \f{(e_0)_0(\chi\psi_n)}{N_0})] - b \rd_t \widetilde{m}(\f 1 i\nabla) [e^{2\gamma_0} (\f{(e_0)_n(\chi\psi_n)}{N_n} - \f{(e_0)_0(\chi\psi_n)}{N_0})] \label{gamma.commute.easy.2}\\
&\:+ b (\rd_t  (e^{2\gamma_n} - e^{2\gamma_0})) [\widetilde{m}(\f 1 i\nabla)(\f{(e_0)_n(\chi\psi_n)}{N_n})] + b (\rd_t  e^{2\gamma_0})) [\widetilde{m}(\f 1i \nabla)(\f{(e_0)_n(\chi\psi_n)}{N_n} - \f{(e_0)_0(\chi\psi_n)}{N_0})] \label{gamma.commute.easy.3}\\
&\: + (\rd_t b)\{ (e^{2\gamma_n}-e^{2\gamma_0}) [\widetilde{m}(\f 1 i\nabla)(\f{(e_0)_n(\chi\psi_n)}{N_n})] + e^{2\gamma_0} [\widetilde{m}(\f 1i \nabla)(\f{(e_0)_n(\chi\psi_n)}{N_n} - \f{(e_0)_0(\chi\psi_n)}{N_0})] \}.\label{gamma.commute.easy.4}
\end{align}

By the Calder\'on commutation theorem (Lemma~\ref{lem:PSIDOs}.6), the fact that $\widetilde{m}(\f 1i\nabla)$ is bounded on $L^2$ (Lemma~\ref{lem:PSIDOs}.4), and the estimates in \eqref{assumption.0}, \eqref{assumption.1} and Proposition~\ref{prop:spatial.imp} and \ref{prop:dtgamma.imp},
\begin{equation}\label{gamma.commute.easy.2.1}
\begin{split}
&\:\| \eqref{gamma.commute.easy.1} + \eqref{gamma.commute.easy.2}\|_{L^2} \\
\ls &\: \| \widetilde{\chi} (e^{2\gamma_n}- e^{2\gamma_0})\|_{W^{1,\infty}} \|\f{(e_0)_n(\chi\psi_n)}{N_n}\|_{L^2} + \|\widetilde{\chi} e^{2\gamma_0}\|_{W^{1,\infty}} \|\f{(e_0)_n(\chi\psi_n)}{N_n} - \f{(e_0)_0(\chi\psi_n)}{N_0} \|_{L^2} \ls \lambda_n^{\f 12}.
\end{split}
\end{equation}
Using again the fact that $\widetilde{m}(\f 1i\nabla)$ is bounded on $L^2$, and the estimates in \eqref{assumption.0}, \eqref{assumption.1} and Proposition~\ref{prop:spatial.imp} and \ref{prop:dtgamma.imp}, the remaining terms can be bounded directly as follows:
\begin{equation}\label{gamma.commute.easy.2.2}
\begin{split}
&\: \|\eqref{gamma.commute.easy.3} + \eqref{gamma.commute.easy.4}\|_{L^2} \\
\ls &\: \| \widetilde{\chi} (e^{2\gamma_n}- e^{2\gamma_0})\|_{W^{1,\infty}} \|\f{(e_0)_n(\chi\psi_n)}{N_n}\|_{L^2} + \| \widetilde{\chi} e^{2\gamma_0}\|_{W^{1,\infty}} \|\f{(e_0)_n(\chi\psi_n)}{N_n} - \f{(e_0)_0(\chi\psi_n)}{N_0} \|_{L^2} \ls \lambda_n^{\f 12}.
\end{split}
\end{equation}

Using \eqref{gamma.commute.easy.2.1}, \eqref{gamma.commute.easy.2.2} and also \eqref{assumption.0} and \eqref{assumption.1}, and the Cauchy--Schwarz inequality, we thus obtain
\begin{equation}\label{gamma.commute.easy.final.2}
\begin{split}
\|\eqref{eq:gamma.commute.easy.2} + \eqref{eq:gamma.commute.easy.3}\|_{L^1} 
\ls &\: \|\f{(e_0)_0(\chi\psi_n)}{N_0} \|_{L^2} (\| \eqref{gamma.commute.easy.1} + \eqref{gamma.commute.easy.2}\|_{L^2}+ \|\eqref{gamma.commute.easy.3} + \eqref{gamma.commute.easy.4}\|_{L^2}) \ls \lambda_n^{\f 12}.
\end{split}
\end{equation}

Combining \eqref{gamma.commute.easy.final.1} and \eqref{gamma.commute.easy.final.2} yields the conclusion. \qedhere
\end{proof}

\subsection{The term \eqref{eq:main.term.commutator.N}}\label{sec:commutator.N}

We first argue as in Section~\ref{sec:commutator.gamma} to control most of the terms. We will identify, however, in Proposition~\ref{prop:N.commute.1} that there is one difficult term that cannot be handled just with techniques in Section~\ref{sec:commutator.N}. In the rest of this subsection, we will then handle the difficult term that is identified in this proposition.

\begin{proposition}\label{prop:N.commute.1}
\begin{equation*}
\begin{split}
&\: \int_{\mathbb R^{2+1}} \f{(e_0)_n(\chi\psi_n)}{N_n} \de^{ij} \{ \rd_i [N_n^2 A (\f{\rd_j (\chi\psi_n)}{N_n})] - A \rd_i [N_n \rd_j(\chi\psi_n)] \}\,\ud x \\
&\: - \int_{\mathbb R^{2+1}} \f{(e_0)_0(\chi\psi_n)}{N_0} \de^{ij} \{ \rd_i [N_0^2 A(\f{\rd_j(\chi\psi_n)}{N_0})] - A \rd_i [N_0 \rd_j(\chi\psi_n)] \}\,\ud x \\
&\: \underbrace{- \int_{\mathbb R^{2+1}} \f{(e_0)_0(\chi\psi_n)}{N_0} \de^{ij} b \{ (N_n^2 - N_0^2) \widetilde{m}(\f 1i\nabla)  (\f{\rd^2_{ij} (\chi\psi_n )}{N_0}) - \widetilde{m}(\f 1i\nabla)  [(N_n^2 - N_0^2) \f{\rd_{ij}^2(\chi\psi_n )}{N_0}] \}\,\ud x }_{=:\mathrm{I}}\to 0.
\end{split}
\end{equation*}
\end{proposition}
\begin{proof}
The idea is to argue as in the proof of Proposition~\ref{prop:gamma.commute.final} until we face a term that does not obviously $\to 0$.

In analogy with \eqref{eq:gamma.commute.easy.1}--\eqref{eq:gamma.commute.easy.3}, we have
\begin{align}
&\: \f{(e_0)_n(\chi\psi_n)}{N_n} \de^{ij} \{ \rd_i [N_n^2 A (\f{\rd_j (\chi\psi_n)}{N_n})] - A \rd_i [N_n \rd_j(\chi\psi_n)] \} \nonumber\\
&\: - \f{(e_0)_0(\chi\psi_n)}{N_0} \de^{ij} \{ \rd_i [N_0^2 A(\f{\rd_j(\chi\psi_n)}{N_0})] - A \rd_i [N_0 \rd_j(\chi\psi_n)] \} \nonumber\\
= &\: ( \f{(e_0)_n(\chi\psi_n)}{N_n} - \f{(e_0)_0(\chi\psi_n)}{N_0} ) \de^{ij} \{ \rd_i [N_n^2 A (\f{\rd_j (\chi\psi_n)}{N_n})] - A \rd_i [N_n \rd_j(\chi\psi_n)] \} \label{eq:N.commute.easy.1}\\
&\: + \f{(e_0)_0(\chi\psi_n)}{N_0} \de^{ij} \left( \{ \rd_i [N_n^2 A (\f{\rd_j (\chi\psi_n)}{N_n})] - A \rd_i [N_n \rd_j(\chi\psi_n)] \}\right.\label{eq:N.commute.easy.2}\\
&\: \left. - \{ \rd_i [N_0^2 A(\f{\rd_j(\chi\psi_n)}{N_0})] - A \rd_i [N_0 \rd_j(\chi\psi_n)] \} \right).\label{eq:N.commute.easy.3}
\end{align}

First, note that \eqref{eq:N.commute.easy.1} can be handled completely analogously as in Step~1 in the proof of Proposition~\ref{prop:gamma.commute.final} using that $\|\f{(e_0)_n(\chi\psi_n)}{N_n} - \f{(e_0)_0(\chi\psi_n)}{N_0}\|_{L^2}\to 0$ and that $\|\rd_i [N_n^2 A (\f{\rd_j (\chi\psi_n)}{N_n})] - A \rd_i [N_n \rd_j(\chi\psi_n)]\|_{L^2}$ is uniformly bounded. 

To control $\eqref{eq:N.commute.easy.2}+\eqref{eq:N.commute.easy.3}$, we first compute as in Step~2 in the proof of Proposition~\ref{prop:gamma.commute.final}.
\begin{align}
&\: \{ \rd_i [N_n^2 A (\f{\rd_j (\chi\psi_n)}{N_n})] - A \rd_i [N_n \rd_j(\chi\psi_n)] \}  - \{ \rd_i [N_0^2 A(\f{\rd_j(\chi\psi_n)}{N_0})] - A \rd_i [N_0 \rd_j(\chi\psi_n)] \}\nonumber\\
=&\: b \rd_i [N_n^2 \widetilde{m}(\f 1i\nabla) (\f{\rd_j (\chi\psi_n)}{N_n} - \f{\rd_j(\chi\psi_n)}{N_0})] - b \rd_i \widetilde{m}(\f 1i\nabla) [N_n^2 (\f{\rd_j(\chi\psi_n)}{N_n}-\f{\rd_j(\chi\psi_n)}{N_0})] \label{N.commute.easy.1}\\
&\: + (\rd_i b)N_n^2 \widetilde{m}(\f 1i\nabla) (\f{\rd_j (\chi\psi_n)}{N_n} - \f{\rd_j(\chi\psi_n)}{N_0}) + (\rd_i b)(N_n^2- N_0^2) \widetilde{m}(\f 1i\nabla) (\f{\rd_j(\chi\psi_n)}{N_0}) \label{N.commute.easy.2}\\
&\: + b [\rd_i (N_n^2- N_0^2)] \widetilde{m}(\f 1i\nabla) (\f{\rd_j(\chi\psi_n)}{N_0}) - b  \widetilde{m}(\f 1i\nabla) [(\rd_i(N_n^2 -N_0^2))\f{\rd_j(\chi\psi_n)}{N_0}] \label{N.commute.easy.3}\\
&\: - b (N_n^2- N_0^2) [\widetilde{m}(\f 1i\nabla) (\f{(\rd_i N_0)\rd_{j}(\chi\psi_n)}{N_0^2})] + b \widetilde{m}(\f 1i\nabla) [(N_n^2 -N_0^2)(\f{(\rd_i N_0)\rd_{j}(\chi\psi_n)}{N_0^2})] \label{N.commute.easy.4}\\
&\: + b (N_n^2- N_0^2) \widetilde{m}(\f 1i\nabla) (\f{\rd^2_{ij}(\chi\psi_n)}{N_0}) - b \widetilde{m}(\f 1i\nabla) [(N_n^2 -N_0^2)\f{\rd^2_{ij}(\chi\psi_n)}{N_0}].\label{N.commute.hard}
\end{align}

Using the Calder\'on commutator theorem (Lemma~\ref{lem:PSIDOs}.6), $L^2$ boundedness of $\widetilde{m}(\f 1i\nabla)$ (Lemma~\ref{lem:PSIDOs}.4), H\"older's inequality, and the estimates in \eqref{assumption.0}, \eqref{assumption.1} and Proposition~\ref{prop:spatial.imp}, we obtain (in a similar manner as \eqref{gamma.commute.easy.2.1} and \eqref{gamma.commute.easy.2.2})
\begin{equation*}
\begin{split}
&\: \|\eqref{N.commute.easy.1}\|_{L^2} + \|\eqref{N.commute.easy.2}\|_{L^2} + \|\eqref{N.commute.easy.3}\|_{L^2} + \|\eqref{N.commute.easy.4}\|_{L^2} \\
\ls &\: \|\widetilde{\chi} N_n^2\|_{W^{1,\infty}}\|\f{\rd_j(\chi\psi_n)}{N_n}-\f{\rd_j(\chi\psi_n)}{N_0}\|_{L^2} + \|\widetilde{\chi} (N_n^2 - N_0^2)\|_{L^\i} \|\f{\rd_j(\chi\psi_n)}{N_0}\|_{L^2} \\
&\: + \|\rd_i(\widetilde{\chi} (N_n^2 - N_0^2))\|_{L^\i} \|\f{\rd_j(\chi\psi_n)}{N_0}\|_{L^2} \ls \lambda_n^{\f 12}.
\end{split}
\end{equation*}
Therefore, the contribution of \eqref{N.commute.easy.1}--\eqref{N.commute.easy.4} to $\eqref{eq:N.commute.easy.2}+\eqref{eq:N.commute.easy.3}$ $\to 0$ in analogy with \eqref{gamma.commute.easy.final.2}.

However, \underline{in contrast to Proposition~\ref{prop:gamma.commute.final}}, it is not clear whether the term \eqref{N.commute.hard} converges to $0$ in $L^2$. This therefore gives rise to the additional term $\mathrm{I}$ in the statement of the proposition. \qedhere
\end{proof}

The remaining task of this subsection is therefore to show that $\mathrm{I}$ in Proposition~\ref{prop:N.commute.1} $\to 0$ as $n\to +\infty$. (This requires a use of the full trilinear structure.) We first perform a series of reductions; see Propositions~\ref{prop:reduction.to.diff} and \ref{prop:reduction.to.frozen} below.

Our first reduction is to show that $\mathrm{I}$ in Proposition~\ref{prop:N.commute.1} has the same limit after the replacement $\psi_n \rightsquigarrow \psi_n - \psi_0$.

\begin{proposition}\label{prop:reduction.to.diff}
\begin{equation*}
\begin{split}
&\: \int_{\mathbb R^{2+1}} \f{(e_0)_0(\chi\psi_n)}{N_0} \de^{ij} b \Big\{ (N_n^2 - N_0^2) \widetilde{m}(\f 1i\nabla)  (\f{\rd^2_{ij} (\chi\psi_n )}{N_0}) - \widetilde{m}(\f 1i\nabla)  \big[(N_n^2 - N_0^2) \f{\rd_{ij}^2(\chi\psi_n )}{N_0}\big] \Big\}\,\ud x \\
&\: -\int_{\mathbb R^{2+1}} \f{(e_0)_0(\chi(\psi_n-\psi_0))}{N_0} \de^{ij} b \Big\{ (N_n^2 - N_0^2) \widetilde{m}(\f 1i\nabla)  (\f{\rd^2_{ij} (\chi(\psi_n-\psi_0) )}{N_0}) \\
&\: \qquad\qquad\qquad\qquad\qquad\qquad\qquad\qquad - \widetilde{m}(\f 1i\nabla)  [(N_n^2 - N_0^2) \f{\rd_{ij}^2(\chi(\psi_n-\psi_0) )}{N_0}] \Big\}\,\ud x \to 0.
\end{split}
\end{equation*}
\end{proposition}
\begin{proof}
It clearly suffices to prove the following three convergences as $n\to +\infty$:
\begin{equation}\label{N.diff.LL}
\begin{split}
\int_{\mathbb R^{2+1}} \f{(e_0)_0(\chi\psi_0)}{N_0} \de^{ij} b \Big\{ (N_n^2 - N_0^2) \widetilde{m}(\f 1i\nabla)  (\f{\rd^2_{ij} (\chi\psi_0)}{N_0}) - \widetilde{m}(\f 1i\nabla) \big[(N_n^2 - N_0^2) \f{\rd^2_{ij}(\chi\psi_0)}{N_0} \big] \Big\}\,\ud x \to 0,
\end{split}
\end{equation}
\begin{equation}\label{N.diff.HL}
\begin{split}
\int_{\mathbb R^{2+1}} \f{(e_0)_0(\chi\psi_n)}{N_0} \de^{ij} b \Big\{ (N_n^2 - N_0^2) \widetilde{m}(\f 1i\nabla)(\f{\rd^2_{ij} (\chi\psi_0)}{N_0}) - \widetilde{m}(\f 1i\nabla) \big[(N_n^2 - N_0^2) \f{\rd^2_{ij}(\chi\psi_0)}{N_0} \big] \Big\}\,\ud x \to 0,
\end{split}
\end{equation}
and
\begin{equation}\label{N.diff.LH}
\begin{split}
\int_{\mathbb R^{2+1}} \f{(e_0)_0(\chi\psi_0)}{N_0} \de^{ij} b \Big\{ (N_n^2 - N_0^2) \widetilde{m}(\f 1i\nabla)  (\f{\rd^2_{ij} (\chi\psi_n)}{N_0}) - \widetilde{m}(\f 1i\nabla) \big[(N_n^2 - N_0^2) \f{\rd^2_{ij}(\chi\psi_n)}{N_0} \big] \Big\}\,\ud x \to 0.
\end{split}
\end{equation}

We will in fact not need to take advantage of the commutator $[(N_n^2 - N_0^2), \widetilde{m}(\f 1i\nabla)]$ in the expressions above. We will simply control the first term in each of \eqref{N.diff.LL}--\eqref{N.diff.LH}; the second term in each line can be handled in exactly the same way.

\pfstep{Step~1: Proof of \eqref{N.diff.LL} and \eqref{N.diff.HL}} The terms \eqref{N.diff.LL} and \eqref{N.diff.HL} are easier because $\psi_0$ is smooth and we can directly bound its second derivatives. More precisely, using H\"older's inequality, the boundedness of $\widetilde{m}(\f 1i\nabla)$ on $L^2$, and the estimate \eqref{assumption.0}, we obtain
\begin{equation*}
\begin{split}
|\mbox{First term in \eqref{N.diff.LL}}| 
\ls &\: \|(e_0)_0(\chi \psi_0)\|_{L^2} \| \widetilde{\chi} (N_n^2 - N_0^2)\|_{L^\i} \|\de^{ij} \rd^2_{ij} (\chi \psi_0)\|_{L^2}
\ls  1 \cdot \lambda_n\cdot 1 \ls \lambda_n \to 0.
\end{split}
\end{equation*}

Similarly, but using in addition \eqref{assumption.1}, we obtain
\begin{equation*}
\begin{split}
|\mbox{First term in \eqref{N.diff.HL}} |
\ls &\: \|(e_0)_0(\chi \psi_n)\|_{L^2} \| \widetilde{\chi} (N_n^2 - N_0^2)\|_{L^\i} \|\de^{ij} \rd^2_{ij} (\chi \psi_0)\|_{L^2} 
\ls  1 \cdot \lambda_n\cdot 1 \ls \lambda_n \to 0.
\end{split}
\end{equation*}

\pfstep{Step~2: Proof of \eqref{N.diff.LH}} The key is an integration by parts to throw the derivatives on the smooth $\psi_0$. More precisely, after integrating by parts, applying H\"older's inequality and the boundedness of $\widetilde{m}(\f 1i\nabla)$ in $L^2$, and using the estimates in \eqref{assumption.0}, \eqref{assumption.1} and Proposition~\ref{prop:spatial.imp}, we obtain
\begin{equation*}
\begin{split}
|\mbox{First term in \eqref{N.diff.LH}}| 
\ls &\: \left| \int_{\mathbb R^{2+1}} \rd_i\f{(e_0)_0 (\chi \psi_0)}{N_0} \de^{ij}(N_n^2 - N_0^2) \widetilde{m}(\f 1i\nabla) \f{\rd_j (\chi \psi_n)}{N_0}  \,\ud x\right| \\
&\:+ \left| \int_{\mathbb R^{2+1}} \f{(e_0)_0(\chi \psi_0)}{N_0} \de^{ij}\rd_i(N_n^2 - N_0^2) \widetilde{m}(\f 1i\nabla)\f{\rd_j (\chi \psi_n)}{N_0}  \,\ud x\right| \\
&\: + \left| \int_{\mathbb R^{2+1}} \f{(e_0)_0(\chi \psi_0)}{N_0} \de^{ij}(N_n^2 - N_0^2) \widetilde{m}(\f 1i\nabla)\f{(\rd_i N_0)\rd_j (\chi \psi_n)}{N_0^2}  \,\ud x\right| \\
\ls &\: \|\rd_i(e_0)_0(\chi \psi_0)\|_{L^2} \| N_n^2 - N_0^2\|_{L^\i} \|\rd_j (\chi \psi_n)\|_{L^2}  \\
&\: + \|(e_0)_0(\chi \psi_0)\|_{L^2} (\| \rd_i(N_n^2 - N_0^2)\|_{L^\i} + \| N_n^2 - N_0^2\|_{L^\i}) \|\rd_j (\chi \psi_n)\|_{L^2} \\
\ls &\: \lambda_n + \lambda_n^{\f 12} \ls \lambda_n^{\f 12} \to 0. 
\end{split}
\end{equation*}
\end{proof}

Our next reduction is to freeze the coefficients (cf.~Section~\ref{sec:freeze.coeff}). We show that the difficult term is essentially the same as a ``frozen coefficient'' version up to error terms which are $o(1)$.
\begin{proposition}\label{prop:reduction.to.frozen}
Let $b_{c,\alp}$, $N_{c,\alp}$, $\bt_{c,\alp}^i$ be as in Proposition~\ref{prop:constants}. Denote moreover $(e_0)_{c,\alp}:= \rd_t - \bt^i_{c,\alp} \rd_i$. Then
\begin{equation*}
\begin{split}
&\: \int_{\mathbb R^{2+1}} \f{(e_0)_0(\chi(\psi_n-\psi_0))}{N_0} \de^{ij} b \{ (N_n^2 - N_0^2) \widetilde{m}(\f 1i\nabla)  (\f{\rd^2_{ij} (\chi(\psi_n-\psi_0) )}{N_0}) - \widetilde{m}(\f 1i\nabla)  [(N_n^2 - N_0^2) \f{\rd_{ij}^2(\chi(\psi_n-\psi_0) )}{N_0}] \}\,\ud x \\
&\: - \sum_\alp \f{b_{c,\alp}}{N_{c,\alp}^2} \int_{\mathbb R^{2+1}} (e_0)_{c,\alp}(\zeta_\alp\chi(\psi_n-\psi_0)) \de^{ij} \{ \zeta_\alp(N_n^2 - N_0^2) \widetilde{m}(\f 1i\nabla)  (\rd^2_{ij} (\zeta_\alp\chi(\psi_n-\psi_0) )) \\
&\: \qquad\qquad\qquad\qquad\qquad\qquad\qquad\qquad\qquad- \widetilde{m}(\f 1i\nabla)  [\zeta_\alp(N_n^2 - N_0^2) \rd_{ij}^2(\zeta_\alp\chi(\psi_n-\psi_0) )] \}\,\ud x \to 0.
\end{split}
\end{equation*}
\end{proposition}
\begin{proof}
We write $1 = \sum_{\alp} \zeta_\alp^3$. Then for every $\alp$, we apply the estimates in Proposition~\ref{prop:constants} and \ref{prop:psi.localized} together with H\"older's inequality and the $L^2$ boundedness of $\widetilde{m}(\f 1i \nabla)$ to obtain the desired result.  \qedhere
\end{proof}

After the series of reductions above, we now finally estimate the term with frozen coefficients. As we have indicated earlier, the frozen coefficients allow us to employ Fourier techniques and exploit crucial cancellations.
\begin{proposition}\label{prop:main.N.term}
Let $b_{c,\alp}$, $N_{c,\alp}$, $\bt_{c,\alp}^i$ and $(e_0)_{c,\alp}$ be as in Proposition~\ref{prop:reduction.to.frozen}. Then
\begin{equation*}
\begin{split}
&\: \sum_\alp \f{b_{c,\alp}}{N_{c,\alp}^2} \int_{\mathbb R^{2+1}} (e_0)_{c,\alp}(\zeta_\alp\chi(\psi_n-\psi_0)) \de^{ij} \{ \zeta_\alp(N_n^2 - N_0^2) \widetilde{m}(\f 1i\nabla)  (\rd^2_{ij} (\zeta_\alp\chi(\psi_n-\psi_0) )) \\
&\: \qquad\qquad\qquad\qquad\qquad\qquad\qquad\qquad\qquad- \widetilde{m}(\f 1i\nabla)  [\zeta_\alp(N_n^2 - N_0^2) \rd_{ij}^2(\zeta_\alp\chi(\psi_n-\psi_0) )] \}\,\ud x \to 0.
\end{split}
\end{equation*}

\end{proposition}

\begin{proof}
We will bound each term in the sum. Since there are $O(\lambda_n^{-3\ep_0})$ terms in the sum (cf.~beginning of Section~\ref{sec:freeze.coeff}), it suffices to show that each term is $o(\lambda_n^{3\ep_0})$. This is what we will show.

From now on fix $\alp$.

\pfstep{Step~1: Frequency space decomposition} Decompose $\zeta_\alp(N_n^2-N_0^2)$ into three pieces in frequency space. For this purpose, define a smooth cutoff function $\Theta:[0,+\infty)\to \mathbb R$ such that 
\begin{equation}\label{def:Theta}
\Theta \geq 0,\quad \Theta(x) = 1 \mbox{ for $x\in [0,1]$},\quad \Theta(x) = 0 \mbox{ for $x\geq 2$}.
\end{equation}
Define now the decomposition of $\zeta_\alp(N_n^2-N_0^2)$ as follows
$$\zeta_\alp(N_n^2-N_0^2) = ({\bm N}_{\mathrm{diff}})_{n,1} + ({\bm N}_{\mathrm{diff}})_{n,2} + ({\bm N}_{\mathrm{diff}})_{n,3},$$
where
\begin{equation}\label{def:N.decomp.1}
\widehat{({\bm N}_{\mathrm{diff}})_{n,1}}(\xi) := \Theta(\lambda_n^{\f 56}|\xi|)\left(\widehat{\zeta_\alp N_n^2}(\xi)-\widehat{\zeta_\alp N_0^2}(\xi)\right),
\end{equation}
\begin{equation}\label{def:N.decomp.2}
\widehat{({\bm N}_{\mathrm{diff}})_{n,2}}(\xi) := (1-\Theta(\lambda_n^{\f 56}|\xi|))(1-\Theta(\f{|\xi_i|}{|\xi|^{\f 58}}))\left(\widehat{\zeta_\alp N_n^2}(\xi)-\widehat{\zeta_\alp N_0^2}(\xi)\right),
\end{equation}
\begin{equation}\label{def:N.decomp.3}
\widehat{({\bm N}_{\mathrm{diff}})_{n,3}}(\xi) := (1-\Theta(\lambda_n^{\f 56}|\xi|))\Theta(\f{|\xi_i|}{|\xi|^{\f 58}})\left(\widehat{\zeta_\alp N_n^2}(\xi)-\widehat{\zeta_\alp N_0^2}(\xi)\right).
\end{equation}
{(We recall here our convention from Section~\ref{sec.notations} that $|\xi_i|$ denote only the spatial part.)}

\pfstep{Step~2: Handling $({\bm N}_{\mathrm{diff}})_{n,1}$} We first deal with the terms involving $({\bm N}_{\mathrm{diff}})_{n,1}$. 

By Bernstein's inequality and Proposition~\ref{prop:g.localized},
\begin{equation}\label{N1.Li}
\|({\bm N}_{\mathrm{diff}})_{n,1}\|_{W^{1,\infty}} \ls \lambda_n^{-\f 56} \|\zeta_\alp (N_n^2 - N_0^2)\|_{L^\i} \ls \lambda_n^{\f 16}.
\end{equation}

{On the other hand, notice that for all sufficiently regular functions $u$ and $f$, we have the commutator estimate
\begin{equation}\label{eq:commutator.detail}
\begin{split}
&\: \| u \widetilde{m}(\f 1i\nabla) \rd_{ij}^2 f - \widetilde{m}(\f 1i\nabla) (u\rd_{ij}^2 f) \|_{L^2} \\
\ls &\: \| u \widetilde{m}(\f 1i\nabla) \rd_{ij}^2 f - \widetilde{m}(\f 1i\nabla) \rd_i (u \rd_j f) \|_{L^2}  + \| \widetilde{m}(\f 1i\nabla) [(\rd_i u) \rd_j f)] \|_{L^2} 
\ls  \| u \|_{W^{1,\infty}} \| \rd_j f\|_{L^2}, 
\end{split}
\end{equation}
where we have used the Calder\'on commutator theorem (Lemma~\ref{lem:PSIDOs}.6) for the first term, and the $L^2$-boundedness of $\widetilde{m}(\f 1i\nabla)$ (by Lemma~\ref{lem:PSIDOs}.4) for the second term. 

Using \eqref{eq:commutator.detail} with $u = ({\bm N}_{\mathrm{diff}})_{n,1}$ and $f = \zeta_\alp\chi(\psi_n-\psi_0)$, and applying the Cauchy--Schwartz inequality, \eqref{N1.Li} and Proposition~\ref{prop:psi.localized},} we have
\begin{equation*}
\begin{split}
&\: \left| \f{b_{c,\alp}}{N_{c,\alp}^2} \int_{\mathbb R^{2+1}} (e_0)_{c,\alp}(\zeta_\alp\chi(\psi_n-\psi_0)) \de^{ij}  ({\bm N}_{\mathrm{diff}})_{n,1} \widetilde{m}(\f 1i\nabla)  (\rd^2_{ij} (\zeta_\alp\chi(\psi_n-\psi_0) )) \right. \\
&\:\left.  \qquad\qquad\qquad\qquad\qquad\qquad\qquad\qquad\qquad - \widetilde{m}(\f 1i\nabla)  [({\bm N}_{\mathrm{diff}})_{n,1} \rd_{ij}^2(\zeta_\alp\chi(\psi_n-\psi_0) )] \,\ud x \right| \\
\ls &\: \|\rd (\zeta_\alp\chi(\psi_n-\psi_0))\|_{L^2} \|({\bm N}_{\mathrm{diff}})_{n,1}\|_{W^{1,\infty}} \|\rd (\zeta_\alp\chi(\psi_n-\psi_0))\|_{L^2} \ls \lambda_n^{\f{3\ep_0}{2}}\cdot \lambda_n^{\f 16} \cdot \lambda_n^{\f{3\ep_0}{2}} = \lambda_n^{\f 16+3\ep_0} = o(\lambda_n^{3\ep_0}),
\end{split}
\end{equation*}
as desired.

\pfstep{Step~3: Handling $({\bm N}_{\mathrm{diff}})_{n,2}$} For this, we make use of the large spatial frequency to obtain a good $L^2$ bound. More precisely, since on the support of $(1-\Theta(\lambda_n^{\f 56}|\xi|))(1-\Theta(\f{|\xi_i|}{|\xi|^{\f 58}}))$, $|\xi_i|^2 \gtrsim |\xi|^{\f 54} \gtrsim \lambda_n^{-\f{25}{24}}$, by the Plancherel theorem and Proposition~\ref{prop:g.localized},
\begin{equation}\label{N.diff.2.improvement}
\|({\bm N}_{\mathrm{diff}})_{n,2}\|_{L^2} \ls \|\Delta^{-1}\Delta ({\bm N}_{\mathrm{diff}})_{n,2}\|_{L^2} \ls \lambda_n^{\f{25}{24}} \|\Delta(\zeta_\alp (N_n^2 -N_0^2))\|_{L^2} \ls \lambda_n^{\f{25}{24}+\f{3\ep_0}{2}}.
\end{equation}

By Lemma~\ref{lem:PSIDOs}.4, $\widetilde{m}(\f 1i\nabla)$ is a bounded operator in $L^4$. Therefore, using H\"older's inequality and the estimates in \eqref{N.diff.2.improvement} and in Proposition~\ref{prop:psi.localized}, we obtain
\begin{equation}
\begin{split}
&\: \left| \f{b_{c,\alp}}{N_{c,\alp}^2}\int_{\mathbb R^{2+1}} (\rd_t - \bt^k_{c,\alp} \rd_k)(\zeta_\alp \chi(\psi_n-\psi_0)) \de^{ij}  \{ ({\bm N}_{\mathrm{diff}})_{n,2} \widetilde{m}(\f 1i\nabla)(\rd^2_{ij} (\zeta_\alp \chi(\psi_n-\psi_0))) \,\ud x \right| \\
&\: + \left| \f{b_{c,\alp}}{N_{c,\alp}^2}\int_{\mathbb R^{2+1}} (\rd_t - \bt^k_{c,\alp} \rd_k)(\zeta_\alp \chi(\psi_n-\psi_0)) \de^{ij} \widetilde{m}(\f 1i\nabla) [({\bm N}_{\mathrm{diff}})_{n,2} \rd^2_{ij}(\zeta_\alp \chi(\psi_n-\psi_0))] \}\,\ud x \right| \\
\ls &\: \|(\rd_t - \bt^k_{c,\alp} \rd_k)(\zeta_\alp \chi(\psi_n-\psi_0))\|_{L^4} \|\de^{ij} \rd^2_{ij} (\zeta_\alp \chi(\psi_n-\psi_0))\|_{L^4} \| ({\bm N}_{\mathrm{diff}})_{n,2}\|_{L^2} \\
\ls &\: \lambda_n^{\f {3\ep_0}{4}} \cdot \lambda_n^{-1+\f {3\ep_0}{4}} \cdot \lambda_n^{\f{25}{24}+\f{3\ep_0}{2}} = \lambda_n^{\f 1{24}+3\ep_0} = o(\lambda_n^{3\ep_0}),
\end{split}
\end{equation}
as desired.

\pfstep{Step~4: Handling $({\bm N}_{\mathrm{diff}})_{n,3}$} To handle $({\bm N}_{\mathrm{diff}})_{n,3}$, we need to compute in Fourier space. Here, we take full advantage of having frozen the coefficients. In order to simplify the formulae, we will denote $({\bm \psi}_{\mathrm{diff}})_n:= \zeta_\alp \chi(\psi_n-\psi_0)$.

\begin{equation}\label{N.Fourier.1}
\begin{split}
&\: \f{b_{c,\alp}}{N_{c,\alp}^2}\int_{\mathbb R^{2+1}} (\rd_t - \bt^k_{c,\alp} \rd_k)({\bm \psi}_{\mathrm{diff}})_n \de^{ij}  \{ ({\bm N}_{\mathrm{diff}})_{n,3} \widetilde{m}(\f 1i\nabla)(\rd^2_{ij} ({\bm \psi}_{\mathrm{diff}})_n) - \widetilde{m}(\f 1i\nabla) [({\bm N}_{\mathrm{diff}})_{n,3} \rd^2_{ij}({\bm \psi}_{\mathrm{diff}})_n] \}\,\ud x\\
=&\: \f{-ib_{c,\alp}}{N_{c,\alp}^2} \iint (\xi_t-\bt^k_{c,\alp}\xi_k)\eta_i\eta_j \de^{ij}\widehat{({\bm \psi}_{\mathrm{diff}})_n}(\xi) \overline{\widehat{({\bm N}_{\mathrm{diff}})_{n,3}}(\xi-\eta)} \overline{\widehat{({\bm \psi}_{\mathrm{diff}})_n}(\eta)}[\widetilde{m}(\eta)-\widetilde{m}(\xi)]\,\ud \eta\,\ud \xi \\
=&\: \f{-ib_{c,\alp}}{N_{c,\alp}^2} \iint (\xi_t-\bt^k_{c,\alp}\xi_k)\eta_i\eta_j \de^{ij}\widehat{({\bm \psi}_{\mathrm{diff}})_n}(\xi) \widehat{({\bm N}_{\mathrm{diff}})_{n,3}}(\eta-\xi) \widehat{({\bm \psi}_{\mathrm{diff}})_n}(-\eta)[\widetilde{m}(\eta)-\widetilde{m}(\xi)]\,\ud \eta\,\ud \xi.
\end{split}
\end{equation}
Exchanging $\xi$ and $\eta$,\eqref{N.Fourier.1} can be given equivalently as 
\begin{equation}\label{N.Fourier.2}
\begin{split}
\mbox{\eqref{N.Fourier.1}}=&\: \f{-ib_{c,\alp}}{N_{c,\alp}^2} \iint (\eta_t-\bt^k_{c,\alp}\eta_k)\xi_i\xi_j \de^{ij}\widehat{({\bm \psi}_{\mathrm{diff}})_n}(\eta) \widehat{({\bm N}_{\mathrm{diff}})_{n,3}}(\xi-\eta) \widehat{({\bm \psi}_{\mathrm{diff}})_n}(-\xi)[\widetilde{m}(\xi)-\widetilde{m}(\eta)]\,\ud \eta\,\ud \xi.
\end{split}
\end{equation}
Changing variables $\xi\mapsto -\xi$ and $\eta \mapsto -\eta$, and using the evenness of $m$, 
\begin{equation}\label{N.Fourier.3}
\begin{split}
\mbox{\eqref{N.Fourier.2}}= &\: \f{-ib_{c,\alp}}{N_{c,\alp}^2} \iint (\eta_t-\bt^k_{c,\alp}\eta_k)\xi_i\xi_j \de^{ij} \widehat{({\bm \psi}_{\mathrm{diff}})_n}(\xi) \widehat{({\bm N}_{\mathrm{diff}})_{n,3}}(\eta-\xi) \widehat{({\bm \psi}_{\mathrm{diff}})_n}(-\eta)[\widetilde{m}(\eta)-\widetilde{m}(\xi)]\,\ud \eta\,\ud \xi.
\end{split}
\end{equation}
Therefore, averaging between \eqref{N.Fourier.1} and \eqref{N.Fourier.3}, we obtain
\begin{equation}\label{N.Fourier.final}
\begin{split}
&\: \f{b_{c,\alp}}{N_{c,\alp}^2}\int_{\mathbb R^{2+1}} (\rd_t - \bt^k_{c,\alp} \rd_k)({\bm \psi}_{\mathrm{diff}})_n \de^{ij}  \{ ({\bm N}_{\mathrm{diff}})_{n,3} \widetilde{m}(\f 1i\nabla)(\rd^2_{ij} ({\bm \psi}_{\mathrm{diff}})_n) - \widetilde{m}(\f 1i\nabla) [({\bm N}_{\mathrm{diff}})_{n,3} \rd^2_{ij}({\bm \psi}_{\mathrm{diff}})_n] \}\,\ud x\\
=&\: \f{-ib_{c,\alp}}{2 N_{c,\alp}^2} \iint \left((\xi_t-\bt^k_{c,\alp}\xi_k)|\eta_i|^2 + (\eta_t-\bt^k_{c,\alp}\eta_k)|\xi_i|^2\right) \\
&\: \qquad \times\widehat{({\bm \psi}_{\mathrm{diff}})_n}(\xi) \widehat{({\bm N}_{\mathrm{diff}})_{n,3}}(\eta-\xi) \widehat{({\bm \psi}_{\mathrm{diff}})_n}(-\eta)[\widetilde{m}(\eta)-\widetilde{m}(\xi)]\,\ud \eta\,\ud \xi.
\end{split}
\end{equation}

\pfstep{Step~4(a): Some manipulation of the Fourier multiplier}
Note that on the support of $\Big(1-\Theta(\lambda_n^{\f 56}|\xi-\eta|)\Big)\Theta\Big(\f{|\xi_i-\eta_i|}{|\xi-\eta|^{\f 58}}\Big)$, we easily have $|\xi_t - \eta_t| \gtrsim |\xi - \eta| - |\xi_j -\eta_j| \gtrsim |\xi - \eta|$. It follows that
\begin{equation}\label{eq:Fourier.lower.bound}
\begin{split}
\Big| (\xi_t - \eta_t) - \bt^j_{c,\alp}(\xi_j - \eta_j) \Big| \Big(1-\Theta(\lambda_n^{\f 56}|\xi-\eta|)\Big)\Theta\Big(\f{|\xi_i-\eta_i|}{|\xi-\eta|^{\f 58}}\Big)  \gtrsim \lambda_n^{-\f 56}.
\end{split}
\end{equation}
In particular, $(\xi_t - \eta_t) - \bt^j_{c,\alp}(\xi_j - \eta_j)$ is bounded away from $0$. Therefore, we can divide by $(\xi_t - \eta_t) - \bt^j_{c,\alp}(\xi_j - \eta_j)$ and a direct computation shows that
\begin{equation}\label{expression.for.sum.of.t.frequency}
\begin{split}
&\: (\eta_t - \bt^j_{c,\alp}\eta_j) + (\xi_t - \bt^j_{c,\alp}\xi_j) \\
= &\: \f{(\eta_t - \bt^j_{c,\alp}\eta_j)^2 - e^{-2\gamma_{c,\alp}} N_{c,\alp}^2 |\eta_i|^2}{(\eta_t - \bt^j_{c,\alp}\eta_j) - (\xi_t - \bt^j_{c,\alp}\xi_j)} - \f{(\xi_t - \bt^j_{c,\alp}\xi_j)^2 - e^{-2\gamma_{c,\alp}} N_{c,\alp}^2 |\xi_i|^2}{(\eta_t - \bt^j_{c,\alp}\eta_j) - (\xi_t - \bt^j_{c,\alp}\xi_j)} + \f{e^{-2\gamma_{c,\alp}} N_{c,\alp}^2 (|\eta_i|^2 -|\xi_i|^2)}{(\eta_t - \bt^j_{c,\alp}\eta_j) - (\xi_t - \bt^j_{c,\alp}\xi_j)}.
\end{split}
\end{equation}

Using \eqref{expression.for.sum.of.t.frequency}, we can therefore write the Fourier multiplier in \eqref{N.Fourier.final} as follows:
\begin{align}
&\: \left( |\xi_i|^2 (\eta_t - \bt^j_{c,\alp}\eta_j) + |\eta_i|^2 (\xi_t - \bt^j_{c,\alp}\xi_j)\right) \Big(1-\Theta(\lambda_n^{\f 56}|\xi-\eta|)\Big)\Theta\Big(\f{|\xi_i-\eta_i|}{|\xi-\eta|^{\f 58}}\Big) \nonumber\\
= &\: \left(\left( |\xi_i|^2 - |\eta_i|^2\right) (\eta_t - \bt^j_{c,\alp}\eta_j) + |\eta_i|^2 \left((\eta_t - \bt^j_{c,\alp}\eta_j) + (\xi_t - \bt^j_{c,\alp}\xi_j) \right)\right) \Big(1-\Theta(\lambda_n^{\f 56}|\xi-\eta|)\Big)\Theta\Big(\f{|\xi_i-\eta_i|}{|\xi-\eta|^{\f 58}}\Big) \nonumber\\
= &\: \de^{ik}(\xi_i-\eta_i)(\xi_k + \eta_k)(\eta_t - \bt^j_{c,\alp}\eta_j )\Big(1-\Theta(\lambda_n^{\f 56}|\xi-\eta|)\Big)\Theta\Big(\f{|\xi_i-\eta_i|}{|\xi-\eta|^{\f 58}}\Big) \label{N.multiplier.1}\\
&\: + |\eta_i|^2 \f{(\eta_t - \bt^j_{c,\alp}\eta_j)^2 - e^{-2\gamma_{c,\alp}} N_{c,\alp}^2 |\eta_i|^2}{(\eta_t - \bt^j_{c,\alp}\eta_j) - (\xi_t - \bt^j_{c,\alp}\xi_j)} \Big(1-\Theta(\lambda_n^{\f 56}|\xi-\eta|)\Big)\Theta\Big(\f{|\xi_i-\eta_i|}{|\xi-\eta|^{\f 58}}\Big) \label{N.multiplier.2}\\
&\: - |\eta_i|^2 \f{(\xi_t - \bt^j_{c,\alp}\xi_j)^2 - e^{-2\gamma_{c,\alp}} N_{c,\alp}^2 |\xi_i|^2}{(\eta_t - \bt^j_{c,\alp}\eta_j) -(\xi_t - \bt^j_{c,\alp}\xi_j)} \Big(1-\Theta(\lambda_n^{\f 56}|\xi-\eta|)\Big)\Theta\Big(\f{|\xi_i-\eta_i|}{|\xi-\eta|^{\f 58}}\Big) \label{N.multiplier.3}\\
&\: + |\eta_i|^2 \f{e^{-2\gamma_{c,\alp}} N_{c,\alp}^2 \de^{k\ell}(\eta_k -\xi_k)(\eta_\ell + \xi_\ell)}{(\eta_t - \bt^j_{c,\alp}\eta_j) - (\xi_t - \bt^j_{c,\alp}\xi_j)}\Big(1-\Theta(\lambda_n^{\f 56}|\xi-\eta|)\Big)\Theta\Big(\f{|\xi_i-\eta_i|}{|\xi-\eta|^{\f 58}}\Big). \label{N.multiplier.4}
\end{align}

\pfstep{Step~4(b): Estimating each term} Define now the term $\mathrm{I}$, $\mathrm{II}$, $\mathrm{III}$ and $\mathrm{IV}$ respectively by inserting \eqref{N.multiplier.1}, \eqref{N.multiplier.2}, \eqref{N.multiplier.3} and \eqref{N.multiplier.4} into (*) below
\begin{equation}
\begin{split}
\f{-ib_{c,\alp}}{2 N_{c,\alp}^2} \iint \left(*\right) \widehat{({\bm \psi}_{\mathrm{diff}})_n}(\xi) {(\widehat{\zeta_\alp N_n^2}- \widehat{\zeta_\alp N_0^2})}(\eta-\xi) \widehat{({\bm \psi}_{\mathrm{diff}})_n}(-\eta) [\widetilde{m}(\eta)-\widetilde{m}(\xi)]\,\ud \eta\,\ud \xi.
\end{split}
\end{equation}

For the term $\mathrm{I}$, by first applying {Fourier inversion} and {then} H\"older's inequality, we obtain
\begin{equation}\label{eq.N3.1}
\begin{split}
\mathrm{|I|} \ls &\: \left(\| \zeta_\alp \chi (\psi_n-\psi_0{)}\|_{L^{{4}}} \|\widetilde{m}(\f 1i\nabla)\partial_i (\rd_t - \bt^j_{c,\alp} \rd_j)(\zeta_\alp \chi (\psi_n - \psi_0) ) \|_{L^4} \right.\\
&\: \left. \qquad+\|\partial_i {(}\zeta_\alp \chi (\psi_n-\psi_0))\|_{L^{{4}}} \|\widetilde{m}(\f 1i\nabla) (\rd_t - \bt^j_{c,\alp} \rd_j)(\zeta_\alp \chi (\psi_n - \psi_0) ) \|_{L^4} \right.\\
&\: \left. \qquad+\|\widetilde{m}(\f 1i\nabla)\partial_i {(} \zeta_\alp \chi (\psi_n-\psi_0))\|_{L^{{4}}} \| (\rd_t - \bt^j_{c,\alp} \rd_j)(\zeta_\alp \chi (\psi_n - \psi_0) ) \|_{L^4} \right.\\
&\: \left. \qquad+\|\widetilde{m}(\f 1i\nabla)(\zeta_\alp \chi (\psi_n-\psi_0))\|_{L^{{4}}} \| \partial_i(\rd_t - \bt^j_{c,\alp} \rd_j)(\zeta_\alp \chi (\psi_n - \psi_0) ) \|_{L^4} \right)\\
&\: \qquad \qquad \times\|(1-\Theta(\lambda_n^{\f 56}|\nabla|))\Theta(\f{|\nabla_i|}{|\nabla|^{\f 58}}) \rd_i (\zeta_\alp(N_n^2 - N_0^2))\|_{L^{{2}}}.
\end{split}
\end{equation}
With the estimates in \eqref{assumption.0}{, \eqref{assumption.1}} and Proposition~\ref{prop:g.localized}, {Plancherel's theorem and H\"older's inequality,} we obtain
\begin{equation}\label{eq.N3.1.1}
\begin{split}
&\: {\|(1-\Theta(\lambda_n^{\f 56}|\nabla|))\Theta(\f{|\nabla_i|}{|\nabla|^{\f 58}}) \rd_i (\zeta_\alp(N_n^2 - N_0^2))\|_{L^{2}}} \\
\ls &\: {\|\rd_i (\zeta_\alp (N_n - N_0))\|_{L^2} \|N_n + N_0\|_{L^\i} + \|\zeta_\alp (N_n - N_0)\|_{L^2} \|\rd_i(N_n + N_0)\|_{L^\i} }\\
\ls &\: {\lambda_n^{\f 12 + \f{3\ep_0}{2}} + \lambda_n^{1 + \f{3\ep_0}{2}} \ls \lambda_n^{\f 12 + \f{3\ep_0}{2}}.}
\end{split}
\end{equation}
{Plugging \eqref{eq.N3.1.1} into \eqref{eq.N3.1} and using using the estimates in Proposition~\ref{prop:psi.localized} together with Lemma~\ref{lem:PSIDOs}.4, we obtain}
\begin{equation}\label{eq.N3.1.final}
|\mathrm{I}|\ls (\lambda_n^{1+\f {3\ep_0}{{4}}}\cdot\lambda_n^{-1+\f{3\ep_0}{4}}+\lambda_n^{\f {3\ep_0}{{4}}}\cdot\lambda_n^{\f{3\ep_0}{4}}) \lambda_n^{\frac{1}{2}+\frac{3\ep_0}{{2}}}= \lambda_n^{\f 12+3\ep_0} = o(\lambda_n^{3\ep_0}).
\end{equation}

To handle the term $\mathrm{II}$, we likewise apply the inverse Fourier transform and then use H\"older's inequality to obtain
\begin{equation}\label{eq.N3.2}
\begin{split}
|\mathrm{II}|\ls & \: \left(\|\widetilde{m}(\f 1i \nabla)\rd_i^2 \widetilde{\Box}_{c,\alp} (\zeta_\alp \chi (\psi_n-\psi_0))\|_{L^4} \|\zeta_\alp \chi (\psi_n-\psi_0)\|_{L^4} \right.\\
&\: \left. \qquad + \|\rd_i^2 \widetilde{\Box}_{c,\alp} (\zeta_\alp \chi (\psi_n-\psi_0))\|_{L^4} \|\widetilde{m}(\f 1i \nabla)(\zeta_\alp \chi (\psi_n-\psi_0))\|_{L^4} \right)\\
&\: \qquad\qquad \times \Big\|(1-\Theta(\lambda_n^{\f 56}|\nabla|))\Theta(\f{|\nabla_i|}{|\nabla|^{\f 58}}) \f{i}{\nabla_t - \bt^j \nabla_j} (\zeta_\alp (N_n^2- N_0^2)) \Big\|_{L^2}.
\end{split}
\end{equation}
By Plancherel's theorem, \eqref{eq:Fourier.lower.bound}, H\"older's inequality, \eqref{assumption.1} and Proposition~\ref{prop:g.localized}, we obtain
\begin{equation}\label{eq.N3.2.1}
\begin{split}
&\: \Big\|(1-\Theta(\lambda_n^{\f 56}|\nabla|))\Theta(\f{|\nabla_i|}{|\nabla|^{\f 58}}) \f{i}{\nabla_t - \bt^j \nabla_j} (\zeta_\alp (N_n^2- N_0^2)) \Big\|_{L^2} \\
\ls &\: \lambda_n^{\f 56}\|\zeta_\alp(N_n^2-N_0^2)\|_{L^2} \ls \lambda_n^{\f 56}\|\zeta_\alp(N_n-N_0)\|_{L^2} \|N_n+N_0\|_{L^\i} \ls \lambda_n^{\f {11}6+ \f{3\ep_0}{2}}.
\end{split}
\end{equation}
Plugging \eqref{eq.N3.2.1} into \eqref{eq.N3.2} and using the estimates in Propositions~\ref{prop:psi.localized} and \ref{prop:approximate.wave} together with Lemma~\ref{lem:PSIDOs}.4, we obtain
\begin{equation}\label{eq.N3.2.final}
\begin{split}
|\mathrm{II}| \ls \lambda_n^{-3+\ep_0+\f{3\ep_0}{4}} \cdot \lambda_n^{1+\f{3\ep_0}{4}} \cdot \lambda_n^{\f {11}6+ \f{3\ep_0}{2}} = \lambda_n^{-\f 16 + 4\ep_0} = o(\lambda_n^{3\ep_0}), 
\end{split}
\end{equation}
since $\ep_0>\f 16$ (cf.~Section~\ref{sec:freeze.coeff}).

$\mathrm{III}$ can be controlled in an entirely analogous manner as $\mathrm{II}$; we omit the details:
\begin{equation}\label{eq.N3.3.final}
|\mathrm{III}| \ls  \lambda_n^{-\f 16 + 4\ep_0} = o(\lambda_n^{3\ep_0}).
\end{equation}

Finally, we handle the term $\mathrm{IV}$. As before, we apply the inverse Fourier transform and then use H\"older's inequality. We then obtain
\begin{equation}\label{eq.N3.4}
\begin{split}
|\mathrm{IV}| \ls &\: \left(\| \widetilde{m}(\f 1i \nabla)\rd_k \rd_i^2 (\zeta_\alp \chi (\psi_n-\psi_0))\|_{L^4} \|\zeta_\alp \chi (\psi_n - \psi_0) \|_{L^4} \right.\\
&\:\quad \left.+ \| \widetilde{m}(\f 1i \nabla)\rd_i^2(\zeta_\alp \chi (\psi_n-\psi_0))\|_{L^4} \|\rd_k(\zeta_\alp \chi (\psi_n - \psi_0) ) \|_{L^4}\right) \\
&\:\qquad \times \Big\|(1-\Theta(\lambda_n^{\f 56}|\nabla|))\Theta(\f{|\nabla_i|}{|\nabla|^{\f 58}}) \rd_i \f{i}{\nabla_t - \bt^j \nabla_j} (\zeta_\alp(N_n^2 - N_0^2)) \Big\|_{L^2} 
\end{split}
\end{equation}
By Plancherel's theorem, \eqref{eq:Fourier.lower.bound}, H\"older's inequality, \eqref{assumption.1} and Proposition~\ref{prop:g.localized}, we obtain
\begin{equation}\label{eq.N3.4.1}
\begin{split}
&\: \Big\|(1-\Theta(\lambda_n^{\f 56}|\nabla|))\Theta(\f{|\nabla_i|}{|\nabla|^{\f 58}}) \rd_i \f{i}{\nabla_t - \bt^j \nabla_j} (\zeta_\alp(N_n^2 - N_0^2)) \Big\|_{L^2} \\
\ls &\: \Big\|(1-\Theta(\lambda_n^{\f 56}|\xi|))\Theta(\f{|\xi_i|}{|\xi|^{\f 58}}) |\xi|^{\f 58-1} (\widehat{\zeta_\alp N_n^2} - \widehat{\zeta_\alp N_0^2}) \Big\|_{L^2} \\
\ls &\: \lambda_n^{(\f 58-1)(-\f 56)} \|\zeta_\alp(N_n^2 - N_0^2)\|_{L^2} \ls \lambda_n^{\f 5{16}} \|\zeta_\alp(N_n-N_0)\|_{L^2} \|N_n+N_0\|_{L^\i} \ls \lambda_n^{\f {21}{16}+ \f{3\ep_0}{2}}.
\end{split}
\end{equation}
Plugging \eqref{eq.N3.4.1} into \eqref{eq.N3.4} and using the estimates in Propositions~\ref{prop:psi.localized} together with Lemma~\ref{lem:PSIDOs}.4, we obtain
\begin{equation}\label{eq.N3.4.final}
\begin{split}
|\mathrm{IV}|\ls \lambda_n^{-2+\f{3\ep_0}{4}} \cdot \lambda_n^{1+\f{3\ep_0}{4}} \cdot \lambda_n^{\f {21}{16}+ \f{3\ep_0}{2}} = \lambda_n^{\f 5{16}+3\ep_0} = o(\lambda_n^{3\ep_0}).
\end{split}
\end{equation}

By \eqref{eq.N3.1.final}, \eqref{eq.N3.2.final}, \eqref{eq.N3.3.final} and \eqref{eq.N3.4.final}, we have thus shown that each of the terms obey the desired estimate. This concludes the proof. \qedhere
\end{proof}

Let us summarize what we have achieved in this subsection. At this point, let us also note that while the computations in this subsection concerns the commutator term involving $\psi_n$, they apply in an identical manner to the commutator term involving $\om_n$.
\begin{proposition}\label{prop:N.commute.final}
\begin{equation*}
\begin{split}
&\: \int_{\mathbb R^{2+1}} \f{(e_0)_n(\chi\psi_n)}{N_n} \de^{ij} \Big\{ \rd_i [N_n^2 A (\f{\rd_j (\chi\psi_n)}{N_n})] - A \rd_i [N_n \rd_j(\chi\psi_n)] \Big\}\,\ud x \\
&\: - \int_{\mathbb R^{2+1}} \f{(e_0)_0(\chi\psi_n)}{N_0} \de^{ij} \Big\{ \rd_i [N_0^2 A(\f{\rd_j(\chi\psi_n)}{N_0})] - A \rd_i [N_0 \rd_j(\chi\psi_n)] \Big\}\,\ud x 
\to 0.
\end{split}
\end{equation*}
A similar statement holds after replacing $\psi_n\rightsquigarrow \om_n$, $\psi_0\rightsquigarrow \om_0$ and $\ud x\rightsquigarrow \f 14 e^{-4\psi_0} \ud x$.
\end{proposition}

\subsection{The term \eqref{eq:main.term.commutator.beta}}\label{sec:commutator.beta}

We now look at the term \eqref{eq:main.term.commutator.beta}. Unlike the terms \eqref{eq:main.term.commutator.gamma} and \eqref{eq:main.term.commutator.N} (cf.~Sections~\ref{sec:commutator.gamma} and \ref{sec:commutator.N}), we will not be able to just compute the limit of \eqref{eq:main.term.commutator.beta}. Instead we will need to \emph{combine} this term with the ``$\mathrm{hard}$'' term in \eqref{eq:energy.id.n.1}.

For this reason, we first consider some reduction of \eqref{eq:main.term.commutator.beta} and the ``$\mathrm{hard}$'' term in \eqref{eq:energy.id.n.1} in \textbf{Section~\ref{sec:reduction.commutator.beta}} and \textbf{Section~\ref{sec:reduction.hard}} respectively. We then consider the limit of the combination in \textbf{Section~\ref{sec:commute.beta.combined}}.

\subsubsection{Reduction for the term \eqref{eq:main.term.commutator.beta}}\label{sec:reduction.commutator.beta}

We first argue as in Proposition~\ref{prop:N.commute.1} and identify the main term in the limit. The proof is essentially the same as Proposition~\ref{prop:N.commute.1} and is omitted.
\begin{proposition}\label{prop:reduction.beta.main.1}
\begin{equation*}
\begin{split}
&\: \int_{\mathbb R^{2+1}} \f{(e_0)_n(\chi\psi_n)}{N_n} \Big\{ \rd_k [e^{2\gamma_n} \bt_n^k A(\f{(e_0)_n(\chi\psi_n)}{N_n})] - A \rd_k [e^{2\gamma_n} \bt_n^k (\f{(e_0)_n(\chi\psi_n)}{N_n})] \Big\}\,\ud x \\
&\: - \int_{\mathbb R^{2+1}} \f{(e_0)_0(\chi\psi_n)}{N_0} \Big\{ \rd_k [e^{2\gamma_0} \bt_0^k A(\f{(e_0)_0(\chi\psi_n)}{N_0})] - A\rd_k  [e^{2\gamma_0} \bt_0^k (\f{{(e_0)_0}(\chi\psi_n)}{N_0})] \Big\}\,\ud x \\
&\: - \int_{\mathbb R^{2+1}} \f{(e_0)_0(\chi\psi_n)}{N_0} b \Big\{ (e^{2\gamma_n} \bt_n^k-e^{2\gamma_0} \bt_0^k) \widetilde{m}(\f 1i\nabla)(\f{\rd_k(e_0)_0(\chi\psi_n)}{N_0}) \\
&\: \qquad\qquad\qquad\qquad\qquad- \widetilde{m}(\f 1i\nabla)  [(e^{2\gamma_n} \bt_n^k-e^{2\gamma_0} \bt_0^k) (\f{\rd_k(e_0)_0(\chi\psi_n)}{N_0})] \Big\}\,\ud x\to 0.
\end{split}
\end{equation*}
\end{proposition}

\begin{proposition}\label{prop:reduction.beta.main.2}
\begin{equation*}
\begin{split}
&\: \int_{\mathbb R^{2+1}} \f{(e_0)_0(\chi\psi_n)}{N_0} b \Big\{ (e^{2\gamma_n} \bt_n^k-e^{2\gamma_0} \bt_0^k) \widetilde{m}(\f 1i\nabla)(\f{\rd_k(e_0)_0(\chi\psi_n)}{N_0}) \\
&\: \qquad\qquad\qquad\qquad\qquad- \widetilde{m}(\f 1i\nabla)  [(e^{2\gamma_n} \bt_n^k-e^{2\gamma_0} \bt_0^k) (\f{\rd_k(e_0)_0(\chi\psi_n)}{N_0})] \Big\}\,\ud x \\
&\: -\int_{\mathbb R^{2+1}} \f{(e_0)_0(\chi\psi_n)}{N_0} b \Big\{ e^{2\gamma_0}(\bt_n^k-\bt_0^k) \widetilde{m}(\f 1i\nabla)(\f{\rd_k(e_0)_0(\chi\psi_n)}{N_0}) \\
&\: \qquad\qquad\qquad\qquad\qquad - \widetilde{m}(\f 1i\nabla)  [e^{2\gamma_0}(\bt_n^k-\bt_0^k) (\f{\rd_k(e_0)_0(\chi\psi_n)}{N_0})] \Big\}\,\ud x \to 0.
\end{split}
\end{equation*}
\end{proposition}
\begin{proof}
In view of
$$e^{2\gamma_n} \bt_n^k-e^{2\gamma_0} \bt_0^k = e^{2\gamma_0} (\bt_n^k - \bt_0^k)+ \bt^k_n (e^{2\gamma_n} - e^{2\gamma_0}),$$
it suffices to show that 
\begin{equation}\label{goal.for.handling.gamma.comm}
\begin{split}
&\: \int_{\mathbb R^{2+1}} \f{(e_0)_0(\chi\psi_n)}{N_0} b \Big\{ \bt^k_n (e^{2\gamma_n} - e^{2\gamma_0}) \widetilde{m}(\f 1i\nabla)(\f{\rd_k(e_0)_0(\chi\psi_n)}{N_0}) \\
&\: \qquad\qquad\qquad\qquad\qquad - \widetilde{m}(\f 1i\nabla)  [\bt^k_n (e^{2\gamma_n} - e^{2\gamma_0}) (\f{\rd_k(e_0)_0(\chi\psi_n)}{N_0})] \Big\}\,\ud x \to 0.
\end{split}
\end{equation}
We compute
\begin{align}
&\: \bt^k_n (e^{2\gamma_n} - e^{2\gamma_0}) \widetilde{m}(\f 1i\nabla)(\f{\rd_k(e_0)_0(\chi\psi_n)}{N_0}) - \widetilde{m}(\f 1i\nabla)  \big[\bt^k_n (e^{2\gamma_n} - e^{2\gamma_0}) (\f{\rd_k(e_0)_0(\chi\psi_n)}{N_0}) \big] \notag\\
= &\: \bt^k_n (e^{2\gamma_n} - e^{2\gamma_0}) \widetilde{m}(\f 1i\nabla)\rd_k (\f{(e_0)_0(\chi\psi_n)}{N_0})  - \widetilde{m}(\f 1i\nabla)  \rd_k \big[\bt^k_n (e^{2\gamma_n} - e^{2\gamma_0}) \f{(e_0)_0(\chi\psi_n)}{N_0}\big] \label{handling.gamma.comm.1}\\
&\:  + \bt^k_n (e^{2\gamma_n} - e^{2\gamma_0}) \widetilde{m}(\f 1i\nabla) \big[\f{((e_0)_0(\chi\psi_n))(\rd_k N_0)}{N_0^2}\big] + \widetilde{m}(\f 1i\nabla)\big[ ((e_0)_0(\chi\psi_n)) \rd_k \f{\bt^k_n (e^{2\gamma_n} - e^{2\gamma_0})}{N_0}\big]. \label{handling.gamma.comm.2}
\end{align}

\eqref{handling.gamma.comm.1} can be controlled using Calder\'on's commutator theorem (Lemma~\ref{lem:PSIDOs}.6) with $T = \widetilde{m}(\f 1i\nabla)\rd_k$ and using the estimates in \eqref{assumption.0}, \eqref{assumption.1}, Propositions~\ref{prop:spatial.imp} and \ref{prop:dtgamma.imp}, we obtain
\begin{equation*}
\begin{split}
\|\eqref{handling.gamma.comm.1}\|_{L^2}\ls \|\widetilde{\chi} \bt^k_n (e^{2\gamma_n} - e^{2\gamma_0})\|_{C^1} \Big\|\f{(e_0)_0(\chi\psi_n)}{N_0} \Big\|_{L^2} \ls \lambda_n^{\f 12}.
\end{split}
\end{equation*}
On the other hand, \eqref{handling.gamma.comm.2} can be bounded using estimates in \eqref{assumption.0}, \eqref{assumption.1} and Proposition~\ref{prop:spatial.imp} as follows:
\begin{equation*}
\begin{split}
\|\eqref{handling.gamma.comm.2}\|_{L^2}\ls &\: \|\widetilde{\chi} \bt^k_n (e^{2\gamma_n} - e^{2\gamma_0})\|_{L^\i} \Big\|\f{((e_0)_0(\chi\psi_n))(\rd_k N_0)}{N_0^2} \Big\|_{L^2}  \\
&\: + \|(e_0)_0(\chi\psi_n)\|_{L^2} \Big\|\widetilde{\chi} \rd_k \f{\bt^k_n (e^{2\gamma_n} - e^{2\gamma_0})}{N_0} \Big\|_{L^\i} \ls \lambda_n + \lambda_n^{\f 12} \ls \lambda_n^{\f 12}.
\end{split}
\end{equation*}
It therefore follows that as $n\to +\infty$,
$$\Big\|\bt^k_n (e^{2\gamma_n} - e^{2\gamma_0}) \widetilde{m}(\f 1i\nabla)(\f{\rd_k(e_0)_0(\chi\psi_n)}{N_0})  - \widetilde{m}(\f 1i\nabla)  [\bt^k_n (e^{2\gamma_n} - e^{2\gamma_0}) (\f{\rd_k(e_0)_0(\chi\psi_n)}{N_0})] \Big\|_{L^2} \to 0.$$

Hence, our goal \eqref{goal.for.handling.gamma.comm} follows from the Cauchy--Schwarz inequality and the estimates in \eqref{assumption.1}. \qedhere
\end{proof}

We now take the main term in Proposition~\ref{prop:reduction.beta.main.2} (i.e.~the term on the last two lines) and show that the limit remains the same after (1) replacing $\psi_n \mapsto \psi_n-\psi_0$ and (2) freezing the coefficients. The proof is entirely analogous to Propositions~\ref{prop:reduction.to.diff} and \ref{prop:reduction.to.frozen} and is omitted.
\begin{proposition}\label{prop:reduction.beta.main.3}
Let $b_{c,\alp}$, $N_{c,\alp}$, $\gamma_{c,\alp}$ and $\bt_{c,\alp}^i$ be as in Proposition~\ref{prop:constants}. Then
\begin{equation*}
\begin{split}
&\: \int_{\mathbb R^{2+1}} \f{(e_0)_0(\chi\psi_n)}{N_0} b \Big\{ e^{2\gamma_0}(\bt_n^k-\bt_0^k) \widetilde{m}(\f 1i\nabla)(\f{\rd_k(e_0)_0(\chi\psi_n)}{N_0}) \\
&\: \qquad\qquad\qquad\qquad\qquad - \widetilde{m}(\f 1i\nabla)  [e^{2\gamma_0}(\bt_n^k-\bt_0^k) (\f{\rd_k(e_0)_0(\chi\psi_n)}{N_0})] \Big\}\,\ud x\\
&\: - \sum_{\alp} \f{b_{c,\alp}e^{2\gamma_{c,\alp}}}{N_{c,\alp}^2} \int_{\mathbb R^{2+1}} (\rd_t - \bt^{\ell}_{c,\alp}\rd_\ell)(\zeta_\alp\chi(\psi_n-\psi_0)) \Big\{  (\bt_n^k - \bt_0^k) \widetilde{m}(\f 1i\nabla)(\rd_k(\rd_t - \bt^m_{c,\alp}\rd_m)(\zeta_\alp\chi(\psi_n-\psi_0)))\\
&\: \qquad \qquad \qquad \qquad \qquad \qquad \qquad  - \widetilde{m}(\f 1i\nabla)  [ (\bt_n^k - \bt_0^k) (\rd_k(\rd_t - \bt^m_{c,\alp}\rd_m)(\zeta_\alp\chi(\psi_n-\psi_0))] \Big\}\,\ud x\to 0.
\end{split}
\end{equation*}
\end{proposition}

\subsubsection{Reduction for the ``$\mathrm{hard}$'' term in \eqref{eq:energy.id.n.1}}\label{sec:reduction.hard}

Note that
$$[(e_0)_n,A](\f{\rd_j(\chi\psi_n)}{N_n}) = -\bt^k_n \rd_k A(\f{\rd_j(\chi\psi_n)}{N_n}) + A[\bt^k_n \rd_k (\f{\rd_j(\chi\psi_n)}{N_n})] + [\partial_t,A](\f{\rd_j(\chi\psi_n)}{N_n}) .$$
Hence, the ``$\mathrm{hard}$'' term in \eqref{eq:energy.id.n.1} has a similar form as the previous commutator terms, can also be treated in a similar manner.

First, we identify one main term for which the limit is difficult to compute. This is similar to Proposition~\ref{prop:reduction.beta.main.1}; we omit the details.
\begin{proposition}\label{prop:reduction.beta.hard.1}
\begin{equation*}
\begin{split}
&\: -\int_{\mathbb R^{2+1}} [\rd_i (\chi\psi_n)] \de^{ij} N_n \Big\{[(e_0)_n,A](\f{\rd_j(\chi\psi_n)}{N_n}) \Big\} \,\ud x + \int_{\mathbb R^{2+1}} [\rd_i (\chi\psi_n)] \de^{ij} N_0 \Big\{[(e_0)_0,A](\f{\rd_j(\chi\psi_n)}{N_0}) \Big\} \,\ud x \\
&\: - \int_{\mathbb R^{2+1}} [\rd_i (\chi\psi_n)] \de^{ij} N_0 b \Big\{(\bt_n^k-\bt_0^k) \widetilde{m}(\f 1i\nabla)(\f{\rd^2_{jk}(\chi\psi_n)}{N_0}) - \widetilde{m}(\f 1i\nabla)[(\bt_n^k-\bt_0^k) \f{\rd^2_{jk}(\chi\psi_n)}{N_0}] \Big\} \,\ud x \to 0
\end{split}
\end{equation*}
\end{proposition}

Next, we show that the limit remains unchanged after replacing $\psi_n\mapsto \psi_n -\psi_0$ and freezing the coefficients. This is similar to Proposition~\ref{prop:reduction.beta.main.3}.
\begin{proposition}\label{prop:reduction.beta.hard.2}
Let $b_{c,\alp}$ be as in Proposition~\ref{prop:constants}. Then
\begin{equation*}
\begin{split}
&\: \int_{\mathbb R^{2+1}} [\rd_i (\chi \psi_n)] \de^{ij} N_0 b \Big\{(\bt_n^k-\bt_0^k) \widetilde{m}(\f 1i\nabla)(\f{\rd^2_{jk}(\chi\psi_n)}{N_0}) - \widetilde{m}(\f 1i\nabla)[(\bt_n^k -\bt_0^k ) \f{\rd^2_{jk}(\chi\psi_n)}{N_0}] \Big\} \,\ud x \\
&\: -\sum_{\alp} b_{c,\alp} \int_{\mathbb R^{2+1}} [\rd_i (\zeta_\alp\chi(\psi_n-\psi_0))] \de^{ij} \Big\{(\bt_n^k -\bt_0^k) \widetilde{m}(\f 1i\nabla)(\rd^2_{jk}(\zeta_\alp \chi(\psi_n-\psi_0)))\\
&\:  \qquad \qquad \qquad \qquad \qquad \qquad \qquad  - \widetilde{m}(\f 1i\nabla)[(\bt_n^k -\bt_0^k) \rd^2_{jk}(\zeta_\alp \chi(\psi_n-\psi_0))] \Big\} \,\ud x \to 0.
\end{split}
\end{equation*}
\end{proposition}

\subsubsection{Computation of the limit}\label{sec:commute.beta.combined}

We now combine the terms in Propositions~\ref{prop:reduction.beta.main.3} and \ref{prop:reduction.beta.hard.2} and compute the limit.
\begin{proposition}
Let $b_{c,\alp}$, $N_{c,\alp}$, $\gamma_{c,\alp}$ and $\bt_{c,\alp}^i$ be as in Proposition~\ref{prop:constants}. Then
\begin{equation*}
\begin{split}
&\: \sum_{\alp} \f{b_{c,\alp}e^{2\gamma_{c,\alp}}}{N_{c,\alp}^2} \int_{\mathbb R^{2+1}} (\rd_t - \bt^{\ell}_{c,\alp}\rd_\ell)(\zeta_\alp\chi(\psi_n-\psi_0)) \Big\{  (\bt_n^k - \bt_0^k) \widetilde{m}(\f 1i\nabla)(\rd_k(\rd_t - \bt^m_{c,\alp}\rd_m)(\zeta_\alp\chi(\psi_n-\psi_0)))\\
&\: \qquad \qquad \qquad \qquad \qquad \qquad \qquad  - \widetilde{m}(\f 1i\nabla)  [ (\bt_n^k - \bt_0^k) (\rd_k(\rd_t - \bt^m_{c,\alp}\rd_m)(\zeta_\alp\chi(\psi_n-\psi_0))] \Big\}\,\ud x \\
&\: +\sum_{\alp} b_{c,\alp} \int_{\mathbb R^{2+1}} [\rd_i (\zeta_\alp\chi(\psi_n-\psi_0))] \de^{ij} \Big\{(\bt_n^k - \bt_0^k) \widetilde{m}(\f 1i\nabla)(\rd^2_{jk}(\zeta_\alp \chi(\psi_n-\psi_0)))\\
&\:  \qquad \qquad \qquad \qquad \qquad \qquad \qquad  - \widetilde{m}(\f 1i\nabla)[(\bt_n^k - \bt_0^k) \rd^2_{jk}(\zeta_\alp \chi(\psi_n-\psi_0))] \Big\} \,\ud x \to 0.
\end{split}
\end{equation*}
\end{proposition}
\begin{proof}
\pfstep{Step~1: Fourier decomposition} 
Let $\Theta$ be as in \eqref{def:Theta}. Define now the decomposition of $\bt_n^i-\bt_0^i$ as follows (compare \eqref{def:N.decomp.1}--\eqref{def:N.decomp.3}):
$${\zeta_\alp}(\bt_n^i-\bt_0^i) = ({\bm \beta}_{\mathrm{diff}})^i_{n,1} + ({\bm \beta}_{\mathrm{diff}})^i_{n,2} + ({\bm \beta}_{\mathrm{diff}})^i_{n,3},$$
where
$$\widehat{({\bm \beta}_{\mathrm{diff}})^i_{n,1}}(\xi) := \Theta(\lambda_n^{\f 56}|\xi|)\left(\widehat{{\zeta_\alp}\bt_n^i}(\xi)-\widehat{{\zeta_\alp}\bt_0^i}(\xi)\right),$$
$$\widehat{({\bm \beta}_{\mathrm{diff}})^i_{n,2}}(\xi) := (1-\Theta(\lambda_n^{\f 56}|\xi|))(1-\Theta(\f{|\xi_i|}{|\xi|^{\f 58}}))\left(\widehat{{\zeta_\alp}\bt_n^i}(\xi)-\widehat{{\zeta_\alp}\bt_0^i}(\xi)\right),$$
$$\widehat{({\bm \beta}_{\mathrm{diff}})^i_{n,3}}(\xi) := (1-\Theta(\lambda_n^{\f 56}|\xi|))\Theta(\f{|\xi_i|}{|\xi|^{\f 58}})\left(\widehat{{\zeta_\alp}\bt_n^i}(\xi)-\widehat{{\zeta_\alp}\bt_0^i}(\xi)\right).$$

We need to estimate the contributions from $({\bm \beta}_{\mathrm{diff}})^i_{n,1}$, $({\bm \beta}_{\mathrm{diff}})^i_{n,2}$ and $({\bm \beta}_{\mathrm{diff}})^i_{n,3}$. The contributions from the terms $({\bm \beta}_{\mathrm{diff}})^i_{n,1}$ and $({\bm \beta}_{\mathrm{diff}})^i_{n,2}$ can be handled as in Steps~2 and 3 in the proof of Proposition~\ref{prop:main.N.term}, where analogous terms were estimated. We note in particular that in Steps~2 and 3 in the proof of Proposition~\ref{prop:main.N.term}, the argument relies only on the frequency support of the corresponding terms and we did not use the precise structure of the nonlinearity. We therefore omit the details about bounding these terms.

On the other hand, the contribution from $({\bm \beta}_{\mathrm{diff}})^i_{n,3}$ requires a more careful treatment. (This is analogous to the term in Step~4 in the proof of Proposition~\ref{prop:main.N.term}, where we fully exploit the precise structure of the term.) We will thus focus on this term in the remainder of the proof. 

\pfstep{Step~2: Estimating the main term}
Denote $({\bm \psi}_{\mathrm{diff}})_n = \zeta_\alp {\chi} (\psi_n - \psi_0)$. We compute
\begin{equation}\label{beta.commute.frozen.1}
\begin{split}
&\: \f{b_{c,\alp}e^{2\gamma_{c,\alp}}}{N_{c,\alp}^2} \int_{\mathbb R^{2+1}} (\rd_t - \bt^\ell_{c,\alp}\rd_\ell)({\bm \psi}_{\mathrm{diff}})_n({\bm \beta}_{\mathrm{diff}})^k_{n,3} [ \widetilde{m}(\f 1 i\nabla)((\rd_t - \bt^m_{c,\alp}\rd_m)\rd_k({\bm \psi}_{\mathrm{diff}})_n)] \,\ud x \\
&\: - \f{b_{c,\alp}e^{2\gamma_{c,\alp}}}{N_{c,\alp}^2} \int_{\mathbb R^{2+1}} (\rd_t - \bt^j_{c,\alp}\rd_j)({\bm \psi}_{\mathrm{diff}})_n \widetilde{m}(\f 1 i\nabla) [({\bm \beta}_{\mathrm{diff}})^k_{n,3} ((\rd_t - \bt^m_{c,\alp}\rd_m)\rd_k({\bm \psi}_{\mathrm{diff}})_n))] \,\ud x \\
= &\: \f{-i e^{2\gamma_{c,\alp}} b_{c,\alp}}{N_{c,\alp}^2}\int_{\mathbb R^{2+1}} (\xi_t - \bt^j_{c,\alp}\xi_j)(\eta_t - \bt^m_{c,\alp}\eta_m) \eta_k \widehat{({\bm \psi}_{\mathrm{diff}})_n}(\xi) \overline{\widehat{({\bm \beta}_{\mathrm{diff}})^k_{n,3}}(\xi-\eta)} \overline{\widehat{({\bm \psi}_{\mathrm{diff}})_n}(\eta)} [\widetilde{m}(\eta) - \widetilde{m}(\xi)]\,\ud \eta\,\ud \xi \\
= &\: \f{-i e^{2\gamma_{c,\alp}} b_{c,\alp}}{N_{c,\alp}^2}\int_{\mathbb R^{2+1}} (\xi_t - \bt^j_{c,\alp}\xi_j)(\eta_t - \bt^m_{c,\alp}\eta_m) \eta_k \widehat{({\bm \psi}_{\mathrm{diff}})_n}(\xi) \widehat{({\bm \beta}_{\mathrm{diff}})^k_{n,3}}(\eta-\xi) \widehat{({\bm \psi}_{\mathrm{diff}})_n}(-\eta) [\widetilde{m}(\eta) - \widetilde{m}(\xi)]\,\ud \eta\,\ud \xi.
\end{split}
\end{equation}

Similarly,
\begin{equation}\label{beta.commute.frozen.2}
\begin{split}
&\: b_{c,\alp} \int_{\mathbb R^{2+1}} \rd_i({\bm \psi}_{\mathrm{diff}})_n \de^{ij} \{ ({\bm \beta}_{\mathrm{diff}})^k_{n,3} [ \widetilde{m}(\f 1 i\nabla)(\rd^2_{jk}({\bm \psi}_{\mathrm{diff}})_n)] - \widetilde{m}(\f 1 i\nabla) [({\bm \beta}_{\mathrm{diff}})^k_{n,3} (\rd^2_{jk} ({\bm \psi}_{\mathrm{diff}})_n))]\}\,\ud x \\
= &\: -i b_{c,\alp} \int_{\mathbb R^{2+1}} \de^{ij}\xi_i\eta_j \eta_k \widehat{({\bm \psi}_{\mathrm{diff}})_n}(\xi) \widehat{({\bm \beta}_{\mathrm{diff}})^k_{n,3}}(\eta-\xi) \widehat{({\bm \psi}_{\mathrm{diff}})_n}(-\eta) [\widetilde{m}(\eta) - \widetilde{m}(\xi)]\,\ud \eta\,\ud \xi.
\end{split}
\end{equation}
We now analyze the Fourier multiplier corresponding to $\eqref{beta.commute.frozen.1} + \eqref{beta.commute.frozen.2}$. First, we compute
\begin{equation}\label{beta.Fourier.mult.prelim}
\begin{split}
&\: \f{e^{2\gamma_{c,\alp}} }{N_{c,\alp}^2}(\xi_t - \bt^j_{c,\alp}\xi_j)(\eta_t - \bt^m_{c,\alp}\eta_m) \eta_k + \de^{ij} \xi_i\eta_j \eta_k \\
=&\: \f{e^{2\gamma_{c,\alp}} }{N_{c,\alp}^2}(\xi_t - \bt^j_{c,\alp}\xi_j)((\xi_t - \bt^j_{c,\alp}\xi_j)+(\eta_t - \bt^m_{c,\alp}\eta_m)) \eta_k + \de^{ij}\xi_i(\eta_j - \xi_j) \eta_k \\
&\: + [-\f{e^{2\gamma_{c,\alp}} }{N_{c,\alp}^2}(\xi_t - \bt^j_{c,\alp}\xi_j)^2 + \de^{ij} \xi_i \xi_j ]\eta_k
\end{split}
\end{equation}
From \eqref{beta.Fourier.mult.prelim} and \eqref{expression.for.sum.of.t.frequency} it follows that
\begin{align}
&\: \left( \f{e^{2\gamma_{c,\alp}} }{N_{c,\alp}^2}(\xi_t - \bt^j_{c,\alp}\xi_j)(\eta_t - \bt^m_{c,\alp}\eta_m) \eta_k + \de^{ij} \xi_i\eta_j \eta_k\right) \Big(1-\Theta(\lambda_n^{\f 56}|\xi-\eta|)\Big)\Theta\Big(\f{|\xi_i-\eta_i|}{|\xi-\eta|^{\f 58}}\Big) \nonumber\\
= &\: + \f{e^{2\gamma_{c,\alp}} }{N_{c,\alp}^2}(\xi_t - \bt^j_{c,\alp}\xi_j) \eta_k \f{(\eta_t - \bt^j_{c,\alp}\eta_j)^2 - e^{-2\gamma_{c,\alp}} N_{c,\alp}^2 |\eta_i|^2}{(\eta_t - \bt^j_{c,\alp}\eta_j) - (\xi_t - \bt^j_{c,\alp}\xi_j)} \Big(1-\Theta(\lambda_n^{\f 56}|\xi-\eta|)\Big)\Theta\Big(\f{|\xi_i-\eta_i|}{|\xi-\eta|^{\f 58}}\Big) \label{beta.multiplier.1}\\
&\: - \f{e^{2\gamma_{c,\alp}} }{N_{c,\alp}^2}(\xi_t - \bt^j_{c,\alp}\xi_j) \eta_k \f{(\xi_t - \bt^j_{c,\alp}\xi_j)^2 - e^{-2\gamma_{c,\alp}} N_{c,\alp}^2 |\xi_i|^2}{(\eta_t - \bt^j_{c,\alp}\eta_j) -(\xi_t - \bt^j_{c,\alp}\xi_j)} \Big(1-\Theta(\lambda_n^{\f 56}|\xi-\eta|)\Big)\Theta\Big(\f{|\xi_i-\eta_i|}{|\xi-\eta|^{\f 58}}\Big) \label{beta.multiplier.2}\\
&\: + \f{e^{2\gamma_{c,\alp}} }{N_{c,\alp}^2}(\xi_t - \bt^j_{c,\alp}\xi_j) \eta_k \f{e^{-2\gamma_{c,\alp}} N_{c,\alp}^2 \de^{k\ell}(\eta_k -\xi_k)(\eta_\ell + \xi_\ell)}{(\eta_t - \bt^j_{c,\alp}\eta_j) - (\xi_t - \bt^j_{c,\alp}\xi_j)}\Big(1-\Theta(\lambda_n^{\f 56}|\xi-\eta|)\Big)\Theta\Big(\f{|\xi_i-\eta_i|}{|\xi-\eta|^{\f 58}}\Big) \label{beta.multiplier.3} \\
&\: +\de^{ij}\xi_i(\eta_j - \xi_j) \eta_k \Big(1-\Theta(\lambda_n^{\f 56}|\xi-\eta|)\Big)\Theta\Big(\f{|\xi_i-\eta_i|}{|\xi-\eta|^{\f 58}}\Big) \label{beta.multiplier.4}\\
&\: +[-\f{e^{2\gamma_{c,\alp}} }{N_{c,\alp}^2}(\xi_t - \bt^j_{c,\alp}\xi_j)^2 + \de^{ij} \xi_i \xi_j ]\eta_k\Big(1-\Theta(\lambda_n^{\f 56}|\xi-\eta|)\Big)\Theta\Big(\f{|\xi_i-\eta_i|}{|\xi-\eta|^{\f 58}}\Big).\label{beta.multiplier.5}
\end{align}

Define now the term $\mathrm{I}$, $\mathrm{II}$, $\mathrm{III}$, $\mathrm{IV}$ and $\mathrm{V}$ respectively by inserting \eqref{beta.multiplier.1}, \eqref{beta.multiplier.2}, \eqref{beta.multiplier.3}, \eqref{beta.multiplier.4} and \eqref{beta.multiplier.5} into (*) below
\begin{equation}\label{beta.multiplier.sub}
\begin{split}
-ib_{c,\alp} \iint \left(*\right) \widehat{({\bm \psi}_{\mathrm{diff}})_n}(\xi) {(\widehat{\zeta_\alp \beta_n^k}- \widehat{\zeta_\alp \bt_0^k})}(\eta-\xi) \widehat{({\bm \psi}_{\mathrm{diff}})_n}(-\eta) [\widetilde{m}(\eta)-\widetilde{m}(\xi)]\,\ud \eta\,\ud \xi.
\end{split}
\end{equation}

We note that the term $\mathrm{I}$ and $\mathrm{II}$ here can be handled in a similar way as the terms $\mathrm{II}$ and $\mathrm{III}$ in Step~4(b) of the proof of Proposition~\ref{prop:main.N.term}. Also, the term $\mathrm{III}$ here can be handled in a similar way as the terms $\mathrm{IV}$ in Step~4(b) of the proof of Proposition~\ref{prop:main.N.term}. In particular, we have
\begin{equation}\label{beta.main.1.2.3}
|\mathrm{I}|+|\mathrm{II}|+|\mathrm{III}| = o(\lambda^{3\ep_0}).
\end{equation}

Inverting the Fourier transform, using H\"older's inequality, Lemma~\ref{lem:PSIDOs}.4, and applying the estimates in Propositions~\ref{prop:psi.localized} and \ref{prop:g.localized}, we obtain
\begin{equation}\label{beta.main.4}
\begin{split}
|\mathrm{IV}| \ls &\: \left(\| \rd_i (\zeta_\alp \chi (\psi_n-\psi_0))\|_{L^4} \|\widetilde{m}(\f 1i\nabla)\rd_k (\zeta_\alp \chi (\psi_n - \psi_0) ) \|_{L^4} \right) \\
&\: \qquad \times\|(1-\Theta(\lambda_n^{\f 56}|\nabla|))\Theta(\f{|\nabla_i|}{|\nabla|^{\f 58}})\rd_i(\zeta_\alp(\bt_n - \bt_0))\|_{L^2} \\
\ls &\: \lambda_n^{\f{3\ep_0}{4}} \cdot \lambda_n^{\f{3\ep_0}{4}}\cdot \lambda_n^{\f 12+\f{3\ep_0}{2}} = \lambda_n^{\f 12+ 3\ep_0} = o(\lambda_n^{3\ep_0}).
\end{split}
\end{equation}

Inverting the Fourier transform, using H\"older's inequality, Lemma~\ref{lem:PSIDOs}.4, and applying the estimates in Propositions~\ref{prop:psi.localized}, \ref{prop:g.localized} and \ref{prop:approximate.wave}, we obtain
\begin{equation}\label{beta.main.5}
\begin{split}
|\mathrm{V}| \ls &\: \left(\| \rd_i (\zeta_\alp \chi (\psi_n-\psi_0))\|_{L^4} \|\widetilde{m}(\f 1i\nabla)\widetilde{\Box}_{c,\alp} (\zeta_\alp \chi (\psi_n - \psi_0) ) \|_{L^4} \right.\\
&\: \left. \qquad+ \|\widetilde{\Box}_{c,\alp} (\zeta_\alp \chi (\psi_n-\psi_0))\|_{L^4} \|\widetilde{m}(\f 1i\nabla)\rd_i(\zeta_\alp \chi (\psi_n - \psi_0) ) \|_{L^4}\right) \\
&\: \qquad \qquad \times\|(1-\Theta(\lambda_n^{\f 56}|\nabla|))\Theta(\f{|\nabla_i|}{|\nabla|^{\f 58}})\zeta_\alp(\bt_n - \bt_0)\|_{L^2} \\
\ls &\: \lambda_n^{\f{3\ep_0}{4}} \cdot \lambda_n^{-1+\ep_0} \lambda_n^{\f{3\ep_0}{4}}\cdot \lambda_n \lambda_n^{\f{3\ep_0}{2}} = \lambda_n^{4\ep_0} = o(\lambda_n^{3\ep_0}).
\end{split}
\end{equation}

Noticing that the sum $\sum_\alp$ has $O(\lambda_n^{-3\ep_0})$ terms (cf.~beginning of Section~\ref{sec:freeze.coeff}), it follows from \eqref{beta.commute.frozen.1}, \eqref{beta.commute.frozen.2}, \eqref{beta.multiplier.1}, \eqref{beta.multiplier.2}, \eqref{beta.multiplier.3}, \eqref{beta.multiplier.4}, \eqref{beta.multiplier.5}, \eqref{beta.multiplier.sub}, \eqref{beta.main.1.2.3}, \eqref{beta.main.4} and \eqref{beta.main.5} that
\begin{equation*}
\begin{split}
&\: \sum_{\alp} \f{b_{c,\alp}e^{2\gamma_{c,\alp}}}{N_{c,\alp}^2} \int_{\mathbb R^{2+1}} (\rd_t - \bt^\ell_{c,\alp}\rd_\ell)({\bm \psi}_{\mathrm{diff}})_n({\bm \beta}_{\mathrm{diff}})^k_{n,3} [ \widetilde{m}(\f 1 i\nabla)((\rd_t - \bt^m_{c,\alp}\rd_m)\rd_k({\bm \psi}_{\mathrm{diff}})_n)] \,\ud x \\
&\: + \sum_{\alp} b_{c,\alp} \int_{\mathbb R^{2+1}} \rd_i({\bm \psi}_{\mathrm{diff}})_n \de^{ij} \{ ({\bm \beta}_{\mathrm{diff}})^k_{n,3} [ \widetilde{m}(\f 1 i\nabla)(\rd^2_{jk}({\bm \psi}_{\mathrm{diff}})_n)] - \widetilde{m}(\f 1 i\nabla) [({\bm \beta}_{\mathrm{diff}})^k_{n,3} (\rd^2_{jk} ({\bm \psi}_{\mathrm{diff}})_n))]\}\,\ud x \to 0.
\end{split}
\end{equation*}
Combining this with the discussions in Step~1, {this concludes} the proof. \qedhere
\end{proof}

We summarize below what we have obtained in this subsection. As in Proposition~\ref{prop:N.commute.final}, we note that the computation for the commutator applies equally well to the term involving $\om_n$.
\begin{proposition}\label{prop:beta.commute.final}
\begin{equation*}
\begin{split}
&\: \int_{\mathbb R^{2+1}} \f{(e_0)_n(\chi\psi_n)}{N_n} \Big\{ \rd_i [e^{2\gamma_n} \bt_n^i A(\f{(e_0)_n(\chi\psi_n)}{N_n})] - A \rd_i [e^{2\gamma_n} \bt_n^i (\f{(e_0)_n(\chi\psi_n)}{N_n})] \Big\}\,\ud x \\
&\: - \int_{\mathbb R^{2+1}} [\rd_i (\chi\psi_n)] \de^{ij} N_n \Big\{[(e_0)_n,A](\f{\rd_j(\chi\psi_n)}{N_n}) \Big\} \,\ud x \\
&\: - \int_{\mathbb R^{2+1}}  \f{(e_0)_0(\chi\psi_n)}{N_0} \Big\{ \rd_i [e^{2\gamma_0} \bt_0^i A(\f{(e_0)_0(\chi\psi_n)}{N_0})] - A \rd_i [e^{2\gamma_0} \bt_0^i (\f{(e_0)_0(\chi\psi_n)}{N_0})] \Big\}\,\ud x \\
&\: + \int_{\mathbb R^{2+1}} [\rd_i (\chi\psi_n)] \de^{ij} N_0 \Big\{[(e_0)_0,A](\f{\rd_j(\chi\psi_n)}{N_0}) \Big\} \,\ud x \to 0.
\end{split}
\end{equation*}
A similar statement holds after replacing $\psi_n\rightsquigarrow \om_n$, $\psi_0\rightsquigarrow \om_0$ and $\ud x\rightsquigarrow \f 14 e^{-4\psi_0} \ud x$.
\end{proposition}

\subsection{Taking limits using the microlocal defect measures}\label{sec:commutator.terms.mdm}

Let us summarize what we have obtained so far. The following is an immediate consequence of Propositions~\ref{prop:gamma.commute.final}, \ref{prop:N.commute.final} and \ref{prop:beta.commute.final}:
\begin{proposition}\label{prop:Box.commute.almost}
\begin{equation*}
\begin{split}
&\: \int_{\mathbb R^{2+1}} \f{(e_0)_n(\chi\psi_n)}{N_n} \Big( \Box_{g_n,A} (\chi\psi_n)- \f{1}{\sqrt{-\det g_n}}A(\sqrt{-\det g_n}\Box_{g_n} (\chi\psi_n)) \Big) \,\mathrm{dVol}_{g_n} \\
&\: \quad - \int_{\mathbb R^{2+1}} [\rd_i (\chi\psi_n)] \de^{ij} N_n \Big\{ [(e_0)_n,A](\f{\rd_j(\chi\psi_n)}{N_n}) \Big\} \,\ud x \\
&\: - \int_{\mathbb R^{2+1}} \f{(e_0)_0 (\chi\psi_n)}{N_0} \Big( \Box_{g_0,A} (\chi\psi_n) - \f{1}{\sqrt{-\det g_0}} A (\sqrt{-\det g_0}\Box_{g_0} (\chi\psi_n)) \Big) \,\mathrm{dVol}_{g_0} \\
&\: \quad + \int_{\mathbb R^{2+1}} [\rd_i (\chi\psi_n)] \de^{ij} N_0 \Big\{[(e_0)_0,A](\f{\rd_j(\chi\psi_n)}{N_0}) \Big\} \,\ud x \to 0.
\end{split}
\end{equation*}
A similar statement holds after replacing $\psi \rightsquigarrow \om$ and $\ud x\rightsquigarrow \f 14 e^{-4\psi_0} \ud x$.
\end{proposition}

In other words, we have reduced the computation of the limit of the first two lines to that of the limit of the last two lines. To proceed, we use Corollary~\ref{cor:nu} to compute the limit of the last two lines on the LHS as $n\to +\infty$. This will be achieved in the next two propositions.

\begin{proposition}\label{prop:Box.commute.almost.1}
\begin{equation*}
\begin{split}
&\: \int_{\mathbb R^{2+1}} [\rd_i (\chi\psi_n)] \de^{ij} N_0 \Big\{[(e_0)_0,A](\f{\rd_j(\chi\psi_n)}{N_0}) \Big\} \,\ud x \\
\to &\: \int_{\mathbb R^{2+1}} [\rd_i (\chi\psi_0)] \de^{ij} N_0 \Big\{[(e_0)_0,A](\f{\rd_j(\chi\psi_0)}{N_0}) \Big\} \,\ud x \\
&\: + \int_{S^*\mathbb R^{2+1}} \Big[\de^{ij} \xi_i\xi_j (\rd_{x^t}a - \bt^k_0 \rd_{x^k}a) + \de^{ij} \xi_i\xi_j (\rd_\mu \bt^k_0)\xi_k \rd_{\xi_\mu} a \Big]\, \f{e^{-2\gamma_0}}{N_0}\f{d\nu^\psi}{|\xi|^2}.
\end{split}
\end{equation*}
A similar statement holds after changing $\psi \rightsquigarrow \om$ and $\ud\nu^\psi\rightsquigarrow e^{-4\psi_0} \,\ud \nu^\om$.
\end{proposition}
\begin{proof}
By Lemma~\ref{lem:PSIDOs}.2,  $[(e_0)_0,A]$ is a $0$-th order pseudo-differential symbol with principal symbol
$$-i \{ i(\xi_t - \bt^k_0 \xi_k), a\} = \rd_{x^t} a - \bt^k_0 \rd_{x^k} a + (\rd_\mu \bt^k_0)\xi_k \rd_{\xi_\mu} a.$$

The conclusion therefore follows from Corollary~\ref{cor:nu} {and that $\sqrt{-\det g_0} = e^{2\gamma} N$ (by \eqref{g.form})}. \qedhere
\end{proof}

\begin{proposition}\label{prop:Box.commute.almost.2}
\begin{equation*}
\begin{split}
&\: \int_{\mathbb R^{2+1}} \f{(e_0)_0 (\chi\psi_n)}{N_0} \Big(\Box_{g_0,A} (\chi\psi_n) - \f{1}{\sqrt{-\det g_0}} A (\sqrt{-\det g_0}\Box_{g_0} (\chi\psi_n)) \Big) \,\mathrm{dVol}_{g_0} \\
\to &\: \int_{\mathbb R^{2+1}} \f{(e_0)_0 (\chi\psi_0)}{N_0} \Big(\Box_{g_0,A} (\chi\psi_0) - \f{1}{\sqrt{-\det g_0}} A (\sqrt{-\det g_0}\Box_{g_0} (\chi\psi_0))\Big) \,\mathrm{dVol}_{g_0}\\
&\: + \int_{S^*\mathbb R^{2+1}} \f 1{N_0} (\xi_t-\bt^k_0\xi_k) \Big[ (g_0^{-1})^{\mu\nu} \xi_\mu (\rd_{x^\nu} a) - \rd_\mu (g_0^{-1})^{\alp\bt} \xi_\alp \xi_\bt (\rd_{\xi_\mu} a) \Big] \, \f{\ud \nu^\psi}{|\xi|^2}\\
&\:  + \int_{S^*\mathbb R^{2+1}} \f{e^{-2\gamma_0} (\rd_\mu \bt^k_0) {(\rd_{\xi_\mu} a)}}{N_0} \de^{ij} \xi_i \xi_j \xi_k \, \f{\ud \nu^\psi}{|\xi|^2}.
\end{split}
\end{equation*}
A similar statement holds after changing $\psi \rightsquigarrow \om$ and $\ud\nu^\psi\rightsquigarrow e^{-4\psi_0} \,\ud \nu\om$.
\end{proposition}
\begin{proof}
\pfstep{Step~1: Computing the limit using Corollary~\ref{cor:nu}}
We compare each of the terms in $\Box_{g_0,A}$ and $\Box_{g_0}$ (cf.~definitions in Section~\ref{sec:def.boxes}). 

By Lemma~\ref{lem:PSIDOs}.2, $[-\rd_t e^{2\gamma_0}, A]$ {(to be understood as $[-\rd_t e^{2\gamma_0}, A] h := -\rd_t (e^{2\gamma_0} A h) + A \rd_t (e^{2\gamma_0} h)$, similarly below)} is a $0$-th order pseudo-differential symbol with principal symbol
$$-i\{ -i\xi_t e^{2\gamma_0}, a \} = - e^{2\gamma_0} \rd_{x^t} a + \xi_t (\rd_{x^\mu} e^{2\gamma_0})(\rd_{\xi_\mu}a).$$
It follows from Corollary~\ref{cor:nu} {(and that $\sqrt{-\det g_0} = e^{2\gamma} N$)} that
\begin{equation*}
\begin{split}
&\: \int_{\mathbb R^{2+1}} \f{(e_0)_0 (\chi\psi_n)}{N_0} \Big\{ - \rd_t [e^{2\gamma_0} A(\f{(e_0)_0(\chi\psi_n)}{N_0})] + A\rd_t [e^{2\gamma_0} (\f{(e_0)_0(\chi\psi_n)}{N_0})] \Big\} \, \ud x \\
\to &\: \int_{\mathbb R^{2+1}} \f{(e_0)_0 (\chi\psi_0)}{N_0} \Big\{ - \rd_t [e^{2\gamma_0} A(\f{(e_0)_0(\chi\psi_0)}{N_0})] + A\rd_t [e^{2\gamma_0} (\f{(e_0)_0(\chi\psi_0)}{N_0})] \Big\} \, \ud x \\
&\: + \int_{S^*\mathbb R^{2+1}} \Big[ \underbrace{- \f 1{N_0^3}(\xi_t - \bt^k_0\xi_k)^2(\rd_{x^t} a)}_{=:\mathrm{I}} + \underbrace{\f {e^{-2\gamma_0}}{N_0^3}\xi_t (\xi_t - \bt^k_0\xi_k)^2(\rd_{x^\mu} e^{2\gamma_0})(\rd_{\xi_\mu}a)}_{=:\mathrm{II}} \Big]\, \f{\ud\nu^\psi}{|\xi|^2}.
\end{split}
\end{equation*}

By Lemma~\ref{lem:PSIDOs}.2, $[\rd_i {N_0^2}, A]$ is a $0$-th order pseudo-differential symbol with principal symbol
$$-i\{ i\xi_i N_0^2, a\} = N_0^2 \rd_{x^i} a - \xi_i (\rd_{x^\mu} N_0^2)(\rd_{\xi_\mu} a).$$
It follows from Corollary~\ref{cor:nu} that
\begin{equation*}
\begin{split}
&\: \int_{\mathbb R^{2+1}} \f{(e_0)_0 (\chi\psi_n)}{N_0} \de^{ij} \Big\{ \rd_i [N_0^2 A(\f{\rd_j(\chi\psi_n)}{N_0})] - A\rd_i [N_0 \rd_j(\chi\psi_n)] \Big\} \, \ud x \\
\to &\: \int_{\mathbb R^{2+1}} \f{(e_0)_0 (\chi\psi_0)}{N_0} \de^{ij} \Big\{ \rd_i [N_0^2 A(\f{\rd_j(\chi\psi_0)}{N_0})] - A\rd_i [N_0 \rd_j(\chi\psi_0)] \Big\} \, \ud x \\
&\: + \int_{S^*\mathbb R^{2+1}} \Big[\underbrace{\f{e^{-2\gamma_0}}{N_0}\de^{ij}(\xi_t-\bt^k_0\xi_k) \xi_j(\rd_{x^i} a)}_{=:\mathrm{III}} \underbrace{- \f{e^{-2\gamma_0}}{N_0^3} \de^{ij} (\xi_t-\bt^k_0\xi_k) \xi_i \xi_j (\rd_{x^\mu} N_0^2)(\rd_{\xi_\mu} a)}_{=:\mathrm{IV}} \Big] \, \f{\ud\nu^\psi}{|\xi|^2}.
\end{split}
\end{equation*}

Finally, by Lemma~\ref{lem:PSIDOs}.2, $[\rd_i e^{2\gamma_0}\bt_0^i, A]$ is a $0$-th order pseudo-differential symbol with principal symbol
$$-i \{i\xi_i e^{2\gamma_0}\bt_0^i, a \} = e^{2\gamma_0}\bt_0^i (\rd_{x^i} a) - \xi_i(\rd_{x^\mu} (e^{2\gamma_0}\bt_0^i) )(\rd_{\xi_\mu} a) = e^{2\gamma_0}\bt_0^i (\rd_{x^i} a) - \xi_i((\rd_{x^\mu} e^{2\gamma_0})\bt_0^i + e^{2\gamma_0} (\rd_{x^\mu}\bt_0^i))(\rd_{\xi_\mu} a).$$
Therefore, by Corollary~\ref{cor:nu},
\begin{equation*}
\begin{split}
&\: \int_{\mathbb R^{2+1}} \f{(e_0)_0 (\chi\psi_n)}{N_0}  \{ \rd_i [e^{2\gamma_0}\bt_0^i A(\f{(e_0)_0(\chi\psi_n)}{N_0})] - A\rd_i [e^{2\gamma_0}\bt_0^i (e_0)_0(\chi\psi_n)] \, \ud x \\
\to &\: \int_{\mathbb R^{2+1}} \f{(e_0)_0 (\chi\psi_0)}{N_0}  \{ \rd_i [e^{2\gamma_0}\bt_0^i A(\f{(e_0)_0(\chi\psi_0)}{N_0})] - A\rd_i [e^{2\gamma_0}\bt_0^i (e_0)_0(\chi\psi_0)] \, \ud x \\
&\: + \int_{S^*\mathbb R^{2+1}} [\underbrace{\f{\bt_0^i}{N_0^3}(\xi_t-\bt^k_0\xi_k)^2  (\rd_{x^i} a)}_{=:\mathrm{V}} \underbrace{- \f{e^{-2\gamma_0}}{N_0^3} (\xi_t-\bt^k_0\xi_k)^2 \xi_i  ((\rd_{x^\mu} e^{2\gamma_0})\bt_0^i + e^{2\gamma_0} (\rd_{x^\mu}\bt_0^i))(\rd_{\xi_\mu} a)}_{=:\mathrm{VI}} ]\, \f{\ud\nu^\psi}{|\xi|^2}.
\end{split}
\end{equation*}

\pfstep{Step~2: Computing $(g_0^{-1})^{\mu\nu} \xi_\mu (\rd_{x^\nu} a)$}
\begin{equation}\label{eq:spatial.transport.a}
\begin{split}
 (g_0^{-1})^{\mu\nu} \xi_\mu (\rd_{x^\nu} a) 
 = &\: - \f 1{N_0^2} (\xi_t - \bt^k_0 \xi_k) (\rd_{x^t} a - \bt^i_0 \rd_{x^i} a) + e^{-2\gamma_0} \de^{ij} \xi_i \rd_{x^j} a.
\end{split}
\end{equation}
Therefore,
\begin{equation*}
\begin{split}
&\: \f 1{N_0} (\xi_t-\bt^k_0\xi_k) (g_0^{-1})^{\mu\nu} \xi_\mu (\rd_{x_\nu}\ a) \\
= &\: - \f 1{N_0^3} (\xi_t - \bt^k_0 \xi_k)^2 (\rd_{x^t} a - \bt^i_0 \rd_{x^i} a) + \f {e^{-2\gamma_0}}{N_0}   (\xi_t - \bt^k_0 \xi_k) \de^{ij} \xi_i (\rd_{x^j} a) = \mathrm{I} + \mathrm{III} + \mathrm{V}.
\end{split}
\end{equation*}

\pfstep{Step~3: Computing $\rd_\mu (g_0^{-1})^{\alp\bt} \xi_\alp \xi_\bt (\rd_{\xi_\mu} a)$}
\begin{equation}\label{eq:transport.second.term}
\begin{split}
&\: \rd_\mu (g_0^{-1})^{\alp\bt} \xi_\alp \xi_\bt (\rd_{\xi_\mu} a) \\
= &\: \left(- \rd_\mu (\f 1{N_0^2}) \xi_t \xi_t  + 2 \rd_\mu (\f {\bt^i_0}{N_0^2}) \xi_t \xi_i  + \rd_{\mu} (e^{-2\gamma_0} \de^{ij} - \f{\bt^i_0\bt^j_0}{N_0^2}) \xi_i \xi_j \right)(\rd_{\xi_\mu} a) \\
= &\: \left(\f{\rd_\mu N_0^2}{N_0^4} (\xi_t - \bt^k_0 \xi_k)^2 + 2 \f{(\rd_\mu \bt^i_0)}{N_0^2} (\xi_t - \bt_0^k \xi_k) \xi_i + (\rd_\mu e^{-2\gamma_0}) \de^{ij} \xi_i \xi_j \right) (\rd_{\xi_\mu} a) \\
= &\: \left(\f{\rd_\mu N_0^2}{N_0^4} (\xi_t - \bt^k_0 \xi_k)^2 + 2 \f{(\rd_\mu \bt^i_0)}{N_0^2} (\xi_t - \bt_0^k \xi_k) \xi_i - e^{-4\gamma_0} (\rd_\mu e^{2\gamma_0}) \de^{ij} \xi_i \xi_j \right) (\rd_{\xi_\mu} a).
\end{split}
\end{equation}

Recall that \underline{on the support of $\ud \nu^\psi$}, \eqref{eq:transport.second.term.null} holds.
Hence, by \eqref{eq:transport.second.term} and \eqref{eq:transport.second.term.null}, \underline{on the support of $\ud \nu^\psi$},
\begin{equation*}
\begin{split}
&\: \f 1{N_0} (\xi_t-\bt^k_0\xi_k) \rd_\mu (g_0^{-1})^{\alp\bt} \xi_\alp \xi_\bt (\rd_{\xi_\mu} a) \\
= &\: \left(\f{\rd_\mu N_0^2}{N_0^5} (\xi_t - \bt^k_0 \xi_k)^3 + 2 \f{(\rd_\mu \bt^i_0)}{N_0^3} (\xi_t - \bt_0^k \xi_k)^2 \xi_i - \f{e^{-4\gamma_0} (\rd_\mu e^{2\gamma_0})}{N_0} (\xi_t - \bt_0^k \xi_k) \de^{ij} \xi_i \xi_j \right) (\rd_{\xi_\mu} a) \\
= &\: \left(\f{e^{-2\gamma_0} (\rd_\mu N_0^2)}{N_0^3} (\xi_t - \bt^k_0 \xi_k) \de^{ij} \xi_i \xi_j + 2 \f{(\rd_\mu \bt^i_0)}{N_0^3} (\xi_t - \bt_0^k \xi_k)^2 \xi_i - \f{e^{-2\gamma_0} (\rd_\mu e^{2\gamma_0})}{N_0^3} (\xi_t - \bt_0^k \xi_k)^3 \right) (\rd_{\xi_\mu} a) \\
= &\: -\mathrm{II} - \mathrm{IV} - \mathrm{VI} + \f{(\rd_\mu \bt^i_0)}{N_0^3} (\xi_t - \bt_0^k \xi_k)^2 \xi_i {(\rd_{\xi_\mu} a)} =  -\mathrm{II} - \mathrm{IV} - \mathrm{VI} + \f{e^{-2\gamma_0} (\rd_\mu \bt^k_0)}{N_0} \de^{ij} \xi_i \xi_j \xi_k {(\rd_{\xi_\mu} a)}.
\end{split}
\end{equation*}

Combining Steps~1,2 and 3 yields the conclusion. \qedhere
\end{proof}

\subsection{Putting everything together}\label{sec:commutator.terms.final}

We summarize what we have obtained so far. Combining Propositions~\ref{prop:Box.commute.almost}, \ref{prop:Box.commute.almost.1} and \ref{prop:Box.commute.almost.2}, {and noting that the two terms of $\int_{S^*\mathbb R^{2+1}} \f{e^{-2\gamma_0} (\rd_\mu \bt^k_0)}{N_0} \de^{ij} \xi_i \xi_j \xi_k \, \f{\ud \nu^\psi}{|\xi|^2}$ cancel,} we immediately obtain

\begin{proposition}\label{prop:Box.commute}
Suppose $A = b(x) \widetilde{m}(\f 1i\nabla)$, {where the principal symbol $a(x,\xi) = b(x) m(\xi)$ (with $m(\xi) = \widetilde{m}(\xi)$ for $|\xi|\geq 1$) is real and supported in $T^*\Omega$, $m(\xi)$ is homogeneous of order $0$, and $m$ and $\widetilde{m}$ are both even.} Then
\begin{equation*}
\begin{split}
&\: \int_{\mathbb R^{2+1}} \f{(e_0)_n(\chi\psi_n)}{N_n} (\Box_{g_n,A} (\chi\psi_n)- \f{1}{\sqrt{-\det g_n}}A(\sqrt{-\det g_n}\Box_{g_n} (\chi\psi_n))) \,\mathrm{dVol}_{g_n} \\
&\: \quad - \int_{\mathbb R^{2+1}} [\rd_i (\chi\psi_n)] \de^{ij} N_n \{[(e_0)_n,A](\f{\rd_j(\chi\psi_n)}{N_n})\} \,\ud x \\
\to &\: + \int_{\mathbb R^{2+1}} \f{(e_0)_0 (\chi\psi_0)}{N_0} (\Box_{g_0,A} (\chi\psi_0) - \f{1}{\sqrt{-\det g_0}} A (\sqrt{-\det g_0}\Box_{g_0} (\chi\psi_0))) \,\mathrm{dVol}_{g_0}\\
&\: -\int_{\mathbb R^{2+1}} [\rd_i (\chi\psi_0)] \de^{ij} N_0 \{[(e_0)_0,A](\f{\rd_j(\chi\psi_0)}{N_0})\} \,\ud x \\
&\: - \int_{S^*\mathbb R^{2+1}} \de^{ij} \xi_i\xi_j (\rd_{x^t}a - \bt^k_0 \rd_{x^k}a) \, \f{e^{-2\gamma_0}}{N_0}\f{d\nu^\psi}{|\xi|^2} \\
&\: + \int_{S^*\mathbb R^{2+1}} \f 1{N_0} (\xi_t-\bt^k_0\xi_k) [(g_0^{-1})^{\mu\nu} \xi_\mu (\rd_{x^\nu} a) - \rd_\mu (g_0^{-1})^{\alp\bt} \xi_\alp \xi_\bt (\rd_{\xi_\mu} a)] \, \f{\ud \nu^\psi}{|\xi|^2}.
\end{split}
\end{equation*}
A similar statement holds after replacing $\psi \rightsquigarrow \om$ and $\ud x\rightsquigarrow \f 14 e^{-4\psi_0} \ud x$.
\end{proposition}

\section{The wave equation terms in Proposition~\ref{prop:energy.id.n.2} and trilinear compensated compactness for three waves}\label{sec:cc}

We continue to work under the assumptions of Theorem~\ref{thm:main} and the reductions in Sections~\ref{sec:compact.reduction} and \ref{sec:reduction}. As above, {let $A$ be a $0$-th order pseudo-differential operator given by $A = b(x) \widetilde{m}(\f 1i\nabla)$, where the principal symbol $a(x,\xi) = b(x) m(\xi)$ (with $m(\xi) = \widetilde{m}(\xi)$ for $|\xi|\geq 1$) is real and supported in $T^*\Omega$, $m(\xi)$ is homogeneous of order $0$, and $m$ and $\widetilde{m}$ are both even.}

In this section, we handle the terms $\mathrm{trilinear}_1$ and $\mathrm{trilinear}_2$ in \eqref{eq:energy.id.n.2} (and the analogous terms in \eqref{eq:energy.id.n.4}. There are two types of terms coming from two types of contribution from $F^\psi_n$ and $F^\om_n$. First, there are terms which are linear in the wave variables $\psi_n$ and $\om_n$ --- these terms are easier and will be handled in \textbf{Section~\ref{sec:limit.chi.terms}}. The remaining terms are nonlinear and will be treated in \textbf{Section~\ref{sec:limit.nonlinear.terms}}. In order to deal with the nonlinear terms, we will need a trilinear compensated compactness result for three waves, which will be established in \textbf{Section~\ref{subsec:cc}}.

\subsection{The linear terms in the wave equation}\label{sec:limit.chi.terms}

\begin{proposition}\label{prop:limit.chi.terms}
The following holds after passing to a subsequence (which we do not relabel):

\begin{equation*}
\begin{split}
&\: -\int_{\mathbb R^{2+1}} \f{(e_0)_n(\chi\psi_n)}{N_n}\{ A[\sqrt{-\det g_n}(2 g_n^{-1}(\ud \chi, \ud \psi_n) + \psi_n \Box_{g_n} \chi) ]\} \,{\ud x}\\
&\: -\int_{\mathbb R^{2+1}} A(\f{(e_0)_n(\chi\psi_n)}{N_n}) (2 g_n^{-1}(\ud \chi, \ud \psi_n) + \psi_n \Box_{g_n} \chi) \,\mathrm{dVol}_{g_n}\\
&\: -\f 14 \int_{\mathbb R^{2+1}} e^{-4\psi_0} \f{(e_0)_n(\chi\om_n)}{N_n} \{ A[\sqrt{-\det g_n}(2 (g_n^{-1})^{\alp\bt} \rd_\alp \chi\rd_\bt \om_n + \om_n \Box_{g_n} \chi) ]\} \,{\ud x}\\
&\: -\f 14 \int_{\mathbb R^{2+1}} e^{-4\psi_0} A(\f{(e_0)_n(\chi\om_n)}{N_n}) (2 g_n^{-1} (\ud \chi, \ud \om_n) + \om_n \Box_{g_n} \chi) \,\mathrm{dVol}_{g_n}\\
\to &\:  -\int_{\mathbb R^{2+1}} \f{(e_0)_0(\chi\psi_0)}{N_0}\{ A[\sqrt{-\det g_0} (2 g_0^{-1} (\ud \chi,\ud \psi_0) + \psi_0 \Box_{g_0} \chi)]\} \,{\ud x} \\
&\: -\int_{\mathbb R^{2+1}} A(\f{(e_0)_0(\chi\psi_0)}{N_0}) (2 g_0^{-1}(\ud \chi, \ud \psi_0) + \psi_0 \Box_{g_0} \chi) \,\mathrm{dVol}_{g_0} \\
&\: -\f 14 \int_{\mathbb R^{2+1}} e^{-4\psi_0} \f{(e_0)_0(\chi\om_0)}{N_0} \{ A[\sqrt{-\det g_0}(2 g_0^{-1} (\ud \chi, \ud \om_0) + \om_0 \Box_{g_0} \chi)] \} \,{\ud x} \\
&\: -\f 14 \int_{\mathbb R^{2+1}} e^{-4\psi_0} A(\f{(e_0)_0(\chi\om_0)}{N_0}) (2 g_0^{-1}(\ud \chi, \ud \om_0) + \om_0 \Box_{g_0} \chi) \,\mathrm{dVol}_{g_0}.
\end{split}
\end{equation*}

\end{proposition}
\begin{proof}
We will only indicate how to obtain the limit of the terms on the first line; all the other terms can be treated similarly.

\pfstep{Step~1: First term on first line} Since $a$ (the symbol of $A$) and $\rd_\bt \chi$ have disjoint support, by Lemma~\ref{lem:PSIDOs}, $A \rd_\bt \chi: L^2\to L^2_{\mathrm{loc}}$ is compact. Therefore, using \eqref{assumption.0} and \eqref{assumption.1}, we see that after passing to a subsequence (not {relabeled})
\begin{equation*}
\| A[ g_n^{-1}(\ud \chi, \ud \psi_n) \sqrt{-\det g_n}] - A[ g_0^{-1} (\ud \chi, \ud \psi_0) \sqrt{-\det g_0}]\|_{L^2(\Omega')} \to 0.
\end{equation*}
On the other hand, by \eqref{assumption.0} and \eqref{assumption.1}, we know that $\f{(e_0)_n(\chi\psi_n)}{N_n} \rightharpoonup \f{(e_0)_0(\chi\psi_0)}{N_0}$ weakly in $L^2$. Therefore,
\begin{equation*}
\begin{split}
&\: -\int_{\mathbb R^{2+1}} \f{(e_0)_n(\chi\psi_n)}{N_n} A[\sqrt{-\det g_n} g_n^{-1} (\ud \chi, \ud \psi_n)]  \,{\ud x} \\
\to &\: -\int_{\mathbb R^{2+1}} \f{(e_0)_0(\chi\psi_0)}{N_0} A[\sqrt{-\det g_0} g_0^{-1} (\ud \chi, \ud \psi_0)]  \,{\ud x}.
\end{split}
\end{equation*}

\pfstep{Step~2: Second term on first line} By \eqref{assumption.0} and \eqref{assumption.1}, $\sqrt{-\det g_n} \psi_n \Box_{g_n} \chi \to \sqrt{-\det g_0} \psi_0 \Box_{g_0} \chi$ in the $L^2$ norm. The fact that $A$ is a $0${-}th order operator then implies that $A[\sqrt{-\det g_n} \psi_n \Box_{g_n} \chi ] \to A[\sqrt{-\det g_0} \psi_0 \Box_{g_0} \chi ]$ in the $L^2$ norm. Using also that $\f{(e_0)_n(\chi\psi_n)}{N_n} \rightharpoonup \f{(e_0)_0(\chi\psi_0)}{N_0}$ weakly in $L^2$ (by \eqref{assumption.0} and \eqref{assumption.1}), it thus follows that
\begin{equation*}
\begin{split}
-\int_{\mathbb R^{2+1}} \f{(e_0)_n(\chi\psi_n)}{N_n}\{ A[\sqrt{-\det g_n} \psi_n \Box_{g_n} \chi ]\} \,{\ud x} \to &\:  -\int_{\mathbb R^{2+1}} \f{(e_0)_0(\chi\psi_0)}{N_0}\{ A[\sqrt{-\det g_0} \psi_0 \Box_{g_0} \chi ]\} \,{\ud x}.
\end{split}
\end{equation*}

The other terms can be treated similarly; we omit the details. \qedhere
\end{proof}

\subsection{A general trilinear compensated compactness result}\label{subsec:cc}
\begin{proposition}\label{prop:trilinear.wave.0}
Let $\{\phi^{(1)}_n\}_{n=1}^{+\infty}$, $\{\phi^{(2)}_n\}_{n=1}^{+\infty}$ and $\{\phi^{(3)}_n\}_{n=1}^{+\infty}$ be three sequences of smooth functions with $\phi^{(i)}_n:\mathbb R^{2+1}\to \mathbb R$. Assume that for any (spacetime) compact set $K \subset \mathbb R^{2+1}$,	
\begin{enumerate}
\item $\max_i \sup_n (\|\rd \phi^{(i)}_n\|_{L^3(K)} + \|\Box_{g_0} \phi^{(i)}_n\|_{L^3(K)}) <+\infty,$
\item $\max_i \|\phi^{(i)}_n\|_{L^3(K)} \to 0$ as $n\to +\infty$.
\end{enumerate}

Then for any smooth vector field $X$, 
$$(X\phi^{(1)}_n) g_0^{-1}(\ud\phi^{(2)}_n, \ud \phi^{(3)}_n) \rightharpoonup 0$$ 
in the sense of distributions.
\end{proposition}
\begin{proof}
We write
$$g_0^{-1}(\ud\phi^{(2)}_n, \ud \phi^{(3)}_n) = \underbrace{\f 12 \Box_{g_0} (\phi^{(2)}_n \phi^{(3)}_n)}_{=:\mathrm{I}} - \underbrace{\f 12 (\Box_{g_0} \phi^{(2)}_n) \phi^{(3)}_n}_{=:\mathrm{II}} - \underbrace{\f 12 \phi^{(2)}_n (\Box_{g_0}\phi^{(3)}_n)}_{=:\mathrm{III}}.$$

By assumptions of the proposition and H\"older's inequality, $\mathrm{II}$ and $\mathrm{III}$ both converge to $0$ in the $L^{\f 32}$ norm as $n\to +\infty$. Together with the assumed uniform $L^3$-boundedness of ${X\phi^{(1)}_n}$, H\"older's inequality implies that $X\phi^{(1)}_n (\mathrm{II} + \mathrm{III})$ in fact converges to $0$ in the $L^1$ norm on any compact set.

It therefore remains to check the contribution from the term $\mathrm{I}$. Let $\vartheta \in C^\infty_c(\mathbb R^{2+1})$ be a smooth function with support $K$. We then compute
\begin{equation*}
\begin{split}
&\: \int_{\mathbb R^{2+1}} \vartheta (X\phi^{(1)}_n) \Box_{g_0} (\phi^{(2)}_n \phi^{(3)}_n) \,\mathrm{dVol}_{g_0}\\
= &\: \int_{\mathbb R^{2+1}} [(\Box_{g_0}\vartheta) (X\phi^{(1)}_n) + \vartheta (X\Box_{g_0}\phi^{(1)}_n) + \vartheta ([\Box_{g_0},X]\phi^{(1)}_n) + 2 (g^{-1}_0)^{\alp\bt}(\rd_{\alp}\vartheta) (\rd_{\bt} X\phi^{(1)}_n)] \phi^{(2)}_n \phi^{(3)}_n \,\mathrm{dVol}_{g_0}\\
=: &\: \mathrm{I_a} + \mathrm{I_b} + \mathrm{I_c} + \mathrm{I_d}.
\end{split}
\end{equation*}

To control $\mathrm{I_a}$ we note that $\vartheta$ is smooth and thus $\Box_{g_0}\vartheta$ is pointwise bounded on $K$. Thus using H\"older's inequality and the bounds in the assumptions of the proposition, we obtain
\begin{equation*}
|\mathrm{I_a}| \ls \|X\phi^{(1)}_n\|_{L^3(K)} \|\phi^{(2)}_n\|_{L^3(K)} \|\phi^{(3)}_n\|_{L^3(K)} \to 0.
\end{equation*} 
For $\mathrm{I_b}$, we integrate by parts to obtain
\begin{equation*}
\begin{split}
\mathrm{I_b} = &\: \int_{K} [-(X\vartheta) (\Box_{g_0}\phi^{(1)}_n) \phi^{(2)}_n \phi^{(3)}_n - \vartheta (\Box_{g_0}\phi^{(1)}_n) (X\phi^{(2)}_n) \phi^{(3)}_n]\\
&\: -\int_{K}  [\vartheta (\Box_{g_0}\phi^{(1)}_n) \phi^{(2)}_n (X\phi^{(3)}_n) + \vartheta (\mathrm{div}_{g_0} X)(\Box_{g_0}\phi^{(1)}_n) \phi^{(2)}_n \phi^{(3)}_n] \,\mathrm{dVol}_{g_0}.
\end{split}
\end{equation*}
Since $g_0$, $X$ and $\vartheta$ are smooth, by H\"older's inequality and the bounds in the assumptions of the proposition, we obtain
\begin{equation*}
|\mathrm{I_b}|\ls \|\Box_{g_0}\phi^{(1)}_n\|_{L^3(K)} (\|\phi^{(2)}_n\|_{L^3(K)} \|\phi^{(3)}_n\|_{L^3(K)}+ \|X\phi^{(2)}_n\|_{L^3(K)} \|\phi^{(3)}_n\|_{L^3(K)}+ \|\phi^{(2)}_n\|_{L^3(K)} \|X\phi^{(3)}_n\|_{L^3(K)}) \to 0.
\end{equation*} 
Note that $[\Box_{g_0}, X]$ is a smooth second order differential operator that can be written as a finite sum $\sum_j Y_jZ_j$ for some smooth vector fields $Y_j$ and $Z_j$. Therefore we can treat $\mathrm{I_c}$ and $\mathrm{I_d}$ simultaneously by bounding a term of the form
$$\int_{K} \varsigma (YZ\phi^{(1)}_n) \phi^{(2)}_n \phi^{(3)}_n \,\mathrm{dVol}_{g_0}$$ 
for some smooth function $\varsigma$ and smooth vector fields $Y$ and $Z$. We integrate by parts and then use H\"older's inequality and the smoothness of $\varsigma$ and $Y$ to show that 
\begin{equation*}
\begin{split}
&\: \left| \int_{K} \varsigma (YZ\phi^{(1)}_n) \phi^{(2)}_n \phi^{(3)}_n \,\mathrm{dVol}_{g_0} \right| \\
\ls &\: \|Z \phi^{(1)}_n\|_{L^3(K)} (\|\phi^{(2)}_n\|_{L^3(K)} \|\phi^{(3)}_n\|_{L^3(K)}+ \|Y\phi^{(2)}_n\|_{L^3(K)} \|\phi^{(3)}_n\|_{L^3(K)}+ \|\phi^{(2)}_n\|_{L^3(K)} \|Y\phi^{(3)}_n\|_{L^3}) \to 0.
\end{split}
\end{equation*}
This shows that $\mathrm{I_c},\,\mathrm{I_d}\to 0$ and finishes the proof. \qedhere
\end{proof}

We next compute the limits in a similar setting but instead with $\phi^{(i)}_n$ converging to a potentially non-zero $\phi^{(i)}_0$.
\begin{proposition}\label{prop:trilinear}
Let $\{\phi^{(1)}_n\}_{n=1}^{+\infty}$, $\{\phi^{(2)}_n\}_{n=1}^{+\infty}$ and $\{\phi^{(3)}_n\}_{n=1}^{+\infty}$ be three sequences of smooth functions. Assume that there exist smooth $\phi^{(i)}_0:\mathbb R^{2+1}\to \mathbb R$ so that for every compact subset $K\subset \mathbb R^{2+1}$,
\begin{enumerate}
\item $\max_i \sup_n (\|\rd (\phi^{(i)}_n-\phi^{(i)}_0)\|_{L^3(K)} + \|\Box_{g_0} (\phi^{(i)}_n-\phi^{(i)}_0)\|_{L^3(K)}) <+\infty,$
\item $\max_i \|\phi^{(i)}_n-\phi^{(i)}_0\|_{L^3(K)} \to 0$ as $n\to +\infty$.
\end{enumerate}

Let $\vartheta \in C^\infty_c(\mathbb R^{2+1})$. Then
\begin{equation*}
\begin{split}
&\: \int_{\mathbb R^{2+1}} \vartheta (X\phi^{(1)}_n) g_0^{-1}(\ud \phi^{(2)}_n, \ud \phi^{(3)}_n) \,\mathrm{dVol}_{g_0} - \int_{\mathbb R^{2+1}} \vartheta (X\phi^{(1)}_n) g_0^{-1}(\ud \phi^{(2)}_0, \ud \phi^{(3)}_n) \,\mathrm{dVol}_{g_0}\\
&\: - \int_{\mathbb R^{2+1}} \vartheta (X\phi^{(1)}_n) g_0^{-1}(\ud \phi^{(2)}_n,\ud \phi^{(3)}_0) \,\mathrm{dVol}_{g_0}\\
\to  &\: -\int_{\mathbb R^{2+1}} \vartheta (X\phi^{(1)}_0) g_0^{-1}(\ud \phi^{(2)}_0,\ud \phi^{(3)}_0) \,\mathrm{dVol}_{g_0}.
\end{split}
\end{equation*}
\end{proposition}
\begin{proof}
Using Proposition~\ref{prop:trilinear.wave.0} (with $\phi^{(i)}_n - \phi^{(i)}_0$ in place of $\phi^{(i)}_n$) and then expanding the terms,
\begin{equation*}
\begin{split}
0=&\: \int_{\mathbb R^{2+1}} \vartheta (X(\phi^{(1)}_n-\phi^{(1)}_0)) g_0^{-1}(\ud (\phi^{(2)}_n-\phi^{(2)}_0), \ud (\phi^{(3)}_n-\phi^{(3)}_0)) \,\mathrm{dVol}_{g_0} \\
=&\: \int_{\mathbb R^{2+1}} \vartheta (X\phi^{(1)}_n) g_0^{-1}(\ud \phi^{(2)}_n, \ud \phi^{(3)}_n) \,\mathrm{dVol}_{g_0} - \int_{\mathbb R^{2+1}} \vartheta (X\phi^{(1)}_0) g_0^{-1}(\ud \phi^{(2)}_0, \ud \phi^{(3)}_0) \,\mathrm{dVol}_{g_0} \\
&\: + \underbrace{\int_{\mathbb R^{2+1}} \vartheta (X\phi^{(1)}_n) g_0^{-1}(\ud \phi^{(2)}_0, \ud \phi^{(3)}_0) \,\mathrm{dVol}_{g_0}}_{=:\mathrm{I}} + \underbrace{\int_{\mathbb R^{2+1}} \vartheta (X\phi^{(1)}_0) g_0^{-1}(\ud \phi^{(2)}_n, \ud \phi^{(3)}_0) \,\mathrm{dVol}_{g_0}}_{=:\mathrm{II}} \\
&\: + \underbrace{\int_{\mathbb R^{2+1}} \vartheta (X\phi^{(1)}_0) g_0^{-1}(\ud \phi^{(2)}_0,\ud \phi^{(3)}_n) \,\mathrm{dVol}_{g_0}}_{=:\mathrm{III}} \underbrace{-\int_{\mathbb R^{2+1}} \vartheta (X\phi^{(1)}_0) g_0^{-1}(\ud \phi^{(2)}_n,\ud \phi^{(3)}_n) \,\mathrm{dVol}_{g_0}}_{=:\mathrm{IV}} \\
&\: - \int_{\mathbb R^{2+1}} \vartheta (X\phi^{(1)}_n) g_0^{-1}(\ud \phi^{(2)}_0, \ud \phi^{(3)}_n) \,\mathrm{dVol}_{g_0} - \int_{\mathbb R^{2+1}} \vartheta (X\phi^{(1)}_n) g_0^{-1}(\ud \phi^{(2)}_n, \ud \phi^{(3)}_0) \,\mathrm{dVol}_{g_0}.
\end{split}
\end{equation*}

Note that each of $\mathrm{I}$, $\mathrm{II}$ and $\mathrm{III}$ has at most one factor depending on $n$. Since our assumptions easily imply that $\rd \phi^{(i)}_n$ converges weakly to $\rd \phi^{(i)}_0$ in $L^3$ (for each $i$), we have
$$\mathrm{I}+ \mathrm{II} + \mathrm{III} \to 3\int_{\mathbb R^{2+1}} \vartheta (X\phi^{(1)}_0) g_0^{-1}(\ud \phi^{(2)}_0, \ud \phi^{(3)}_0) \,\mathrm{dVol}_{g_0}.$$

Next, we apply Lemma~\ref{lem:limitwave} with $p_0 = 3$. Noting that since $L^3(K)\subset L^2(K)$ and $L^3(K)\subset L^{\f 32}(K)$ (for any compact set $K$), by Lemma~\ref{lem:limitwave}, $g_0^{-1}(\ud \phi^{(2)}_n, \ud \phi^{(3)}_n)$ converges to $g_0^{-1}(\ud \phi^{(2)}_0, \ud \phi^{(3)}_0)$ in the sense of distributions. Hence,
\begin{equation*}
\begin{split}
\mathrm{IV}
\to &\: -\int_{\mathbb R^{2+1}} \vartheta (X\phi^{(1)}_0) g_0^{-1}(\ud \phi^{(2)}_0, \ud \phi^{(3)}_0) \,\mathrm{dVol}_{g_0}.
\end{split}
\end{equation*}

Finally, rearranging yields the conclusion. \qedhere
\end{proof}

\subsection{Computation of the remaining terms using trilinear compensated compactness}\label{sec:limit.nonlinear.terms}

We now look at the contributions in $F^\psi_n$ and $F^\om_n$ which are nonlinear in the derivatives of $\psi_n$ and $\om_n$. There are four relevant terms. For these terms, we need the trilinear compensated compactness in Section~\ref{subsec:cc}.
\begin{proposition}\label{prop:limit.nonlinear.terms}
The following holds after passing to a subsequence (which we do not relabel):
\begin{equation*}
\begin{split}
&\: \f 12 \int_{\mathbb R^{2+1}} A(\f{(e_0)_n(\chi\psi_n)}{N_n}) \chi e^{-4\psi_n} g^{-1}_n (\ud \om_n, \ud \om_n)  \,\mathrm{dVol}_{g_n} \\
&\:- \int_{\mathbb R^{2+1}} e^{-4\psi_0} A(\f{(e_0)_n(\chi\om_n)}{N_n}) \chi g^{-1}_n (\ud \om_n, \ud \psi_n) \,\mathrm{dVol}_{g_n}\\
&\: + \f 12 \int_{\mathbb R^{2+1}} \f{(e_0)_n (\chi\psi_n)}{N_n} A[\sqrt{-\det g_n} \chi e^{-4\psi_n} g^{-1}_n (\ud \om_n, \ud \om_n)]  \,\ud x \\
&\: - \int_{\mathbb R^{2+1}} e^{-4\psi_0} \f{(e_0)_n(\chi\om_n)}{N_n} A[\sqrt{-\det g_n}\chi g^{-1}_n (\ud \om_n, \ud \psi_n)] \,\ud x\\
\to &\: \mbox{corresponding terms on the RHS of \eqref{eq:energy.id.0.2} and \eqref{eq:energy.id.0.4}}\\
&\: -2\int_{S^*\mathbb R^{2+1}} e^{-4\psi_0}  \f{(g^{-1}_0)^{\alp\bt} (\rd_\bt \psi_0)}{N_0} (\xi_t - \bt^k_0 \xi_k)\xi_\alp a  \, \f{\ud \nu^{\om}}{|\xi|^2}.
\end{split}
\end{equation*}
\end{proposition}

\begin{proof}

We will compute the limit of each term. Since the computation is largely similar, we will give the details for the first term (Step~1) and 
only give the results for the remaining terms (Step~2).

\pfstep{Step~1: Detailed computation for the first term}
{Notice} that on the support of $a$, $\chi \equiv 1$. In particular, by Lemma~\ref{lem:PSIDOs}, $A(1-\chi)(1+\chi)$ and $(1-\chi)(1+\chi) A$ are both pseudo-differential operators of order $-1$ and hence compact on $L^2$. We write $1 = (1-\chi)(1+\chi) + \chi^2$ and compute each contribution.

\begin{equation}\label{eq:cutoffterm.main}
\begin{split}
&\: \f 12 \int_{\mathbb R^{2+1}} A(\f{(e_0)_n(\chi\psi_n)}{N_n}) \chi e^{-4\psi_n} g^{-1}_n(\ud \om_n, \ud \om_n)  \,\mathrm{dVol}_{g_n}\\
= &\: \underbrace{\f 12 \int_{\mathbb R^{2+1}} A(\f{(e_0)_n(\chi\psi_n)}{N_n}) (1-\chi)(1+\chi)\chi e^{-4\psi_n} g^{-1}_n(\ud \om_n, \ud \om_n)  \,\mathrm{dVol}_{g_n} }_{=:\mathrm{I}}\\
&\: + \underbrace{\f 12 \int_{\mathbb R^{2+1}} A(\f{(e_0)_n(\chi\psi_n)}{N_n}) \chi^3 e^{-4\psi_n} g^{-1}_n(\ud \om_n, \ud \om_n)  \,\mathrm{dVol}_{g_n}}_{=:\mathrm{II}}.
\end{split}
\end{equation}

To handle $\mathrm{I}$, we use the following two facts:
\begin{itemize}
\item $(1-\chi)(1+\chi) A(\f{(e_0)_n(\chi\psi_n)}{N_n})$ converges in the $L^2$ \underline{norm} to $(1-\chi)(1+\chi) A(\f{(e_0)_0(\chi\psi_0)}{N_0})$ after passing to a subsequence (by Lemmas~\ref{lem:PSIDOs}.1 and \ref{lem:PSIDOs}.5).
\item By the pointwise convergence in \eqref{assumption.0}, the bound in \eqref{assumption.1} and Lemma~\ref{lem:limitwave}, $\chi e^{-4\psi_n} g^{-1}_n(\ud \om_n, \ud \om_n)$ converges to $\chi e^{-4\psi_0} g^{-1}_0( \ud \om_0, \ud \om_0)$ in the sense of distributions. Using \eqref{assumption.0} and \eqref{assumption.1} again then implies that the said convergence holds \underline{weakly} in $L^2$.
\end{itemize}
Then, we obtain that, up to a subsequence (which we do not relabel),
\begin{equation}\label{eq:cutoffterm.I}
\begin{split}
\mathrm{I} \to &\: \f 12 \int_{\mathbb R^{2+1}} A(\f{(e_0)_0(\chi\psi_0)}{N_0}) (1-\chi)(1+\chi)\chi e^{-4\psi_0} g^{-1}_0 (\ud \om_0, \ud \om_0)  \,\mathrm{dVol}_{g_0}.
\end{split}
\end{equation}

For $\mathrm{II}$, we further compute
\begin{equation}\label{eq:cutoffterm.II.1}
\begin{split}
\mathrm{II} = &\: \underbrace{\f 12 \int_{\mathbb R^{2+1}} A(\f{(e_0)_n(\chi\psi_n)}{N_n}) \chi e^{-4\psi_n} g^{-1}_n (\ud (\chi\om_n), \ud (\chi\om_n))  \,\mathrm{dVol}_{g_n} }_{=:\mathrm{II_a}}\\
&\: + \underbrace{\f 12 \int_{\mathbb R^{2+1}} A(\f{(e_0)_n(\chi\psi_n)}{N_n}) \chi e^{-4\psi_n} \om_n[ g^{-1}_n (\ud \chi, \ud (\chi\om_n)) + \chi g^{-1}_n (\ud \chi, \ud \om_n)]  \,\mathrm{dVol}_{g_n} }_{=:\mathrm{II_b}} \\
\end{split}
\end{equation}

First, using the fact that $g_n$ and $\psi_n$ converges in $C^0$ to their limits (see \eqref{assumption.0}), $\mathrm{II}_a$ has the same limit as
\begin{equation}\label{eq:cutoffterm.II.2}
\begin{split}
\mathrm{II}_a':= \f 12 \int_{\mathbb R^{2+1}} A(\f{(e_0)_0(\chi\psi_n)}{N_0}) \chi e^{-4\psi_0} g^{-1}_0(\ud (\chi\om_n), \ud(\chi\om_n))  \,\mathrm{dVol}_{g_0}.
\end{split}
\end{equation}
We now apply the result on trilinear compensated compactness (Proposition~\ref{prop:trilinear}). 
\begin{equation}\label{eq:cutoffterm.II.3}
\begin{split} 
\mathrm{II}_a' =&\: \underbrace{ \f 12 \int_{\mathbb R^{2+1}} \f{(e_0)_0(A(\chi\psi_n))}{N_0} \chi e^{-4\psi_0} g^{-1}_0(\ud (\chi\om_n) ,\ud (\chi\om_n))  \,\mathrm{dVol}_{g_0} }_{=:\mathrm{II}'_{a,1}}\\
&\: + \underbrace{\f 12 \int_{\mathbb R^{2+1}} [A(\f{(e_0)_0(\chi\psi_n)}{N_0}) -\f{(e_0)_0(A(\chi\psi_n))}{N_0}] \chi e^{-4\psi_0} g^{-1}_0 (\ud (\chi\om_n), \ud (\chi\om_n))  \,\mathrm{dVol}_{g_0} }_{=:\mathrm{II}'_{a,2}}.
\end{split}
\end{equation}

Note that by Lemmas~\ref{lem:PSIDOs}.2 and \ref{lem:PSIDOs}.4, $A$, $[\rd_\alp, A]$ are both bounded $:L^3 \to L^3$ and $[\Box_{g_0}, A]$ is bounded $: W^{1,3} \to L^3$. Therefore, $\phi^{(1)}_n = A(\chi\psi_n)$, $\phi^{(2)}_n = \phi^{(3)}_n = \chi\om_n$ satisfy the estimates of Proposition~\ref{prop:trilinear}. Hence, by Proposition~\ref{prop:trilinear}, {Corollary~\ref{cor:nu}} and the fact that $\chi\equiv 1$ on the support of $a$,
\begin{equation}\label{eq:cutoffterm.II.4}
\begin{split}
\mathrm{II}'_{a,1} \to &\: \f 12 \int_{\mathbb R^{2+1}} \f{(e_0)_0(A(\chi\psi_0))}{N_0} \chi e^{-4\psi_0} g^{-1}_0 (\ud(\chi\om_0), \ud(\chi\om_0))  \,\mathrm{dVol}_{g_0} \\
&\: + \int_{S^*\mathbb R^{2+1}} e^{-4\psi_0} \f{(g^{-1}_0)^{\alp\bt} (\rd_\bt \om_0)}{N_0} a \,((\ud \sigma^{\mathrm{cross}})^*_{\alp t} - \bt^k_0 (\ud \sigma^{\mathrm{cross}})^*_{\alp k}).
\end{split}
\end{equation}

For $\mathrm{II_{a,2}}$, we note the following:
\begin{itemize}
\item By Lemma~\ref{lem:PSIDOs}, $[A, \f 1{N_0}]:L^2 \to L^2_{\mathrm{loc}}$ and $[A, (e_0)_0]:H^1\to L^2_{\mathrm{loc}}$ are compact so that (after passing to a subsequence) $[A(\f{(e_0)_0(\chi\psi_n)}{N_0}) -\f{(e_0)_0(A(\chi\psi_n))}{N_0}]$ converges in the $L^2$ \underline{norm} to $[A(\f{(e_0)_0(\chi\psi_0)}{N_0}) -\f{(e_0)_0(A(\chi\psi_0))}{N_0}]$.
\item By Lemma~\ref{lem:limitwave}, \eqref{assumption.0} and \eqref{assumption.1}, $g^{-1}_0(\ud (\chi\om_n), \ud (\chi\om_n))$ converges \underline{weakly} in $L^2$ to $g^{-1}_0(\ud (\chi\om_0),\ud (\chi\om_0))$.
\end{itemize}
It follows that 
\begin{equation}\label{eq:cutoffterm.II.5}
\mathrm{II}'_{a,2} \to \f 12 \int_{\mathbb R^{2+1}} [A(\f{(e_0)_0(\chi\psi_0)}{N_0}) -\f{(e_0)_0(A(\chi\psi_0))}{N_0}] \chi e^{-4\psi_0} g^{-1}_0(\ud (\chi\om_0), \ud (\chi\om_0))  \,\mathrm{dVol}_{g_0}.
\end{equation}

We now return to the term $\mathrm{II}_b$ in \eqref{eq:cutoffterm.II.1}. Notice now that $\rd\chi$ and $a$ have disjoint support. Therefore by Lemmas~\ref{lem:PSIDOs}.1 and \ref{lem:PSIDOs}.5, $\rd \chi A: L^2 \to L^2$ is compact. As a result, using also \eqref{assumption.0} and \eqref{assumption.1}, we obtain
\begin{equation}\label{eq:cutoffterm.II.6}
\begin{split}
\mathrm{II}_{b} \to &\: \f 12 \int_{\mathbb R^{2+1}} A(\f{(e_0)_0(\chi\psi_0)}{N_0}) \chi e^{-4\psi_0} \om_0[ g^{-1}_0 (\ud \chi, \ud (\chi\om_0)) + \chi g^{-1}_0 (\ud \chi, \ud \om_0)]  \,\mathrm{dVol}_{g_0}.
\end{split}
\end{equation}

Combining \eqref{eq:cutoffterm.II.1}--\eqref{eq:cutoffterm.II.6}, we obtain
\begin{equation}\label{eq:cutoffterm.II}
\begin{split}
\mathrm{II} \to  &\: \f 12 \int_{\mathbb R^{2+1}} A(\f{(e_0)_0(\chi\psi_0)}{N_0}) \chi^3 e^{-4\psi_0} g^{-1}_0 (\ud \om_0, \ud \om_0)  \,\mathrm{dVol}_{g_0} \\
&\: + \int_{S^*\mathbb R^{2+1}} e^{-4\psi_0} \f{(g^{-1}_0)^{\alp\bt} (\rd_\bt \om_0)}{N_0} a \,((\ud \sigma^{\mathrm{cross}})^*_{\alp t} - \bt^k_0 (\ud \sigma^{\mathrm{cross}})^*_{\alp k}).
\end{split}
\end{equation}

Combining \eqref{eq:cutoffterm.main}, \eqref{eq:cutoffterm.I} and \eqref{eq:cutoffterm.II}, we obtain
\begin{equation}\label{eq:cutoff.1}
\begin{split}
&\: \f 12 \int_{\mathbb R^{2+1}} A(\f{(e_0)_n(\chi\psi_n)}{N_n}) \chi e^{-4\psi_n} g^{-1}_n (\ud \om_n, \ud \om_n)  \,\mathrm{dVol}_{g_n}\\
\to &\: \f 12 \int_{\mathbb R^{2+1}} A(\f{(e_0)_0 (\chi\psi_0)}{N_0}) \chi e^{-4\psi_0} g^{-1}_0 (\ud \om_0, \ud \om_0)  \,\mathrm{dVol}_{g_0} \\
&\: + \int_{S^*\mathbb R^{2+1}} e^{-4\psi_0} \f{(g^{-1}_0)^{\alp\bt} (\rd_\bt \om_0)}{N_0} a \,((\ud \sigma^{\mathrm{cross}})^*_{\alp t} - \bt^k_0 (\ud \sigma^{\mathrm{cross}})^*_{\alp k}).
\end{split}
\end{equation}

\pfstep{Step~2: Computing the second to fourth terms} Arguing as in the derivation of \eqref{eq:cutoff.1} in Step~1, we obtain
\begin{equation}\label{eq:cutoff.2}
\begin{split}
&\: - \int_{\mathbb R^{2+1}} e^{-4\psi_0} A(\f{(e_0)_n(\chi\om_n)}{N_n}) \chi g^{-1}_n (\ud \om_n, \ud \psi_n) \,\mathrm{dVol}_{g_n} \\
\to &\: - \int_{\mathbb R^{2+1}} e^{-4\psi_0} A(\f{(e_0)_0(\chi\om_0)}{N_0}) \chi g^{-1}_0 (\ud \om_0, \ud \psi_0) \,\mathrm{dVol}_{g_0} \\
&\: -\int_{S^*\mathbb R^{2+1}} e^{-4\psi_0}  \f{(g^{-1}_0)^{\alp\bt} (\rd_\bt \psi_0)}{N_0} (\xi_t - \bt^k_0 \xi_k)\xi_\alp a  \, \f{\ud \nu^{\om}}{|\xi|^2} \\
&\: -\int_{S^*\mathbb R^{2+1}} e^{-4\psi_0}  \f{(g^{-1}_0)^{\alp\bt} (\rd_\alp \om_0)}{N_0}  a \, (\ud \sigma^{\mathrm{cross}}_{\bt t} - \bt^k_0 \ud \sigma^{\mathrm{cross}}_{\bt k}).
\end{split}
\end{equation}

Similarly, we also obtain
\begin{equation}\label{eq:cutoff.3}
\begin{split}
&\: \f 12 \int_{\mathbb R^{2+1}} \f{(e_0)_n (\chi\psi_n)}{N_n} A[\sqrt{-\det g_n} \chi e^{-4\psi_n} g^{-1}_n (\ud \om_n, \ud \om_n)]  \,\ud x\\
\to &\: \f 12 \int_{\mathbb R^{2+1}} \f{(e_0)_0 (\chi\psi_0)}{N_0} A[\sqrt{-\det g_0} \chi e^{-4\psi_0} g^{-1}_0 (\ud \om_0, \ud \om_0)]  \, \ud x \\
&\: + \int_{S^*\mathbb R^{2+1}} e^{-4\psi_0}  \f{(g^{-1}_0)^{\alp\bt} (\rd_\alp \om_0)}{N_0}  a \, (\ud {\sigma}^{\mathrm{cross}}_{t\bt}- \bt^k_0 \ud {\sigma}^{\mathrm{cross}}_{k\bt}),
\end{split}
\end{equation}
and
\begin{equation}\label{eq:cutoff.4}
\begin{split}
&\: - \int_{\mathbb R^{2+1}} e^{-4\psi_0} \f{(e_0)_n(\chi\om_n)}{N_n} A[\sqrt{-\det g_n}\chi g^{-1}_n (\ud \om_n, \ud \psi_n)] \,\ud x\\
\to &\: - \int_{\mathbb R^{2+1}} e^{-4\psi_0} \f{(e_0)_0(\chi\om_0)}{N_0} A[\sqrt{-\det g_n}\chi g^{-1}_0 (\ud \om_0, \ud \psi_0)] \,\ud x \\
&\: -\int_{S^*\mathbb R^{2+1}} e^{-4\psi_0} \f{(g^{-1}_0)^{\alp\bt} (\rd_\bt \psi_0)}{N_0} (\xi_t - \bt_0^k \xi_k)\xi_\alp a \, \f{\ud \nu^{\om}}{|\xi|^2} \\
&\: -\int_{S^*\mathbb R^{2+1}} e^{-4\psi_0} \f{(g^{-1}_0)^{\alp\bt} (\rd_\alp \om_0)}{N_0} a \, \,((\ud \sigma^{\mathrm{cross}})^*_{t\bt} - \bt^k_0 (\ud \sigma^{\mathrm{cross}})^*_{k\bt}).
\end{split}
\end{equation}

\pfstep{Step~3: Putting everything together} Adding \eqref{eq:cutoff.1}, \eqref{eq:cutoff.2}, \eqref{eq:cutoff.3} and \eqref{eq:cutoff.4}, and noticing a cancellation using Proposition~\ref{prop:dsigma.is.symmetric}, we finish the proof. \qedhere

\end{proof}

\subsection{Putting everything together} 
We now combining the results in Section~\ref{sec:elliptic.wave.tri} and this section. More precisely, subtracting the expression in Proposition~\ref{prop:limit.nonlinear.terms} from the sum of the expressions for $\psi$ and $\om$ in Proposition~\ref{prop:Box.commute}, we obtain the following:
\begin{proposition}\label{prop:RHS.other}
Let $\ud\nu$ be defined as in \eqref{def:dnu}. Suppose $A = b(x) \widetilde{m}(\f 1i\nabla)$, {where the principal symbol $a(x,\xi) = b(x) m(\xi)$ (with $m(\xi) = \widetilde{m}(\xi)$ for $|\xi|\geq 1$) is real and supported in $T^*\Omega$, $m(\xi)$ is homogeneous of order $0$, and $m$ and $\widetilde{m}$ are both even.} Then, after passing to a subsequence (which we do not relabel),
\begin{equation*}
\begin{split}
&\:\mbox{(RHS of \eqref{eq:energy.id.n.2}) $+$ (RHS of \eqref{eq:energy.id.n.4})} \\
&\: - \int_{\mathbb R^{2+1}} [\rd_i (\chi\psi_n)] \de^{ij} N_n \{[(e_0)_n,A](\f{\rd_j(\chi\psi_n)}{N_n})\} \,\ud x - \f 14\int_{\mathbb R^{2+1}} e^{-4\psi_0}[\rd_i (\chi\om_n)] \de^{ij} N_n \{[(e_0)_n,A](\f{\rd_j(\chi\om_n)}{N_n})\} \,\ud x \\
\to &\: \mbox{(RHS of \eqref{eq:energy.id.0.2}) $+$ (RHS of \eqref{eq:energy.id.0.4})} \\
&\: - \int_{\mathbb R^{2+1}} [\rd_i (\chi\psi_0)] \de^{ij} N_0 \{[(e_0)_0,A](\f{\rd_j(\chi\psi_0)}{N_0})\} \,\ud x - \f 14 \int_{\mathbb R^{2+1}} e^{-4\psi_0} [\rd_i (\chi\om_0)] \de^{ij} N_0 \{[(e_0)_0,A](\f{\rd_j(\chi\om_0)}{N_0})\} \,\ud x \\
&\: - \f 12\int_{S^*\mathbb R^{2+1}} \de^{ij} \xi_i\xi_j (\rd_{x^t}a - \bt^k_0 \rd_{x^k}a) \f{e^{-2\gamma_0}}{N_0}\f{\ud \nu}{|\xi|^2}\\
&\: + \f 12\int_{S^*\mathbb R^{2+1}} \f 1{N_0} (\xi_t-\bt^k_0\xi_k) [(g_0^{-1})^{\mu\nu} \xi_\mu (\rd_{x^\nu} a) - \rd_\mu (g_0^{-1})^{\alp\bt} \xi_\alp \xi_\bt (\rd_{\xi_\mu} a)] \, \f{\ud \nu}{|\xi|^2}\\
&\: +2\int_{S^*\mathbb R^{2+1}} e^{-4\psi_0}  \f{(g^{-1}_0)^{\alp\bt} (\rd_\bt \psi_0)}{N_0} (\xi_t - \bt^k_0 \xi_k)\xi_\alp a  \, \f{\ud \nu^{\om}}{|\xi|^2}.
\end{split}
\end{equation*}
where $X = \f 1{N_0} (\rd_t -\bt^i_0 \rd_i)$. 
\end{proposition}

\section{Transport equation for the microlocal defect measure and conclusion of the proof of Theorem~\ref{thm:main}}\label{sec:final}

Our goal in this section is to combine Propositions~\ref{prop:RHS.1}, \ref{prop:RHS.3} and \ref{prop:RHS.other} to prove that the measure $\ud\nu$ indeed satisfies a transport equation as in \eqref{eq:transport.def}. This will allow us to conclude the proof of Theorem~\ref{thm:main}.

\begin{proposition}\label{prop:main.transport.prelim}
Let $\ud\nu$ be defined as in \eqref{def:dnu}.

Suppose $a: T^*\mathbb R^{2+1}\to \mathbb R$ be a smooth function which is homogeneous of order $0$ in $\xi$ and is supported in $T^*\Omega$. Then
$$\int_{S^*\mathbb R^{2+1}} (2 (g_0^{-1})^{\alp\bt} \xi_\alp \rd_{x^\bt} (\f{(\xi_t-\bt^k_0 \xi_k) a}{N_0}) - (\rd_{\mu}(g_0^{-1})^{\alp\bt})\xi_\alp \xi_\bt \rd_{\xi_\mu} (\f{(\xi_t-\bt^k_0 \xi_k) a}{N_0})) \,\f{\ud\nu}{|\xi|^2} = 0.$$
\end{proposition}
\begin{proof}
By Proposition~\ref{prop:main.reduction}, it suffices to consider the case where $a(x,\xi) = b(x) m(\xi)$, where $m$ is homogeneous of order $0$ and \underline{even}. We make this assumption for the remainder of the proof (so that we can apply results in earlier sections).

Note that $\mbox{RHS of \eqref{eq:energy.id.n.1}} = \mbox{RHS of \eqref{eq:energy.id.n.2}}$, $\mbox{RHS of \eqref{eq:energy.id.n.3}} = \mbox{RHS of \eqref{eq:energy.id.n.4}}$, $\mbox{RHS of \eqref{eq:energy.id.0.1}} = \mbox{RHS of \eqref{eq:energy.id.0.2}}$ and $\mbox{RHS of \eqref{eq:energy.id.0.3}} = \mbox{RHS of \eqref{eq:energy.id.0.4}}$ {(because the LHSs all agree)}. Therefore, combining Propositions~\ref{prop:RHS.1}, \ref{prop:RHS.3} and \ref{prop:RHS.other}, we obtain
\begin{equation}\label{main.transport.identity.0}
\begin{split}
0= &\: -\int_{S^*\mathbb R^{2+1}} ((g^{-1}_0)^{\alp\bt} (\rd_\bt X^\gamma) \xi_\alp \xi_\gamma -\f 12 X^\mu \rd_\mu (g^{-1}_0)^{\alp\gamma}\xi_\alp \xi_\gamma) a \,\f{d\nu}{|\xi|^2} \\
&\: + \f 12\int_{S^*\mathbb R^{2+1}} [-\de^{ij} \xi_i (\xi_t-\bt^k_0\xi_k) \rd_{x^j} a ]\, \f{e^{-2\gamma_0}}{N_0}\f{d\nu}{|\xi|^2}\\
&\: + \f 12\int_{S^*\mathbb R^{2+1}} \de^{ij} \xi_i\xi_j (\rd_{x^t}a - \bt^k_0 \rd_{x^k}a) \f{e^{-2\gamma_0}}{N_0}\f{\ud \nu}{|\xi|^2}\\
&\: - \f 12\int_{S^*\mathbb R^{2+1}} \f 1{N_0} (\xi_t-\bt^k_0\xi_k) [(g_0^{-1})^{\alp\bt} \xi_\alp (\rd_{x^\beta} a) - \rd_\mu (g_0^{-1})^{\alp\bt} \xi_\alp \xi_\bt (\rd_{\xi_\mu} a)] \, \f{\ud \nu}{|\xi|^2},
\end{split}
\end{equation}
where $X = \f 1{N_0} (\rd_t -\bt^i_0 \rd_i)$ as before. (Note that the two terms of $2\int_{S^*\mathbb R^{2+1}} \f{e^{-4\psi_0}}{N_0} (g_0^{-1})^{\alp\bt} (\rd_\alp \psi_0) \xi_\bt (\xi_t-\bt^k_0 \xi_k) a\, \f{\ud \nu^\om}{|\xi|^2}$ cancel.)

Since $\f 1{N_0^2} (\xi_t -\bt^k_0 \xi_k)^2 = e^{-2\gamma_0} \de^{ij} \xi_i \xi_j$ on the support of $\ud \nu$ (by Proposition~\ref{prop:psiom.localized}),
\begin{equation}\label{complicated.transport.terms}
\begin{split}
 &\: \f 12\int_{S^*\mathbb R^{2+1}} [-\de^{ij} \xi_i (\xi_t-\bt^k_0\xi_k) \rd_{x^j} a ]\, \f{e^{-2\gamma_0}}{N_0}\f{d\nu}{|\xi|^2} + \f 12\int_{S^*\mathbb R^{2+1}} \de^{ij} \xi_i\xi_j (\rd_{x^t}a - \bt^k_0 \rd_{x^k}a) \f{e^{-2\gamma_0}}{N_0}\f{\ud \nu}{|\xi|^2} \\
= &\: \f 12\int_{S^*\mathbb R^{2+1}} [-\de^{ij} \xi_i (\xi_t-\bt^k_0\xi_k) \rd_{x^j} a ]\, \f{e^{-2\gamma_0}}{N_0}\f{d\nu}{|\xi|^2} +\f 12 \int_{S^*\mathbb R^{2+1}} (\xi_t - \bt^i_0 \xi_i)^2 (\rd_{x^t}a - \bt^k_0 \rd_{x^k}a) \f{1}{N_0^3}\f{\ud \nu}{|\xi|^2}\\
= &\: \f 12\int_{S^*\mathbb R^{2+1}} (\xi_t-\bt^k_0\xi_k) [- e^{-2\gamma_0}\de^{ij} \xi_i  \rd_{x^j} a + \f{(\xi_t - \bt^i_0 \xi_i)}{N_0^2} (\rd_{x^t}a - \bt^k_0 \rd_{x^k}a)]\, \f{1}{N_0}\f{d\nu}{|\xi|^2}.
\end{split}
\end{equation}
By \eqref{eq:spatial.transport.a}, it then follows that
\begin{equation}\label{complicated.transport.terms.simplified}
\mbox{\eqref{complicated.transport.terms}}= - \f 12\int_{S^*\mathbb R^{2+1}} (\xi_t-\bt^k_0\xi_k) (g_0^{-1})^{\alp\bt} \xi_\alpha (\rd_{x^\beta} a) \, \f{1}{N_0}\f{d\nu}{|\xi|^2}.
\end{equation}

Plugging \eqref{complicated.transport.terms.simplified} into \eqref{main.transport.identity.0}, we then obtain
\begin{equation}\label{main.transport.identity}
\begin{split}
0= &\: -\int_{S^*\mathbb R^{2+1}}\left({-}(g_0^{-1})^{\alpha \mu}\frac{\partial_\mu \beta_0^j}{N_0}\xi_j \xi_\alpha + (\xi_t - \bt^i_0 \xi_i) ((g^{-1}_0)^{\alp\bt} (\rd_\bt \f{1}{N_0}) \xi_\alp 
-\f 12 \f 1{N_0} ((e_0)_0 (g^{-1}_0)^{\alp\gamma}) \xi_\alp {\xi_\gamma}) a \,\right)\f{d\nu}{|\xi|^2} \\
&\: - \int_{S^*\mathbb R^{2+1}} \f 1{N_0} (\xi_t-\bt^k_0\xi_k) [(g_0^{-1})^{\alp\bt} \xi_\alp (\rd_{x^\beta} a) - \f 12 \rd_\mu (g_0^{-1})^{\alp\bt} \xi_\alp \xi_\bt (\rd_{\xi_\mu} a)] \, \f{\ud \nu}{|\xi|^2}\\
= &\: - \int_{S^*\mathbb R^{2+1}}  [(g_0^{-1})^{\alp\bt} \xi_\alp \rd_{x^\beta} (\f{(\xi_t-\bt^k_0\xi_k) a}{N_0}) - \f 12 (\rd_\mu (g_0^{-1})^{\alp\bt}) \xi_\alp \xi_\bt \rd_{\xi_\mu} (\f{(\xi_t-\bt^k_0\xi_k) a}{N_0})] \, \f{\ud \nu}{|\xi|^2},
\end{split}
\end{equation}
as desired. \qedhere
\end{proof}

\begin{proposition}\label{prop:main.transport}
Let $\ud\nu$ be defined as in \eqref{def:dnu}.

Suppose $\wht{a}: T^*\mathbb R^{2+1}\to \mathbb R$ be a smooth function which is homogeneous of order $1$ in $\xi$ and is supported in $T^*\Omega$. Then
$$\int_{S^*\mathbb R^{2+1}} (2 (g_0^{-1})^{\alp\bt} \xi_\alp \rd_{x^\bt} \wht a - (\rd_{\mu}(g_0^{-1})^{\alp\bt})\xi_\alp \xi_\bt \rd_{\xi_\mu} \wht a) \,\f{\ud\nu}{|\xi|^2} = 0.$$
\end{proposition}
\begin{proof}
Suppose $\wht a$ is homogeneous {of} order $+1$ {in $\xi$} with support in $S^*\Omega$. Since $\ud \nu$ is supported on $\{(x,\xi): g_0^{-1}(\xi,\xi)=0 \}$, $\xi_t-\bt^k_0\xi_k\neq 0$ on the support of $\ud\nu$. It follows that we can define $a$ to be {homogeneous of} order $0$ {in $\xi$} supported in $S^*\Omega$ so that $\f{(\xi_t-\bt^k_0\xi_k) a}{N_0} \equiv \wht{a}$ in a neighborhood of the support of $\ud \nu$. Applying Proposition~\ref{prop:main.transport.prelim} to this $a$ then yields
$$\int_{S^*\mathbb R^{2+1}} (2 (g_0^{-1})^{\alp\bt} \xi_\alp \rd_{x^\bt} \wht a - (\rd_{\mu}(g_0^{-1})^{\alp\bt})\xi_\alp \xi_\bt \rd_{\xi_\mu} \wht a) \,\f{\ud\nu}{|\xi|^2},$$
which is what we want to prove. \qedhere
\end{proof}

\begin{proof}[Proof of Theorem~\ref{thm:main}]
In view of Theorem~\ref{thm:prelim} and Proposition~\ref{prop:compact.reduction}, it suffices to prove is that under the additional assumption of Theorem~\ref{thm:main}, the transport equation \eqref{eq:transport.def} holds in $\Omega$. This is exactly provided by Proposition~\ref{prop:main.transport}. \qedhere
\end{proof}

\end{document}